\documentclass[12pt,oneside,reqno]{amsart}
\usepackage[latin9]{inputenc}
\usepackage{geometry}
\usepackage{mathrsfs}
\usepackage{mathtools}
\usepackage{amstext}
\usepackage{amsthm}
\usepackage{amssymb}
\usepackage{mathdots}
\usepackage{graphicx}
\usepackage[unicode=true,pdfusetitle,
 bookmarks=true,bookmarksnumbered=true,bookmarksopen=true,bookmarksopenlevel=3,
 breaklinks=false,pdfborder={0 0 1},backref=section,colorlinks=false]
 {hyperref}
\hypersetup{
 hypertex,bookmarksnumbered=true}

\makeatletter

\DeclareTextSymbolDefault{\textquotedbl}{T1}
\providecommand{\tabularnewline}{\\}

\numberwithin{equation}{section}
\numberwithin{figure}{section}
\theoremstyle{plain}
\newtheorem{thm}{\protect\theoremname}
\theoremstyle{remark}
\newtheorem{rem}[thm]{\protect\remarkname}
\theoremstyle{definition}
\newtheorem{defn}[thm]{\protect\definitionname}
\theoremstyle{plain}
\newtheorem{prop}[thm]{\protect\propositionname}
\theoremstyle{plain}
\newtheorem{cor}[thm]{\protect\corollaryname}
\theoremstyle{plain}
\newtheorem{lem}[thm]{\protect\lemmaname}

\usepackage{chngcntr}
\counterwithout{figure}{section}
\newtheorem{assumption}{Assumption}

\makeatother

\providecommand{\corollaryname}{Corollary}
\providecommand{\definitionname}{Definition}
\providecommand{\lemmaname}{Lemma}
\providecommand{\propositionname}{Proposition}
\providecommand{\remarkname}{Remark}
\providecommand{\theoremname}{Theorem}

\begin{document}
\title{Analytic theory of coupled-cavity traveling wave tubes}
\author{Alexander Figotin}
\address{University of California at Irvine, CA 92967}
\begin{abstract}
Coupled-cavity traveling wave tube (CCTWT) is a high power microwave
(HPM) vacuum electronic device used to amplify radio-frequency (RF)
signals. CCTWTS have numerous applications, including radar, radio
navigation, space communication, television, radio repeaters, and
charged particle accelerators. The microwave-generating interactions
in CCTWTs take place mostly in coupled resonant cavities positioned
periodically along the electron beam axis. Operational features of
a CCTWT particularly the amplification mechanism are similar to those
of a multicavity klystron (MCK). We advance here a Lagrangian field
theory of CCTWTs with the space being represented by one-dimensional
continuum. The theory integrates into it the space-charge effects
including the so-called debunching (electron-to-electron repulsion).
The corresponding Euler-Lagrange field equations are ODEs with coefficients
varying periodically in the space. Utilizing the system periodicity
we develop the instrumental features of the Floquet theory including
the monodromy matrix and its Floquet multipliers. We use them to derive
closed form expressions for a number of physically significant quantities.
Those include in particular the dispersion relations and the frequency
dependent gain foundational to the RF signal amplification. Serpentine
(folded, corrugated) traveling wave tubes are very similar to CCTWTs
and our theory applies to them also. 
\end{abstract}

\maketitle

\section{Introduction\label{sec:intcctwt}}

We start with the general principles of the microwave radiation generation
and the amplification of RF signals, \cite[4]{Shev}:
\begin{quotation}
``ANY generating or amplifying device converts d.c. energy into high-frequency
electric field energy, and this conversion is effected by means of
an electron beam. All energy exchanges between the electron beam and
the alternating electric field are a result of acceleration or retardation
of the electrons. The kinetic energy of the electrons is converted
into electromagnetic energy, and vice versa. Therefore, although the
mechanisms of various devices are different, in each of them power
is transferred from the constant voltage source to the alternating
electromagnetic field. This is brought about in the oscillatory system
by means of a density-modulated electron beam in which electrons are
accelerated in the constant electric field, and retarded in the alternating
electric field. Density modulation of the electron beam makes it possible
to retard a greater number of electrons than are accelerated by the
same alternating field, thus producing the transfer of energy.''
\end{quotation}
The last sentence in the above quote underlines the critical role
played by the density modulation of the electron beam (known also
as ``electron bunching'') in the energy transfer from the electron
beam to the electromagnetic (EM) radiation.

A coupled-cavity traveling wave tube (CCTWT) shown schematically in
Fig. \ref{fig:cctwt} is the primary subject we pursue here. The CCTWT
is special type of traveling wave tube (TWT) that utilizes coupled-cavity
structure (CCS) as a slow-wave structure (SWS), \cite[15]{Gilm1},
\cite[4]{MAEAD}. The CCS commonly is a periodic linear chain of several
tens of cavities coupled by coupling holes or slots and a beam tunnel,
\cite[8.7.5]{Tsim}. The cavities can be similar to those in klystrons.
As to the physical implementation cavities are often constructed of
sections of a slow-wave structure that are made resonant by suitable
terminations. The quality factor of each cavity is required to be
sufficiently high so that the RF field distribution in each separate
cavity is substantially unaffected by the interaction with the beam,
\cite{ChoWes}.
\begin{figure}[h]
\centering{}\hspace{1.3cm}\includegraphics[bb=0bp 120bp 1280bp 620bp,clip,scale=0.4]{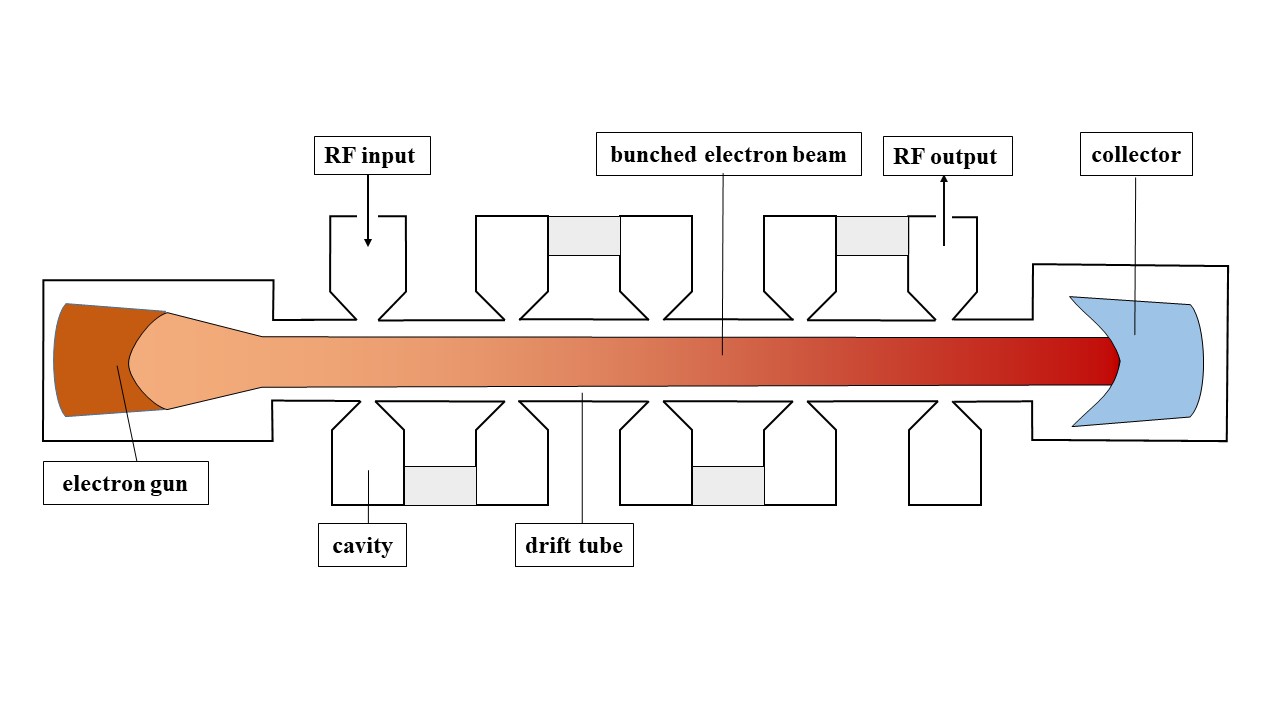}\caption{\label{fig:cctwt} Schematic representation of a coupled-cavity traveling
wave tube (CCTWT) composed of a periodic array of coupled cavities
(often of toroidal shape) interacting with the pencil-like electron
beam. The interaction causes the electron bunching and consequent
amplification of the RF signal.}
\end{figure}

By its very design the CCS is mechanically and thermally more robust
than a helix, which is often used as the SWS, allowing much greater
average power, especially in the short-wave bands of microwave range,
\cite{ChoCra}, \cite{GraParArm}, \cite[8.3, 8.7.5]{Tsim}. Serpentine
(folded, corrugated) TWTs are based on the corresponding waveguides
that permit electron interaction below the velocity of light, \cite[15.1, 15.2]{Gilm1}.
They are very similar to CCTWTs and the results of our studies apply
also to serpentine (folded, corrugated) TWTs.

A distinct and important feature of a CCTWT is that the interaction
between the e-beam and the \emph{coupled-cavity structure} (CCS),
particularly the electron velocity modulation by high-frequency EM
field, is limited mostly to the EM cavities positioned periodically
along the e-beam axis. The cavity properties that are significant
for an effective interaction with the e-beam are as follows, \cite[2]{Shev}: 
\begin{quotation}
``In order to be used in an electron tube, a cavity resonator must
have a region with a relatively strong high-frequency field which
is polarized along the direction of electron flow. This region should,
in the majority of cases, be so small that the electron transit time
is less than the period of change of the field. Hollow toroidal resonators
satisfy these conditions. Toroidal resonators consist of cylinders
with a very prominent \textquotedbl bulge\textquotedbl{} in the middle.''
\end{quotation}
The region in the above quote is commonly referred to as the \emph{cavity
gap} or just \emph{gap}, and it is there the electron velocity is
modulated leading to the electron bunching and consequent RF signal
amplification. If $\mathring{v}$ is the stationary (dc) velocity
of the electron flow and $l_{\mathrm{g}}$ is the length of the cavity
gap then the mentioned condition of smallness of $l_{\mathrm{g}}$
and the electron transit time $\tau_{g}=\frac{l_{\mathrm{g}}}{\mathring{v}}$
can be written as
\begin{equation}
l_{\mathrm{g}}<\frac{2\pi\mathring{v}}{\omega},\label{eq:transL1a}
\end{equation}
It is a common assumption for one-dimensional models for charge-waves
that the e-beam is not dense in the sense that the operational frequency
$\omega$ satisfies, \cite[p. 277]{Tsim}:
\begin{equation}
\omega_{\mathrm{p}}\ll\omega,\label{eq:transL1b}
\end{equation}
where $\omega_{\mathrm{p}}$ is the relevant plasma frequency. Then
combing inequalities (\ref{eq:transL1b}) and (\ref{eq:transL1a})
we obtain the following upper bound on the gap length:
\begin{equation}
l_{\mathrm{g}}<\frac{2\pi\mathring{v}}{\omega}\ll\lambda_{\mathrm{p}}=\frac{2\pi\mathring{v}}{\omega_{\mathrm{p}}}.\label{eq:transL1c}
\end{equation}
\emph{Idealized theories including the one we advance here assume
that the narrow cavity gaps are just of zero width} \emph{corresponding
to zero transit time} \emph{of the electron, }\cite[II.5]{Shev},
\cite[III.3]{Werne}. That is we make the following simplifying assumption:
\begin{equation}
l_{\mathrm{g}}=0,\quad\tau_{g}=\frac{l_{\mathrm{g}}}{\mathring{v}}=0.\label{eq:transL1d}
\end{equation}

The primary subject of our studies here is the construction of one-dimensional
Lagrangian field theory of a coupled-cavity traveling wave tube (CCTWT)
a schematic sketch of which is shown in Fig. \ref{fig:cctwt} (compare
it with Fig. \ref{fig:mck} with a schematic sketch of a \emph{multicavity
klystron (MCK)}). This theory integrates into it (i) our one-dimensional
Lagrangian field theory for TWTs introduced and studied in \cite[4, 24]{FigTWTbk}
and reviewed in Section \ref{sec:twtmod}; (ii) one-dimensional Lagrangian
field theory for multicavity klystron we developed in \cite{FigKly}
in reviewed in Section \ref{sec:mckrev}. The theory takes into account
the space-charge effects, and it applies also to serpentine (folded,
corrugated) traveling wave tubes.

This paper is organized as follows. In Section \ref{sec:twtmod},
we review concisely the one-dimensional Lagrangian field theory for
TWTs introduced and studied in \cite[4, 24]{FigTWTbk}. In Section
\ref{sec:cctwtmod}, we construct the Lagrangian of the CCTWT, derive
the corresponding Euler-Lagrange equations and introduce the CCTWT
constitutive subsystems: coupled cavity structure and the e-beam.
In Section \ref{sec:cctwtsol}, we use the Floquet theory to study
solutions to the Euler-Lagrange equations. In particular we construct
the monodromy matrix. In Section \ref{sec:floqmul}, we analyze the
Floquet multiplies which are solution to the characteristic equations.
In Section \ref{sec:disprel} we construct the dispersion relations
and study their properties. In Section \ref{sec:cctwtgain}, we derive
expressions for the frequency dependent gain associated with the CCTWT
eigenmodes. In Section \ref{sec:mckrev}, we review concisely the
one-dimensional Lagrangian field theory for multicavity klystrons
developed in \cite{FigKly} that allows to see some of its features
in the CCTWT. In Section \ref{sec:ccs}, we study couple-cavity structure
when it is not coupled to the e-beam. That is allows us to see some
of its features in the properties of the CCTWT. In a numbers of appendices
we review for the reader's convenience a number of mathematical and
physical subjects relevant to the analysis of the CCTWT. The kinetic
and field points of view on the gap interaction is considered in Section
\ref{sec:fitkit}. The Lagrangian variational framework of our analytical
theory is developed in Section \ref{sec:lagvar}. In Section \ref{sec:degpol4},
we consider special polynomials of the forth degree and their root
degeneracies that are useful for our studies of the CCTWT exceptional
points of degeneracy. In a number of Appendices we review some mathematical
and physical subject relevant to our studies.

While quoting monographs, we identify the relevant sections as follows.
Reference {[}X,Y{]} refers to Section/Chapter ``Y'' of monograph
(article) ``X'', whereas {[}X, p. Y{]} refers to page ``Y'' of
monograph (article) ``X''. For instance, reference {[}2, VI.3{]}
refers to monograph {[}2{]}, Section VI.3; reference {[}2, p. 131{]}
refers to page 131 of monograph {[}2{]}.

\section{Sketch of the analytic model of the traveling wave tube\label{sec:twtmod}}

When constructing the CCTWT Lagrangian, we use the elements of the
analytic theory of TWTs developed in \cite[4, 24]{FigTWTbk}. The
purpose of this section is to introduce those elements as well as
the relevant variables of the analytic model of TWT. TWT converts
the energy of the electron beam (e-beam) into the EM energy of the
amplified RF signal. A schematic sketch of typical TWT is shown in
Fig. \ref{fig:TWT1}. To facilitate the energy conversion and signal
amplification, the e-beam is enclosed in the so-called \emph{slow
wave structure} (SWS), that supports waves that are slow enough to
effectively interact with the electron flow. As a result of this interaction,
the kinetic energy of electrons is converted into the EM energy stored
in the field, \cite{Gilm1}, \cite{Tsim}, \cite[2.2]{Nusi}, \cite[4]{SchaB}.
Consequently, the \emph{key operational principle of a TWT is a positive
feedback interaction between the slow-wave structure and the flow
of electrons}. The physical mechanism of the radiation generation
and its amplification is the electron bunching caused by the acceleration
and deceleration of electrons along the e-beam (see quotes in Section
\ref{sec:intcctwt}).
\begin{figure}[h]
\centering{}\includegraphics[scale=0.5]{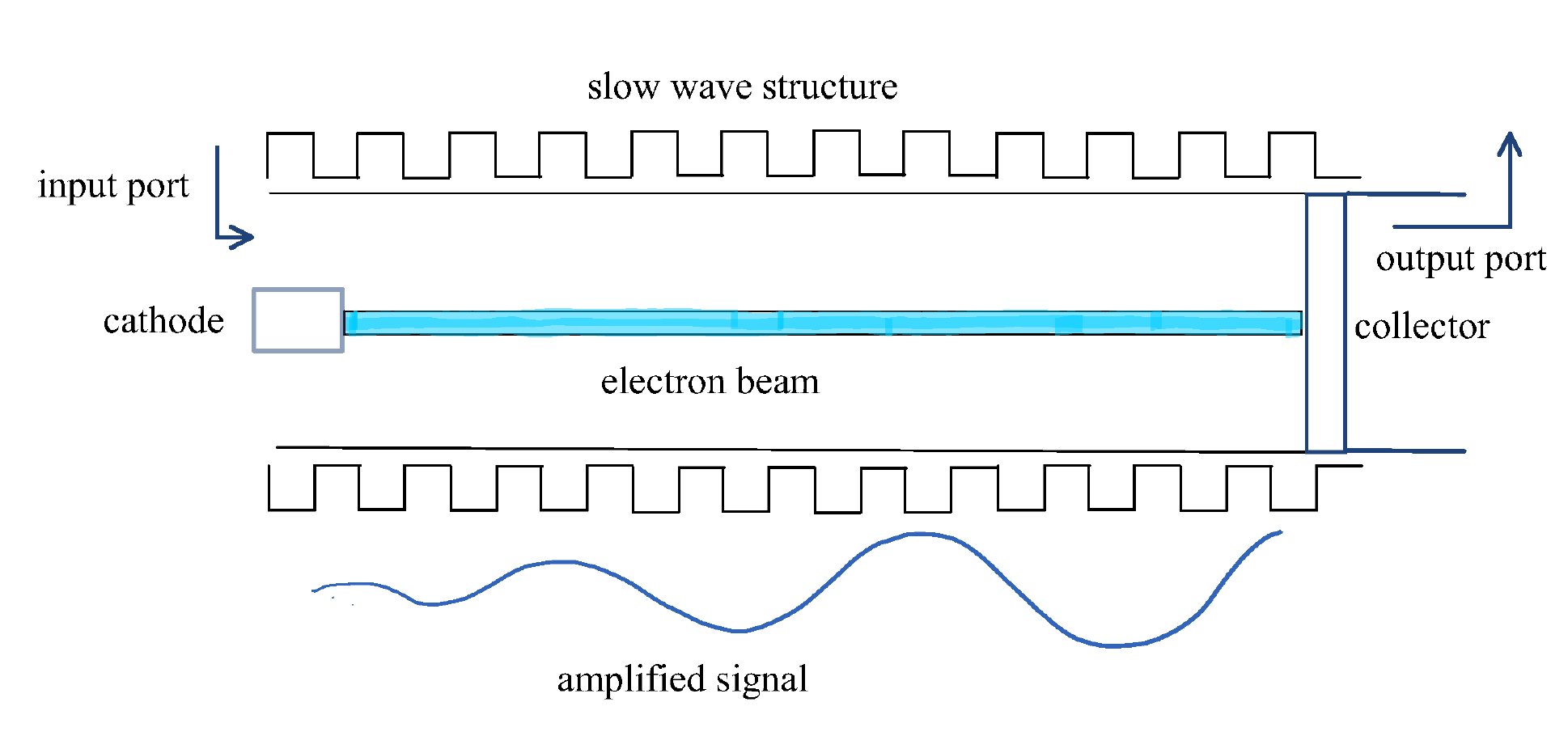}\caption{\label{fig:TWT1} The upper picture is a schematic representation
of a traveling wave tube. The lower picture shows an RF perturbation
in the form of a space-charge wave which is amplified exponentially
as it propagates through the traveling wave tube.}
\end{figure}

A typical TWT consists of a vacuum tube containing the e-beam that
passes down the middle of an SWS such as an RF circuit. It operates
as follows. The left end of the RF circuit is fed with a low-powered
RF signal to be amplified. The SWS electromagnetic field acts upon
the e-beam causing electron bunching and the formation of the so-called
\emph{space-space-charge wave}. In turn, the electromagnetic field
generated by the space-charge wave induces more current back into
the RF circuit with a consequent enhancement of electron bunching.
As a result, the EM field is amplified as the RF signal passes down
the structure until a saturation regime is reached and a large RF
signal is collected at the output. The role of the SWS is to provide
slow-wave modes to match up with the velocity of the electrons in
the e-beam. This velocity is usually a small fraction of the speed
of light. Importantly, synchronism is required for effective in-phase
interaction between the SWS and the e-beam with optimal extraction
of the kinetic energy of the electrons. A typical simple SWS is the
helix, which reduces the speed of propagation according to its pitch.
The TWT is designed so that the RF signal travels along the tube at
nearly the same speed as electrons in the e-beam to facilitate effective
coupling. Technical details on the designs and operation of TWTs can
be found in \cite{Gilm1}, \cite[4]{Nusi} \cite{PierTWT}, \cite{Tsim}.
As for a rich and interesting history of traveling wave tubes, we
refer the reader to \cite{MAEAD} and references therein.

An effective mathematical model for a TWT interacting with the e-beam
was introduced by Pierce \cite[I]{Pier51}, \cite{PierTWT}. The Pierce
model is one-dimensional; it accounts for the wave amplification,
energy extraction from the e-beam and its conversion into microwave
radiation in the TWT \cite{Gilm1}, \cite{Gilm}, \cite[4]{Nusi},
\cite[4]{SchaB}, \cite{Tsim}. This model captures remarkably well
significant features of the wave amplification and the beam-wave energy
transfer, and is still used for basic design estimates. In our paper
\cite{FigRey1}, we have constructed a Lagrangian field theory by
generalizing and extending the Pierce theory to the case of a possibly
inhomogeneous MTL coupled to the e-beam. This work was extended to
an analytic theory of multi-stream electron beams in traveling wave
tubes in \cite{FigTWTbk}. We concisely review here this theory. According
to the simplest version of the theory an ideal TWT is represented
by a single-stream electron beam (e-beam) interacting with a single
transmission line (TL) just as in the Pierce model \cite[I]{Pier51}.
The main parameter describing the single-stream e-beam is e-beam intensity
\begin{equation}
\beta=\frac{\sigma_{\mathrm{B}}}{4\pi}R_{\mathrm{sc}}^{2}\omega_{\mathrm{p}}^{2}=\frac{e^{2}}{m}R_{\mathrm{sc}}^{2}\sigma_{\mathrm{B}}\mathring{n},\quad\omega_{\mathrm{p}}^{2}=\frac{4\pi\mathring{n}e^{2}}{m},\label{eq:T1B1betas1a}
\end{equation}
where $-e$ is the electron charge with $e>0$, $m$ is the electron
mass, $\omega_{\mathrm{p}}$ is the e-beam plasma frequency, $\sigma_{\mathrm{B}}$
is the area of the cross-section of the e-beam, s $\mathring{v}>0$
is stationary velocity of electrons in the e-beam and $\mathring{n}$
is the density of the number of electrons. The constant $R_{\mathrm{sc}}$
is the\emph{ }plasma frequency reduction factor that accounts phenomenologically
for finite dimensions of the e-beam cylinder as well as geometric
features of the slow-wave structure, \cite{BraMih}, \cite[9.2]{Gilm1},
\cite[3.3.3]{Nusi}. The frequency
\begin{equation}
\omega_{\mathrm{rp}}=R_{\mathrm{sc}}\omega_{\mathrm{p}}\label{eq:omred1a}
\end{equation}
is known as reduced plasma frequency, \cite[9.2]{Gilm1}.

Assume the Gaussian system of units of the physical dimensions of
a complete set of the e-beam parameters, as in Tables \ref{tab:e-beam-units}
and \ref{tab:ebeam-par}.
\begin{table}[h]
\centering{}%
\begin{tabular}{|l||l||l|}
\hline 
\noalign{\vskip\doublerulesep}
Frequency & Plasma frequency & $\omega_{\mathrm{p}}=\sqrt{\frac{4\pi\mathring{n}e^{2}}{m}}$\tabularnewline[0.2cm]
\hline 
\hline 
\noalign{\vskip\doublerulesep}
Velocity & e-beam velocity & $\mathring{v}$\tabularnewline[0.2cm]
\hline 
\hline 
\noalign{\vskip\doublerulesep}
Wavenumber &  & $k_{\mathrm{q}}=\frac{\omega_{\mathrm{rp}}}{\mathring{v}}=\frac{R_{\mathrm{sc}}\omega_{\mathrm{p}}}{\mathring{v}}$\tabularnewline[0.2cm]
\hline 
\hline 
\noalign{\vskip\doublerulesep}
Length & Wavelength for $k_{\mathrm{q}}$ & $\lambda_{\mathrm{rp}}=\frac{2\pi\mathring{v}}{\omega_{\mathrm{rp}}},\:\omega_{\mathrm{rp}}=R_{\mathrm{sc}}\omega_{\mathrm{p}}$\tabularnewline[0.2cm]
\hline 
\noalign{\vskip\doublerulesep}
Time & Wave time period & $\mathring{\tau}=\frac{2\pi}{\omega_{\mathrm{p}}}$\tabularnewline[0.2cm]
\hline 
\end{tabular}\vspace{0.3cm}
\caption{\label{tab:e-beam-units} Natural units relevant to the e-beam.}
\end{table}
\begin{table}
\centering{}%
\begin{tabular}{|r||r||r|}
\hline 
\noalign{\vskip\doublerulesep}
$i$ & current & $\frac{\left[\text{charge}\right]}{\left[\text{time}\right]}$\tabularnewline[0.2cm]
\hline 
\noalign{\vskip\doublerulesep}
$q$ & charge & $\left[\text{charge}\right]$\tabularnewline[0.2cm]
\hline 
\noalign{\vskip\doublerulesep}
$\mathring{n}$ & number of electrons p/u of volume & $\frac{\left[\text{1}\right]}{\left[\text{length}\right]^{3}}$\tabularnewline[0.2cm]
\hline 
\noalign{\vskip\doublerulesep}
$\lambda_{\mathrm{rp}}=\frac{2\pi\mathring{v}}{\omega_{\mathrm{rp}}},\:\omega_{\mathrm{rp}}=R_{\mathrm{sc}}\omega_{\mathrm{p}}$ & the electron plasma wavelength & $\left[\text{length}\right]$\tabularnewline[0.2cm]
\hline 
\noalign{\vskip\doublerulesep}
$g_{\mathrm{B}}=\frac{\sigma_{\mathrm{B}}}{4\lambda_{\mathrm{rp}}}$ & the e-beam spatial scale & $\left[\text{length}\right]$\tabularnewline[0.2cm]
\hline 
\noalign{\vskip\doublerulesep}
$\beta=\frac{\sigma_{\mathrm{B}}}{4\pi}R_{\mathrm{sc}}^{2}\omega_{\mathrm{p}}^{2}=\frac{e^{2}}{m}R_{\mathrm{sc}}^{2}\sigma_{\mathrm{B}}\mathring{n}$ & e-beam intensity & $\frac{\left[\text{length}\right]^{2}}{\left[\text{time}\right]^{2}}$\tabularnewline[0.2cm]
\hline 
\noalign{\vskip\doublerulesep}
$\beta^{\prime}=\frac{\beta}{\mathring{v}^{2}}=\frac{\pi\sigma_{\mathrm{B}}}{\lambda_{\mathrm{rp}}^{2}}=\frac{4\pi g_{\mathrm{B}}}{\lambda_{\mathrm{rp}}}$ & dimensionless e-beam intensity & $\left[\text{dim-less}\right]$\tabularnewline[0.2cm]
\hline 
\end{tabular}\vspace{0.3cm}
\caption{\label{tab:ebeam-par}Physical dimensions of the e-beam parameters.
Abbreviations: dimensionless \textendash{} dim-less, p/u \textendash{}
per unit.}
\end{table}

We would like to point to an important spatial scale related to the
e-beam, namely
\begin{equation}
\lambda_{\mathrm{rp}}=\frac{2\pi\mathring{v}}{R_{\mathrm{sc}}\omega_{\mathrm{p}}},\quad\omega_{\mathrm{rp}}=R_{\mathrm{sc}}\omega_{\mathrm{p}},\label{eq:aBvom1a}
\end{equation}
which is the distance passed by an electron for the time period $\frac{2\pi}{\omega_{\mathrm{rp}}}$
associated with the plasma oscillations at the reduced plasma frequency
$\omega_{\mathrm{rp}}$. This scale is well known in the theory of
klystrons and is referred to as \emph{the electron plasma wavelength},
\cite[9.2]{Gilm1}. Another spatial scale related to the e-beam that
arises in our analysis later on is
\begin{equation}
g_{\mathrm{B}}=\frac{\sigma_{\mathrm{B}}}{4\lambda_{\mathrm{rp}}},\label{eq:Bvom1ba}
\end{equation}
and we will refer to it as\emph{ e-beam spatial scale}. Using these
spatial scales we obtain the following representation for the dimensionless
form $\beta^{\prime}$ of the e-beam intensity
\begin{equation}
\beta^{\prime}=\frac{\beta}{\mathring{v}^{2}}=\frac{\pi\sigma_{\mathrm{B}}}{\lambda_{\mathrm{rp}}^{2}}=\frac{4\pi g_{\mathrm{B}}}{\lambda_{\mathrm{rp}}}.\label{eq:Bvom1b}
\end{equation}

As for the single transmission line, its shunt capacitance per unit
of length is a real number $C>0$ and its inductance per unit of length
is another real number $L>0$. The coupling constant $0<b\leq1$ is
also a number, see \cite[3]{FigTWTbk} for more details. The TL single
characteristic velocity $w$ and the single \emph{TL principal coefficient}
$\theta$ are defined by
\begin{equation}
w=\frac{1}{\sqrt{CL}},\quad\theta=\frac{b^{2}}{C}.\label{eq:T1B1betas1b}
\end{equation}
Following \cite[3]{FigTWTbk}, we assume that
\begin{equation}
0<\mathring{v}<w.\label{eq:T1B1betas1c}
\end{equation}

\subsection{TWT Lagrangian and evolution equations\label{subsec:twtlagev}}

Following the developments in \cite{FigTWTbk}, we introduce the \emph{TWT
principal parameter} $\bar{\gamma}=\theta\beta$. This parameter in
view of equations (\ref{eq:T1B1betas1a}) and (\ref{eq:T1B1betas1b})
can be represented as follows
\begin{equation}
\gamma=\theta\beta=\frac{b^{2}}{C}\frac{\sigma_{\mathrm{B}}}{4\pi}R_{\mathrm{sc}}^{2}\omega_{\mathrm{p}}^{2}=\frac{b^{2}}{C}\frac{e^{2}}{m}R_{\mathrm{sc}}^{2}\sigma_{\mathrm{B}}\mathring{n},\quad\theta=\frac{b^{2}}{C},\quad\beta=\frac{e^{2}}{m}R_{\mathrm{sc}}^{2}\sigma_{\mathrm{B}}\mathring{n}.\label{eq:T1B1be1a}
\end{equation}
The TWT Lagrangian $\mathcal{L}{}_{\mathrm{TB}}$ in the simplest
case of a single transmission line and one stream e-beam is of the
following form, \cite[4, 24]{FigTWTbk}:
\begin{gather}
\mathcal{L}\left(\left\{ Q\right\} ,\left\{ q\right\} \right)=\mathcal{L}_{\mathrm{Tb}}\left(\left\{ Q\right\} ,\left\{ q\right\} \right)+\mathcal{L}_{\mathrm{B}}\left(\left\{ q\right\} \right),\label{eq:T1B1be1b}\\
\mathcal{L}_{\mathrm{Tb}}=\frac{L}{2}\left(\partial_{t}Q\right)^{2}-\frac{1}{2C}\left(\partial_{z}Q+b\partial_{z}q\right)^{2},\;\mathcal{L}_{\mathrm{B}}=\frac{1}{2\beta}\left(\partial_{t}q+\mathring{v}\partial_{z}q\right)^{2}-\frac{2\pi}{\sigma_{\mathrm{B}}}q^{2},\nonumber 
\end{gather}
where where $b>0$ is a coupling coefficient and
\begin{gather}
\left\{ Q\right\} =Q,\partial_{z}Q,\partial_{t}Q,\quad Q=Q\left(z,t\right);\quad\left\{ q\right\} =q,\partial_{z}q,\partial_{t}q,\quad q=q\left(z,t\right),\label{eq:T1B1be1c}
\end{gather}
where $q\left(z,t\right)$ and $Q\left(z,t\right)$\emph{ are charges}
associated respectively with the e-beam and the TL. The charges defined
as time integrals of the corresponding e-beam currents $i(z,t)$ and
TL current $I(z,t)$, that is
\begin{equation}
q(z,t)=\int^{t}q(z,t^{\prime})\,\mathrm{d}t^{\prime},\quad.Q(z,t)=\int^{t}I(z,t^{\prime})\,\mathrm{d}t^{\prime}.\label{eq:T1B1be1d}
\end{equation}
Note that the term $-\frac{2\pi}{\sigma_{\mathrm{B}}}q^{2}$ in the
Lagrangian $\mathcal{L}_{\mathrm{B}}$ defined in equations (\ref{eq:T1B1be1b})
represents space-charge effects including the so-called debunching
(electron-to-electron repulsion). The corresponding Euler-Lagrange
equations is represented by the following system of second-order differential
equations
\begin{gather}
L\partial_{t}^{2}Q-\partial_{z}\left[C^{-1}\left(\partial_{z}Q+b\partial_{z}q\right)\right]=0,\label{eq:T1B1be1e}\\
\frac{1}{\beta}\left(\partial_{t}+\mathring{v}\partial_{z}\right)^{2}q+\frac{4\pi}{\sigma_{\mathrm{B}}}q-b\partial_{z}\left[C^{-1}\left(\partial_{z}Q+b\partial_{z}q\right)\right]=0,\label{eq:T1B1be1f}
\end{gather}
where $\mathring{v}$ is the stationary velocity of electrons in the
e-beam, $\sigma_{\mathrm{B}}$ is the area of the cross-section of
the e-beam and $\beta$ is the e-beam intensity defined by equations
(\ref{eq:T1B1be1a}).

\subsection{Space-charge wave velocity and electron density fields\label{subsec:cwvel}}

Following to the field theory constructions in \cite[22]{FigTWTbk}
we consider the total electron density $N=\mathring{n}+n$ and $V=\mathring{v}+v$
where $\mathring{n}$ and $\mathring{v}$ are respectively densities
of the electron number and the electron velocity of the stationary
dc electron flow, and $n=n\left(z,t\right)$ and $v=v\left(z,t\right)$
are respectively position and time dependent ac densities of the electron
number and the electron velocity of the space-charge wave field $q=q\left(z,t\right)$
and the electric field $E=E\left(z,t\right)$. To comply with the
linear theory approximation we assume that $n$ and $v$ to be relatively
small:
\begin{equation}
\left|n\right|\ll\mathring{n},\quad\left|v\right|\ll\mathring{v}.\label{eq:nnvvE1a}
\end{equation}
We consider also the ac current density field $j=j(z,t)$ that satisfies
\begin{equation}
j=j(z,t)=-e\left(\mathring{n}v+\mathring{v}n\right),\quad\partial_{t}\left(-en\right)+\partial_{z}j=0.\label{eq:chavar1a}
\end{equation}
Then the following relations between charge density field $q=q\left(z,t\right)$
and fields $n=n\left(z,t\right)$ and $v=v\left(z,t\right)$ hold,
\cite[22.2]{FigTWTbk}:
\begin{equation}
\partial_{z}q=\sigma_{\mathrm{B}}en,\qquad\partial_{t}q=\sigma_{\mathrm{B}}j=J,\label{eq:chavar1aa}
\end{equation}
where $J$ is the e-beam current. The first equation (\ref{eq:chavar1aa})
readily implies
\begin{equation}
n=\frac{1}{\sigma_{\mathrm{B}}e}\partial_{z}q.\label{eq:chavar1ab}
\end{equation}

The relations between charge variable $q=q(z,t)$ defined in Section
\ref{subsec:twtlagev}, the velocity $v=v(z,t)$ and associated with
it current $J_{v}=J_{v}(z,t)$ are as follows, \cite[22.2]{FigTWTbk}:
\begin{equation}
e\sigma_{\mathrm{B}}\mathring{n}v=-D_{t}q,\qquad J_{v}=-e\sigma_{\mathrm{B}}\mathring{n}v=D_{t}q,\quad D_{t}=\partial_{t}+\mathring{v}\partial_{z}\label{eq:charvar1b}
\end{equation}
where $D_{t}$ is the so-called \emph{material time derivative}. The
second equation in (\ref{eq:charvar1b}) implies evidently that current
$J_{v}$ is exactly $D_{t}q$, whereas the first equation in (\ref{eq:charvar1b})
yields the following representations of the velocity $v$
\begin{equation}
v=-\frac{1}{e\sigma_{\mathrm{B}}\mathring{n}}D_{t}q=-\frac{J_{v}}{e\sigma_{\mathrm{B}}\mathring{n}}.\label{eq:chavar1ba}
\end{equation}
The electric field $E=E(z,t)$ associated with the space-charge wave
satisfies the Poisson equation, \cite[22.2]{FigTWTbk}:
\begin{equation}
\partial_{z}E=-\frac{4\pi}{\sigma_{\mathrm{B}}}\partial_{z}q=-4\pi en.\label{eq:charvar1c}
\end{equation}
Under the additional natural assumption ``if there is no charges
there is no electric field'', that is, $\bar{q}=0$ must imply $E=0$,
the above equation yields
\begin{equation}
E=-\frac{4\pi}{\sigma_{\mathrm{B}}}\bar{q.}\label{eq:charvar1d}
\end{equation}
If we introduce the e-beam voltage $V_{\mathrm{b}}=V_{\mathrm{b}}\left(z,t\right)=-\partial_{z}E$
then the Poisson equation (\ref{eq:charvar1c}) can be recast as
\begin{equation}
\partial_{z}^{2}V_{\mathrm{b}}=4\pi en.\label{eq:charvar1f}
\end{equation}

\section{Analytic model of coupled-cavity traveling wave tube\label{sec:cctwtmod}}

When integrating into the mathematical model significant features
of the CCTWT, we make a number of simplifying assumptions. In particular,
we use the following basic assumptions of one-dimensional model of
space-charge waves in velocity-modulated beams, \cite[7.6.1]{Tsim}:
(i) all quantities of interest depend only on a single space variable
$z$; (ii) the electric field has only an $z$-component; (iii) there
are no transverse velocities of electrons; (iv) ac values are small
compared with dc values; (v) electrons have a constant dc velocity
which is much smaller than the speed of light, and (vi) electron beams
are nondense. The list of preliminary assumptions of our ideal model
for the CCTWT is as follows.

\begin{assumption} \label{ass:idealmod}(ideal model of the e-beam
and the TL interaction).
\begin{enumerate}
\item E-beam is represented by a cylinder of infinitesimally small radius
having as its axis the $z$-axis (see Fig. \ref{fig:cctwt}).
\item Coupled-cavity structure (CCS) is represented mathematically by a
periodic array of adjacent segments of a transmission line (TL) of
length $a>0$ connected by cavities at at points $a\ell$, $\ell\in\mathbb{Z}$
by cavities.
\item Every cavity carries shunt capacitance $c_{0}$. The e-beam interacts
with the CCS exclusively through the cavities located at a discrete
set of equidistant points, that is the lattice
\begin{equation}
a\mathbb{Z}:\mathbb{Z}=\left\{ \ldots,-2,-1,0,1,2,\ldots\right\} ,\label{eq:aZep1a}
\end{equation}
where $a>0$ and we refer to this parameter as the \emph{CCS period
or just period. The cavity width $l_{\mathrm{g}}$ and the corresponding
transit time $\tau_{g}$ are assumed to be zero, see equations (\ref{eq:transL1d})
and comments above it.}
\end{enumerate}
\end{assumption}

The CCTWT state is described by charges $q=q\left(z,t\right)$ and
$Q=Q\left(z,t\right)$ for respectively the e-beam and the TL defined
as the time integrals of the corresponding currents
\begin{equation}
Q=Q\left(z,t\right)=\int^{t}I\left(z,t^{\prime}\right)\,\mathrm{d}t^{\prime},\quad q=q\left(z,t\right)=\int^{t}i\left(z,t^{\prime}\right)\,\mathrm{d}t^{\prime}.\label{eq:aZep1b}
\end{equation}
Since according to the formulated above assumptions the interaction
should occur only at the discrete set $a\mathbb{Z}$ (lattice) of
points embedded into one-dimensional continuum of real numbers $\mathbb{R}$
some degree of singularity of functions $Q\left(z,t\right)$ and $q\left(z,t\right)$
is expected. As the analysis shows it is appropriate to impose the
following \emph{jump and continuity conditions} on charge functions
$Q\left(z,t\right)$ and $q\left(z,t\right)$.

\begin{assumption} \label{ass:jumpcon}(jump-continuity of charge
functions).
\begin{enumerate}
\item Functions $Q\left(z,t\right)$ and $q\left(z,t\right)$ and their
time derivatives $\partial_{t}^{j}Q\left(z,t\right)$ and $\partial_{t}^{j}q\left(z,t\right)$
for $j=1,2$ are continuous for all real $t$ and $z$.
\item Derivatives $\partial_{t}^{j}Q\left(z,t\right)$, $\partial_{t}^{j}Q\left(z,t\right)$,
$\partial_{z}^{j}q\left(z,t\right)$, $\partial_{z}^{j}q\left(z,t\right)$
for $j=1,2$, and the mixed derivatives $\partial_{z}\partial_{t}Q\left(z,t\right)=\partial_{t}\partial_{z}Q\left(z,t\right)$,
$\partial_{z}\partial_{t}q\left(z,t\right)=\partial_{t}\partial_{z}q\left(z,t\right)$
exist and continuous for all real real $t$ and $z$ except for the
interaction points in the lattice $a\mathbb{Z}$.
\item Let for a function $F\left(z\right)$ and a real number $b$ symbols
$F$$\left(b-0\right)$ and $F$$\left(b+0\right)$ stand for its
left and right limit at $b$ assuming their existence, that is
\begin{equation}
F\left(b\pm0\right)=\lim_{z\rightarrow b\pm0}F\left(z\right).\label{eq:limpm1a}
\end{equation}
Let us also denote by $\left[F\right]\left(b\right)$ the jump of
function $F\left(z\right)$ at $b$, that is
\begin{equation}
\left[F\right]\left(b\right)=F\left(b+0\right)-F\left(b-0\right).\label{eq:limpm1b}
\end{equation}
The following right and left limits exist
\begin{equation}
\partial_{z}^{j}Q\left(a\ell\pm0,t\right),\quad\partial_{z}^{j}q\left(a\ell\pm0,t\right),\;j=1,2;\;\ell\in\mathbb{Z},\label{eq:limpm1c}
\end{equation}
and these limits are continuously differentiable functions of $t$.
\emph{The values $\partial_{z}Q\left(a\ell\pm0,t\right)$ as well
as $\partial_{z}q\left(a\ell\pm0,t\right)$ can be different and consequently
the jumps $\left[\partial_{z}Q\right]\left(a\ell,t\right)$ and $\left[\partial_{z}q\right]\left(a\ell,t\right)$
can be nonzero}.
\end{enumerate}
\end{assumption}
\begin{rem}[physical significance of jumps]
\emph{ }\label{rem:physjump}Though according to Assumption \ref{ass:idealmod}
we neglect the widths of the EM cavities their interaction with the
electron flow is represented through jumps $\left[\partial_{z}q\right]\left(a\ell,t\right)$
which are of the direct physical significance. Indeed, the field interpretation
of the kinetic properties of the electron flow in Section \ref{subsec:fiekin},
namely equations (\ref{eq:naLvaL1a}), imply
\begin{equation}
\left[\partial_{z}q\right]\left(a\ell,t\right)=\sigma_{\mathrm{B}}e\left[n\right]\left(a\ell,t\right).\label{eq:dzqsn1a}
\end{equation}
Equation (\ref{eq:dzqsn1a}) shows that jump $\left[\partial_{z}q\right]\left(a\ell,t\right)$
up to a multiplicative constant $\frac{1}{\sigma_{\mathrm{B}}e}$
represents jump $\left[n\right]\left(a\ell,t\right)=\frac{\left[\partial_{z}q\right]\left(a\ell,t\right)}{\sigma_{\mathrm{B}}e}$
in the number of electron density \emph{manifesting the electron bunching}
that occurs in the EM cavity centered at $a\ell$. In view of equations
(\ref{eq:naLvaL1b}) we also have $\left[v\right]\left(a\ell,t\right)=-\frac{\mathring{v}\left[\partial_{z}q\right]\left(a\ell,t\right)}{e\sigma_{\mathrm{B}}\mathring{n}}$
manifesting the ac electron velocity modulation in the EM cavity centered
at $a\ell$. 
\end{rem}

The physical dimensions of quantities related to the cavities and
the TL are summarized respectively in Tables \ref{tab:cavity-par}
and \ref{tab:MTL-phys-dim}.
\begin{table}
\centering{}%
\begin{tabular}{|r||r||r|}
\hline 
\noalign{\vskip\doublerulesep}
$I$ & Current & $\frac{\left[\text{charge}\right]}{\left[\text{time}\right]}$\tabularnewline
\hline 
\noalign{\vskip\doublerulesep}
$Q$ & Charge & $\left[\text{charge}\right]$\tabularnewline
\hline 
\noalign{\vskip\doublerulesep}
$c_{0}$ & Cavity capacitance & $\left[\text{length}\right]$\tabularnewline
\hline 
\noalign{\vskip\doublerulesep}
$l_{0}$ & Cavity inductance & $\frac{\left[\text{time}\right]^{2}}{\left[\text{length}\right]}$\tabularnewline
\hline 
\noalign{\vskip\doublerulesep}
$b$ & Coupling parameter & $\left[\text{dim-less}\right]$\tabularnewline
\hline 
\end{tabular}\vspace{0.3cm}
\caption{\label{tab:cavity-par}Physical dimensions of cavity related quantities.
Abbreviations: dimensionless \textendash{} dim-less}
\end{table}
\begin{table}
\centering{}%
\begin{tabular}{|r||r||r|}
\hline 
\noalign{\vskip\doublerulesep}
$I$ & Current & $\frac{\left[\text{charge}\right]}{\left[\text{time}\right]}$\tabularnewline
\hline 
\noalign{\vskip\doublerulesep}
$Q$ & Charge & $\left[\text{charge}\right]$\tabularnewline
\hline 
\noalign{\vskip\doublerulesep}
$C$ & Shunt capacitance p/u of length & $\left[\text{dim-less}\right]$\tabularnewline
\hline 
\noalign{\vskip\doublerulesep}
$L$ & Series inductance p/u of length & $\frac{\left[\text{time}\right]^{2}}{\left[\text{length}\right]^{2}}$\tabularnewline
\hline 
\end{tabular}\vspace{0.3cm}
\caption{Physical dimensions of the TL related quantities. Abbreviations: dimensionless
\textendash{} dim-less, p/u \textendash{} per unit.\label{tab:MTL-phys-dim}}
\end{table}

\subsection{CCTWT Lagrangian and the Euler-Lagrange equations\label{subsec:cctwt-lag}}

To simplify expressions we use the following notations:
\begin{equation}
\left\{ Q\right\} =Q,\:\partial_{z}Q,\:\partial_{t}Q,\quad Q=Q\left(z,t\right);\quad\left\{ q\right\} =q,\:\partial_{z}q,\:\partial_{t}q,\quad q=q\left(z,t\right),\label{eq:aZep1Q}
\end{equation}
\begin{equation}
x=\left[\begin{array}{r}
Q\\
q
\end{array}\right],\quad\left\{ x\right\} =\left[\begin{array}{r}
Q\\
q
\end{array}\right],\:\left[\begin{array}{r}
\partial_{z}Q\\
\partial_{z}q
\end{array}\right],\:\left[\begin{array}{r}
\partial_{t}Q\\
\partial_{t}q
\end{array}\right].\label{eq:eZep1x}
\end{equation}
The CCTWT Lagrangian $\mathcal{L}$ is defined as the sum of its three
components: (i) $\mathcal{L}_{\mathrm{T}}$ is the TL Lagrangian,
(ii) $\mathcal{L}_{\mathrm{B}}$ is the e-beam Lagrangian; (iii) $\mathcal{L}_{\mathrm{TB}}$
represent the TL and e-beam interaction Lagrangian. That is,

\begin{equation}
\mathcal{L}\left(\left\{ x\right\} \right)=\mathcal{L}_{\mathrm{T}}\left(\left\{ Q\right\} \right)+\mathcal{L}_{\mathrm{B}}\left(\left\{ q\right\} \right)+\mathcal{L}_{\mathrm{TB}}\left(x\right),\label{eq:aZep1L}
\end{equation}
where we used notations (\ref{eq:aZep1Q}) and (\ref{eq:eZep1x}).
The expressions for $\mathcal{L}_{\mathrm{T}}$ and $\mathcal{L}_{\mathrm{B}}$
are similar to the TWT Lagrangian components in equations (\ref{eq:T1B1be1b}),
(\ref{eq:T1B1be1c}), namely

\begin{gather}
\mathcal{L}_{\mathrm{T}}\left(\left\{ Q\right\} \right)=\left.\mathcal{L}_{\mathrm{Tb}}\right|_{b=0}=\frac{L}{2}\left(\partial_{t}Q\right)^{2}-\frac{1}{2C}\left(\partial_{z}Q\right)^{2},\quad\mathcal{L}_{\mathrm{B}}\left(\left\{ q\right\} \right)=\frac{1}{2\beta}\left(\partial_{t}q+\mathring{v}\partial_{z}q\right)^{2}-\frac{2\pi}{\sigma_{\mathrm{B}}}q^{2},\label{eq:aZep1c}
\end{gather}
where $\mathcal{L}_{\mathrm{Tb}}$ is defined by equation (\ref{eq:T1B1be1b})
and the interaction Lagrangian $\mathcal{L}_{\mathrm{TB}}$ is defined
by
\begin{equation}
\mathcal{L}_{\mathrm{TB}}\left(x\right)=\sum_{\ell=-\infty}^{\infty}\delta\left(z-a\ell\right)\left\{ \frac{l_{0}}{2}\left(\partial_{t}Q\left(a\ell\right)\right)^{2}-\frac{1}{2c_{0}}\left[Q\left(a\ell\right)+bq\left(a\ell\right)\right]^{2}\right\} .\label{eq:aZep1d}
\end{equation}
Parameters $C$ and $L$ are respectively distributed shunt capacitance
and inductance of the TL, $\sigma_{\mathrm{B}}$ is the area of the
cross-section and $\beta$ is the e-beam intensity defined in Section
\ref{sec:twtmod}. Lagrangian $\mathcal{L}_{\mathrm{B}}$ in equations
(\ref{eq:T1B1be1c}) represents the e-beam and the term $-\frac{2\pi}{\sigma_{\mathrm{B}}}q^{2}$
models the space-charge effects including the so-called debunching
(electron-to-electron repulsion). Lagrangian $\mathcal{L}_{\mathrm{Tb}}$
in equations (\ref{eq:T1B1be1b}) integrates into it the interactions
between the TL and the e-beam whereas for the CCTWT the Lagrangian
$\mathcal{L}_{\mathrm{T}}$ corresponds to the decoupled TL. This
is why we set $\mathcal{L}_{\mathrm{T}}$ to be $\mathcal{L}_{\mathrm{Tb}}$
for $b=0$. Note also that (i) expression (\ref{eq:aZep1d}) for the
interaction Lagrangian $\mathcal{L}_{\mathrm{TB}}$ is limited by
design to the interaction points $a\ell$ as indicated by delta functions
$\delta\left(z-a\ell\right)$ and (ii) the factors before delta functions
$\delta\left(z-a\ell\right)$ are expressions similar to density $\mathcal{L}_{\mathrm{Tb}}$
in equations (\ref{eq:T1B1be1b}) adapted to lattice $a\mathbb{Z}$
of discrete interaction points $a\ell$; (iii) capacitance $c_{0}$
is of the most significance for the interaction between the TL and
the e-beam and we refer to it as the \emph{cavity capacitance}. \emph{It
follows from equations (\ref{eq:aZep1L}), (\ref{eq:aZep1c}) and
(\ref{eq:aZep1d}) that $\mathcal{L}$ is a periodic Lagrangian of
the period $a$. }

As we derive in Section \ref{sec:lagvar} the Euler-Lagrange (EL)
equations for points $z$ outside the lattice $a\mathbb{Z}$ are

\begin{gather}
L\partial_{t}^{2}Q-C^{-1}\partial_{z}^{2}Q=0,\quad\frac{1}{\beta}\left(\partial_{t}+\mathring{v}\partial_{z}\right)^{2}q+\frac{4\pi}{\sigma_{\mathrm{B}}}q=0,\quad z\neq a\ell,\quad\ell\in\mathbb{Z},\label{eq:eZep1e}
\end{gather}
or equivalently
\begin{gather}
\partial_{z}^{2}Q-\frac{1}{w^{2}}\partial_{t}^{2}Q=0,\quad\left(\frac{1}{v}\partial_{t}+\partial_{z}\right)^{2}q+\frac{4\pi\beta}{\sigma_{\mathrm{B}}\mathring{v}^{2}}q=0,\quad z\neq a\ell,\quad\ell\in\mathbb{Z}.\label{eq:eZep1ea}
\end{gather}
The EL equations at the interaction points $a\ell$ are
\begin{equation}
\left[Q\right]\left(a\ell\right)=0,\quad\left[q\right]\left(a\ell\right)=0,\label{eq:eZep1fa}
\end{equation}
\begin{gather}
\left[\partial_{z}Q\right]\left(a\ell\right)=C_{0}\left[\left(\frac{\partial_{t}^{2}}{\omega_{0}^{2}}+1\right)Q\left(a\ell\right)+bq\left(a\ell\right)\right],\quad\left[\partial_{z}q\right]\left(a\ell\right)=-\frac{b\beta_{0}}{\mathring{v}^{2}}\left[Q\left(a\ell\right)+bq\left(a\ell\right)\right],\label{eq:eZep1f}
\end{gather}
where we make use of parameters
\begin{equation}
C_{0}=\frac{C}{c_{0}},\quad\omega_{0}=\frac{1}{\sqrt{l_{0}c_{0}}},\quad\beta_{0}=\frac{\beta}{c_{0}},\label{eq:ezep1h}
\end{equation}
and jumps $\left[Q\right]\left(a\ell\right)$, $\left[q\right]\left(a\ell\right)$,
$\left[\partial_{z}Q\right]\left(a\ell\right)$ and $\left[\partial_{z}q\right]\left(a\ell\right)$
are defined by equation (\ref{eq:limpm1b}). We refer to $\beta_{0}$
as \emph{nodal e-beam interaction parameter}. Note that equations
(\ref{eq:eZep1fa}) is just an acknowledgment of the continuity of
charges $Q\left(z,t\right)$ and $q\left(z,t\right)$ at the interaction
points in consistency with Assumption \ref{ass:jumpcon}. Equations
(\ref{eq:eZep1fa}), (\ref{eq:eZep1f}) can be viewed as the boundary
conditions at the interaction points that are complementary to the
differential equations (\ref{eq:eZep1e}) and (\ref{eq:eZep1ea}).

\emph{The Euler-Lagrange differential equations (\ref{eq:eZep1e}),
(\ref{eq:eZep1ea}) together with the boundary conditions (\ref{eq:eZep1fa}),
(\ref{eq:eZep1f}) form the complete set of equation describing the
CCTWT evolution.} Boundary conditions (\ref{eq:eZep1fa}), (\ref{eq:eZep1f})
can be recast into the following matrix form:
\begin{equation}
\left[\partial_{z}x\right]\left(a\ell\right)=P_{\mathrm{TB}}x\left(a\ell\right),\;x=\left[\begin{array}{r}
Q\\
q
\end{array}\right],\;P_{\mathrm{TB}}=\left[\begin{array}{rr}
C_{0}\left(\frac{\partial_{t}^{2}}{\omega_{0}^{2}}+1\right) & C_{0}b\\
-\frac{b\beta_{0}}{\mathring{v}^{2}} & -\frac{b^{2}\beta_{0}}{\mathring{v}^{2}}
\end{array}\right],\label{eq:eZep1g}
\end{equation}
where parameters $C_{0}$ and $\beta_{0}$ are defined by equations
(\ref{eq:ezep1h}). Hence, the complete set of the boundary conditions
at interaction points $a\ell$ can be concisely written as
\begin{gather}
x\left(+a\ell\right)=x\left(-a\ell\right),\quad\partial_{z}x\left(+a\ell\right)=\partial_{z}x\left(-a\ell\right)+P_{\mathrm{TB}}x.\label{eq:eZep2a}
\end{gather}
Consequently
\begin{equation}
X\left(a\ell+0\right)=\mathsf{S}_{\mathrm{b}}X\left(a\ell-0\right),\quad\mathsf{S}_{\mathrm{b}}=\left[\begin{array}{rr}
\mathbb{I} & 0\\
P_{\mathrm{TB}} & \mathbb{I}
\end{array}\right],\quad X=\left[\begin{array}{r}
x\\
\partial_{z}x
\end{array}\right],\quad x=\left[\begin{array}{r}
Q\\
q
\end{array}\right],\label{eq:eZep2b}
\end{equation}
where $\mathbb{I}$ is $2\times2$ identity matrix and in view of
equations (\ref{eq:eZep1g}) we have
\begin{equation}
\mathsf{S}_{\mathrm{b}}=\left[\begin{array}{rr}
\mathbb{I} & 0\\
P_{\mathrm{TB}} & \mathbb{I}
\end{array}\right]=\left[\begin{array}{rrrr}
1 & 0 & 0 & 0\\
0 & 1 & 0 & 0\\
C_{0}\left(\frac{\partial_{t}^{2}}{\omega_{0}^{2}}+1\right) & C_{0}b & 1 & 0\\
-\frac{b\beta_{0}}{\mathring{v}^{2}} & -\frac{b^{2}\beta_{0}}{\mathring{v}^{2}} & 0 & 1
\end{array}\right],\quad C_{0}=\frac{C}{c_{0}},\quad\beta_{0}=\frac{\beta}{c_{0}}.\label{eq:eZep2c}
\end{equation}

\subsection{Natural units and dimensionless parameters\label{subsec:dim-par}}

The natural units relevant to the e-beam in CCTWT are shown in Table
\ref{tab:units}.
\begin{table}[h]
\centering{}%
\begin{tabular}{|r||r||r|}
\hline 
\noalign{\vskip\doublerulesep}
Velocity & e-beam velocity & $\mathring{v}$\tabularnewline[0.2cm]
\hline 
\noalign{\vskip\doublerulesep}
Length & period & $a$\tabularnewline[0.2cm]
\hline 
\noalign{\vskip\doublerulesep}
Length & Wavelength & $\mathring{\lambda}=\frac{1}{\mathring{k}}=\frac{\mathring{v}}{\omega_{\mathrm{p}}}$\tabularnewline[0.2cm]
\hline 
\noalign{\vskip\doublerulesep}
Frequency & Period frequency & $\omega_{a}=\frac{\mathring{v}}{a}$\tabularnewline[0.2cm]
\hline 
\noalign{\vskip\doublerulesep}
Frequency & Plasma frequency & $\omega_{\mathrm{p}}=\sqrt{\frac{4\pi\mathring{n}e^{2}}{m}}$\tabularnewline[0.2cm]
\hline 
\hline 
\noalign{\vskip\doublerulesep}
Time & Time of passing the period $a$ & $\frac{1}{\omega_{a}}$=$\frac{a}{\mathring{v}}$\tabularnewline[0.2cm]
\hline 
\hline 
\noalign{\vskip\doublerulesep}
Time & Plasma oscillation time period & $\mathring{\tau}=\frac{1}{\omega_{\mathrm{p}}}$\tabularnewline[0.2cm]
\hline 
\end{tabular}\vspace{0.3cm}
\caption{\label{tab:units} Natural units relevant to the e-beam in CCTWT.}
\end{table}

Another variables that arise in our analysis are
\begin{gather}
\lambda_{\mathrm{rp}}=\frac{2\pi\mathring{v}}{\omega_{\mathrm{rp}}},\quad\:\omega_{\mathrm{rp}}=R_{\mathrm{sc}}\omega_{\mathrm{p}},\quad\chi=\frac{w}{\mathring{v}}=\frac{1}{\mathring{v}\sqrt{CL}},\label{eq:unitL1a}\\
f_{\mathrm{B}}=\sqrt{\frac{4\pi\beta}{\sigma_{\mathrm{B}}\mathring{v}^{2}}}=\frac{\omega_{\mathrm{rp}}}{\mathring{v}}=\frac{2\pi}{\lambda_{\mathrm{rp}}},\quad C_{0}=\frac{C}{c_{0}},\quad\beta_{0}=\frac{\beta}{c_{0}},\nonumber 
\end{gather}
where $\omega_{\mathrm{rp}}$ and $\lambda_{\mathrm{rp}}$ are respectively
the \emph{reduced plasma frequency} and \emph{the electron plasma
wavelength}.

The dimensionless variables of importance are
\begin{equation}
z^{\prime}=\frac{z}{a},\quad\partial_{z^{\prime}}=a\partial_{z},\quad t^{\prime}=\frac{\mathring{v}}{a}t,\quad\partial_{t^{\prime}}=\frac{a}{\mathring{v}}\partial_{t},\quad\omega^{\prime}=\frac{\omega}{\omega_{a}}=\frac{a}{2\pi\mathring{v}}\omega,\quad\omega_{0}^{\prime}=\frac{\omega_{0}}{\omega_{a}}=\frac{a}{2\pi\mathring{v}}\omega_{0},\label{eq:unitL1b}
\end{equation}
\begin{equation}
L^{\prime}=\mathring{v}^{2}L,\quad C^{\prime}=C,\quad\beta^{\prime}=\frac{\beta}{\mathring{v}^{2}},\quad\sigma_{\mathrm{B}}^{\prime}=\frac{\sigma_{\mathrm{B}}}{a^{2}},\quad c_{0}^{\prime}=\frac{c_{0}}{a},\quad l_{0}^{\prime}=\frac{\mathring{v}^{2}}{a}l_{0},\label{eq:unitL1c}
\end{equation}
\begin{equation}
C_{0}^{\prime}=\frac{C^{\prime}}{c_{0}^{\prime}}=\frac{aC}{c_{0}}=aC_{0},\quad\beta_{0}^{\prime}=\frac{\beta^{\prime}}{c_{0}^{\prime}}=\frac{a\beta}{c_{0}\mathring{v}^{2}},\quad f_{\mathrm{B}}^{\prime}=af_{\mathrm{B}}=\sqrt{\frac{4\pi\beta^{\prime}}{\sigma_{\mathrm{B}}^{\prime}}}=\frac{a\omega_{\mathrm{rp}}}{\mathring{v}}=\frac{2\pi\omega_{\mathrm{rp}}}{\omega_{a}}=\frac{2\pi a}{\lambda_{\mathrm{rp}}}.\label{eq:unitL1d}
\end{equation}
Notice also that since $w=\chi v$
\begin{equation}
L^{\prime}C^{\prime}=\frac{\mathring{v}^{2}}{w^{2}}=\frac{1}{\chi^{2}},\quad a\delta\left(z\right)=\delta\left(z^{\prime}\right),\quad z=az^{\prime}.\label{eq:unitL1e}
\end{equation}
For the reader convenience we collected in Table \ref{tab:cctwt-par}
all significant parameters associated with CCTWT.
\begin{table}
\centering{}%
\begin{tabular}{|r||r||r|}
\hline 
\noalign{\vskip\doublerulesep}
$a$ & the MCK period & $\left[\text{length}\right]$\tabularnewline[0.2cm]
\hline 
\noalign{\vskip\doublerulesep}
$\mathring{v}$ & the e-beam stationary velocity & $\frac{\left[\text{length}\right]}{\left[\text{time}\right]}$\tabularnewline[0.2cm]
\hline 
\noalign{\vskip\doublerulesep}
$\omega_{a}=\frac{2\pi\mathring{v}}{a}$ & the period frequency & $\frac{\left[\text{1}\right]}{\left[\text{time}\right]}$\tabularnewline[0.2cm]
\hline 
\noalign{\vskip\doublerulesep}
$\omega_{\mathrm{p}}=\sqrt{\frac{4\pi\mathring{n}e^{2}}{m}}$ & the plasma frequency & $\frac{\left[\text{1}\right]}{\left[\text{time}\right]}$\tabularnewline[0.2cm]
\hline 
\noalign{\vskip\doublerulesep}
$\lambda_{\mathrm{rp}}=\frac{2\pi\mathring{v}}{\omega_{\mathrm{rp}}},\:\omega_{\mathrm{rp}}=R_{\mathrm{sc}}\omega_{\mathrm{p}}$ & the electron plasma wavelength & $\left[\text{length}\right]$\tabularnewline[0.2cm]
\hline 
\noalign{\vskip\doublerulesep}
$g_{\mathrm{B}}=\frac{\sigma_{\mathrm{B}}}{4\lambda_{\mathrm{rp}}}$ & the e-beam spatial scale & $\left[\text{length}\right]$\tabularnewline[0.2cm]
\hline 
\noalign{\vskip\doublerulesep}
$f^{\prime}=af=\frac{2\pi\omega_{\mathrm{rp}}}{\omega_{a}}=\frac{2\pi a}{\lambda_{\mathrm{rp}}}$ & normalized period in units of $\frac{\lambda_{\mathrm{rp}}}{2\pi}$ & $\left[\text{dim-less}\right]$\tabularnewline[0.2cm]
\hline 
\noalign{\vskip\doublerulesep}
$\mathring{n}$ & the number of electrons p/u of volume & $\frac{\left[\text{1}\right]}{\left[\text{length}\right]^{3}}$\tabularnewline[0.2cm]
\hline 
\noalign{\vskip\doublerulesep}
$c_{0},\;l_{0}$ & the cavity capacitance, inductance & $\left[\text{length}\right]$$,\;\frac{\left[\text{time}\right]^{2}}{\left[\text{length}\right]}$\tabularnewline[0.2cm]
\hline 
\noalign{\vskip\doublerulesep}
$\omega_{0}=\frac{1}{\sqrt{l_{0}c_{0}}}$ & the cavity resonant frequency & $\frac{\left[\text{1}\right]}{\left[\text{time}\right]}$\tabularnewline[0.2cm]
\hline 
\noalign{\vskip\doublerulesep}
$\beta=\frac{\sigma_{\mathrm{B}}R_{\mathrm{sc}}^{2}\omega_{\mathrm{p}}^{2}}{4\pi}=\frac{e^{2}R_{\mathrm{sc}}^{2}\sigma_{\mathrm{B}}\mathring{n}}{m}=\frac{\pi\sigma_{\mathrm{B}}\mathring{v}^{2}}{\lambda_{\mathrm{rp}}^{2}}$ & the e-beam intensity & $\frac{\left[\text{length}\right]^{2}}{\left[\text{time}\right]^{2}}$\tabularnewline[0.2cm]
\hline 
\noalign{\vskip\doublerulesep}
$\beta^{\prime}=\frac{\beta}{\mathring{v}^{2}}=\frac{\pi\sigma_{\mathrm{B}}}{\lambda_{\mathrm{rp}}^{2}}=\frac{4\pi g_{\mathrm{B}}}{\lambda_{\mathrm{rp}}}$ & dim-less e-beam intensity & $\left[\text{dim-less}\right]$\tabularnewline[0.2cm]
\hline 
\noalign{\vskip\doublerulesep}
$\beta_{0}^{\prime}=\frac{\beta^{\prime}}{c_{0}^{\prime}}=\frac{a\beta}{c_{0}\mathring{v}^{2}}$ & the first interaction par. & $\left[\text{dim-less}\right]$\tabularnewline[0.2cm]
\hline 
\noalign{\vskip\doublerulesep}
$B\left(\omega\right)=B_{0}\frac{\omega^{2}}{\omega^{2}-\omega_{0}^{2}},\quad B_{0}=b^{2}\beta_{0}^{\prime}$ & the second interaction par. & $\left[\text{dim-less}\right]$\tabularnewline[0.2cm]
\hline 
\noalign{\vskip\doublerulesep}
$K_{0}=\frac{B_{0}}{2f}=\frac{b^{2}\beta_{0}^{\prime}}{2f}=\frac{b^{2}\sigma_{\mathrm{B}}}{4\lambda_{\mathrm{rp}}c_{0}}=\frac{b^{2}g_{\mathrm{B}}}{c_{0}}$ & the gain coefficient & $\left[\text{dim-less}\right]$\tabularnewline[0.2cm]
\hline 
\noalign{\vskip\doublerulesep}
$K\left(\omega\right)=\frac{B\left(\omega\right)}{2f}=K_{0}\frac{\omega^{2}}{\omega^{2}-\omega_{0}^{2}}$ & the gain parameter & $\left[\text{dim-less}\right]$\tabularnewline[0.2cm]
\hline 
\noalign{\vskip\doublerulesep}
$C_{0}^{\prime}=\frac{C^{\prime}}{c_{0}^{\prime}}=\frac{aC}{c_{0}}=aC_{0}$ & the capacitance parameter & $\left[\text{dim-less}\right]$\tabularnewline[0.2cm]
\hline 
\end{tabular}\vspace{0.3cm}
\caption{\label{tab:cctwt-par} The CCTWT significant parameters. Abbreviations:
dimensionless: dim-less, p/u: per unit, par.: parameter. For the sake
of simplicity of the notation, we often omit \textquotedblleft prime\textquotedblright{}
super-index indicating that the dimensionless version of the relevant
parameter is involved when it is clear from the context.}
\end{table}

\subsection{Euler-Lagrange equations in dimensionless variables\label{subsec:cctwtEL-dim}}

The component Lagrangians represented in dimensionless variables are
as follows:
\begin{equation}
\mathcal{L}_{\mathrm{T}}^{\prime}=\frac{L^{\prime}}{2}\left(\partial_{t^{\prime}}Q\right)^{2}-\frac{1}{2C^{\prime}}\left(\partial_{z^{\prime}}Q\right)^{2},\quad\mathcal{L}_{\mathrm{B}}^{\prime}=\frac{1}{2\beta^{\prime}}\left(\partial_{t^{\prime}}q+\partial_{z^{\prime}}q\right)^{2}-\frac{2\pi}{\sigma_{\mathrm{B}}^{\prime}}q^{2},\quad\ell\in\mathbb{Z},\label{Lagdim1a}
\end{equation}
\begin{equation}
\mathcal{L}_{\mathrm{TB}}^{\prime}=-\frac{1}{2c_{0}^{\prime}}\sum_{\ell=-\infty}^{\infty}\delta\left(z^{\prime}-\ell\right)\left[Q\left(\ell\right)+bq\left(\ell\right)\right]^{2},\label{eq:Lagdim1b}
\end{equation}
The corresponding EL equations are
\begin{gather}
\partial_{t^{\prime}}^{2}Q-\frac{1}{\chi^{2}}\partial_{z^{\prime}}^{2}Q=0,\quad\left(\partial_{t^{\prime}}+\partial_{z^{\prime}}\right)^{2}q+f_{\mathrm{B}}^{\prime2}q=0,\quad z^{\prime}\neq\ell;\quad f_{\mathrm{B}}^{\prime}=\sqrt{\frac{4\pi\beta^{\prime}}{\sigma_{\mathrm{B}}^{\prime}}}=\frac{R_{\mathrm{sc}}\omega_{\mathrm{p}}}{\omega_{a}},\label{eq:Lagdim1c}
\end{gather}
\begin{equation}
\left[\partial_{z^{\prime}}Q\right]\left(\ell\right)=C_{0}^{\prime}\left[\left(\frac{\partial_{t^{\prime}}^{2}}{\omega_{0}^{\prime2}}+1\right)Q\left(\ell\right)+bq\left(\ell\right)\right],\quad\left[\partial_{z^{\prime}}q\right]\left(\ell\right)=-b\beta_{0}^{\prime}\left[Q\left(\ell\right)+bq\left(\ell\right)\right],\quad\ell\in\mathbb{Z}.\label{eq:Lagdim1d}
\end{equation}

\emph{To simplify notations, we will omit the prime symbol identifying
the dimensionless variables in equations but rather simply will acknowledge
their dimensionless form}. Hence, we will use from now on the following
\emph{dimensionless form of the EL equations} (\ref{eq:Lagdim1c})
and (\ref{eq:unitL1d}):
\begin{gather}
\partial_{t}^{2}Q-\frac{1}{\chi^{2}}\partial_{z}^{2}Q=0,\quad\left(\partial_{t}+\partial_{z}\right)^{2}q+f^{2}q=0,\quad z\neq\ell,\quad\ell\in\mathbb{Z};\quad f=\frac{R_{\mathrm{sc}}\omega_{\mathrm{p}}}{\omega_{a}},\label{eq:Lagdim1e}
\end{gather}
\begin{equation}
\left[\partial_{z}Q\right]\left(\ell\right)=C_{0}\left[\left(\frac{\partial_{t}^{2}}{\omega_{0}^{2}}+1\right)Q\left(\ell\right)+bq\left(\ell\right)\right],\quad\left[\partial_{z}q\right]\left(\ell\right)=-b\beta_{0}\left[Q\left(\ell\right)+bq\left(\ell\right)\right],\quad\ell\in\mathbb{Z}.\label{eq:Lagdim1f}
\end{equation}

Equations (\ref{eq:Lagdim1e}), (\ref{eq:Lagdim1f}) are linear partial
differential equations in time and space variables. Their analysis
is simplified considerably if we recast them as equations in frequency
and space variable. With that in mind we apply the Fourier transform
in $t$ (see Appendix \ref{sec:four}) to equations (\ref{eq:Lagdim1e}),
(\ref{eq:Lagdim1f}) and obtain the following equations
\begin{equation}
\partial_{z}^{2}\check{Q}+\frac{\omega^{2}}{\chi^{2}}\check{Q}=0,\quad\left(\partial_{z}-\mathrm{i}\omega\right)^{2}\check{q}+f^{2}\check{q}=0,\quad z\neq\ell,\label{eq:Lagdim2a}
\end{equation}
\begin{equation}
\left[\partial_{z}Q\right]\left(\ell\right)=C_{0}\left[\frac{\omega_{0}^{2}-\omega^{2}}{\omega_{0}^{2}}Q\left(\ell\right)+bq\left(\ell\right)\right],\quad\left[\partial_{z}q\right]\left(\ell\right)=-b\beta_{0}\left[Q\left(\ell\right)+bq\left(\ell\right)\right],\quad\ell\in\mathbb{Z},\label{eq:Lagdim2aa}
\end{equation}
where $\check{Q}$ and $\check{q}$ are the time Fourier transform
of the corresponding quantities. Equations (\ref{eq:Lagdim2a}) and
(\ref{eq:Lagdim2aa}) are ODE equations in space variable $z$ with
frequency dependent coefficients.

To apply the constructions of the Floquet theory reviewed in Appendix
\ref{sec:floquet} we recast system of equations (\ref{eq:Lagdim2a}),
(\ref{eq:Lagdim2aa}) yet another time into the following manifestly
spatially periodic vector ODE:
\begin{equation}
\partial_{z}^{2}x+A_{1}\partial_{z}x+A_{0}x+\sum_{\ell=-\infty}^{\infty}\delta\left(z-\ell\right)P_{\mathrm{TB}}x=0,\quad x=\left[\begin{array}{r}
\check{Q}\\
\check{q}
\end{array}\right],\label{eq:Lagdim2b}
\end{equation}
where $2\times2$ matrices $A_{j}$ and $P_{\mathrm{TB}}$ are defined
by
\begin{equation}
A_{1}=A_{1}\left(\omega\right)=\left[\begin{array}{rr}
0 & 0\\
0 & -2\mathrm{i}\omega
\end{array}\right],\quad A_{0}=A_{0}\left(\omega\right)=\left[\begin{array}{rr}
\frac{\omega^{2}}{\chi^{2}} & 0\\
0 & f^{2}-\omega^{2}
\end{array}\right],\label{eq:Lagdim2c}
\end{equation}
\begin{equation}
P_{\mathrm{TB}}=P_{\mathrm{TB}}\left(\omega\right)=\left[\begin{array}{rr}
C_{0}\frac{\omega_{0}^{2}-\omega^{2}}{\omega_{0}^{2}} & C_{0}b\\
-b\beta_{0} & -b^{2}\beta_{0}
\end{array}\right].\label{eq:Lagdim2d}
\end{equation}
\emph{Equations (\ref{eq:Lagdim2b})-(\ref{eq:Lagdim2d}) are evidently
the second-order vector ODE with spatially periodic frequency dependent
singular matrix potential $\sum_{\ell=-\infty}^{\infty}\delta\left(z-\ell\right)P_{\mathrm{TB}}\left(\omega\right)$.
These equations becomes is the object of our studies below.}

According to Appendix \ref{sec:dif-jord} the second-order differential
equation (\ref{eq:Lagdim2b})-(\ref{eq:Lagdim2d}) is equivalent to
the first-order spatially periodic differential equation of the form:
\begin{gather}
\partial_{z}X=A\left(z\right)X,\quad A\left(z\right)=A\left(z,\omega\right)=\left[\begin{array}{rr}
0 & \mathbb{I}\\
-A_{0}-P\left(z\right) & -A_{1}
\end{array}\right],\quad X=\left[\begin{array}{r}
x\\
\partial_{z}x
\end{array}\right],\label{eq:Lagdim2e}\\
P\left(z\right)=P\left(z,\omega\right)=\sum_{\ell=-\infty}^{\infty}\delta\left(z-\ell\right)P_{\mathrm{TB}}\left(\omega\right),\nonumber 
\end{gather}
where frequency dependent matrices $A_{0}$, $A_{1}$ and $P_{\mathrm{TB}}$
satisfy equations (\ref{eq:Lagdim2c}) and (\ref{eq:Lagdim2d}).

Using results of Appendix \ref{subsec:speHam} we find that the spatially
periodic vector ODE (\ref{eq:Lagdim2e}) is Hamiltonian with the following
choice of nonsingular Hermitian matrix $G=G\left(\omega\right)$:
\begin{equation}
G=G^{*}=\left[\begin{array}{rrrr}
0 & 0 & \mathrm{i} & 0\\
0 & -\frac{2\omega C_{0}}{\beta_{0}} & 0 & -\mathrm{i}\frac{C_{0}}{\beta_{0}}\\
-\mathrm{i} & 0 & 0 & 0\\
0 & \mathrm{i}\frac{C_{0}}{\beta_{0}} & 0 & 0
\end{array}\right],\quad\det\left\{ G\right\} =\frac{C_{0}^{2}}{\beta_{0}^{2}}.\label{eq:Lagdim3a}
\end{equation}
Indeed, it is an elementary exercise to verify that for each value
of $z$ matrix $A\left(z\right)$ is $G$-skew-Hermitian, that is
\begin{equation}
GA\left(z\right)+A^{*}\left(z\right)G=0.\label{eq:Lagdim3b}
\end{equation}
Then if $\Phi\left(z\right)$ is the matrizant of the Hamiltonian
equation (\ref{eq:Lagdim2e}) then according to results of Appendix
\ref{sec:Ham} $\Phi\left(z\right)$ is $G$-unitary matrix satisfying
\begin{equation}
\Phi^{*}\left(z\right)G\Phi\left(z\right)=G.\label{eq:Lagdim3c}
\end{equation}
 Consequently its spectrum $\sigma\left\{ \Phi\left(z\right)\right\} $
is invariant with respect to the inversion transformation $\zeta\rightarrow\frac{1}{\bar{\zeta}}$,
that is it is symmetric with respect to the unit circle: 
\begin{equation}
\zeta\in\sigma\left\{ \Phi\left(z\right)\right\} \Rightarrow\frac{1}{\bar{\zeta}}\in\sigma\left\{ \Phi\left(z\right)\right\} .\label{eq:Lagdim3d}
\end{equation}

\subsection{CCTWT subsystems: the coupled cavity structure and the e-beam\label{subsec:ccsebeam}}

It is instructive to take a view on the CCTWT system as a composition
of its integral components which are the \emph{coupled cavity structure
(CCS) and the electron beam (e-beam)}. It comes at no surprise that
special features of the CCS and the e-beam are manifested in fundamental
properties of the CCTWT justifying their thorough analysis. This section
provides the initial steps of the analysis whereas more detailed studies
of the CCS features are pursued in Section \ref{sec:ccs}.

One way to identify the CCS and the e-beam components of the CCTWT
is to use its analysis carried out in previous sections setting there
the coupling coefficient $b$ to be zero. With that in mind we consider
monodromy matrix matrix $\mathscr{T}$ defined by equations (\ref{eq:monS1c})-(\ref{eq:monS1e})
and set there $b=0$. To separate variables relevant to the CCS and
the e-beam we use permutation matrix $P_{23}$ defined by equation
(\ref{eq:ATB2sc}) and transform $\left.\mathscr{T}\right|_{b=0}$
as follows:
\begin{equation}
P_{23}\left.\mathscr{T}\right|_{b=0}P_{23}^{-1}=\left[\begin{array}{rr}
\mathscr{T}_{\mathrm{C}} & 0\\
0 & \mathscr{T}_{\mathrm{B}}
\end{array}\right],\label{eq:PTP1a}
\end{equation}
where $\mathscr{T}_{\mathrm{C}}$ and $\mathscr{T}_{\mathrm{B}}$
are $2\times2$ matrices defined by
\begin{equation}
\mathscr{T}_{\mathrm{C}}=\left[\begin{array}{rr}
\cos\left(\frac{{\it \omega}}{{\it \chi}}\right) & \frac{{\it \chi}}{{\it \omega}}\sin\left(\frac{{\it \omega}}{{\it \chi}}\right)\\
\left(1-\frac{\omega^{2}}{\omega_{0}^{2}}\right)C_{0}\cos\left(\frac{{\it \omega}}{{\it \chi}}\right)-\frac{{\it \omega}}{{\it \chi}}\sin\left(\frac{{\it \omega}}{{\it \chi}}\right) & \left(1-\frac{\omega^{2}}{\omega_{0}^{2}}\right)\frac{C_{0}{\it \chi}}{{\it \omega}}\sin\left(\frac{{\it \omega}}{{\it \chi}}\right)+\cos\left(\frac{{\it \omega}}{{\it \chi}}\right)
\end{array}\right],\label{eq:PTP1b}
\end{equation}
\begin{equation}
\mathscr{T}_{\mathrm{B}}={\it {\rm e}^{\mathrm{i}\omega}}\left[\begin{array}{rr}
\cos\left(f\right)-\mathrm{i}{\it \frac{\sin\left(f\right)}{f}\omega} & \frac{\sin\left(f\right){\it {\rm e}^{\mathrm{i}\omega}}}{f}\\
\frac{\left({\it \omega}^{2}-f^{2}\right)\sin\left(f\right)}{f} & \left(\cos\left(f\right)+\mathrm{i}{\it \frac{\sin\left(f\right)}{f}\omega}\right){\it {\rm e}^{\mathrm{i}\omega}}
\end{array}\right].\label{eq:PTP1c}
\end{equation}
Evidently $\mathscr{T}_{\mathrm{C}}$ defined by equation (\ref{eq:PTP1b})
is the CCS monodromy matrix and $\mathscr{T}_{\mathrm{B}}$ defined
by equation (\ref{eq:PTP1c}) is the e-beam monodromy matrix.

It follows from equation (\ref{eq:PTP1b}) for matrix $\mathscr{T}_{\mathrm{C}}$
that the corresponding characteristic equation $\det\left\{ \mathscr{T}_{\mathrm{C}}-s\mathbb{I}\right\} =0$
for the Floquet multipliers $s$ (see Appendix \ref{sec:floquet})
is
\begin{gather}
\det\left\{ \mathscr{T}_{\mathrm{C}}-s\mathbb{I}\right\} =s^{2}+2W_{\mathrm{C}}\left(\omega\right)s+1=0,\quad s=\exp\left\{ \mathrm{i}k\right\} ,\label{eq:PTP1d}\\
W_{\mathrm{C}}\left(\omega\right)=\frac{C_{0}}{2}\left(\frac{\omega}{\omega_{0}^{2}}-\frac{1}{\omega}\right)\chi\sin\left(\frac{\omega}{\chi}\right)-\cos\left(\omega\right).\nonumber 
\end{gather}
Real-valued function $W_{\mathrm{C}}\left(\omega\right)$ in the second
equation in (\ref{eq:PTP1d}) plays an important role in the analysis
of the CCS and its plot is depicted in Fig. \ref{fig:dis-ccs-wom2}(b).
We refer to $W_{\mathrm{C}}\left(\omega\right)$ as the \emph{CCS
instability parameter} for as we will find that it completely determines
if the Floquet multipliers satisfy the instability criterion $\left|s\right|>1$.

It also follows from equation (\ref{eq:PTP1c}) for matrix $\mathscr{T}_{\mathrm{B}}$
that the corresponding characteristic equation $\det\left\{ \mathscr{T}_{\mathrm{B}}-s\mathbb{I}\right\} =0$
for the Floquet multipliers $s$ is
\begin{equation}
\det\left\{ \mathscr{T}_{\mathrm{B}}-s\mathbb{I}\right\} =s^{2}-2\cos\left(f\right){\it {\rm e}^{\mathrm{i}\omega}}s+{\it {\rm e}^{2\mathrm{i}\omega}}=0,\quad s=\exp\left\{ \mathrm{i}k\right\} .\label{eq:PTP1e}
\end{equation}
In view of equations (\ref{eq:PTP1a}), (\ref{eq:PTP1d}) and (\ref{eq:PTP1e})
the following factorization holds for the characteristic function
of the monodromy matrix $\left.\mathscr{T}\right|_{b=0}$ of the decoupled
system 
\begin{equation}
\det\left\{ \left.\mathscr{T}\right|_{b=0}-s\mathbb{I}\right\} =\left(s^{2}+2W_{\mathrm{C}}\left(\omega\right)s+1\right)\left(s^{2}-2\cos\left(f\right){\it {\rm e}^{\mathrm{i}\omega}}s+{\it {\rm e}^{2\mathrm{i}\omega}}\right).\label{eq:PTP1f}
\end{equation}

\section{Solutions to the coupled-cavity TWT equations\label{sec:cctwtsol}}

Dimensionless form of the EL equations (\ref{eq:Lagdim2b})-(\ref{eq:Lagdim2d})
and their solutions can be analyzed by applying the Floquet theory
reviewed in Appendix \ref{sec:floquet}. To use the Floquet theory
we recast first the second-order vector ODE as the first-order vector
ODE following to our review on subject in Appendix \ref{sec:dif-jord}.

\subsection{Solutions to the Euler-Lagrange equation inside the period}

We begin with introducing expressions for (i) the characteristic scalar
polynomial $A_{\mathrm{T}}\left(s\right)$ associated with the first
equation in (\ref{eq:Lagdim2a}) for the TL; and (ii) the characteristic
scalar polynomial $A_{\mathrm{B}}\left(s\right)$ associated with
the second equation in (\ref{eq:Lagdim2a}) for the e-beam:

\begin{equation}
A_{\mathrm{T}}\left(s\right)=\frac{{\it \omega}^{2}}{{\it \chi}^{2}}+s^{2},\quad A_{\mathrm{B}}=s^{2}-2\mathrm{i}{\it \omega}s+f^{2}-{\it {\it \omega}^{2}}.\label{eq:ATBs1a}
\end{equation}
Note that in the accordance with the general theory of differential
equations (see Appendices \ref{sec:co-mat}, \ref{sec:mat-poly} and
\ref{sec:dif-jord}) the spectral parameter $s$ in expressions for
the characteristic polynomials $A_{\mathrm{T}}\left(s\right)$ and
$A_{\mathrm{B}}\left(s\right)$ represents symbolically the differential
operator $\partial_{z}$.

The $2\times2$ companion matrix $\mathscr{C}_{\mathrm{T}}$ of the
scalar characteristic polynomial $A_{\mathrm{T}}\left(s\right)$ (see
Appendices \ref{sec:co-mat}, \ref{sec:mat-poly} and \ref{sec:dif-jord})
is
\begin{equation}
\mathscr{C}_{\mathrm{T}}=\left[\begin{array}{rr}
0 & 1\\
-\frac{{\it \omega}^{2}}{{\it \chi}^{2}} & 0
\end{array}\right]=\mathscr{Z}_{\mathrm{T}}\left[\begin{array}{rr}
\mathrm{i}\frac{{\it \omega}}{{\it \chi}} & 0\\
0 & -\mathrm{i}\frac{{\it \omega}}{{\it \chi}}
\end{array}\right]\mathscr{Z}_{\mathrm{T}}^{-1},\quad\mathscr{Z}_{\mathrm{T}}=\left[\begin{array}{rr}
-\mathrm{i}\frac{{\it \chi}}{{\it \omega}} & \mathrm{i}\frac{{\it \chi}}{{\it \omega}}\\
1 & 1
\end{array}\right],\quad x=\left[\begin{array}{r}
Q\\
\partial_{z}Q
\end{array}\right],\label{eq:ATBs1b}
\end{equation}
where the columns of matrix $\mathscr{Z}_{\mathrm{T}}$ are eigenvectors
of the companion matrix $\mathscr{C}_{\mathrm{T}}$ with the corresponding
eigenvalues being the relevant entries of the diagonal matrix in equations
(\ref{eq:ATBs1b}). Expression of vector $x$ in equations (\ref{eq:ATBs1b})
clarifies the meaning of the entries of relevant matrices. Consequently,
the exponent $\exp\left\{ z\mathscr{C}_{\mathrm{T}}\right\} $ which
is the fundamental matrix solution to the first-order ODE associated
with first equation in (\ref{eq:Lagdim2a}) satisfies
\begin{equation}
\exp\left\{ z\mathscr{C}_{\mathrm{T}}\right\} =\left[\begin{array}{rr}
\cos\left(\frac{{\it \omega z}}{{\it \chi}}\right) & \frac{{\it \chi}}{{\it \omega}}\sin\left(\frac{{\it \omega z}}{{\it \chi}}\right)\\
-\frac{{\it \omega}}{{\it \chi}}\sin\left(\frac{{\it \omega}}{{\it \chi}}\right) & \cos\left(\frac{z{\it \omega}}{{\it \chi}}\right)
\end{array}\right]=\mathscr{Z}_{\mathrm{T}}\exp\left\{ \left[\begin{array}{rr}
\mathrm{i}\frac{{\it \omega z}}{{\it \chi}} & 0\\
0 & -\mathrm{i}\frac{{\it \omega z}}{{\it \chi}}
\end{array}\right]\right\} \mathscr{Z}_{\mathrm{T}}^{-1}.\label{eq:ATBs1c}
\end{equation}
The $2\times2$ companion matrix $\mathscr{C}_{\mathrm{B}}$ of the
scalar characteristic polynomial $A_{\mathrm{B}}\left(s\right)$
\begin{gather}
\mathscr{C}_{\mathrm{B}}=\left[\begin{array}{rr}
0 & 1\\
{\it {\it \omega}^{2}}-f^{2} & 2\mathrm{i}{\it \omega}
\end{array}\right]=\mathscr{Z}_{\mathrm{B}}\left[\begin{array}{rr}
\mathrm{i}\left(\omega-f\right) & 0\\
0 & \mathrm{i}\left(\omega+f\right)
\end{array}\right]\mathscr{Z}_{\mathrm{B}}^{-1},\quad\mathscr{Z}_{\mathrm{B}}=\left[\begin{array}{rr}
-\frac{\mathrm{i}}{\omega-f} & -\frac{\mathrm{i}}{\omega+f}\\
1 & 1
\end{array}\right],\label{eq:ATBs1d}\\
x=\left[\begin{array}{r}
q\\
\partial_{z}q
\end{array}\right],\nonumber 
\end{gather}
where the columns of matrix $\mathscr{Z}_{\mathrm{B}}$ are eigenvectors
of the companion matrix $\mathscr{C}_{\mathrm{B}}$ with the corresponding
eigenvalues being the relevant entries of the diagonal matrix in equations
(\ref{eq:ATBs1d}). Expression of vector $x$ in equations (\ref{eq:ATBs1d})
clarifies the meaning of the entries of relevant matrices. Consequently,
the exponent $\exp\left\{ z\mathscr{C}_{\mathrm{B}}\right\} $ which
is the fundamental matrix solution to the first-order ODE associated
with the second equation in (\ref{eq:Lagdim2a}) satisfies
\begin{gather}
\exp\left\{ z\mathscr{C}_{\mathrm{B}}\right\} =\frac{1}{f}\exp\left\{ \mathrm{i}\omega z\right\} \left[\begin{array}{rr}
f\cos\left(fz\right)-\mathrm{i}{\it \omega}\sin\left(fz\right) & \sin\left(fz\right)\\
\left({\it {\it \omega}^{2}}-f^{2}\right)\sin\left(fz\right) & f\cos\left(fz\right)+\mathrm{i}{\it \omega}\sin\left(fz\right)
\end{array}\right]=\label{eq:ATBs1e}\\
=\mathscr{Z}_{\mathrm{B}}\exp\left\{ \left[\begin{array}{rr}
\mathrm{i}\left(\omega-f\right)z & 0\\
0 & \mathrm{i}\left(\omega+f\right)z
\end{array}\right]\right\} \mathscr{Z}_{\mathrm{B}}^{-1}.\nonumber 
\end{gather}
The $2\times2$ matrix characteristic polynomial $A_{\mathrm{TB}}\left(s\right)$
of \emph{non-interacting} $TL$ and the e-beam is the following diagonal
matrix polynomial
\begin{equation}
A_{\mathrm{TB}}\left(s\right)=\left[\begin{array}{rr}
A_{\mathrm{T}}\left(s\right) & 0\\
0 & A_{\mathrm{B}}\left(s\right)
\end{array}\right]=\left[\begin{array}{rr}
\frac{{\it \omega}^{2}}{{\it \chi}^{2}}+s^{2} & 0\\
0 & s^{2}-2\mathrm{i}{\it \omega}s+f^{2}-{\it {\it \omega}^{2}}
\end{array}\right].\label{eq:ATBs1f}
\end{equation}
The $4\times4$ companion matrix of the matrix polynomial $A_{\mathrm{TB}}\left(s\right)$
$4\times4$ is (see Appendices \ref{sec:co-mat}, \ref{sec:mat-poly}
and \ref{sec:dif-jord})
\begin{equation}
\mathscr{C}_{\mathrm{TB}}=\left[\begin{array}{rrrr}
0 & 0 & 1 & 0\\
0 & 0 & 0 & 1\\
-\frac{{\it \omega}^{2}}{{\it \chi}^{2}} & 0 & 0 & 0\\
0 & {\it {\it \omega}^{2}}-f^{2} & 0 & 2\mathrm{i}{\it \omega}
\end{array}\right],\quad X=\left[\begin{array}{r}
Q\\
q\\
\partial_{z}Q\\
\partial_{z}q
\end{array}\right],\label{eq:ATB1s1g}
\end{equation}
where vector $X$ clarifies the meaning of the entries of matrix $\mathscr{C}_{\mathrm{TB}}$.

The exponential $\exp\left\{ z\mathscr{C}_{\mathrm{TB}}\right\} $
of the component matrix $\mathscr{C}_{\mathrm{TB}}$ is the fundamental
matrix solution to the first-order ODE associated with the system
of equations (\ref{eq:Lagdim2a}) and it satisfies
\begin{gather}
\exp\left\{ z\mathscr{C}_{\mathrm{TB}}\right\} =\label{eq:ATBs2a}\\
\left[\begin{array}{rrrr}
\cos\left(\frac{{\it \omega z}}{{\it \chi}}\right) & 0 & \frac{{\it \chi}}{{\it \omega}}\sin\left(\frac{{\it \omega z}}{{\it \chi}}\right) & 0\\
0 & e^{\mathrm{i}\omega z}\left[\cos\left(fz\right)-{\it \mathrm{i}\frac{\omega}{f}}\sin\left(fz\right)\right] & 0 & \frac{1}{f}e^{\mathrm{i}\omega}\sin\left(fz\right)\\
-\frac{{\it \omega}}{{\it \chi}}\sin\left(\frac{{\it \omega z}}{{\it \chi}}\right) & 0 & \cos\left(\frac{{\it \omega z}}{{\it \chi}}\right) & 0\\
0 & \frac{1}{f}e^{\mathrm{i}\omega z}\left({\it {\it \omega}^{2}}-f^{2}\right)\sin\left(zf\right) & 0 & e^{\mathrm{i}\omega z}\left[\cos\left(fz\right)-{\it \mathrm{i}\frac{\omega}{f}}\sin\left(fz\right)\right]
\end{array}\right].\nonumber 
\end{gather}
In particular for $z=1,$ which is the period in dimensionless variables,
we get
\begin{gather}
\exp\left\{ \mathscr{C}_{\mathrm{TB}}\right\} =\label{eq:ATBs2b}\\
\left[\begin{array}{rrrr}
\cos\left(\frac{{\it \omega}}{{\it \chi}}\right) & 0 & \frac{{\it \chi}}{{\it \omega}}\sin\left(\frac{{\it \omega}}{{\it \chi}}\right) & 0\\
0 & e^{\mathrm{i}\omega}\left[\cos\left(f\right)-{\it \mathrm{i}\frac{\omega}{f}}\sin\left(f\right)\right] & 0 & \frac{1}{f}e^{\mathrm{i}\omega}\sin\left(f\right)\\
-\frac{{\it \omega}}{{\it \chi}}\sin\left(\frac{{\it \omega}}{{\it \chi}}\right) & 0 & \cos\left(\frac{{\it \omega}}{{\it \chi}}\right) & 0\\
0 & \frac{1}{f}e^{\mathrm{i}\omega}\left({\it {\it \omega}^{2}}-f^{2}\right)\sin\left(f\right) & 0 & e^{\mathrm{i}\omega}\left[\cos\left(f\right)-{\it \mathrm{i}\frac{\omega}{f}}\sin\left(f\right)\right]
\end{array}\right]\nonumber 
\end{gather}
Using $4\times4$ matrix $P_{23}$ that permutes the second and the
third coordinates, that is
\begin{equation}
P_{23}=\left[\begin{array}{rrrr}
1 & 0 & 0 & 0\\
0 & 0 & 1 & 0\\
0 & 1 & 0 & 0\\
0 & 0 & 0 & 1
\end{array}\right]=P_{23}^{-1},\label{eq:ATB2sc}
\end{equation}
we get the following representation of $\exp\left\{ z\mathscr{C}_{\mathrm{TB}}\right\} $
in terms and $\exp\left\{ z\mathscr{C}_{\mathrm{T}}\right\} $ and
$\exp\left\{ z\mathscr{C}_{\mathrm{B}}\right\} $:
\begin{equation}
\exp\left\{ z\mathscr{C}_{\mathrm{TB}}\right\} =P_{23}\left[\begin{array}{rr}
\exp\left\{ z\mathscr{C}_{\mathrm{T}}\right\}  & 0\\
0 & \exp\left\{ z\mathscr{C}_{\mathrm{B}}\right\} 
\end{array}\right]P_{23}^{-1}.\label{eq:ATBs2d}
\end{equation}
One can also verify the correctness of the identity (\ref{eq:ATBs2d})
by tedious by straightforward evaluation.

\subsection{The boundary conditions and the monodromy matrix in dimensionless
variables\label{subsec:mon-mat4}}

The fundamental matrix solution $\exp\left\{ z\mathscr{C}_{\mathrm{TB}}\right\} $
represented by equation (\ref{eq:ATBs2d}) provides the solution of
the relevant first-order vector ODE strictly inside the period $\left(0,1\right)$.
The complete solution on the interval $\left(0,1\right]$ that includes
the boundary $1$ has to account for the boundary jump conditions
represented by equations (\ref{eq:Lagdim1f}). These boundary conditions
are equivalent to
\begin{equation}
X\left(\ell+0\right)=\mathsf{S}_{\mathrm{b}}X\left(\ell-0\right),\quad\mathsf{S}_{\mathrm{b}}=\left[\begin{array}{rr}
\mathbb{I} & 0\\
P_{\mathrm{TB}} & \mathbb{I}
\end{array}\right],\quad X=\left[\begin{array}{r}
x\\
\partial_{z}x
\end{array}\right],\quad x=\left[\begin{array}{r}
\check{Q}\\
\check{q}
\end{array}\right],\label{eq:bmonX1a}
\end{equation}
where the\emph{ boundary matrix} $\mathsf{S}_{\mathrm{b}}$ satisfies
\begin{equation}
\mathsf{S}_{\mathrm{b}}=\left[\begin{array}{rr}
\mathbb{I} & 0\\
P_{\mathrm{TB}} & \mathbb{I}
\end{array}\right]=\left[\begin{array}{rrrr}
1 & 0 & 0 & 0\\
0 & 1 & 0 & 0\\
\left(1-\frac{\omega^{2}}{\omega_{0}^{2}}\right)C_{0} & C_{0}b & 1 & 0\\
-b\beta_{0} & -b^{2}\beta_{0} & 0 & 1
\end{array}\right].\label{eq:bmonX1b}
\end{equation}
The Floquet theory (see Appendix \ref{sec:floquet}) when applied
to the first-order ODE equivalent to the EL equations (\ref{eq:Lagdim2b})-(\ref{eq:Lagdim2d})
yields the the following equations for the fundamental $4\times4$
matrix solution $\Phi\left(z\right)$:
\begin{gather}
\Phi\left(z\right)=\exp\left\{ \left(z-\ell\right)\mathscr{C}_{\mathrm{TB}}\right\} \Phi\left(\ell+0\right),\quad\ell<z<\ell+1,\label{eq:bmonX1c}\\
\Phi\left(0+0\right)=\mathbb{I},\quad\Phi\left(\ell+0\right)=\mathsf{S}_{\mathrm{b}}\Phi\left(\ell-0\right),\quad\ell\in\mathbb{Z},\nonumber 
\end{gather}
where $\exp\left\{ z\mathscr{C}_{\mathrm{TB}}\right\} $ is defined
by equation (\ref{eq:ATBs2a}). In particular, according to the Floquet
theorem \ref{thm:floquet} the monodromy matrix is $\mathscr{T}=\Phi\left(1+0\right)$
and in view of equations (\ref{eq:bmonX1b}) and (\ref{eq:bmonX1c})
we have
\begin{equation}
\mathscr{T}=\Phi\left(1+0\right)=\mathsf{S}_{\mathrm{b}}\exp\left\{ \mathscr{C}_{\mathrm{TB}}\right\} =\left[\begin{array}{rr}
\mathscr{T}_{11} & \mathscr{T}_{12}\\
\mathscr{T}_{21} & \mathscr{T}_{22}
\end{array}\right],\label{eq:monS1b}
\end{equation}
where $\exp\left\{ \mathscr{C}_{\mathrm{TB}}\right\} $ is defined
by equation (\ref{eq:ATBs2b}) and $2\times2$ matrix blocks of $\mathscr{T}$
are as follows
\begin{equation}
\mathscr{T}_{11}=\left[\begin{array}{cc}
\cos\left(\frac{{\it \omega}}{{\it \chi}}\right) & 0\\
0 & \left(\cos\left(f\right)-\mathrm{i}{\it \frac{\sin\left(f\right)}{f}\omega}\right){\it {\rm e}^{\mathrm{i}\omega}}
\end{array}\right],\quad\mathscr{T}_{12}=\left[\begin{array}{cc}
\frac{{\it \chi}}{{\it \omega}}\sin\left(\frac{{\it \omega}}{{\it \chi}}\right) & 0\\
0 & \frac{\sin\left(f\right){\it {\rm e}^{\mathrm{i}\omega}}}{f}
\end{array}\right],\label{eq:monS1c}
\end{equation}
\begin{equation}
\mathscr{T}_{21}=\left[\begin{array}{cc}
\left(1-\frac{\omega^{2}}{\omega_{0}^{2}}\right)C_{0}\cos\left(\frac{{\it \omega}}{{\it \chi}}\right)-\frac{{\it \omega}}{{\it \chi}}\sin\left(\frac{{\it \omega}}{{\it \chi}}\right) & C_{0}b\left(\cos\left(f\right)-\mathrm{i}{\it \frac{\sin\left(f\right)}{f}\omega}\right){\it {\rm e}^{\mathrm{i}\omega}}\\
-\beta_{0}b\cos\left(\frac{{\it \omega}}{{\it \chi}}\right) & \left[\beta_{0}b^{2}\left(\mathrm{i}{\it \frac{\sin\left(f\right)}{f}\omega-\cos\left(f\right)}\right)+\frac{\left({\it \omega}^{2}-f^{2}\right)\sin\left(f\right)}{f}\right]{\it {\rm e}^{\mathrm{i}\omega}}
\end{array}\right],\label{eq:monS1d}
\end{equation}
\begin{equation}
\mathscr{T}_{22}=\left[\begin{array}{cc}
\left(1-\frac{\omega^{2}}{\omega_{0}^{2}}\right)\frac{C_{0}{\it \chi}}{{\it \omega}}\sin\left(\frac{{\it \omega}}{{\it \chi}}\right)+\cos\left(\frac{{\it \omega}}{{\it \chi}}\right) & \frac{C_{0}b\sin\left(f\right){\it {\rm e}^{\mathrm{i}\omega}}}{f}\\
-\frac{\beta_{0}b\chi}{{\it \omega}}\sin\left(\frac{{\it \omega}}{{\it \chi}}\right) & \left[\left(\mathrm{i}{\it \frac{\sin\left(f\right)}{f}\omega+\cos\left(f\right)}\right)-\frac{\beta_{0}b^{2}\sin\left(f\right)}{f}\right]{\it {\rm e}^{\mathrm{i}\omega}}
\end{array}\right],\label{eq:monS1e}
\end{equation}
and the involved above dimensionless constants satisfy (see Section
\ref{subsec:dim-par})
\begin{equation}
\chi=\frac{w}{\mathring{v}}=\frac{1}{\mathring{v}\sqrt{CL}},\quad C_{0}=\frac{aC}{c_{0}},\quad\beta_{0}=\frac{a\beta}{c_{0}\mathring{v}^{2}},\quad f=\frac{aR_{\mathrm{sc}}\omega_{\mathrm{p}}}{\mathring{v}}=\frac{R_{\mathrm{sc}}\omega_{\mathrm{p}}}{\omega_{a}}.\label{eq:monS1g}
\end{equation}

An important special and simpler case of the monodromy matrix $\mathscr{T}$
defined by equations (\ref{eq:monS1c})-(\ref{eq:monS1e}) is when
the following equations hold
\begin{equation}
\chi=1,\quad b=1.\label{eq:monS2a}
\end{equation}
The condition $\chi=1$ according to equations (\ref{eq:unitL1a})
is equivalent to the equality of the phase velocity $w$ associated
with TL and the e-beam stationary flow velocity $\mathring{v}$. Equation
$b=1$ signifies the maximal coupling between the TL and e-beam at
interaction points. In this case the monodromy matrix $\mathscr{T}$
turns into
\begin{equation}
\mathscr{T}_{1}=\left[\begin{array}{rr}
\mathscr{T}_{11} & \mathscr{T}_{12}\\
\mathscr{T}_{21} & \mathscr{T}_{22}
\end{array}\right],\quad\chi=1,\quad b=1,\label{eq:monS2b}
\end{equation}
where
\begin{equation}
\mathscr{T}_{11}=\left[\begin{array}{cc}
\cos\left({\it \omega}\right) & 0\\
0 & \left(\cos\left(f\right)-\mathrm{i}{\it \frac{\sin\left(f\right)}{f}\omega}\right){\it {\rm e}^{\mathrm{i}\omega}}
\end{array}\right],\quad\mathscr{T}_{12}=\left[\begin{array}{cc}
\frac{\sin\left({\it \omega}\right)}{{\it \omega}} & 0\\
0 & \frac{\sin\left(f\right){\it {\rm e}^{\mathrm{i}\omega}}}{f}
\end{array}\right],\label{eq:monS2c}
\end{equation}
\begin{equation}
\mathscr{T}_{21}=\left[\begin{array}{cc}
{\it \left(1-\frac{\omega^{2}}{\omega_{0}^{2}}\right)C_{0}}\cos\left({\it \omega}\right)-{\it \omega}\sin\left({\it \omega}\right) & {\it C_{0}}\left(\cos\left(f\right)-\mathrm{i}{\it \frac{\sin\left(f\right)}{f}\omega}\right){\it {\rm e}^{\mathrm{i}\omega}}\\
-{\it \beta_{0}}\cos\left({\it \omega}\right) & \left[{\it \beta_{0}}\left(\mathrm{i}{\it \frac{\sin\left(f\right)}{f}\omega-\cos\left(f\right)}\right)+\frac{\left({\it \omega}^{2}-f^{2}\right)\sin\left(f\right)}{f}\right]{\it {\rm e}^{\mathrm{i}\omega}}
\end{array}\right],\label{eq:monS2d}
\end{equation}
\begin{equation}
\mathscr{T}_{22}=\left[\begin{array}{cc}
\left(1-\frac{\omega^{2}}{\omega_{0}^{2}}\right)\frac{{\it C_{0}}\sin\left({\it \omega}\right)}{{\it \omega}}+\cos\left({\it \omega}\right) & \frac{{\it C_{0}}\sin\left(f\right){\it {\rm e}^{\mathrm{i}\omega}}}{f}\\
-\frac{{\it \beta_{0}}\sin\left({\it \omega}\right)}{{\it \omega}} & \left[\left(\mathrm{i}{\it \frac{\sin\left(f\right)}{f}\omega+\cos\left(f\right)}\right)-\frac{{\it \beta_{0}}\sin\left(f\right)}{f}\right]{\it {\rm e}^{\mathrm{i}\omega}}
\end{array}\right],\label{eq:monS2e}
\end{equation}
where dimensionless constants ${\it C_{0}}$, ${\it \beta_{0}}$ and
$f$ satisfy equations (\ref{eq:monS1g}).

\section{Characteristic equation, the Floquet multipliers and the dispersion
relation\label{sec:floqmul}}

We turn now to the four \emph{Floquet multipliers} $s$ which are
the eigenvalues of the CCTWT monodromy matrix $\mathscr{T}$ defined
by equations (\ref{eq:monS1b})-(\ref{eq:monS1e}). Consequently,
$s$ are solutions to \emph{characteristic equation} $\det\left\{ \mathscr{T}-s\mathbb{I}\right\} =0$
which is the following polynomial equation of the order 4:

\begin{equation}
S^{4}+c_{3}S^{3}+c_{2}S^{2}+\bar{c}_{3}S+1=0,\quad s=\exp\left\{ \mathrm{i}k\right\} ,\quad S=s{\rm e}^{-\mathrm{i}{\it \frac{\omega}{2}}}=\exp\left\{ \mathrm{i}\left(k-{\it \frac{\omega}{2}}\right)\right\} ,\label{eq:monS2f}
\end{equation}
where $k$ is the wavenumber that can be real or complex-valued and
coefficients $c_{3}$ and $c_{2}$ satisfy
\begin{equation}
c_{3}=2{\rm e}^{\frac{\mathrm{i}}{2}\omega}b_{f}^{\infty}+{\rm e}^{-\frac{\mathrm{i}}{2}\omega}\left[C_{0}\frac{\chi}{\omega}\sin\left(\frac{\omega}{\chi}\right)\left(\frac{\omega^{2}}{\omega_{0}^{2}}-1\right)-2\cos\left(\frac{\omega}{\chi}\right)\right],\label{eq:monS2g}
\end{equation}
\begin{equation}
c_{2}=2b_{f}^{\infty}\left[C_{0}\frac{\chi}{\omega}\sin\left(\frac{\omega}{\chi}\right)\frac{\omega^{2}}{\omega_{0}^{2}}-2\,\cos\left(\frac{\omega}{\chi}\right)\right]+2\cos\left(f\right)C_{0}\frac{\chi}{\omega}\sin\left(\frac{\omega}{\chi}\right)+2\,\cos\left({\it \omega}\right),\label{eq:monS2h}
\end{equation}
where
\begin{equation}
b_{f}^{\infty}=K_{0}\sin\left(f\right)-\cos\left(f\right),\quad K_{0}=\frac{b^{2}\beta_{0}}{2f}=\frac{b^{2}g_{\mathrm{B}}}{c_{0}},\quad g_{\mathrm{B}}=\frac{\sigma_{\mathrm{B}}}{4\lambda_{\mathrm{rp}}}.\label{eq:monS2j}
\end{equation}
Note that the quantity $b_{f}^{\infty}$ in equation (\ref{eq:monS2j})
arises also in the theory of the MCK reviewed in Section \ref{sec:mckrev}
(see equation \ref{eq:bfKinf1s}). Note also that coefficient $c_{2}$
defined by equation (\ref{eq:monS2h}) is manifestly real. The utility
of representing the Floquet multipliers $s$ in the form $s=S{\rm e}^{\mathrm{i}{\it \frac{\omega}{2}}}$
is explained by the fact that equation (\ref{eq:monS2f}) for $S$
possesses a manifest symmetry: if $S$ is a solution to equation (\ref{eq:monS2f})
then $\frac{1}{\bar{S}}$ is its solution as well. The forth-order
polynomials carrying this special symmetry are considered in Section
\ref{sec:degpol4}.

It what follows to simplify analytical evaluations we make the following
assumption.

\begin{assumption} (exact synchronism) \label{ass:cavcoup}. To assure
efficient cavity coupling we assume the so-called exact synchronism
condition, that is $\chi=1$ meaning that TL velocity $w$ equals
exactly to the e-beam stationary velocity $\mathring{v}$, namely
$w=\mathring{v}$. It is also convenient to choose frequency units
so that $\omega_{0}=1$. Combining these two conditions we assume
\begin{equation}
\chi=1,\quad\omega_{0}=1.\label{eq:ccTL2f}
\end{equation}

\end{assumption}
\begin{figure}[h]
\begin{centering}
\includegraphics[scale=0.12]{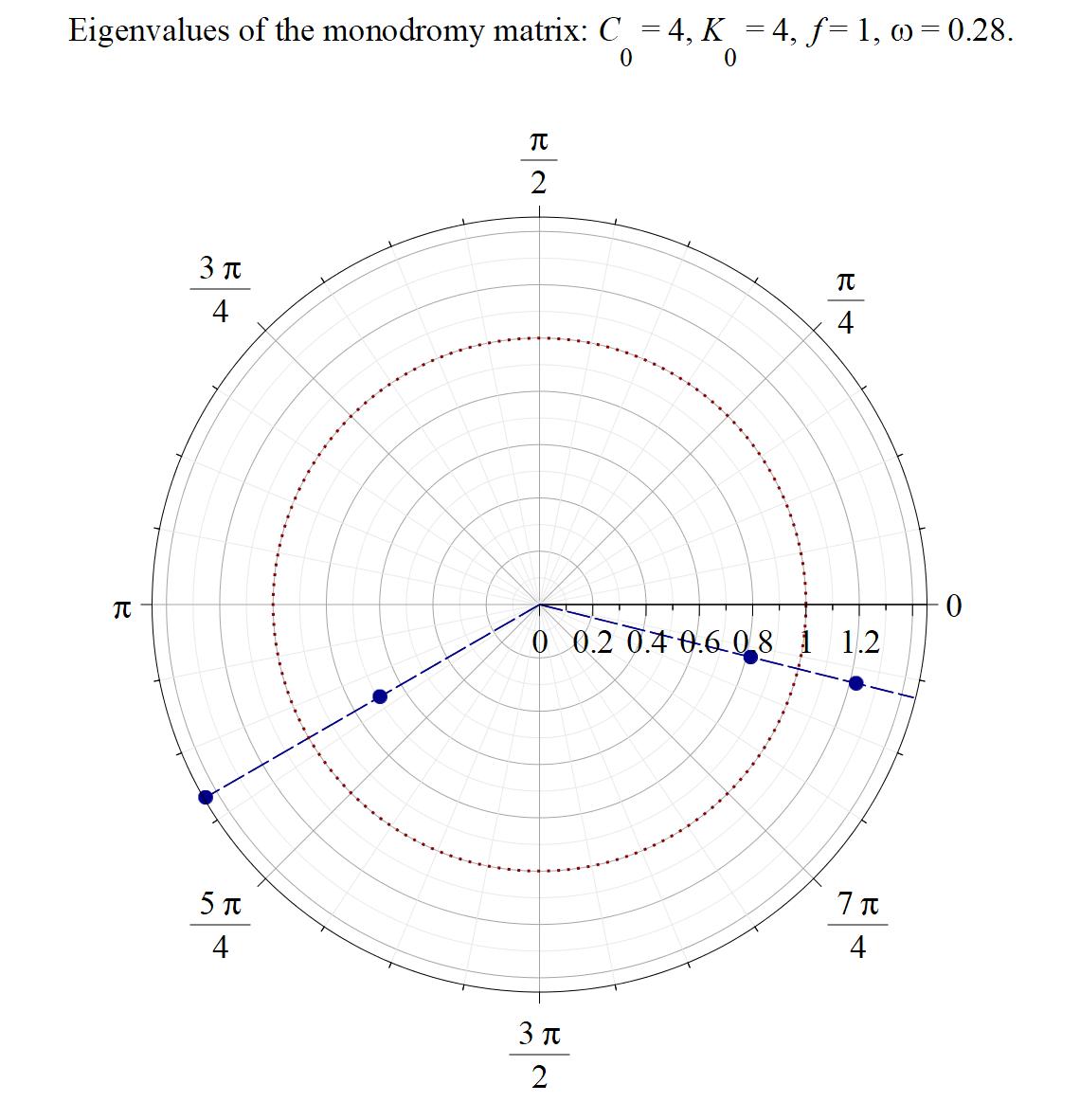}\hspace{0.2cm}\includegraphics[scale=0.115]{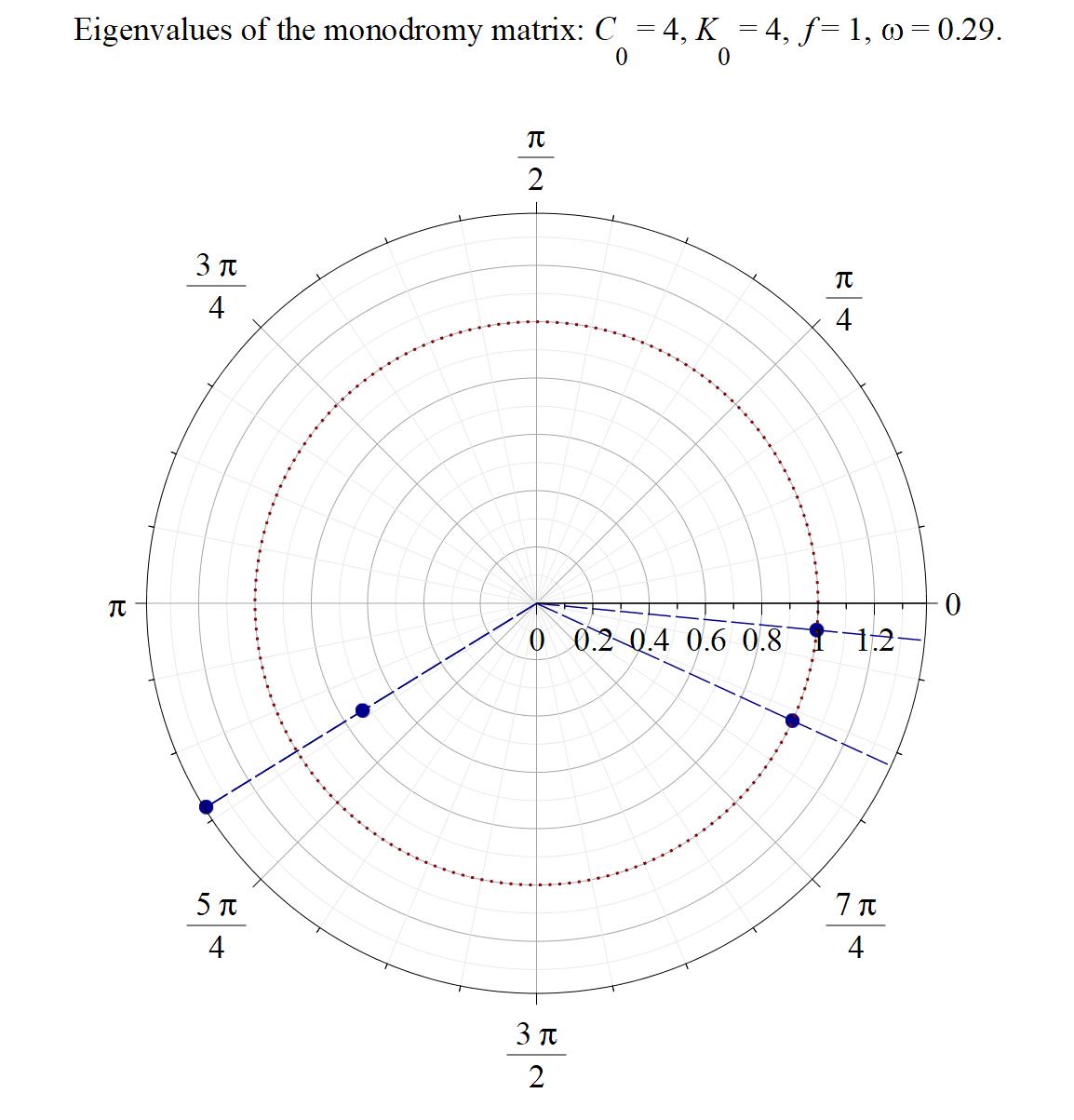}\hspace{0.2cm}\includegraphics[scale=0.115]{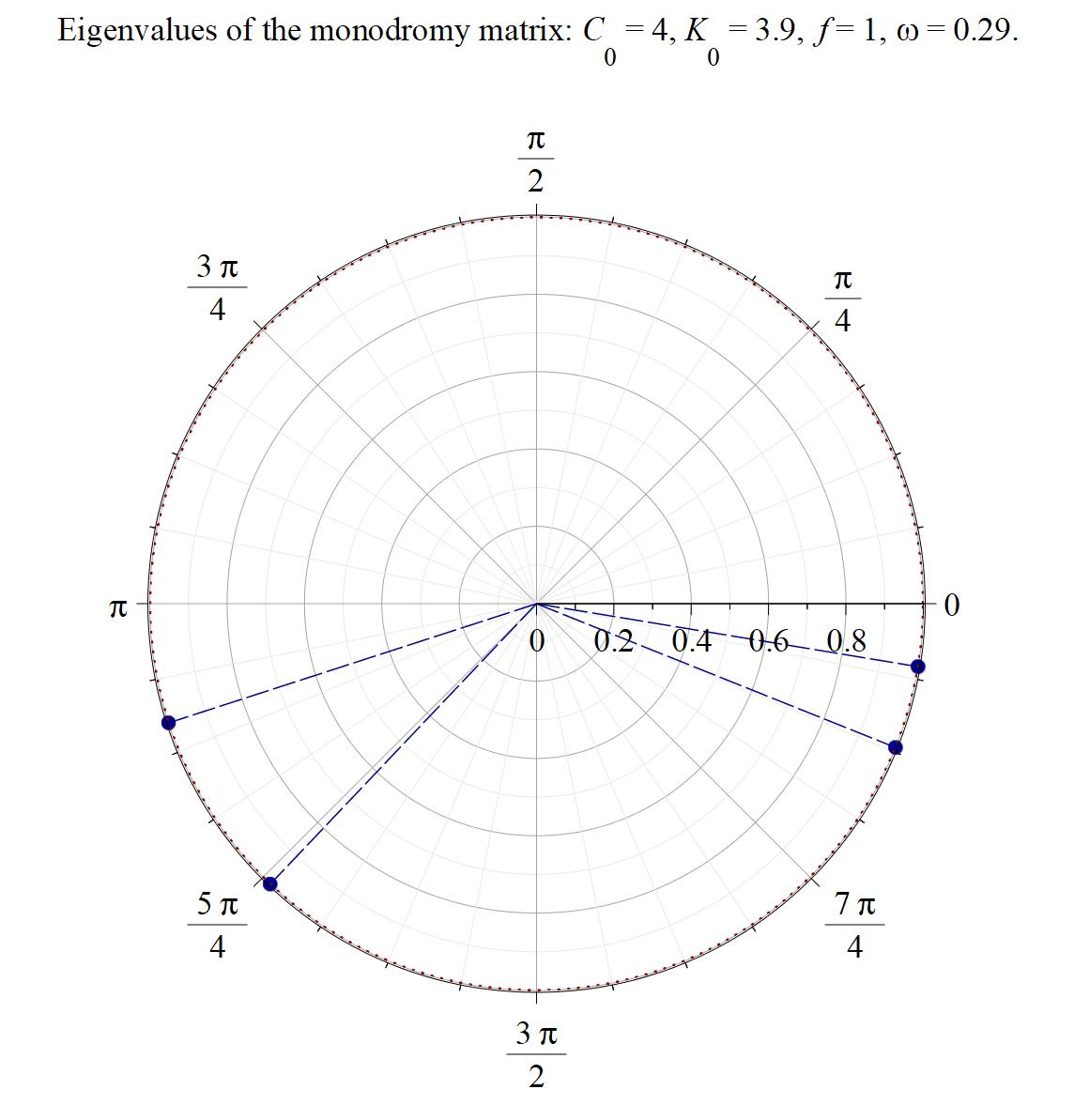}
\par\end{centering}
\centering{}(a)\hspace{4.3cm}(b)\hspace{4.3cm}(c)\caption{\label{fig:eig-mon-serp} The plots of the four complex eigenvalues
(the Floquet multipliers) $s=Se^{\mathrm{i}\frac{\omega}{2}}$ that
solve the characteristic equation (\ref{eq:monS2f}) for the monodromy
matrix $\mathscr{T}_{1}$ defined by equations (\ref{eq:monS2b})-(\ref{eq:monS2e})
in case when $\chi=1$, $\omega_{0}=1$, $f=1$, and $C_{0}=4$: (a)
$K_{0}=4$, $\omega=0.28$; (b) $K_{0}=4$, $\omega=0.29$ ; (c) $K_{0}=3.9,\omega=0.29$.
The horizontal and vertical axes represent respectively $\Re\left\{ s\right\} $
and $\Im\left\{ s\right\} $. The eigenvalues are shown by solid (blue)
dots.}
\end{figure}

\section{Dispersion relations\label{sec:disprel}}

The CCTWT evolution is governed by spatially periodic ODE (\ref{eq:Lagdim2b})-(\ref{eq:Lagdim2e})
implying that the dispersion relations as the relations between frequency
$\omega$ and wavenumber $k$ is constructed based on the Floquet
theory reviewed in Appendix \ref{sec:floquet}. Specifically in view
of the relation $s=\exp\left\{ \mathrm{i}k\right\} $ between the
Floquet multiplier $s$ and the wave number $k$ (see Section \ref{sec:floquet}
and Remark \ref{rem:disprel}) \emph{the characteristic equation (\ref{eq:monS2f})-
(\ref{eq:monS2j}) can be viewed as an expression of the dispersion
relations between the frequency $\omega$ and the wavenumber $k$
and we will refer to it as the CCTWT dispersion relations or just
the dispersion relations.}

Under simplifying exact synchronism Assumption \ref{ass:cavcoup},
that is $\chi=1$ and $\omega_{0}=1$, the dispersion relations described
by equations (\ref{eq:monS2f})-(\ref{eq:monS2j}) turn into
\begin{equation}
S^{4}+c_{3}S^{3}+c_{2}S^{2}+\bar{c}_{3}S+1=0,\quad s=\exp\left\{ \mathrm{i}k\right\} ,\quad S=\exp\left\{ \mathrm{i}\left(k-{\it \frac{\omega}{2}}\right)\right\} ,\label{eq:disp4S1a}
\end{equation}
\begin{equation}
c_{3}=2{\rm e}^{\frac{\mathrm{i}}{2}\omega}b_{f}^{\infty}+{\rm e}^{-\frac{\mathrm{i}}{2}\omega}\left[C_{0}\frac{\omega^{2}-1}{\omega}\sin\left(\omega\right)-2\cos\left(\omega\right)\right],\label{eq:disp4S1b}
\end{equation}
\begin{equation}
c_{2}=2b_{f}^{\infty}\left[C_{0}\omega\sin\left(\omega\right)-2\,\cos\left(\omega\right)\right]+2\cos\left(f\right)C_{0}\frac{\sin\left(\omega\right)}{\omega}+2\,\cos\left({\it \omega}\right).\label{eq:disp4S1c}
\end{equation}
where
\begin{equation}
b_{f}^{\infty}=K_{0}\sin\left(f\right)-\cos\left(f\right),\quad K_{0}=\frac{b^{2}\beta_{0}}{2f}=\frac{b^{2}g_{\mathrm{B}}}{c_{0}},\quad g_{\mathrm{B}}=\frac{\sigma_{\mathrm{B}}}{4\lambda_{\mathrm{rp}}}.\label{eq:disp4S1d}
\end{equation}

Yet another form of the dispersion relations (\ref{eq:disp4S1a}),
under the same exact synchronism assumption $\chi=1$ and $\omega_{0}=1,$
is its high-frequency form, namely
\begin{equation}
D\left(\omega,k\right)=D^{\left(0\right)}\left(\omega,k\right)+\frac{D^{\left(1\right)}\left(\omega,k\right)}{\omega}+\frac{D^{\left(2\right)}\left(\omega,k\right)}{\omega^{2}}=0,\label{eq:disDD1a}
\end{equation}
where
\begin{equation}
D^{\left(0\right)}\left(\omega,k\right)={\it C_{0}\sin\left(\omega\right)}\left(b_{f}^{\infty}+\cos\left(\omega-k\right)\right),\label{eq:disDD1b}
\end{equation}
\begin{equation}
D^{\left(1\right)}\left(\omega,k\right)=2\left(\cos\left(k\right)-\cos\left(\omega\right)\right){\it b_{f}^{\infty}+\cos\left(2\,k-\omega\right)}-\cos\left(k-2\,\omega\right)+\cos\left(\omega\right)-\cos\left(k\right),\label{eq:disDD1c}
\end{equation}
\emph{
\begin{equation}
D^{\left(2\right)}\left(\omega,k\right)=\left(\cos\left(f\right)-\cos\left(k-\omega\right)\right)\sin\left(\omega\right){\it {\it C_{0}}},\label{eq:disDD1d}
\end{equation}
}where $b_{f}^{\infty}$ is defined by equations (\ref{eq:disp4S1d}).
We refer to function $D\left(\omega,k\right)$ in equation (\ref{eq:disDD1a})
as \emph{the CCTWT dispersion function}.

There exists a remarkable in its simplicity relations between the
CCTWT dispersion function $D\left(\omega,k\right)$ and the dispersion
functions $D_{\mathrm{C}}\left(\omega,k\right)$ and $D_{\mathrm{K}}\left(\omega,k\right)$
for respectively the CCS and the MCK systems. These relations can
be verified by tedious but elementary algebraic evaluations and they
are subjects of the following theorem.
\begin{thm}[dispersion function factorization]
\label{thm:dispfac} Let us assume that $\chi=\omega_{0}=1$. Let
the CCTWT, the CCS and the MCK dispersion functions $D\left(\omega,k\right)$,
$D_{\mathrm{C}}\left(\omega,k\right)$ and $D_{\mathrm{K}}\left(\omega,k\right)$
be defined by respectively equations (\ref{eq:disp4S1a}), (\ref{eq:monTc1f})
and (\ref{eq:DKomk1a}). Then the following identity hold:
\begin{gather}
D\left(\omega,k\right)-D_{\mathrm{C}}\left(\omega,k\right)D_{\mathrm{K}}\left(\omega,k\right)=\frac{K_{0}}{\omega}\left[\frac{2\left(\cos\left(\omega\right)-\cos\left(k\right)\right)}{\omega^{2}-1}-\frac{C_{0}\sin\left(\omega\right)\left(1-\sin\left(f\right)\right)}{\omega}\right],\label{eq:disDD2a}\\
K_{0}=\frac{b^{2}\beta_{0}}{2f}=\frac{b^{2}g_{\mathrm{B}}}{c_{0}},\quad g_{\mathrm{B}}=\frac{\sigma_{\mathrm{B}}}{4\lambda_{\mathrm{rp}}}.\label{eq:disDD2aK}
\end{gather}
In the case of the high-frequency approximation the following identity
holds:
\begin{equation}
D^{\left(0\right)}\left(\omega,k\right)=D_{\mathrm{C}}^{\left(0\right)}\left(\omega,k\right)D_{\mathrm{K}}^{\left(0\right)}\left(\omega,k\right)={\it C_{0}\sin\left(\omega\right)}\left(b_{f}^{\infty}+\cos\left(\omega-k\right)\right).\label{eq:disDD2b}
\end{equation}
The dispersion function identities (\ref{eq:disDD2a}) and (\ref{eq:disDD2b})
signify a very particular way the CCS and the MCK subsystems are coupled
and integrated into the CCTWT system. The right-hand side of the identity
(\ref{eq:disDD2a}) can be naturally viewed as a measure of coupling
between the CCS and the MCK subsystems
\end{thm}

\begin{rem}[graphical confirmation of the dispersion factorization]
\label{rem:dispfac} The statements of the Theorem \ref{thm:dispfac}
are well illustrated by Figs. \ref{fig:dis-serp-K0}(f), \ref{fig:dis-serp-C0}(f)
, \ref{fig:dis-serp-f}(f) and \ref{fig:disp-serp-t} when compared
with Figure \ref{fig:dis-ccs-wom2} for the CCS and Figure \ref{fig:mck-disp3s}
for the MCK. One can confidently identify in the CCTWT dispersion-instability
graphs the patterns of the dispersion-instability graphs of its integral
components - the CCS and the MCK.
\end{rem}

Let us consider now the conventional dispersion relation assuming
that $k$ and $\omega$ must be real numbers. Then dividing equation
(\ref{eq:disp4S1a}) by $S^{2}$ and carrying elementary transformations
we arrive at the following trigonometric form of the conventional
dispersion relation:
\begin{gather}
\cos\left(2k-\omega\right)+\left|c_{3}\right|\cos\left(k-\frac{\omega}{2}+\alpha\right)=-\frac{c_{2}}{2},\label{eq:disp4S1e}\\
c_{3}=\left|c_{3}\right|\exp\left\{ \mathrm{i}\alpha\right\} ,\quad S=\exp\left\{ \mathrm{i}\left(k-\frac{\omega}{2}\right)\right\} ,\quad k,\omega\in\mathbb{R},\nonumber 
\end{gather}
where $c_{3}=c_{3}\left(\omega\right)$, $c_{2}=c_{2}\left(\omega\right)$
and $\alpha=\arg\left\{ c_{3}\left(\omega\right)\right\} $ are frequency
dependent parameters satisfying equations (\ref{eq:disp4S1b})-(\ref{eq:disp4S1d}). 

\subsection{Graphical representation of the dispersion relations}

As to the graphical representation of the dispersion relation recall
that the conventional dispersion relations are defined as the relations
between real-valued frequency $\omega$ and real-valued wavenumber
$k$ associated with the relevant eigenmodes. In the case of interest
$k$ can be complex-valued and to represent all system modes geometrically
we follow to \cite[7]{FigTWTbk}. First, we parametrize every mode
of the system uniquely by the pair $\left(k\left(\omega\right),\omega\right)$
where $\omega$ is its frequency and $k\left(\omega\right)$ is its
wavenumber. If $k\left(\omega\right)$ is degenerate, it is counted
a number of times according to its multiplicity. In view of the importance
to us of the mode instability, that is, when $\Im\left\{ k\left(\omega\right)\right\} \neq0$,
we partition all the system modes represented by pairs $\left(k\left(\omega\right),\omega\right)$
into two distinct classes \textendash{} oscillatory modes and unstable
ones \textendash{} based on whether the wavenumber $k\left(\omega\right)$
is real- or complex-valued with $\Im\left\{ k\left(\omega\right)\right\} \neq0$.
We refer to a mode (eigenmode) of the system as an \emph{oscillatory
mode} if its wavenumber $k\left(\omega\right)$ is real-valued. We
associate with such an oscillatory mode point $\left(k\left(\omega\right),\omega\right)$
in the $k\omega$-plane with $k$ being the horizontal axis and $\omega$
being the vertical one. Similarly, we refer to a mode (eigenmode)
of the system as a \emph{(convective) unstable mode} if its wavenumber
$k$ is complex-valued with a nonzero imaginary part, that is, $\Im\left\{ k\left(\omega\right)\right\} \neq0$.
We associate with such an unstable mode point $\left(\Re\left\{ k\left(\omega\right)\right\} ,\omega\right)$
in the $k\omega$-plane. Since we consider here only \emph{convective
unstable modes}, we refer to them shortly as \emph{unstable modes}.
Notice that every point $\left(\Re\left\{ k\left(\omega\right)\right\} ,\omega\right)$
is in fact associated with two complex conjugate system modes with
$\pm\Im\left\{ k\left(\omega\right)\right\} $.
\begin{figure}[h]
\begin{centering}
\includegraphics[scale=0.12]{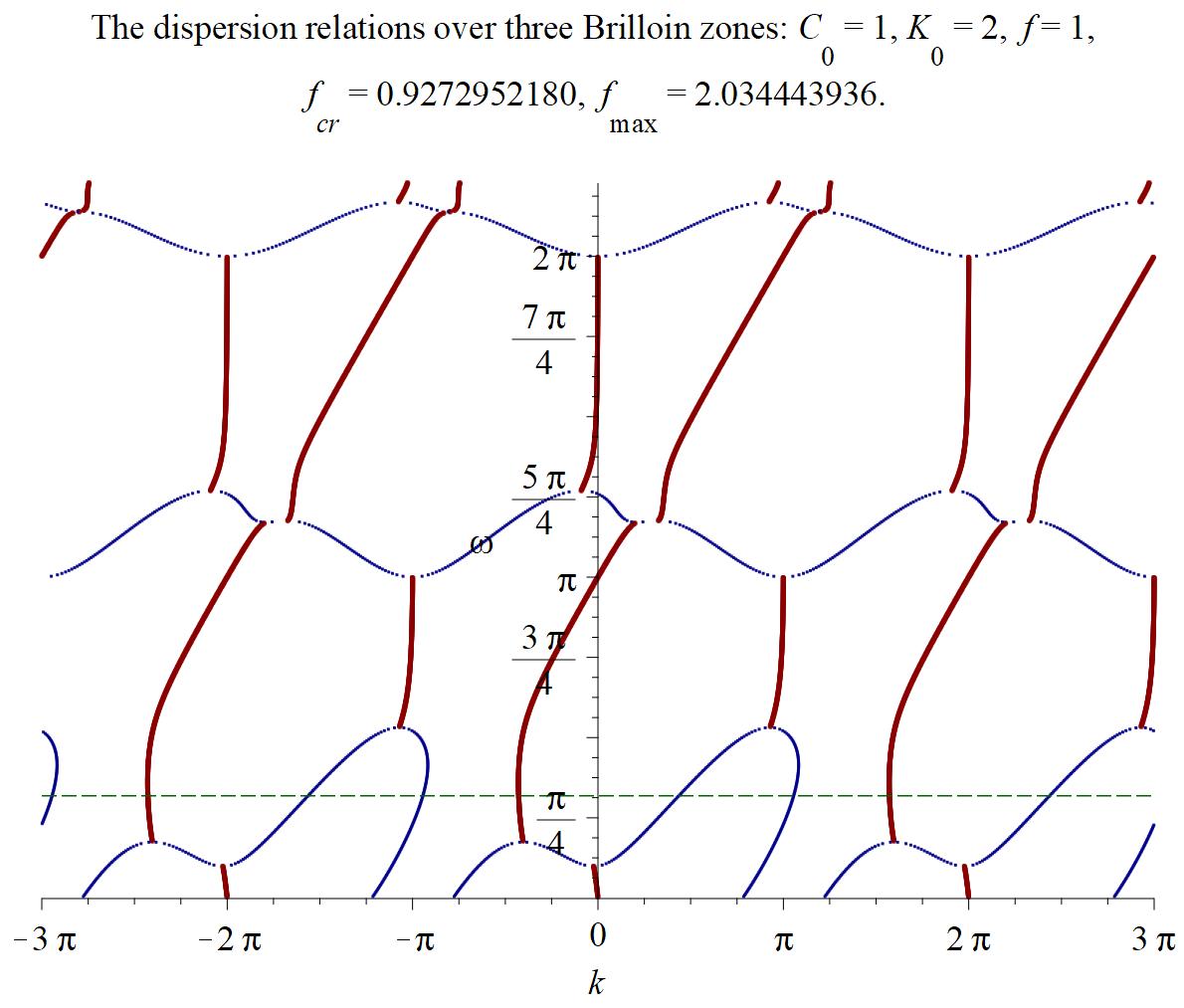}\hspace{0.2cm}\includegraphics[scale=0.12]{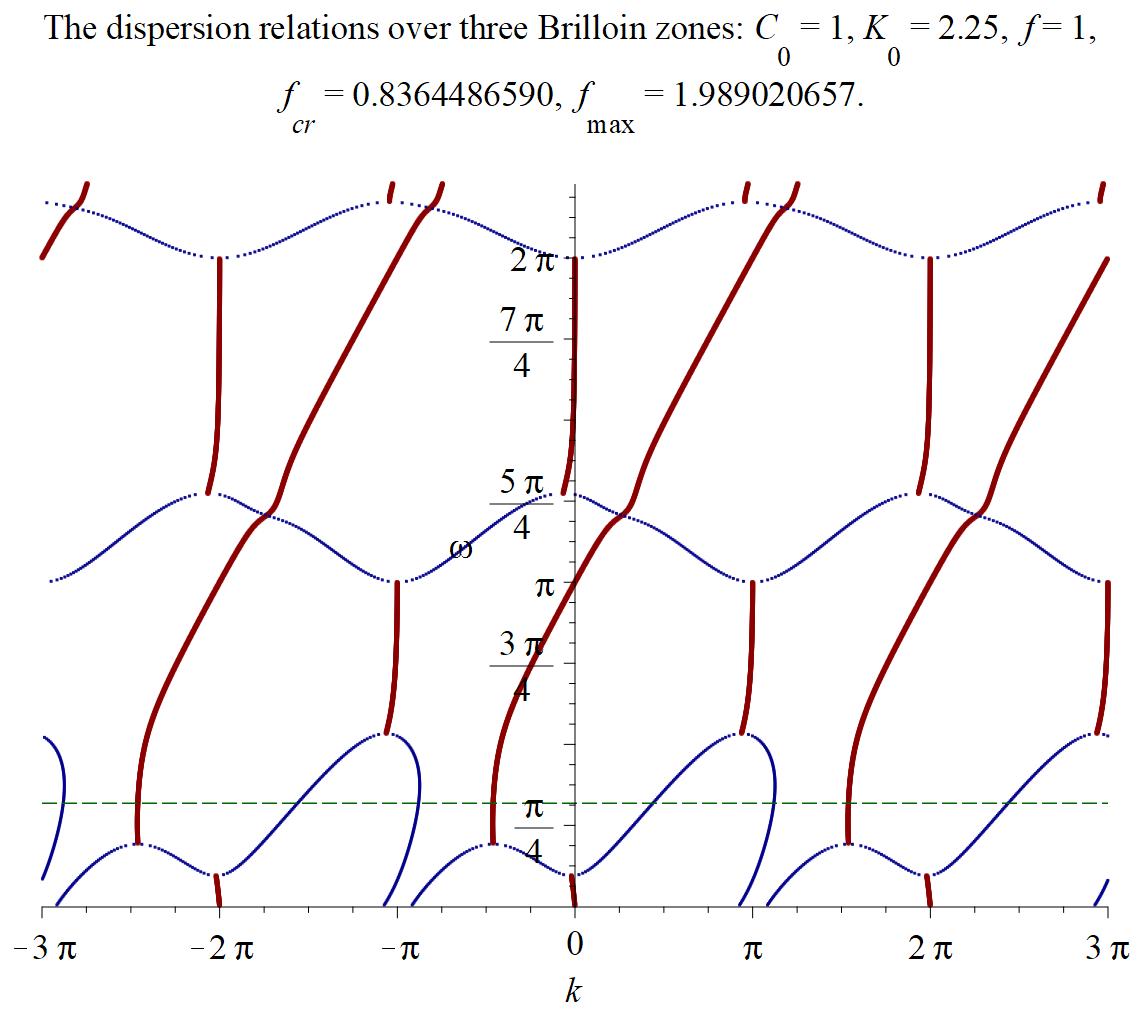}\hspace{0.2cm}\includegraphics[scale=0.12]{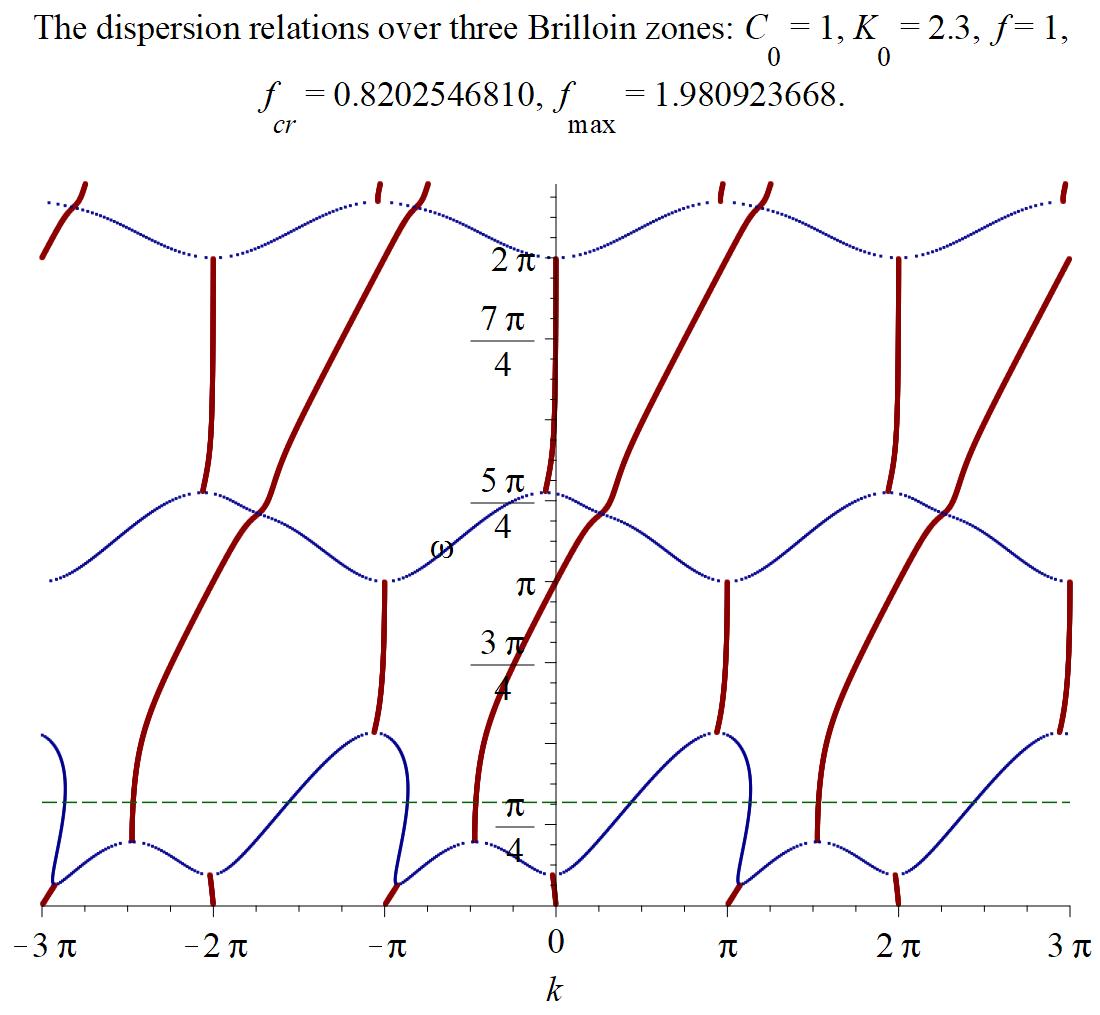}
\par\end{centering}
\begin{centering}
(a)\hspace{4.3cm}(b)\hspace{4.3cm}(c)
\par\end{centering}
\begin{centering}
\includegraphics[scale=0.12]{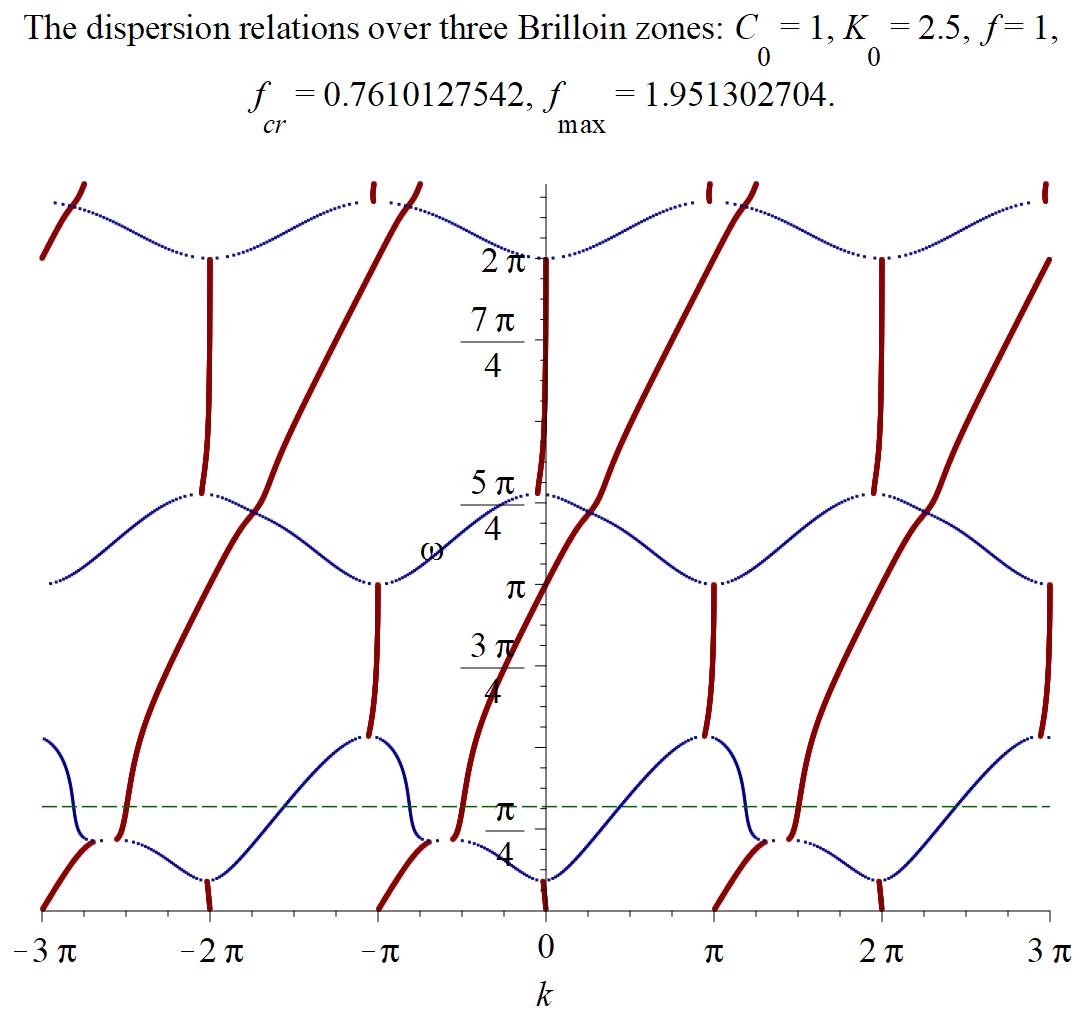}\hspace{0.2cm}\includegraphics[scale=0.12]{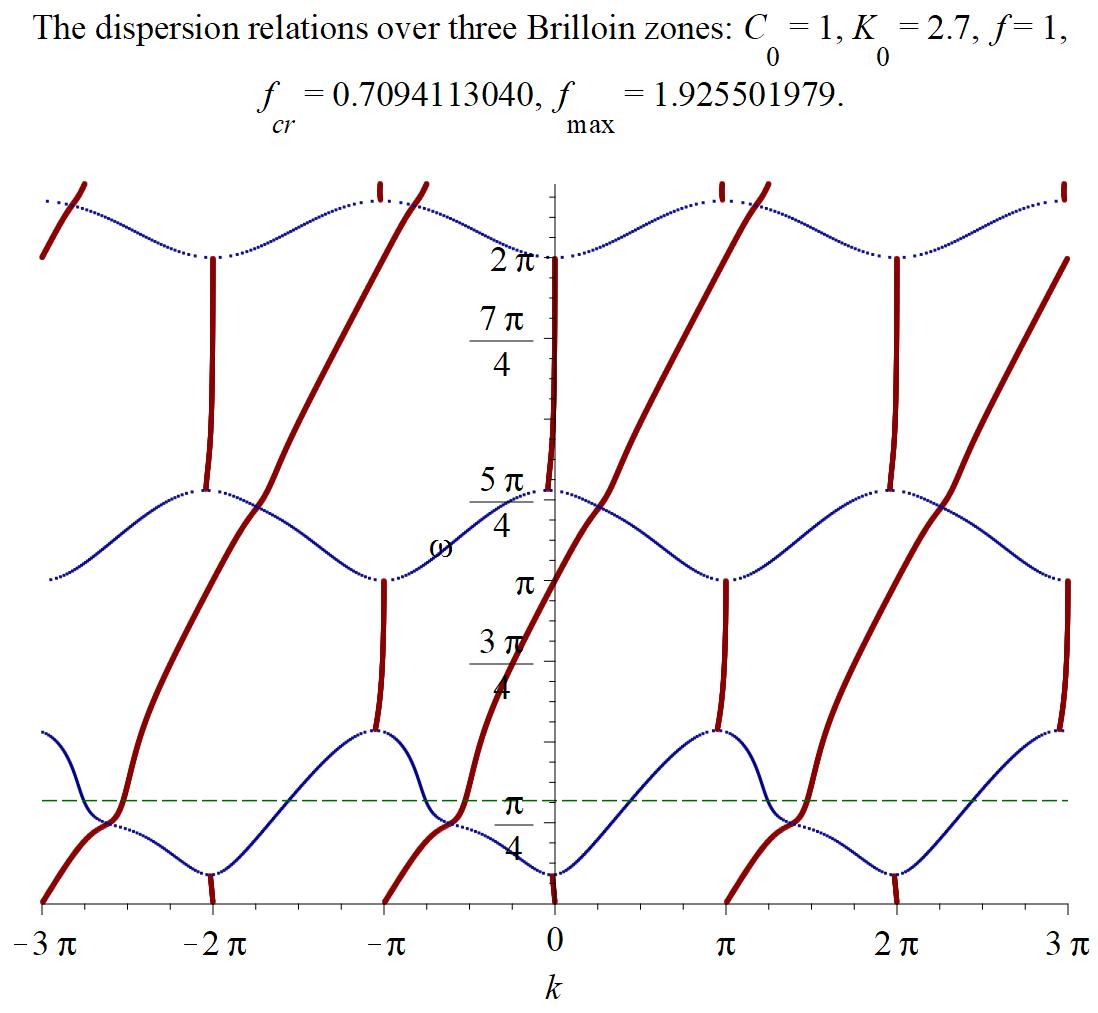}\hspace{0.2cm}\includegraphics[scale=0.12]{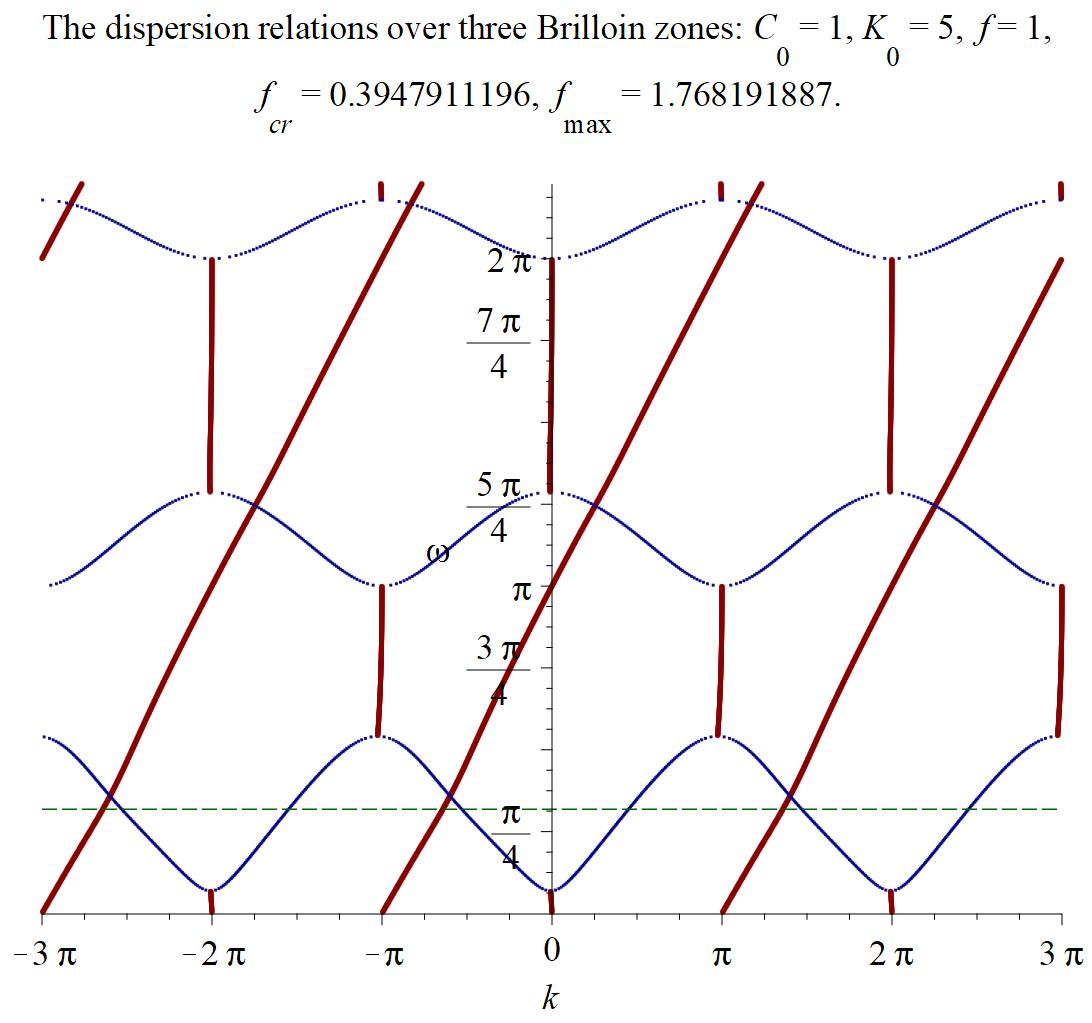}
\par\end{centering}
\centering{}(d)\hspace{4.3cm}(e)\hspace{4.3cm}(f)\caption{\label{fig:dis-serp-K0} The dispersion-instability graphs for the
CCTWT as the gain coefficient $K_{0}$ varies. In all plots the horizontal
and vertical axes represent respectively $\Re\left\{ k\right\} $
and $\omega$. Each of the plots shows 3 bands of the dispersion of
the CCTWT described by equations (\ref{eq:monS2f})-(\ref{eq:monS2j})
over 3 Brillouin zones $\Re\left\{ k\right\} \in\left[-3\pi,3\pi\right]$
for $\chi=1$, $\omega_{0}=1$, $f=1$, and $C_{0}=1$: (a) $K_{0}=2$,
$f_{\mathrm{cr}}\protect\cong0.927292180$ and $f_{\mathrm{max}}\protect\cong2.034$;
(b) $K_{0}=2.25$, $f_{\mathrm{cr}}\protect\cong0.837$ and $f_{\mathrm{max}}\protect\cong1.99$;
(c) $K_{0}=2.3$, $f_{\mathrm{cr}}\protect\cong0.820$ and $f_{\mathrm{max}}\protect\cong1.981$;
(d) $K_{0}=2.5$ , $f_{\mathrm{cr}}\protect\cong0.761$ and $f_{\mathrm{max}}\protect\cong1.951$;
(e) $K_{0}=2.7$, $f_{\mathrm{cr}}\protect\cong0.709$ and $f_{\mathrm{max}}\protect\cong1.926$;
(f) $K_{0}=5$, $f_{\mathrm{cr}}\protect\cong0.395$ and $f_{\mathrm{max}}\protect\cong1.768$
(see Theorem \ref{thm:dispfac}). When $\Im\left\{ k_{\pm}\left(\omega\right)\right\} =0$,
that is the case of oscillatory modes, and $\Re\left\{ k\left(\omega\right)\right\} =k\left(\omega\right)$
the corresponding branches are shown as solid (blue) curves. When
$\Im\left\{ k_{\pm}\left(\omega\right)\right\} \protect\neq0$, that
is there is an instability, and $\Re\left\{ k_{+}\left(\omega\right)\right\} =\Re\left\{ k_{-}\left(\omega\right)\right\} $
then the corresponding branches overlap, they are shown as bold solid
curves in brown color, and each point of these branches represents
exactly two modes with complex-conjugate wave numbers $k_{\pm}$.}
\end{figure}

Based on the above discussion, we represent the set of all oscillatory
and unstable modes of the system geometrically by the set of the corresponding
modal points $\left(k\left(\omega\right),\omega\right)$ and $\left(\Re\left\{ k\left(\omega\right)\right\} ,\omega\right)$
in the $k\omega$-plane. We name this set the \emph{dispersion-instability
graph}. To distinguish graphically points $\left(k\left(\omega\right),\omega\right)$
associated oscillatory modes when $k\left(\omega\right)$ is real-valued
from points $\left(\Re\left\{ k\left(\omega\right)\right\} ,\omega\right)$
associated unstable modes when $k\left(\omega\right)$ is complex-valued
with $\Im\left\{ k\left(\omega\right)\right\} \neq0$ we show points
$\Im\left\{ k\left(\omega\right)\right\} =0$ in blue color whereas
points with $\Im\left\{ k\left(\omega\right)\right\} \neq0$ are shown
in brown color. We remind once again that every point $\left(\omega,\Re\left\{ k\left(\omega\right)\right\} \right)$
with $\Im\left\{ k\left(\omega\right)\right\} \neq0$ represents exactly
two complex conjugate unstable modes associated with $\pm\Im\left\{ k\left(\omega\right)\right\} $.

When $\Im\left\{ k_{\pm}\left(\omega\right)\right\} \neq0$ and $\Re\left\{ k_{+}\left(\omega\right)\right\} =\Re\left\{ k_{-}\left(\omega\right)\right\} $
and consequently the corresponding branches overlap with each point
on the segments representing two modes with complex-conjugate wave
numbers $k_{\pm}$. These branches represent exponentially growing
or decaying in the space modes and shown in plot (c) in brown color.

We generated three sets of dispersion-instability graphs for the CCTWT
shown in Figs. \ref{fig:dis-serp-K0}, \ref{fig:dis-serp-C0} and
\ref{fig:dis-serp-f} to demonstrate their dependence on the gain
coefficient $K_{0}$, the capacitance parameter $C_{0}$ and the normalized
period $f$ as they vary in indicated ranges. Figs. \ref{fig:dis-serp-K0}(f),
\ref{fig:dis-serp-C0}(f), \ref{fig:dis-serp-f}(f) and \ref{fig:disp-serp-t}
when compared with Figure \ref{fig:dis-ccs-wom2} for the CCS and
Figure \ref{fig:mck-disp3s} for the MCK clearly indicate that the
CCTWT dispersion-instability graph is composed of the dispersion-instability
graphs of its integral components - the CCS and the MCK. The later
is important since the CCS and the MCK are significantly simpler systems
compare to the original CCTWT.
\begin{figure}[h]
\begin{centering}
\includegraphics[scale=0.12]{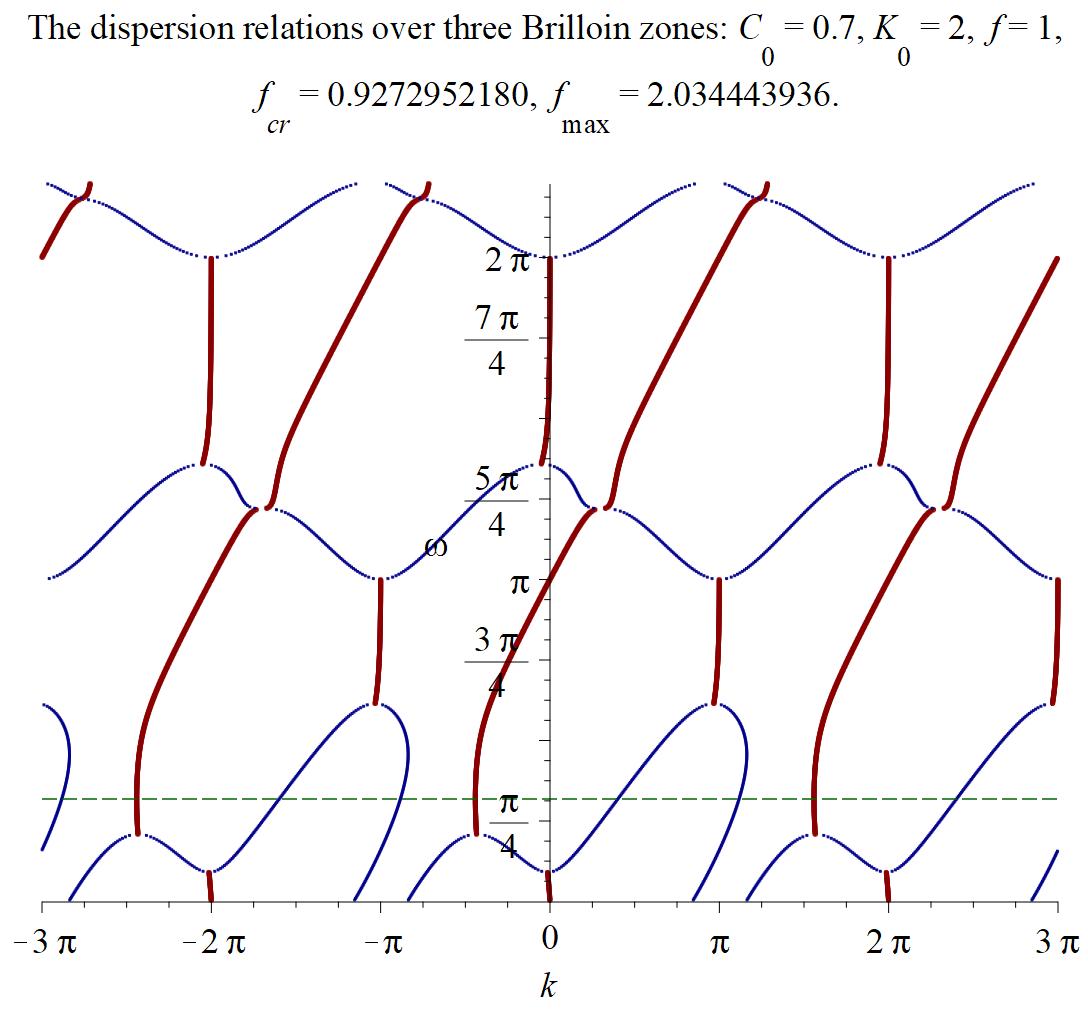}\hspace{0.2cm}\includegraphics[scale=0.12]{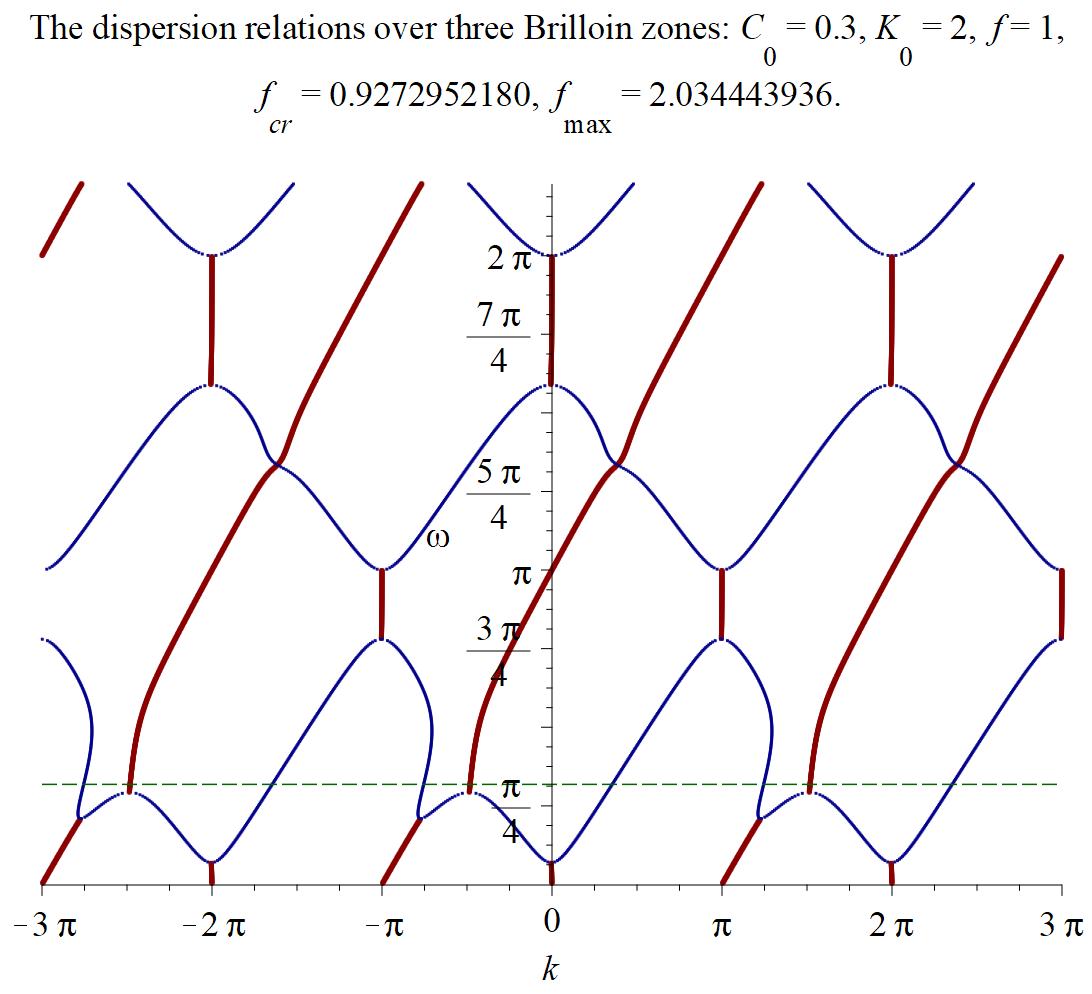}\hspace{0.2cm}\includegraphics[scale=0.12]{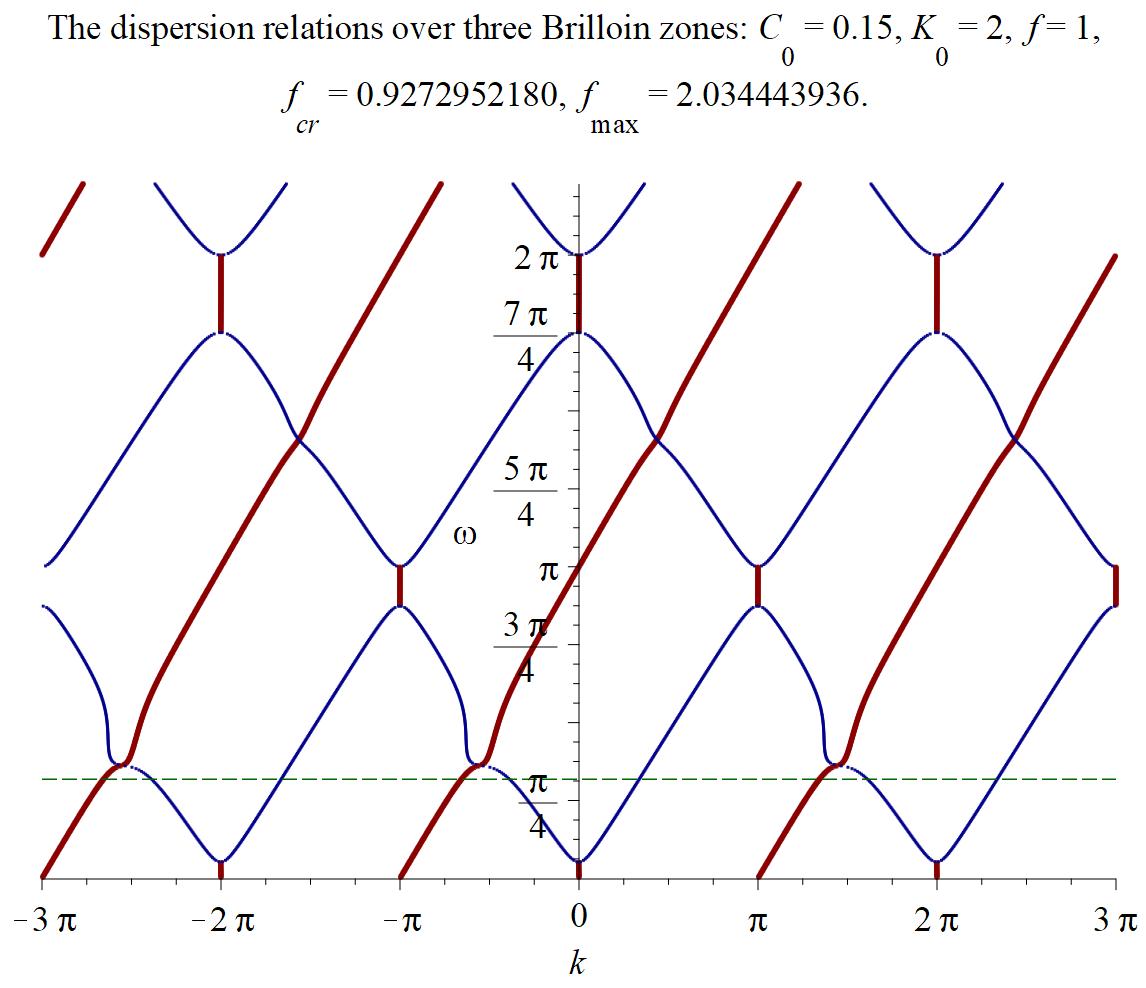}
\par\end{centering}
\begin{centering}
(a)\hspace{4.3cm}(b)\hspace{4.3cm}(c)
\par\end{centering}
\begin{centering}
\includegraphics[scale=0.12]{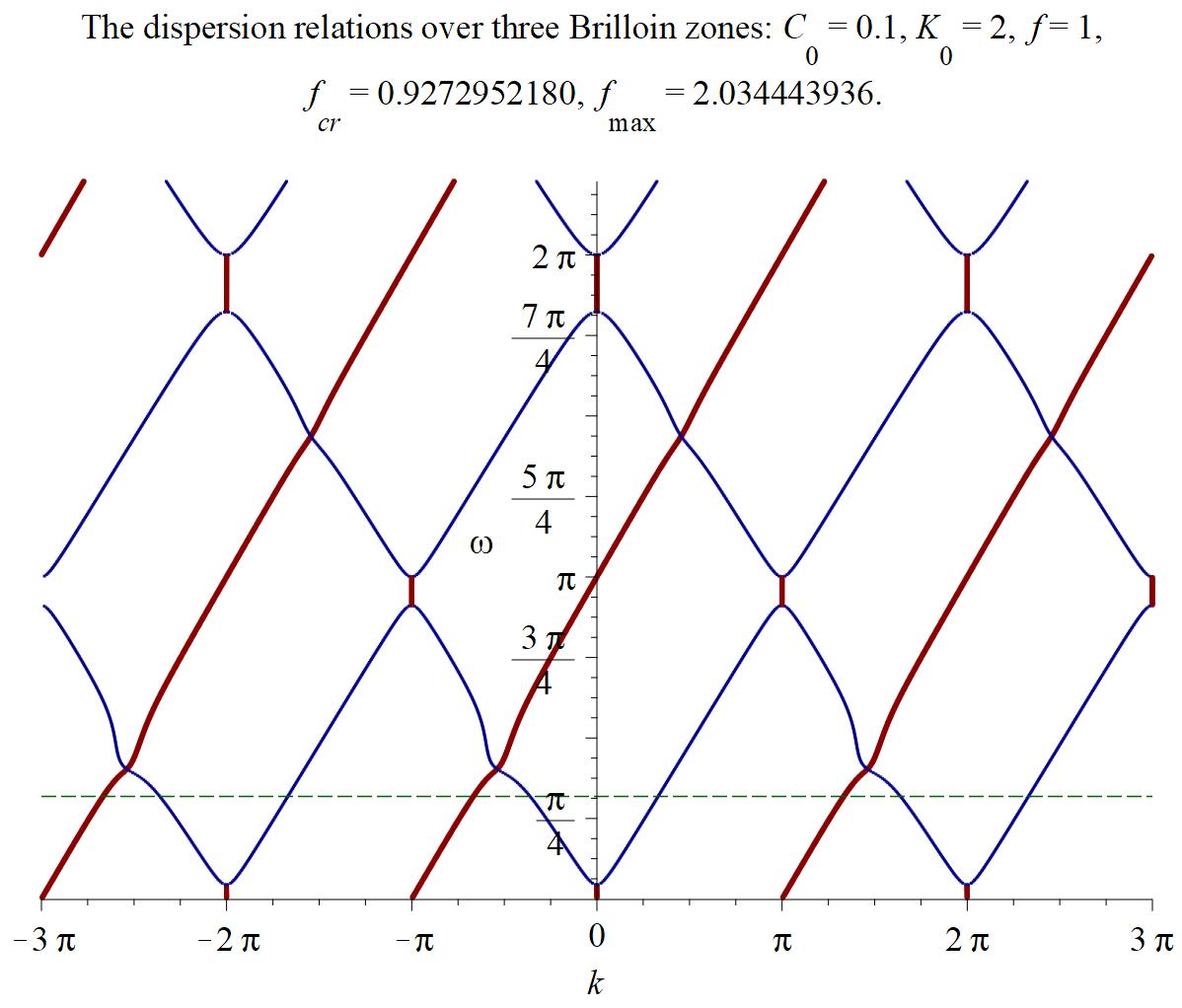}\hspace{0.2cm}\includegraphics[scale=0.12]{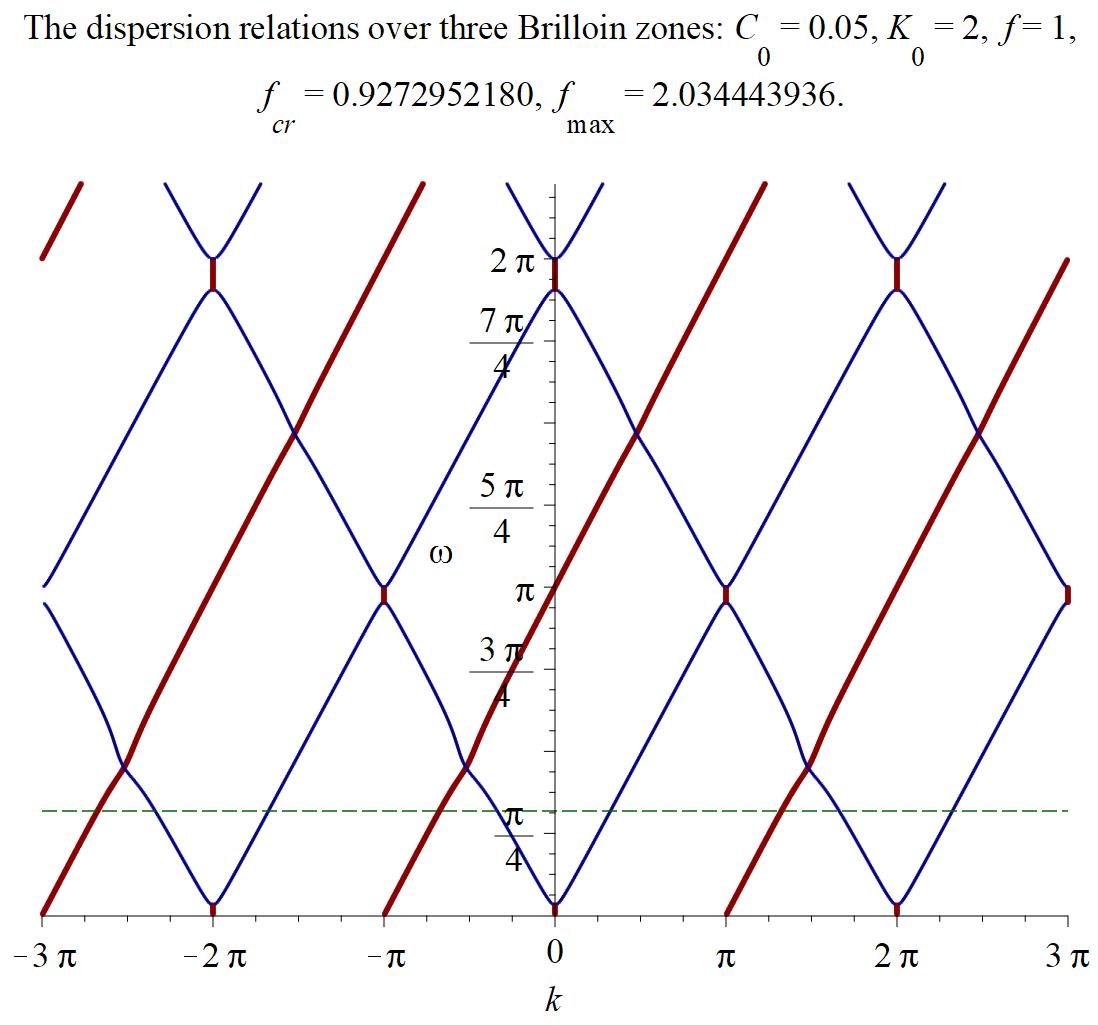}\hspace{0.2cm}\includegraphics[scale=0.12]{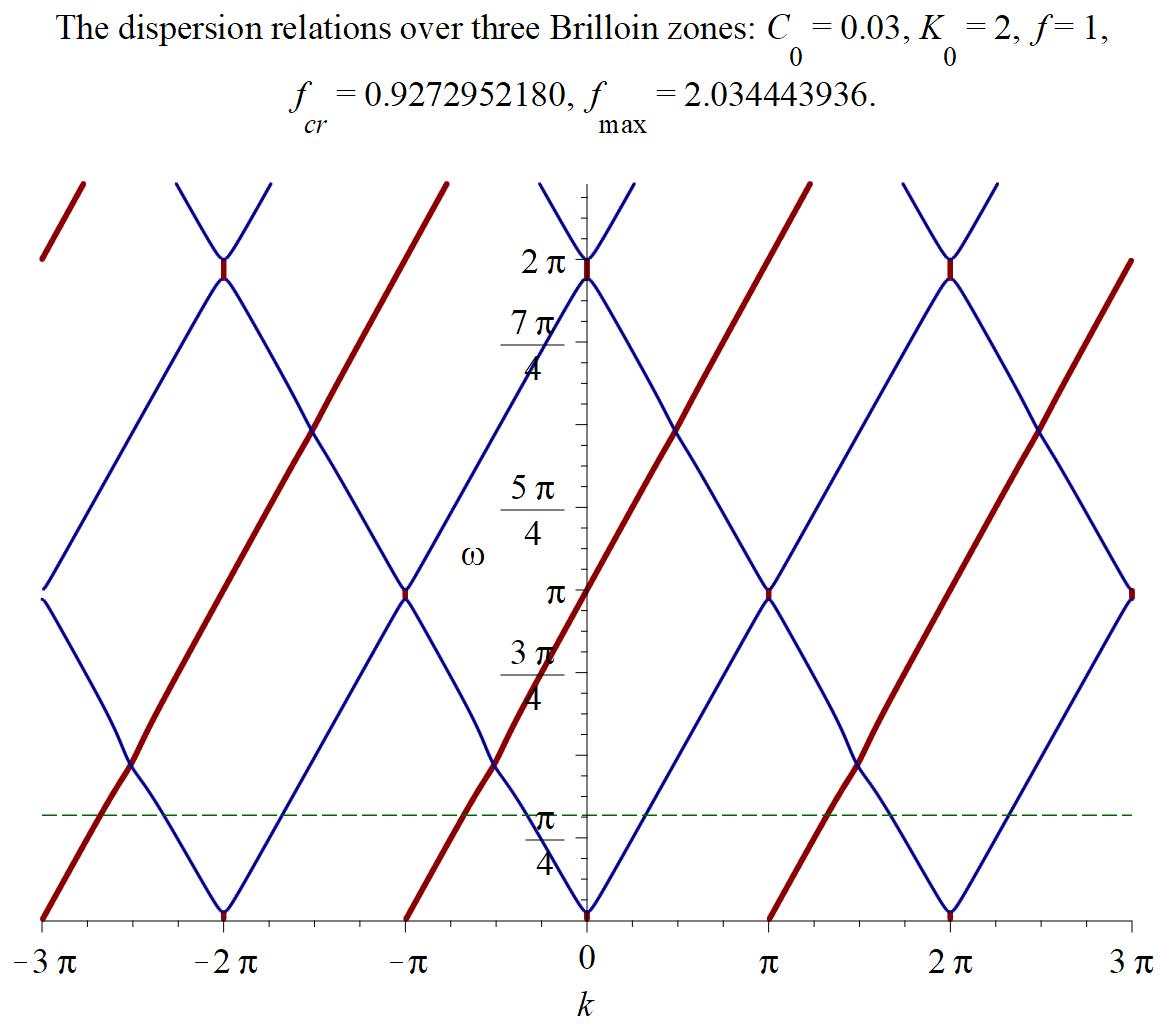}
\par\end{centering}
\centering{}(d)\hspace{4.3cm}(e)\hspace{4.3cm}(f)\caption{\label{fig:dis-serp-C0} The dispersion-instability graphs for the
CCTWT as the capacitance parameter $C_{0}$ varies. In all plots the
horizontal and vertical axes represent respectively $\Re\left\{ k\right\} $
and $\omega$. Each of the plots shows 3 bands of the dispersion of
the CCTWT described by equations (\ref{eq:monS2f})-(\ref{eq:monS2j})
over 3 Brillouin zones $\Re\left\{ k\right\} \in\left[-3\pi,3\pi\right]$
for $\chi=1$, $\omega_{0}=1$, $f=1$, and $K_{0}=2$, $f_{\mathrm{cr}}\protect\cong0.927292180$
and $f_{\mathrm{max}}\protect\cong2.034$ : (a) $C_{0}=0.7$ ; (b)
$C_{0}=0.3$ ; (c) $C_{0}=0.15$ ; (d) $C_{0}=0.1$ ; (e) $C_{0}=0.05$
; (f) $C_{0}=0.03$ (see Theorem \ref{thm:dispfac}). When $\Im\left\{ k_{\pm}\left(\omega\right)\right\} =0$,
that is the case of oscillatory modes, and $\Re\left\{ k\left(\omega\right)\right\} =k\left(\omega\right)$
the corresponding branches are shown as solid (blue) curves. When
$\Im\left\{ k_{\pm}\left(\omega\right)\right\} \protect\neq0$, that
is there is an instability, and $\Re\left\{ k_{+}\left(\omega\right)\right\} =\Re\left\{ k_{-}\left(\omega\right)\right\} $
then the corresponding branches overlap, they are shown as bold solid
(brown) curves in brown color, and each point of these branches represents
exactly two modes with complex-conjugate wave numbers $k_{\pm}$.}
\end{figure}
\begin{figure}[h]
\begin{centering}
\includegraphics[scale=0.12]{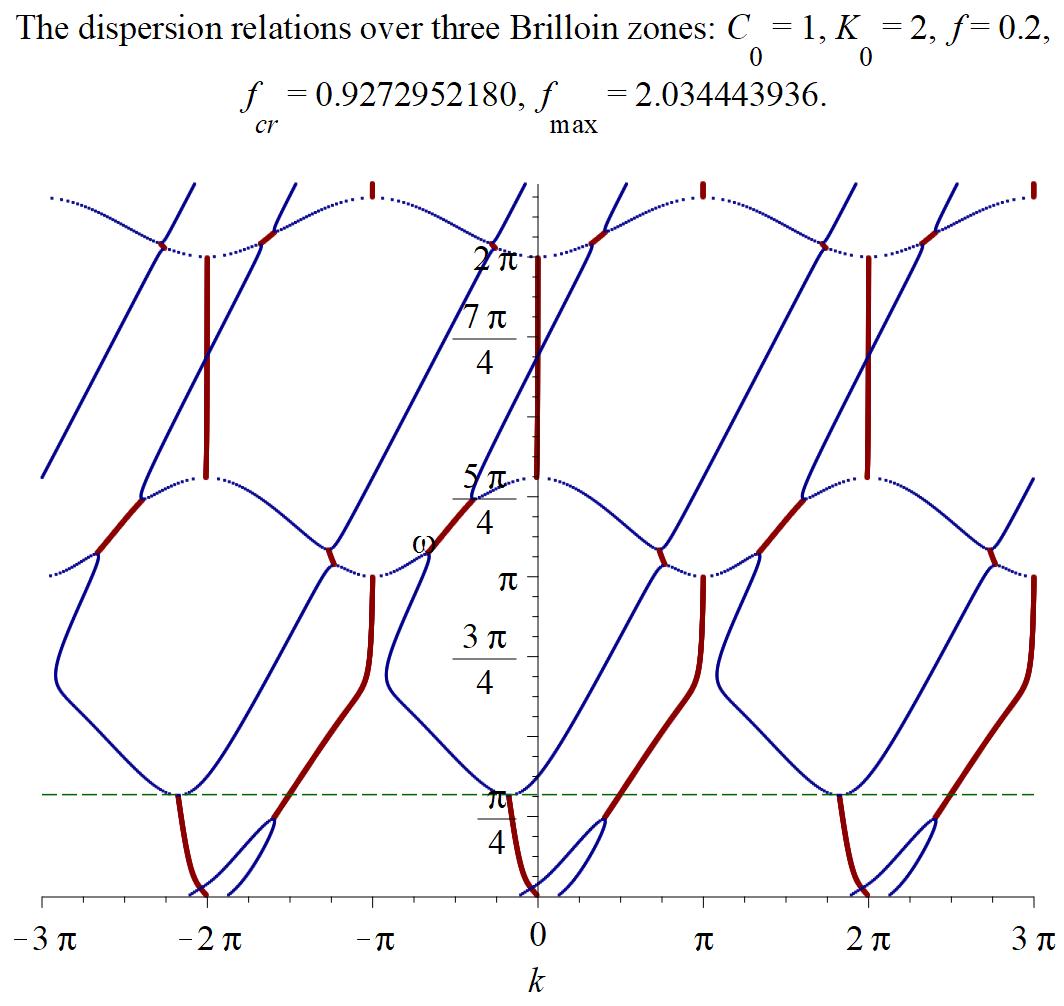}\hspace{0.2cm}\includegraphics[scale=0.12]{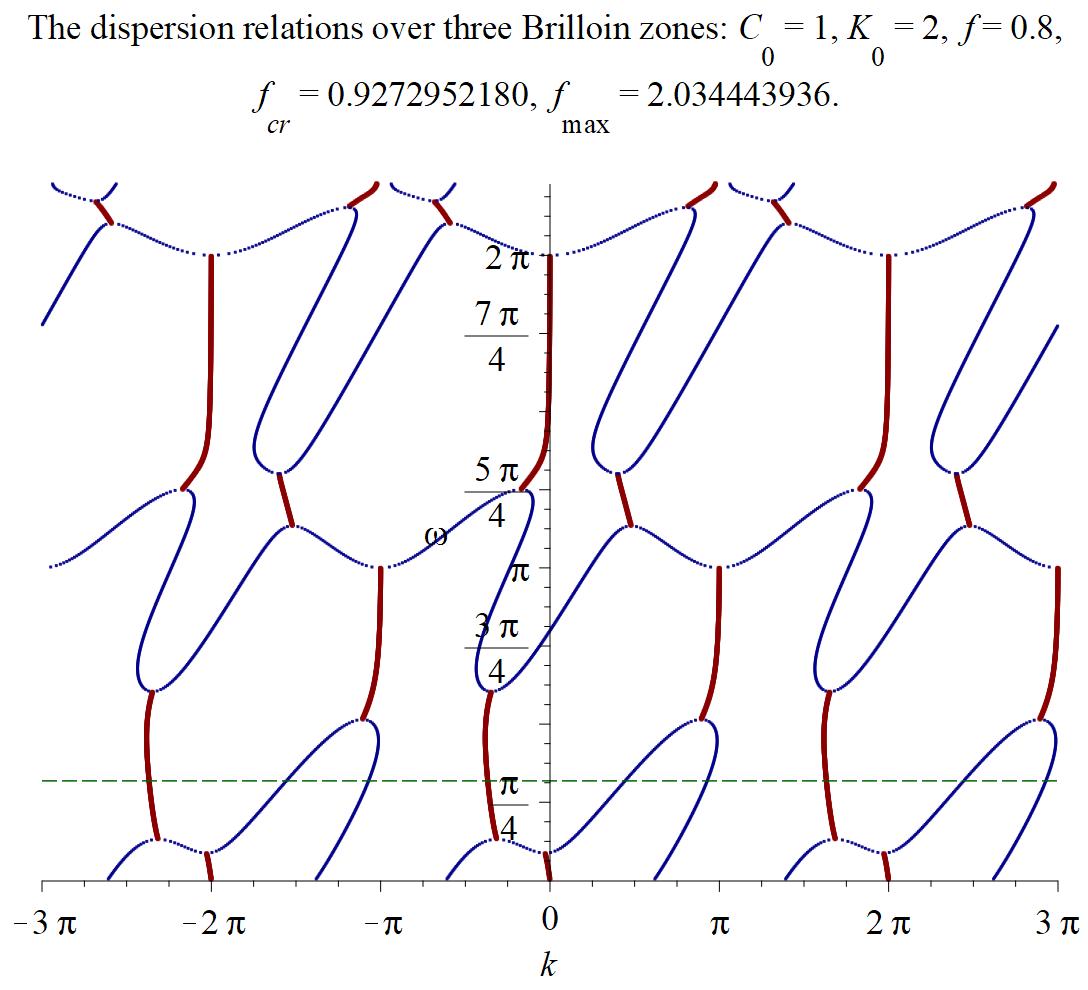}\hspace{0.2cm}\includegraphics[scale=0.12]{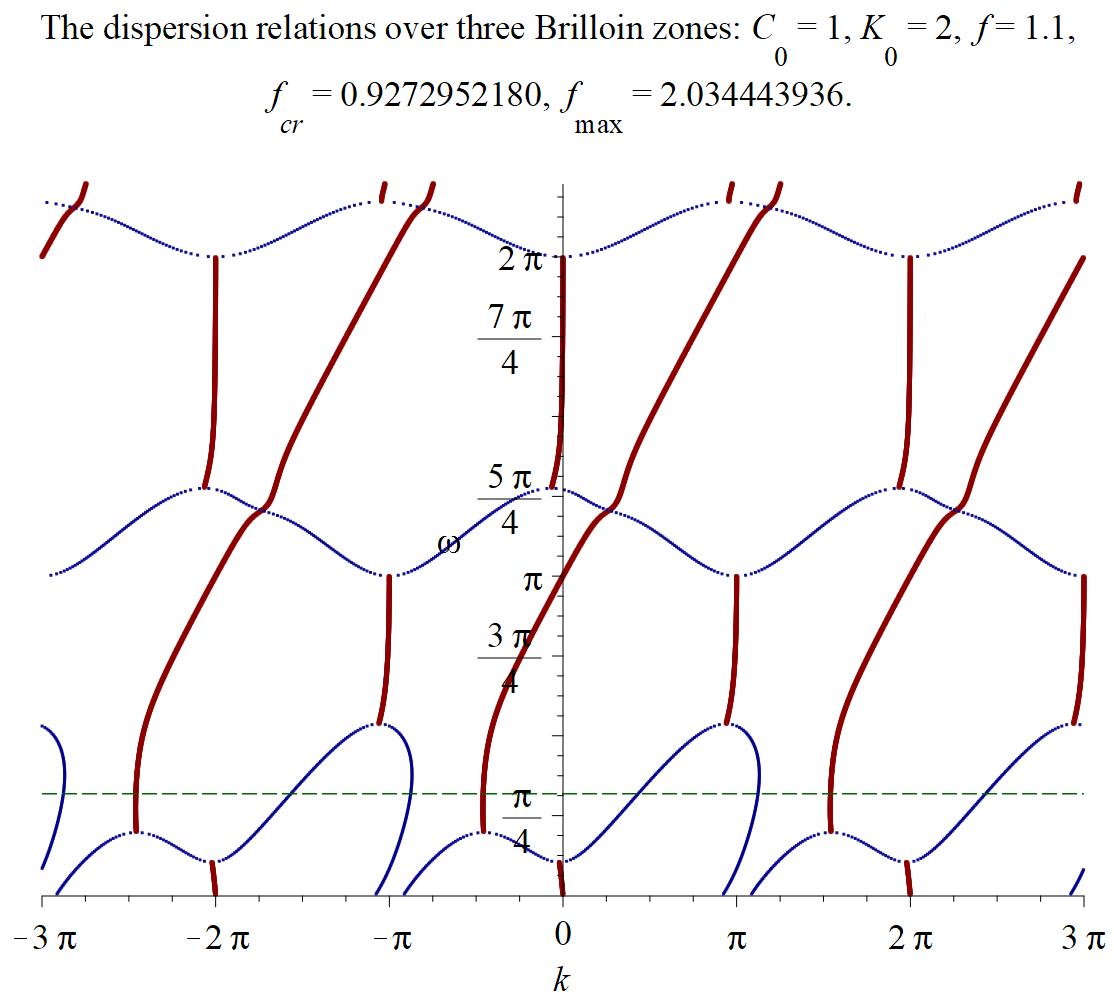}
\par\end{centering}
\begin{centering}
(a)\hspace{4.3cm}(b)\hspace{4.3cm}(c)
\par\end{centering}
\begin{centering}
\includegraphics[scale=0.12]{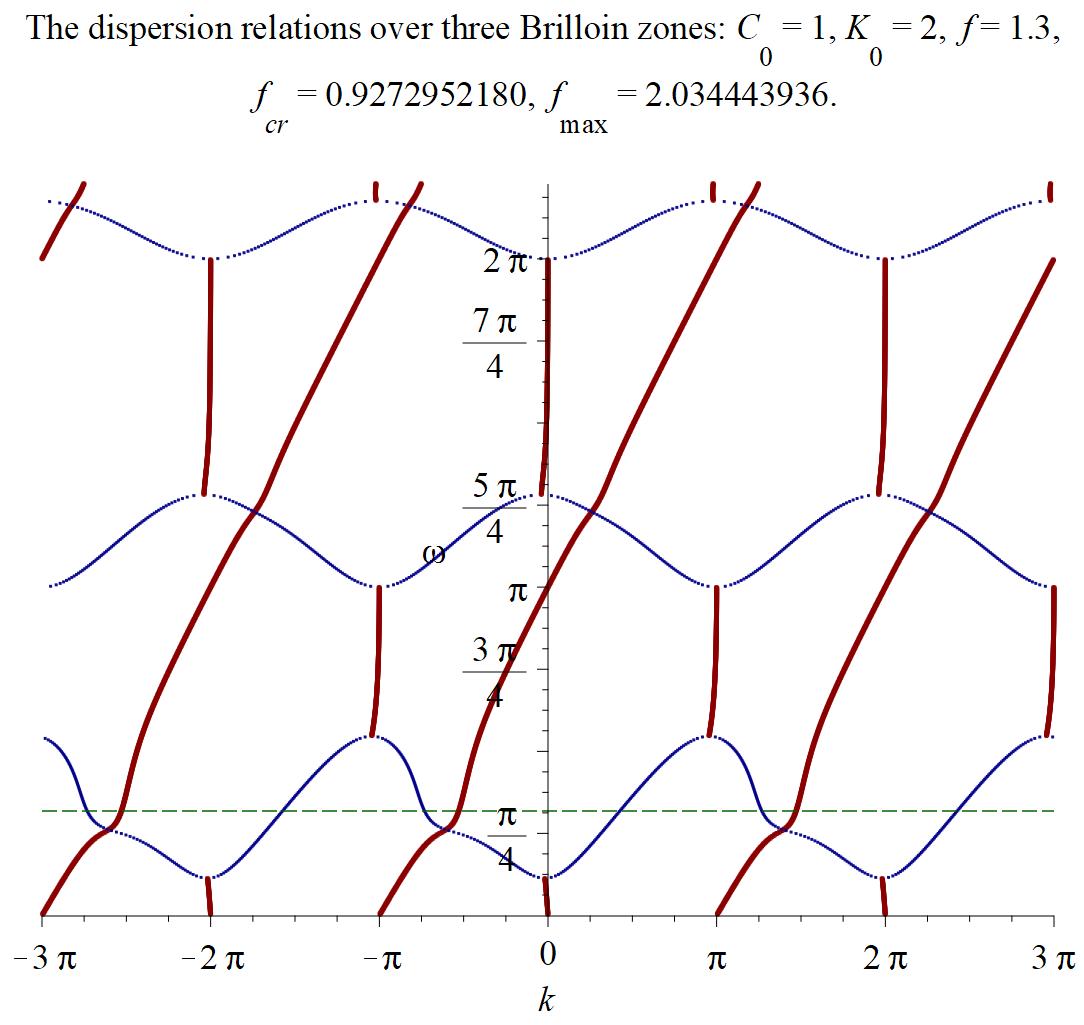}\hspace{0.2cm}\includegraphics[scale=0.12]{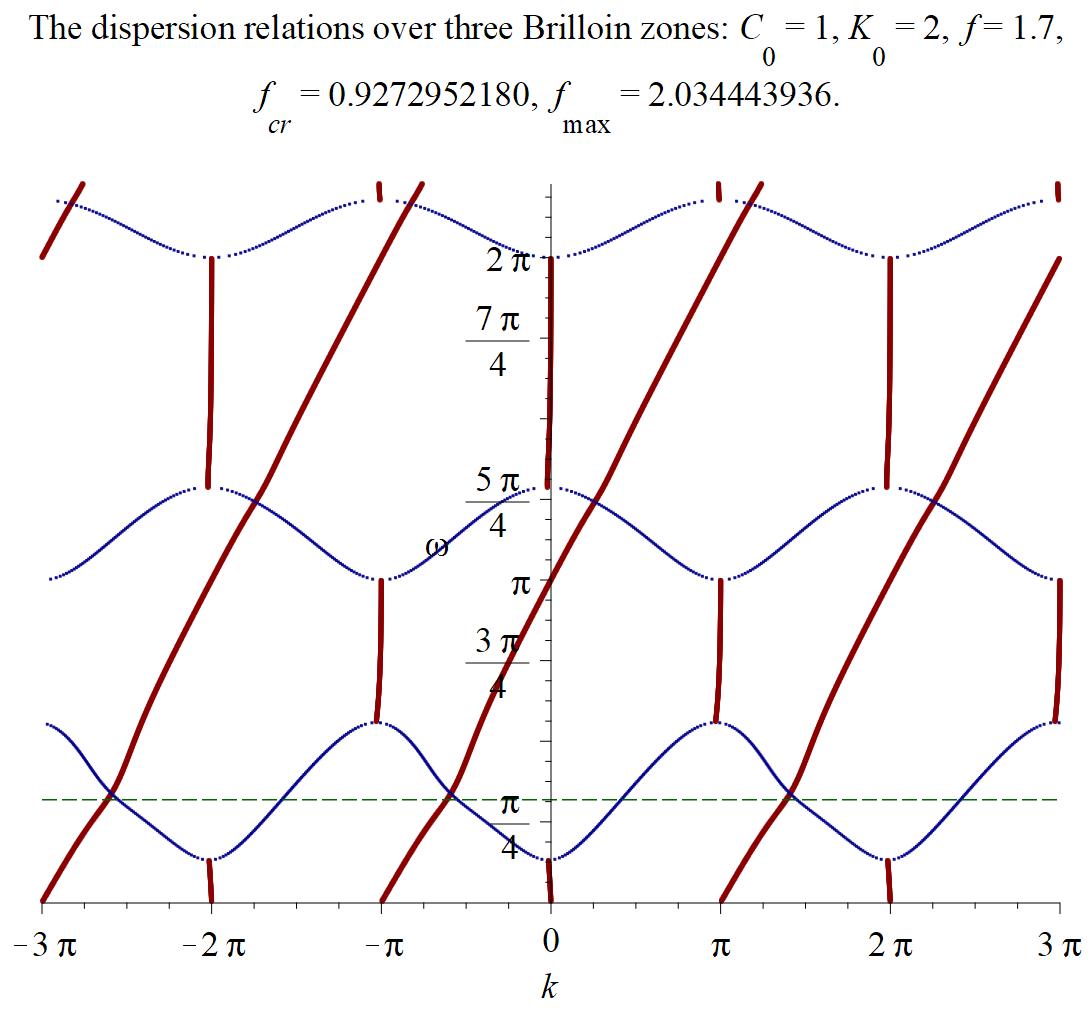}\hspace{0.2cm}\includegraphics[scale=0.12]{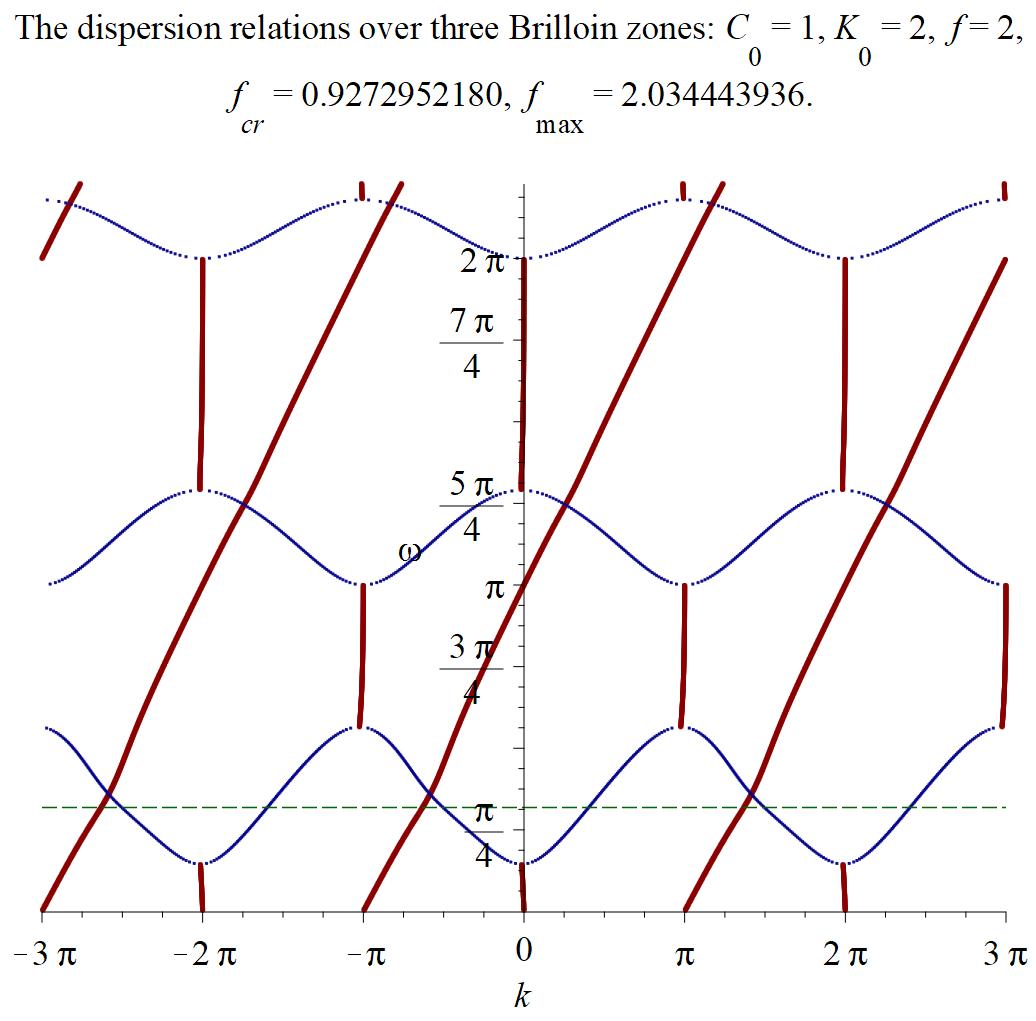}
\par\end{centering}
\centering{}(d)\hspace{4.3cm}(e)\hspace{4.3cm}(f)\caption{\label{fig:dis-serp-f} The dispersion-instability graphs for the
CCTWT as the normalized period $f$ varies. In all plots the horizontal
and vertical axes represent respectively $\Re\left\{ k\right\} $
and $\omega$. Each of the plots shows 3 bands of the dispersion of
the CCTWT described by equations (\ref{eq:monS2f})-(\ref{eq:monS2j})
over 3 Brillouin zones $\Re\left\{ k\right\} \in\left[-3\pi,3\pi\right]$
for $\chi=1$, $\omega_{0}=1$, $C_{0}=1$, $K_{0}=2$, $f_{\mathrm{cr}}\protect\cong0.927292180$
and $f_{\mathrm{max}}\protect\cong2.034$: (a) $f=0.2$; (b) $f=0.8$;
(c) $f=1.1$; (d) $f=1.3$; (e) $f=1.7$; (f) $f=2$ (see Theorem
\ref{thm:dispfac}). When $\Im\left\{ k_{\pm}\left(\omega\right)\right\} =0$,
that is the case of oscillatory modes, and $\Re\left\{ k\left(\omega\right)\right\} =k\left(\omega\right)$
the corresponding branches are shown as solid (blue) curves. When
$\Im\left\{ k_{\pm}\left(\omega\right)\right\} \protect\neq0$, that
is there is an instability, and $\Re\left\{ k_{+}\left(\omega\right)\right\} =\Re\left\{ k_{-}\left(\omega\right)\right\} $
then the corresponding branches overlap, they are shown as bold solid
(brown) curves in brown color, and each point of these branches represents
exactly two modes with complex-conjugate wave numbers $k_{\pm}$.}
\end{figure}
\begin{figure}[h]
\centering{}\includegraphics[scale=0.2]{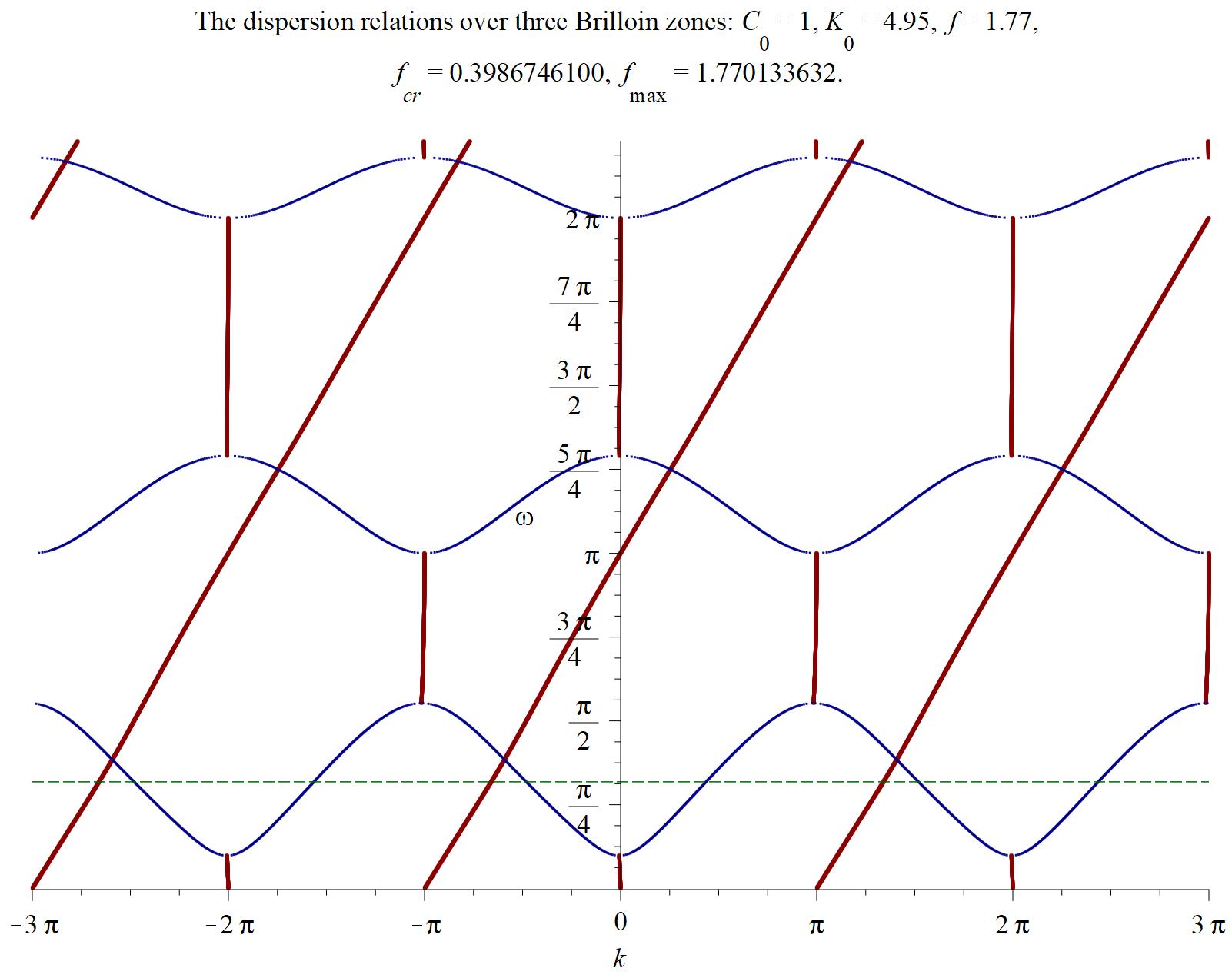}\caption{\label{fig:disp-serp-t} The dispersion-instability graphs for the
CCTWT as the normalized period $f$ varies (see Theorem \ref{thm:dispfac}
and Remark \ref{rem:dispfac}). The horizontal and vertical axes represent
respectively $\Re\left\{ k\right\} $ and $\omega$. The plot shows
3 bands of the dispersion of the CCTWT described by equations (\ref{eq:monS2f})-(\ref{eq:monS2j})
over 3 Brillouin zones $\Re\left\{ k\right\} \in\left[-3\pi,3\pi\right]$
for $\chi=1$, $\omega_{0}=1$, $C_{0}=1$, $f\protect\cong1.77$,
$K_{0}=K_{0\mathrm{T}}=4.95$ (typical value) and consequently $f_{\mathrm{cr}}\protect\cong0.397$,
$f_{\mathrm{max}}\protect\cong1.770$. When $\Im\left\{ k_{\pm}\left(\omega\right)\right\} =0$,
that is the case of oscillatory modes, and $\Re\left\{ k\left(\omega\right)\right\} =k\left(\omega\right)$
the corresponding branches are shown as solid (blue) curves. When
$\Im\left\{ k_{\pm}\left(\omega\right)\right\} \protect\neq0$, that
is there is an instability, and $\Re\left\{ k_{+}\left(\omega\right)\right\} =\Re\left\{ k_{-}\left(\omega\right)\right\} $
then the corresponding branches overlap, they are shown as bold solid
(brown) curves in brown color, and each point of these branches represents
exactly two modes with complex-conjugate wave numbers $k_{\pm}$.
(see Theorem \ref{thm:dispfac} and Remark \ref{rem:dispfac})}
\end{figure}

\subsection{Exceptional points of degeneracy\label{subsec:cctwt-epd}}

Jordan eigenvector degeneracy, which is a degeneracy of the system
evolution matrix when not only some eigenvalues coincide but the corresponding
eigenvectors coincide also, is sometimes referred to as exceptional
point of degeneracy (EPD), \cite[II.1]{Kato}. Our prior studies of
traveling wave tubes (TWT) in \cite[4, 7, 13, 14, 54, 55]{FigTWTbk}
demonstrate that TWTs always have EPDs. A particularly important class
of applications of EPDs is sensing, \cite{CheN}, \cite{KNAC}, \cite{OGC},
\cite{Wie}, \cite{Wie1}. For applications of EPDs for traveling
wave tubes see \cite{FigtwtEPD}, \cite{OTC}, \cite{OVFC}, \cite{OVFC1},
\cite{VOFC}.

Applying the results of Appendix \ref{sec:degpol4} particularly the
system of equations (\ref{eq:speSab1d}) we obtain the following trigonometric
form of equations for \emph{exceptional points of degeneracy (EPDs)}:
\begin{gather}
\cos\left(2k-\omega\right)+\left|c_{3}\right|\cos\left(k-\frac{\omega}{2}+\alpha\right)=-\frac{c_{2}}{2},\quad2\sin\left(2k-\omega\right)+\left|c_{3}\right|\sin\left(k-\frac{\omega}{2}+\alpha\right)=0,\label{eq:epdkom1a}\\
c_{3}=\left|c_{3}\right|\exp\left\{ \mathrm{i}\alpha\right\} \quad S=\exp\left\{ \mathrm{i}\left(k-\frac{\omega}{2}\right)\right\} ,\quad k,\omega\in\mathbb{R}.\nonumber 
\end{gather}
Note that the first equation in (\ref{eq:epdkom1a}) is identical
to the trigonometric form (\ref{eq:disp4S1e}) of the CCTWT dispersion
relations.

Figure \ref{fig:dis-serp-deg1} shows examples of the dispersion-instability
graphs with EPDs as points which are the points of the transition
to instability. In particular, Figure \ref{fig:dis-serp-deg1}(c)
when compared with Figure \ref{fig:dis-ccs-wom2} for the CCS and
Figure \ref{fig:mck-disp3s} for the MCK indicates convincingly that
the components of the CCTWT dispersion-instability graph can be attributed
to the dispersion-instability graphs of its integral components -
the CCS and the MCK (see Theorem \ref{thm:dispfac} and Remark \ref{rem:dispfac}).
\begin{figure}[h]
\begin{centering}
\includegraphics[scale=0.12]{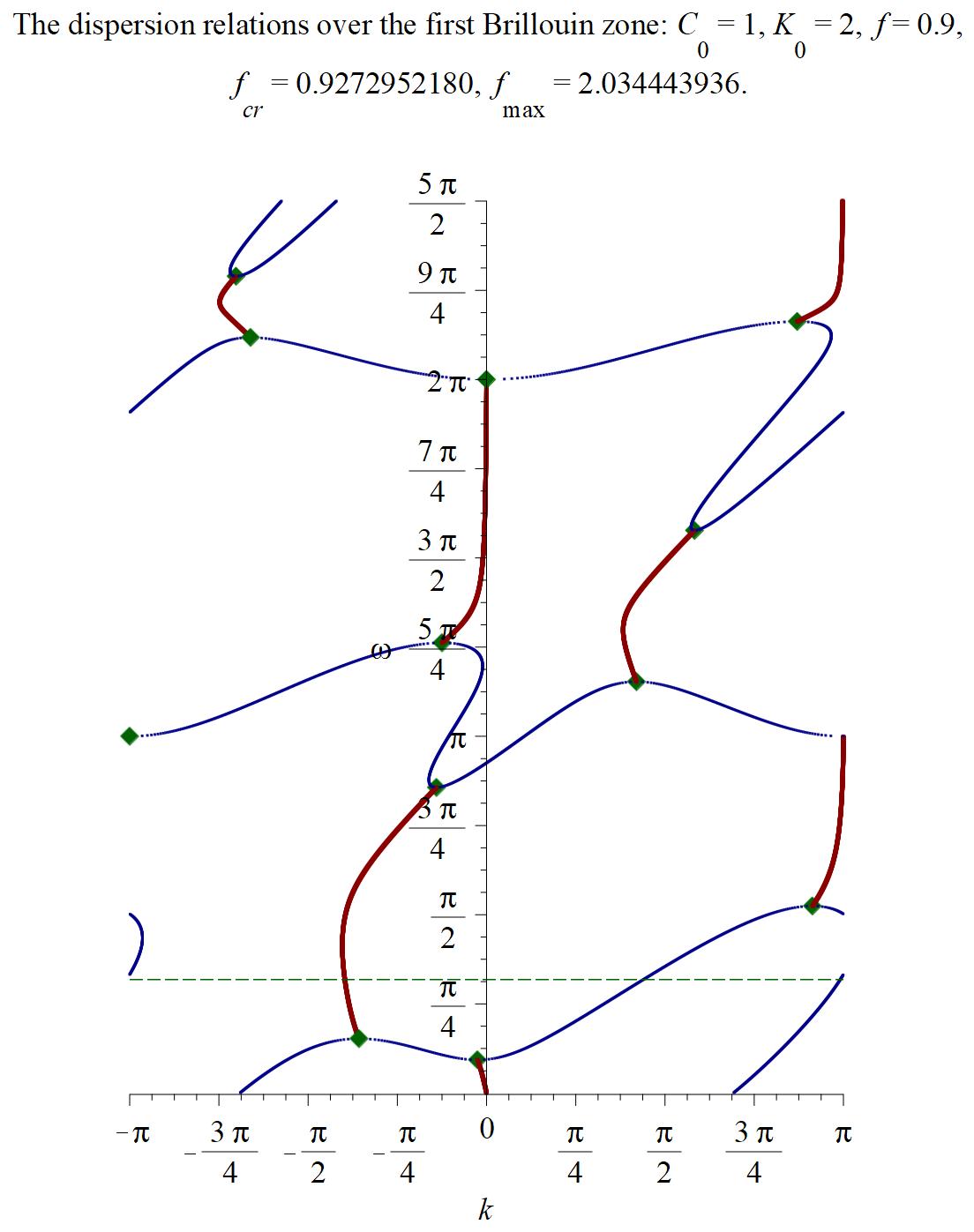}\hspace{0.3cm}\includegraphics[scale=0.12]{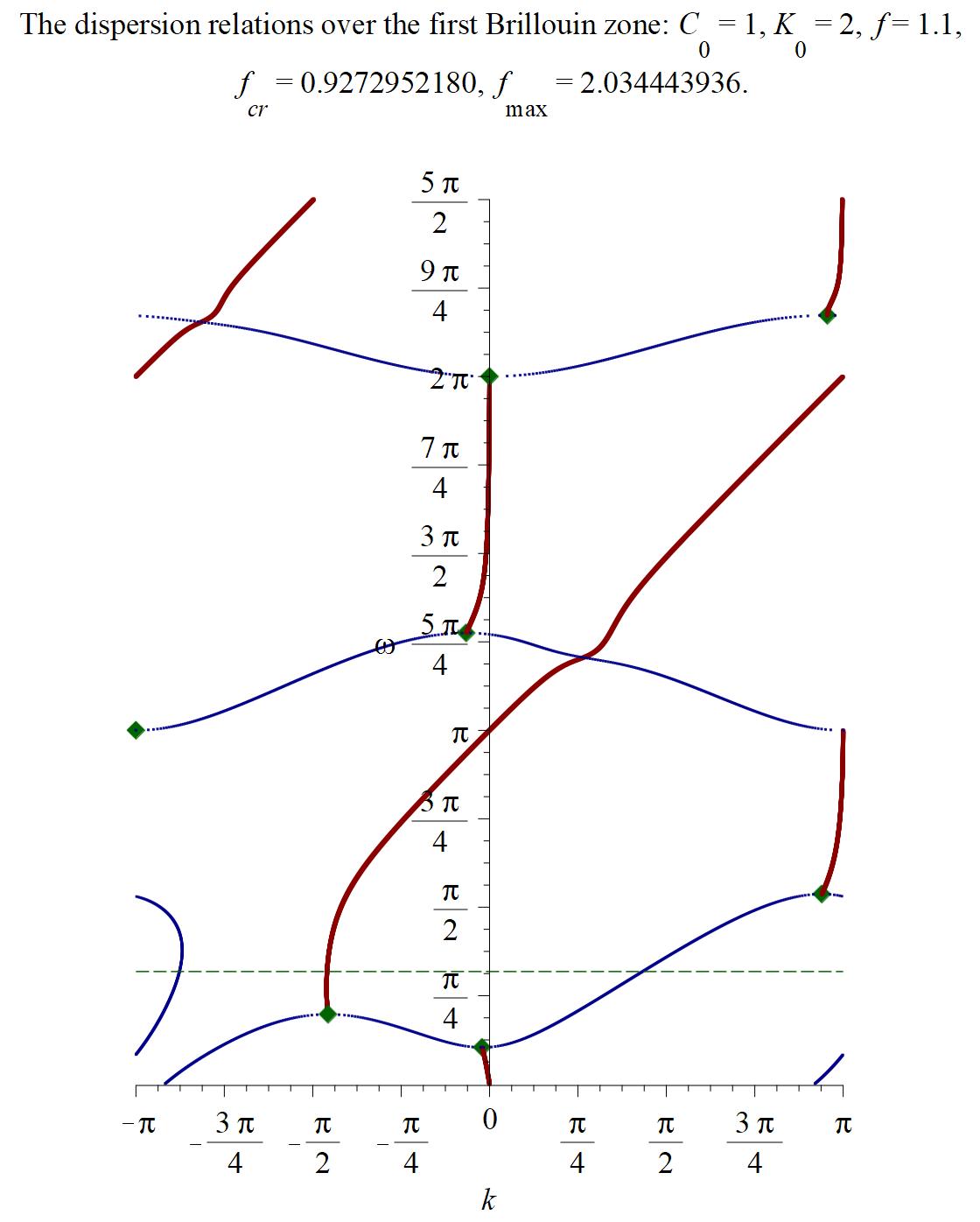}\hspace{0.2cm}\includegraphics[scale=0.12]{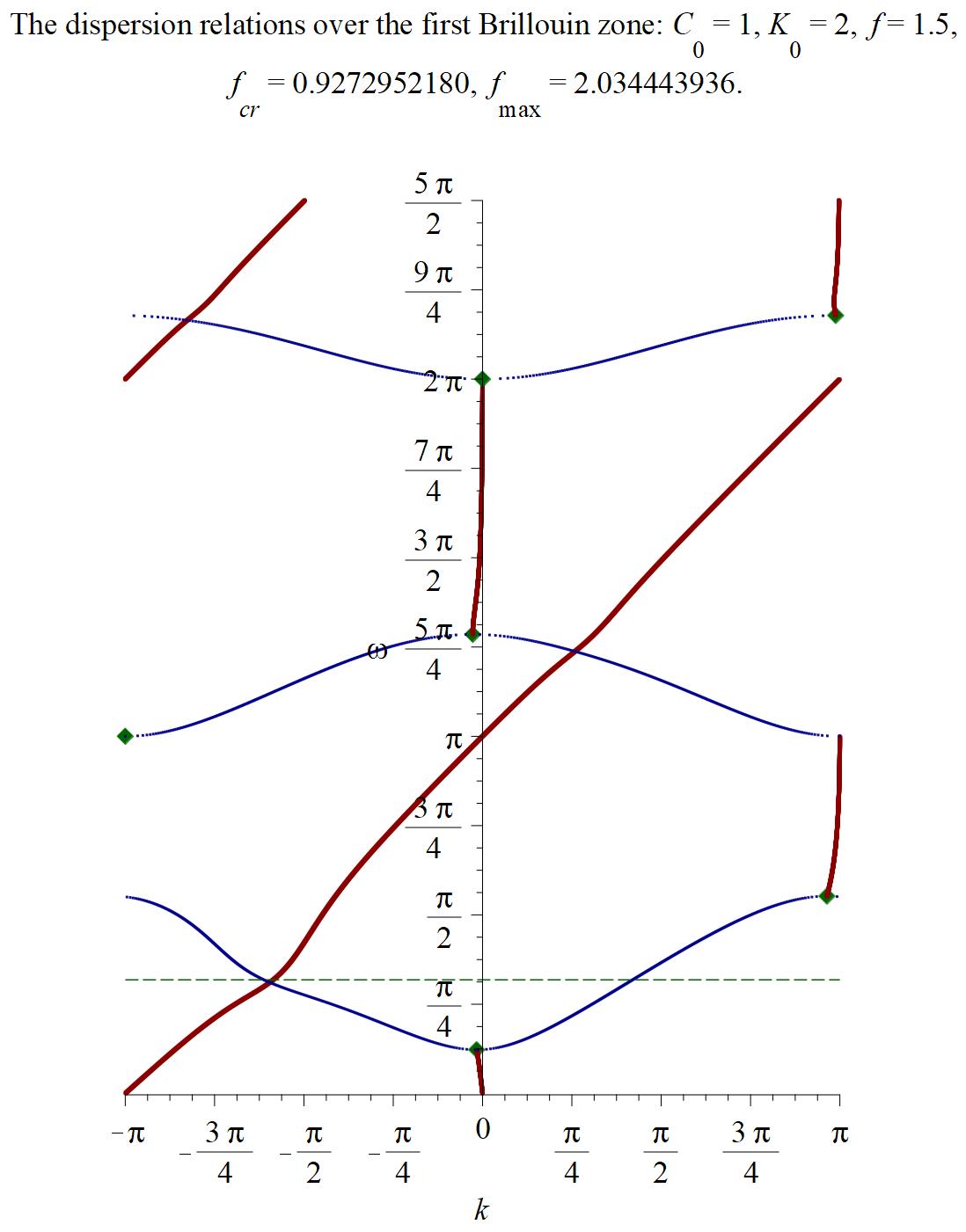}
\par\end{centering}
\centering{}(a)\hspace{4.1cm}(b)\hspace{4.1cm}(c)\caption{\label{fig:dis-serp-deg1} The dispersion-instability graphs for the
CCTWT showing the degeneracy (transition to instability) points as
diamond (green) dots for $C_{0}=1$, $K_{0}=2$ and consequently $f_{\mathrm{cr}}\protect\cong0.927$,
$f_{\mathrm{max}}\protect\cong1.770$: (a) $f=0.9<f_{\mathrm{cr}}$
; (b) $f=1.1>f_{\mathrm{cr}}$; (c) $f=1.5>f_{\mathrm{cr}}$ (see
Theorem \ref{thm:dispfac} and Remark \ref{rem:dispfac}). In all
plots the horizontal and vertical axes represent respectively $\Re\left\{ k\right\} $
and $\omega$. Solid (grin) diamond dots identify points of the transition
from the instability to the stability which are also EPD points}
\end{figure}

\section{Gain expression in terms of the Floquet multipliers\label{sec:cctwtgain}}

Based on the prior analysis we introduce the CCTWT gain $G$ in $\mathrm{dB}$
per one period as the rate of the exponential growth of the CCTWT
eigenmodes associated with the Floquet multipliers $s$ which are
the solutions to the characteristic equations (\ref{eq:monS2f})-(\ref{eq:monS2j}),
namely
\begin{equation}
G=G\left(f,\omega,K_{0}\right)=20\left|\log\left(\left|s\left(f,\omega,K_{0}\right)\right|\right)\right|.\label{eq:GdBgain1a}
\end{equation}
Consequently, to analyze the gain expression (\ref{eq:GdBgain1a})
we have to turn to the indicated characteristic equations.

In turns out that the gain expression (\ref{eq:GdBgain1a}) can be
significantly simplified under exact synchronism Assumption \ref{ass:cavcoup},
that is $\chi=1$ and $\omega_{0}=1$. Under these conditions the
CCTWT characteristic equations (\ref{eq:monS2f})-(\ref{eq:monS2j})
can be recast into the following form:
\begin{equation}
P=P\left(f,\omega,K_{0}\right)=s^{4}+p_{3}s^{3}+p_{2}s^{2}+p_{1}s+1=0,\label{eq:Pgas1a}
\end{equation}
where the coefficients of the \emph{CCTWT characteristic polynomial}
$P$ have the following expressions:
\begin{equation}
p_{3}=2{\rm e}^{\mathrm{i}{\it \omega}}b_{f}^{\infty}+C_{0}\frac{\omega^{2}-1}{\omega}\sin\left(\omega\right)-2\,\cos\left({\it \omega}\right),\label{eq:Pgas1b}
\end{equation}
\begin{gather}
p_{2}={\rm e}^{\mathrm{i}{\it \omega}}\left\{ 2b_{f}^{\infty}\left[C_{0}\omega\sin\left(\omega\right)-2\,\cos\left(\omega\right)\right]+2\cos\left(f\right)C_{0}\frac{\sin\left(\omega\right)}{\omega}+2\cos\left({\it \omega}\right)\right\} ,\label{eq:Pgas1c}
\end{gather}
\begin{equation}
p_{1}={\rm e}^{2\mathrm{i}{\it \omega}}\left[C_{0}\frac{\omega^{2}-1}{\omega}\sin\left(\omega\right)-2\,\cos\left({\it \omega}\right)\right]+2{\rm e}^{\mathrm{i}{\it \omega}}b_{f}^{\infty}={\rm e}^{2\mathrm{i}{\it \omega}}\bar{p}_{3},\label{eq:Pgasd}
\end{equation}
where
\begin{equation}
b_{f}^{\infty}=K_{0}\sin\left(f\right)-\cos\left(f\right),\quad K_{0}=\frac{b^{2}\beta_{0}}{2f}=\frac{b^{2}g_{\mathrm{B}}}{c_{0}},\quad g_{\mathrm{B}}=\frac{\sigma_{\mathrm{B}}}{4\lambda_{\mathrm{rp}}}.\label{eq:Pgas1e}
\end{equation}
According to equations (\ref{eq:monTc1b}) the CCS characteristic
polynomial $P_{\mathrm{C}}$ has the following expression:
\begin{equation}
P_{\mathrm{C}}\left(\omega,s\right)=s^{2}+\left[C_{0}\left(\omega-\frac{1}{\omega}\right)\sin\left(\omega\right)-2\cos\left(\omega\right)\right]s+1=0.\label{eq:Pgas2a}
\end{equation}
The MCK characteristic equation (\ref{eq:bfKpsi1bas}) can be recast
as follows:
\begin{equation}
P_{\mathrm{K}}\left(\omega,s\right)=s^{2}+2{\rm e}^{\mathrm{i}{\it \omega}}\left(\frac{2K_{0}\sin\left(f\right)\omega^{2}}{\omega^{2}-1}-\cos\left(f\right)\right)s+{\rm e}^{2\mathrm{i}{\it \omega}}=0,\label{eq:Pgas2b}
\end{equation}
and we refer to $P_{\mathrm{K}}$ as the \emph{MCK characteristic
polynomial}.

Using the definitions (\ref{eq:Pgas1a})-(\ref{eq:Pgas1a}), (\ref{eq:Pgas2a})
and (\ref{eq:Pgas2b}) for respectively the characteristic polynomials
$P\left(f,\omega,K_{0}\right)$, $P_{\mathrm{C}}\left(\omega,s\right)$
and $P_{\mathrm{K}}\left(\omega,s\right)$ one can identify their
leading terms $P^{\left(0\right)}$, $P_{\mathrm{C}}^{\left(0\right)}\left(\omega,s\right)$
and $P_{\mathrm{K}}^{\left(0\right)}$ as $\omega\rightarrow\infty$
which are:
\begin{equation}
P^{\left(0\right)}\left(\omega,s\right)={\it C_{0}\sin\left(\omega\right)s}\left(s^{2}+2b_{f}^{\infty}e^{\mathrm{i}\omega}s+{\rm e}^{2\mathrm{i}{\it \omega}}\right),\label{eq:Pgas2c}
\end{equation}
\begin{equation}
P_{\mathrm{C}}^{\left(0\right)}\left(\omega,s\right)=C_{0}\sin\left(\omega\right)s,\quad P_{\mathrm{K}}^{\left(0\right)}\left(\omega,s\right)=s^{2}+2b_{f}^{\infty}e^{\mathrm{i}\omega}s+{\rm e}^{2\mathrm{i}{\it \omega}},\label{eq:Pgas2d}
\end{equation}
where $b_{f}^{\infty}$ is defined by equations (\ref{eq:Pgas1e}).

Just as in the case of the dispersion relations that we analyzed in
Section \ref{sec:disprel} there are simple relations between the
CCTWT characteristic function $P\left(\omega,s\right)$ and the characteristic
polynomials $P_{\mathrm{C}}\left(\omega,s\right)$ and $P_{\mathrm{K}}\left(\omega,s\right)$
for respectively the CCS and MCK systems. These relations can be verified
by tedious but elementary algebraic evaluations and they are subjects
of the following theorem that relates the characteristic polynomials
for CCTTX, CCS and MCK systems.
\begin{thm}[characteristic polynomial factorization]
\label{thm:charfac} Let us assume that $\chi=\omega_{0}=1$. Let
the CCTWT, the CCS and the MCK dispersion functions $P\left(\omega,s\right)$,
$P_{\mathrm{C}}\left(\omega,s\right)$ and $P_{\mathrm{K}}\left(\omega,s\right)$
be defined by respectively equations (\ref{eq:Pgas1a})-(\ref{eq:Pgas1a}),
(\ref{eq:Pgas2a}) and (\ref{eq:Pgas2b}). Then the following identity
hold:
\begin{gather}
P\left(\omega,s\right)-P_{\mathrm{C}}\left(\omega,s\right)P_{\mathrm{K}}\left(\omega,s\right)=-\frac{2K_{0}\sin\left(f\right){\rm e}^{\mathrm{i}{\it \omega}}s}{\omega}\frac{\left(s^{2}-2\cos\left(\omega\right)s+1\right)}{\omega^{2}-1},\label{eq:charPP1a}\\
K_{0}=\frac{b^{2}\beta_{0}}{2f}=\frac{b^{2}g_{\mathrm{B}}}{c_{0}},\quad g_{\mathrm{B}}=\frac{\sigma_{\mathrm{B}}}{4\lambda_{\mathrm{rp}}}.\label{eq:charPP1b}
\end{gather}
In the case of the high-frequency approximation the following identity
holds:
\begin{equation}
P^{\left(0\right)}\left(\omega,s\right)=P_{\mathrm{C}}^{\left(0\right)}\left(\omega,s\right)P_{\mathrm{K}}^{\left(0\right)}\left(\omega,s\right)={\it C_{0}\sin\left(\omega\right)s}\left(s^{2}+2b_{f}^{\infty}e^{\mathrm{i}\omega}s+{\rm e}^{2\mathrm{i}{\it \omega}}\right).\label{eq:disDD2b-1}
\end{equation}
The identities (\ref{eq:charPP1a}) and (\ref{eq:disDD2b-1}) represent
a particular way the CCS and the MCK subsystems are coupled and integrated
into the CCTWT system. The right-hand side of the identity (\ref{eq:charPP1a})
can be naturally viewed as a measure of coupling between the CCS and
the MCK subsystems
\end{thm}

\begin{rem}[graphical confirmation of the characteristic polynomial factorization]
\label{rem:dispfac-1} The statements of the Theorem \ref{thm:charfac}
are well illustrated by Figures \ref{fig:gain-cctwt1} and \ref{fig:gain-cctwt2}
when compared with Figure \ref{fig:gain-ccs} for the CCS and Figure
\ref{fig:mck-gain1s} for the MCK. One can confidently recognize in
components of the graph of the gain CCTWT the patents of the graphs
for the gain of the CCS and the MCK.
\end{rem}

Let us consider now the convent
\begin{figure}[h]
\begin{centering}
\includegraphics[scale=0.1]{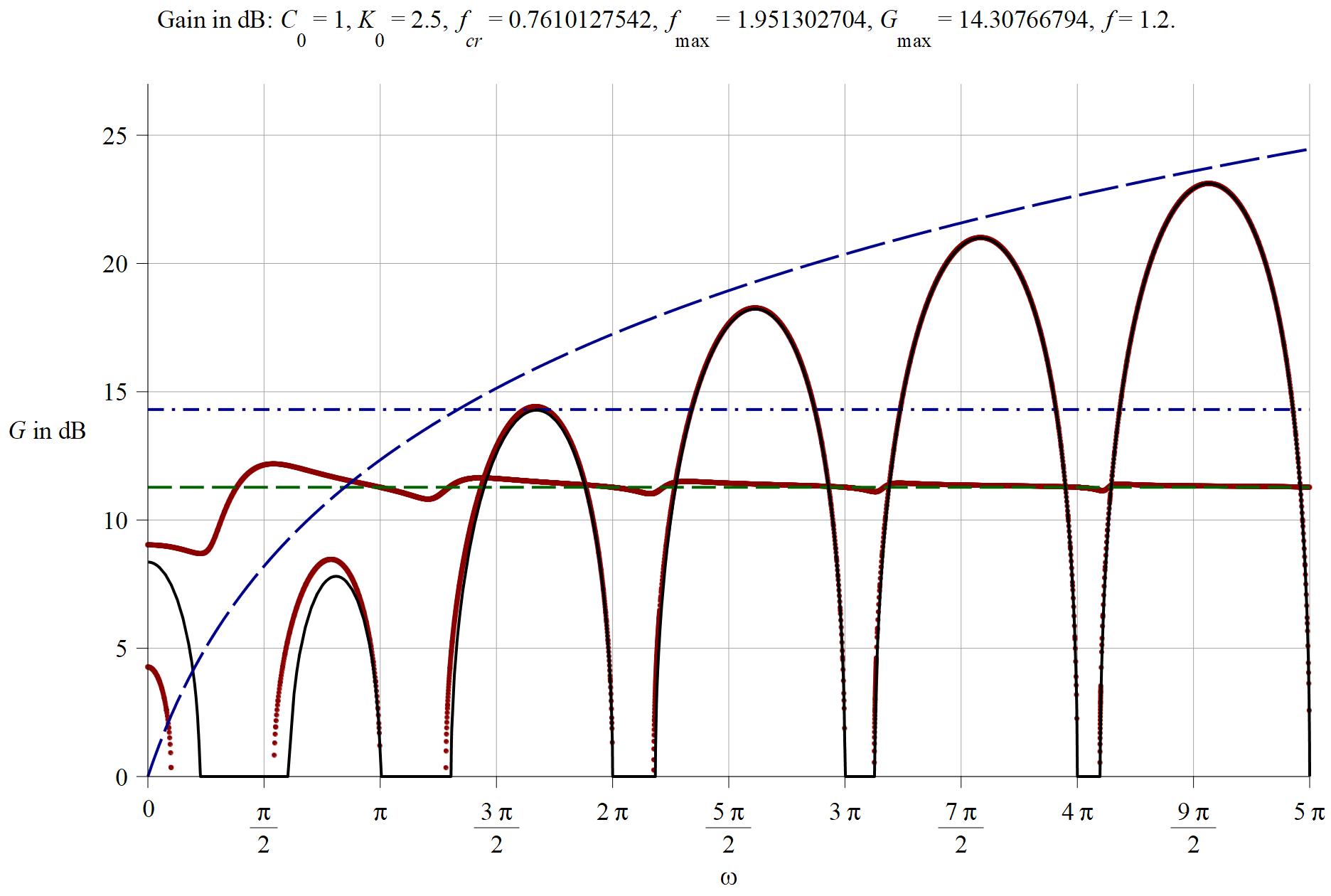}\hspace{1cm}\includegraphics[scale=0.1]{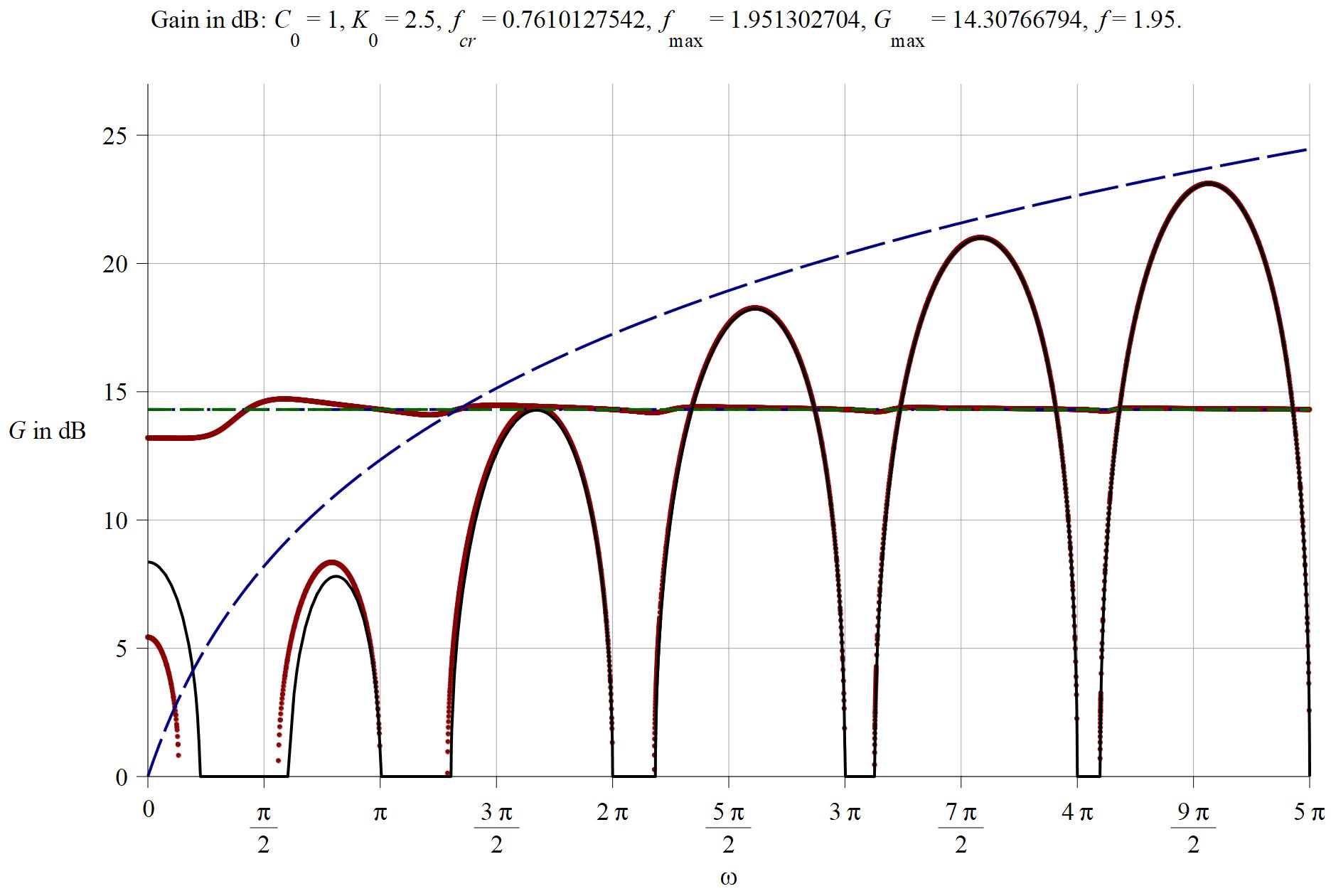}
\par\end{centering}
\centering{}(a)\hspace{7cm}(b)\caption{\label{fig:gain-cctwt1} Plots of gain $G$ per one period in $\mathrm{dB}$
as a function of frequency $\omega$ defined by equations (\ref{eq:GdBgain1a})
for $\omega_{0}=1$, $K_{0}=2.5$ and consequently $f_{\mathrm{cr}}\protect\cong0.7610127542$,
$f_{\mathrm{max}}\protect\cong1.951302704$, $G_{\mathrm{max}}=14.30766794$
(see Section \ref{subsec:mck-mon-gain} for the definition of the
MCK quantities $f_{\mathrm{cr}}$, $f_{\mathrm{max}}$ and $G_{\mathrm{max}}$)
and: (a) $f=1.2>f_{\mathrm{cr}}$; (b) $f=1.95\approx f_{\mathrm{max}}$.
In all plots the horizontal and vertical axes represent respectively
frequency $\omega$ and gain $G$ in $\mathrm{dB}$. The solid (brown)
curves represent gain $G$ as a function of frequency $\omega$; the
dashed horizontal (blue) line $G=G_{\mathrm{max}}$ represents the
maximal $G_{\mathrm{max}}$value of $G$ in the high frequency limit
(see Section \ref{subsec:mck-mon-gain}); the dashed horizontal (green)
line represents the value of $G$ in the high frequency limit for
given value of $f$ (see Section \ref{subsec:mck-mon-gain}). The
envelope of the local maxima of the gain for large values of frequency
$\omega$ behaves as $20\left|\log\left(C_{0}\omega\right)\right|$
(see captions to Fig. \ref{fig:gain-ccs}) and it is is shown as dashed
(blue) curve.}
\end{figure}
\begin{figure}[h]
\centering{}\includegraphics[scale=0.2]{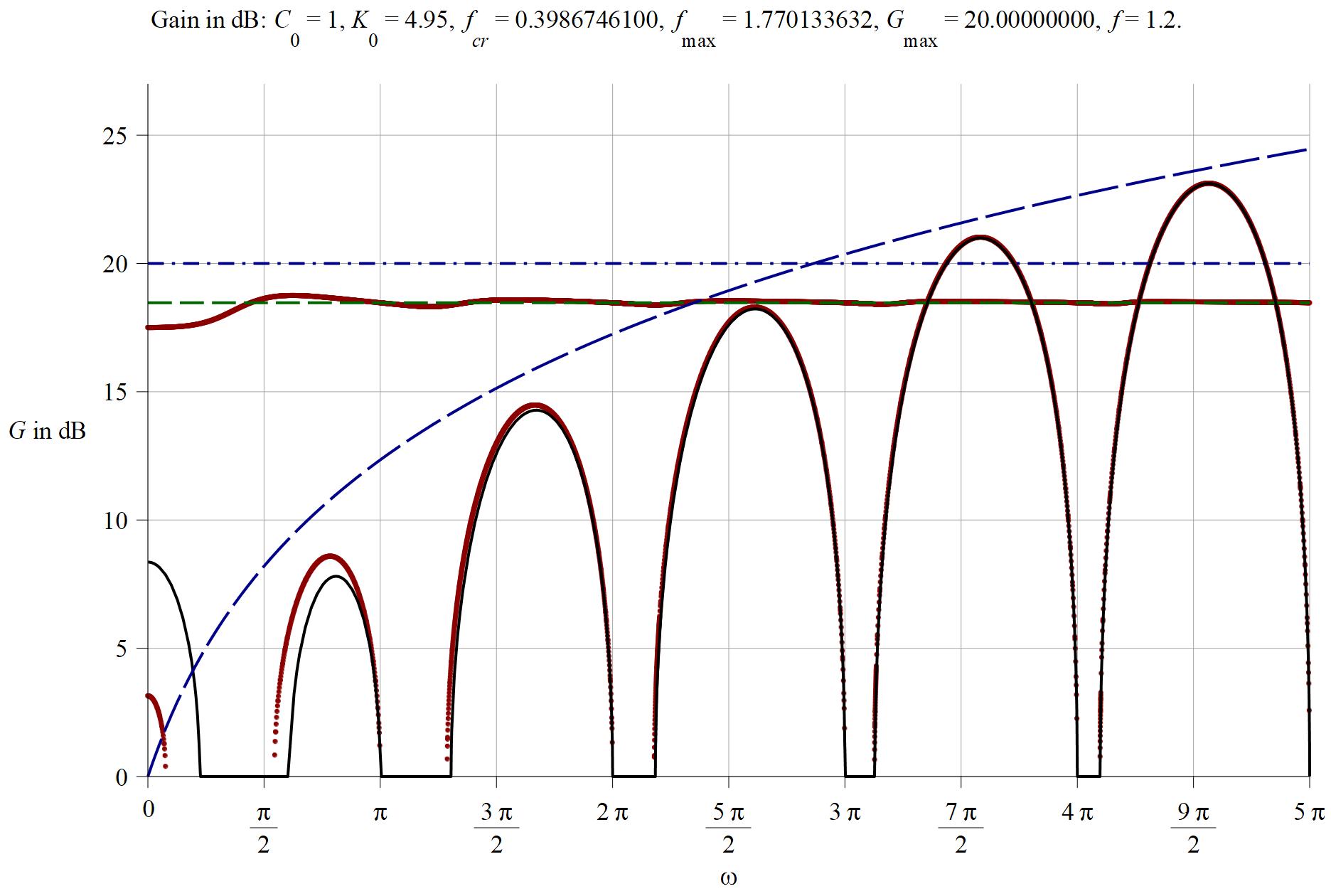}\caption{\label{fig:gain-cctwt2} Plot of gain $G$ per one period in $\mathrm{dB}$
as a function of frequency $\omega$ defined by equations (\ref{eq:ccWom2d})
for $\omega_{0}=1$, $K_{0}=K_{0\mathrm{T}}=4.95$ and consequently
$f_{\mathrm{cr}}\protect\cong0.3986746100$, $f_{\mathrm{max}}\protect\cong1.770133632$,
$G_{\mathrm{max}}=20$ (see Section \ref{subsec:mck-mon-gain} for
the definition of the MCK quantities $f_{\mathrm{cr}}$, $f_{\mathrm{max}}$
and $G_{\mathrm{max}}$) and $f=1.2>f_{\mathrm{cr}}$. The horizontal
and vertical axes represent respectively frequency $\omega$ and gain
$G$ in $\mathrm{dB}$. The solid (brown) curves represent gain $G$
as a function of frequency $\omega$; the dashed horizontal (blue)
line $G=G_{\mathrm{max}}$ represents the maximal $G_{\mathrm{max}}$value
of $G$ in the high frequency limit (see Section \ref{subsec:mck-mon-gain});
the dashed horizontal (green) line represents the value of $G$ in
the high frequency limit for $f=f_{\mathrm{cr}}$ (see Section \ref{subsec:mck-mon-gain}).
The envelope of the local maxima of the gain for large values of frequency
$\omega$ behaves as $20\left|\log\left(C_{0}\omega\right)\right|$
(see captions to Fig. \ref{fig:gain-ccs}) and it is shown as dashed
(blue) curve.}
\end{figure}

\begin{rem}[amplification in stopbands]
\label{rem:ampstop} E-beam interactions in periodic slow-wave structures
were studied by V. Solntsev in \cite{Solnt}. Under the condition
of exact synchronism as in our Assumption \ref{ass:cavcoup} the amplification
was observed in \emph{stopbands}, known also as spectral gaps in the
system (oscillatory) spectrum. Our theory accounts for this general
spectral phenomenon too as indicated by growing in magnitude ``bumps''
in Figures \ref{fig:gain-cctwt1} and \ref{fig:gain-cctwt2}. One
can also see similar bumps in Figure \ref{fig:gain-ccs} for the CCS.
\end{rem}

It is instructive to relate and compare the frequency dependent gain
$G$ per one period in $\mathrm{dB}$ for CCTWT defined by equations
(\ref{eq:GdBgain1a}), (\ref{eq:monS2f})-(\ref{eq:monS2j}) with
its expressions by equation (\ref{eq:sSbf1cs}) for MCK gain in Section
\ref{subsec:mck-mon-gain} and equations (\ref{eq:ccWom2d}) for CCS
gain $G_{\mathrm{C}}=G_{\mathrm{C}}\left(\omega,C_{0}\right)$.

\section{Sketch of the multicavity klystron analytical model\label{sec:mckrev}}

Usage of cavity resonators in the klystron was a revolutionary idea
of Hansen and the Varians, \cite[7.1]{Tsim}. In the pursuit of higher
power and efficiency the original design of Vairan klystrons evolve
significantly over years featuring today multiple cavities and multiple
electron beam, \cite[7.7]{Tsim}. The advantages of klystrons are
their high power and efficiency, potentially wide bandwidth, phase
and amplitude stability, \cite[9.1]{BenSweScha}.

The construction of an analytic model for the multicavity klystron
(MCK) in \cite{FigKly} utilizes elements of the analytic model of
the traveling wave tube (TWT) introduced and studied in our monograph
\cite[4, 24]{FigTWTbk}, see Section \ref{sec:twtmod}. Multicavity
klystron, known also as cascade amplifier, \cite[IIb]{Werne}, is
composed of the e-beam interacting with a periodic array of electromagnetic
cavities, see Fig. \ref{fig:mck}. Consequently the MCK can be naturally
viewed as a subsystem of the CCTWT that contributes to the properties
of CCTWT.
\begin{figure}[h]
\centering{}\hspace{1.2cm}\includegraphics[bb=0bp 120bp 1280bp 620bp,clip,scale=0.4]{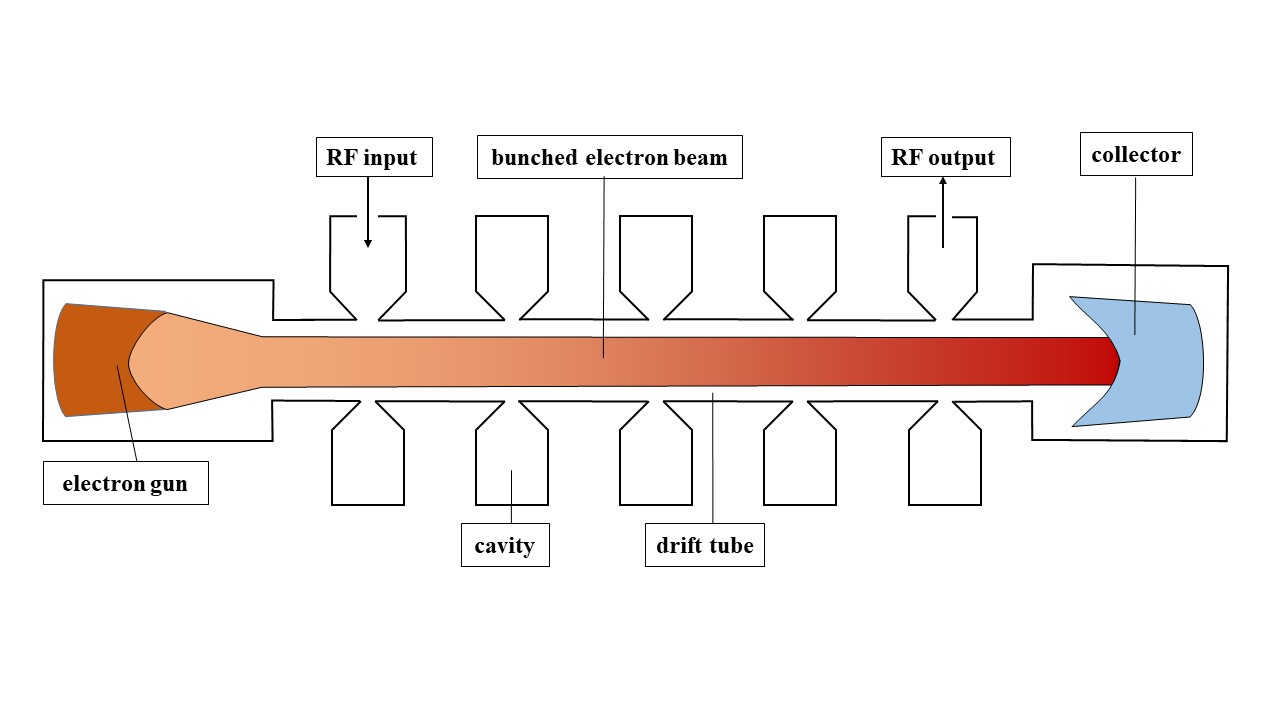}\caption{\label{fig:mck} A schematic representation of a multicavity klystron
(MCK) that exploits constructive interaction between the pencil-like
electron beam and an array of electromagnetic cavities (often of toroidal
shape). The interaction causes the electron bunching and consequent
amplification of the RF signal.}
\end{figure}

\subsection{The Euler-Lagrange equations in dimensionless variables\label{subsec:mck-EL}}

As to basic variables related to the e-beam and the klystron cavities
we refer the reader to Sections \ref{sec:twtmod}, Section \ref{subsec:dim-par}
and Tables \ref{tab:ebeam-par}, \ref{tab:cavity-par}Table \ref{tab:cctwt-par}
. The dimensionless form of the $\mathcal{L}^{\prime}$ of the Lagrangians
is as follows:
\begin{equation}
\mathcal{L}^{\prime}=\mathcal{L}_{\mathrm{B}}^{\prime}+\mathcal{L}_{\mathrm{CB}}^{\prime};\quad\mathcal{L}_{\mathrm{B}}^{\prime}=\frac{1}{2\beta^{\prime}}\left(\partial_{t^{\prime}}q+\partial_{z^{\prime}}q\right)^{2}-\frac{2\pi}{\sigma_{\mathrm{B}}^{\prime}}q^{2},\quad\ell\in\mathbb{Z},\label{eq:Lagdim1ak}
\end{equation}
\begin{equation}
\mathcal{L}_{\mathrm{CB}}^{\prime}=\sum_{\ell=-\infty}^{\infty}\delta\left(z^{\prime}-\ell\right)\left\{ \frac{l_{0}^{\prime}}{2}\left(\partial_{t}Q\left(a\ell\right)\right)^{2}-\frac{1}{2c_{0}^{\prime}}\left[Q\left(a\ell\right)+bq\left(a\ell\right)\right]^{2}\right\} .\label{eq:Lagdim1bk}
\end{equation}
\emph{Just as we did before to simplify notations we will omit prime
symbol in equations but rather will simply acknowledge their dimensionless
form}. The dimensionless form of the EL equations for the MCK is
\begin{gather}
\left(\partial_{t}+\partial_{z}\right)^{2}q+f^{2}q=0,\quad z\neq\ell,\quad\ell\in\mathbb{Z};\quad f=\frac{2\pi R_{\mathrm{sc}}\omega_{\mathrm{p}}}{\omega_{a}}=\frac{2\pi a}{\lambda_{\mathrm{rp}}},\label{eq:Lagdim1ek}
\end{gather}
\begin{gather}
\partial_{\prime}^{2}Q\left(a\ell\right)+\omega_{0}^{2}\left[Q\left(a\ell\right)+bq\left(a\ell\right)\right]=0,\quad\omega_{0}=\frac{1}{\sqrt{l_{0}c_{0}}},\quad\beta_{0}=\frac{\beta}{c_{0}},\label{eq:Lagdim1fk}\\
\left[\partial_{z}q\right]\left(a\ell\right)=-b\beta_{0}\left[Q\left(a\ell\right)+bq\left(a\ell\right)\right],\quad\ell\in\mathbb{Z}.\nonumber 
\end{gather}
Note that term $-\frac{2\pi}{\sigma_{\mathrm{B}}^{\prime}}q^{2}$
in the Lagrangian $\mathcal{L}_{\mathrm{B}}^{\prime}$ defined in
equations (\ref{eq:Lagdim1ak}) represents space-charge effects including
the so-called debunching (electron-to-electron repulsion).

The Fourier transform in $t$ (see Appendix \ref{sec:four}) of equations
(\ref{eq:Lagdim1ek}), (\ref{eq:Lagdim1fk}) is
\begin{equation}
\left(\partial_{z}-\mathrm{i}\omega\right)^{2}\check{q}+f^{2}\check{q}=0,\quad z\neq\ell,\label{eq:Lagdim2ak}
\end{equation}
subjects to the boundary conditions at the interaction points
\begin{gather}
\left[\check{q}\right]\left(a\ell\right)=0,\quad\left[\partial_{z}\check{q}\right]\left(a\ell\right)=-B\left(\omega\right)\check{q}\left(a\ell\right)\quad\ell\in\mathbb{Z},\label{eq:Lagdim2aak}
\end{gather}
where $\check{q}$ is the time Fourier transform of $q$ and $B\left(\omega\right)$
is an important parameter defined by
\begin{equation}
B=B\left(\omega\right)=\frac{b^{2}\beta_{0}\omega^{2}}{\omega^{2}-\omega_{0}^{2}}=B_{0}\frac{\omega^{2}}{\omega^{2}-\omega_{0}^{2}},\quad B_{0}=b^{2}\beta_{0}=\frac{b^{2}\beta}{c_{0}},\label{eq:Lagdim2B1}
\end{equation}
we refer to it as\emph{ cavity e-beam interaction parameter}. The
Fourier transform in time of equation (\ref{eq:Lagdim1fk}) yields
\begin{equation}
\check{Q}\left(a\ell\right)=\frac{\omega_{0}^{2}}{\omega^{2}-\omega_{0}^{2}}b\check{q}\left(a\ell\right),\quad\ell\in\mathbb{Z},\label{eq:Lagdim2abk}
\end{equation}
where $\check{Q}$ is the time Fourier transform of $Q$, and equation
(\ref{eq:Lagdim2abk}) was used to obtain the second equation in (\ref{eq:Lagdim2aak}).

Boundary conditions (\ref{eq:Lagdim2aak}) can be recast into the
matrix form as follows
\begin{equation}
X\left(a\ell+0\right)=\mathsf{S}_{\mathrm{b}}X\left(a\ell-0\right),\quad\mathsf{S}_{\mathrm{b}}=\left[\begin{array}{rr}
1 & 0\\
-B\left(\omega\right) & 1
\end{array}\right],\quad X=\left[\begin{array}{r}
\check{q}\\
\partial_{z}\check{q}
\end{array}\right],\quad B\left(\omega\right)=\frac{b^{2}\beta_{0}\omega^{2}}{\omega^{2}-\omega_{0}^{2}}.\label{eq:Lagdim2ack}
\end{equation}

In order to use the standard form of the Floquet theory reviewed in
Appendix \ref{sec:floquet} we recast the ordinary differential equations
(\ref{eq:Lagdim2ak}) with boundary (interface) conditions (\ref{eq:Lagdim2aak})
as the following single second-order ordinary differential equation
with singular, frequency dependent, periodic potential:
\begin{equation}
\partial_{z}^{2}\check{q}-2\mathrm{i}\omega\partial_{z}\check{q}+\left(f^{2}-\omega^{2}\right)\check{q}-B\left(\omega\right)p\left(z\right)\check{q}=0,\quad p\left(z\right)=\sum_{\ell=-\infty}^{\infty}\delta\left(z-\ell\right),\quad\check{q}=\check{q}\left(z\right),\label{eq:Lagdim4a}
\end{equation}
where the second interaction parameter $B\left(\omega\right)$ is
defined by equation (\ref{eq:Lagdim2B1}).

Analysis of equations (\ref{eq:Lagdim4a}) based on the Floquet theory
(see Appendix \ref{sec:floquet}) becomes now the primary subject
of studies. The second-order ordinary differential equation (\ref{eq:Lagdim4a})
can in turn be recast into the following matrix ordinary differential
equation
\begin{gather}
\partial_{z}X=A_{\mathrm{K}}\left(z\right)X,\quad A_{\mathrm{K}}\left(z\right)=A_{\mathrm{K}}\left(z,\omega\right)=\left[\begin{array}{rr}
0 & 1\\
\omega^{2}-f^{2}+B\left(\omega\right)p\left(z\right) & 2\mathrm{i}\omega
\end{array}\right],\quad X=\left[\begin{array}{r}
q\\
\partial_{z}q
\end{array}\right],\label{eq:Lagdim4b}\\
B\left(\omega\right)=\frac{b^{2}\beta_{0}\omega^{2}}{\omega^{2}-\omega_{0}^{2}}=2fK_{0}\frac{\omega^{2}}{\omega^{2}-\omega_{0}^{2}},\quad p\left(z\right)=\sum_{\ell=-\infty}^{\infty}\delta\left(z-\ell\right).\nonumber 
\end{gather}
Note that normalized period $f=\frac{2\pi a}{\lambda_{\mathrm{rp}}}$
and the MCK gain coefficient $K_{0}=\frac{b^{2}g_{\mathrm{B}}}{c_{0}}$
play particularly significant roles for the MCK properties. 

One can verify by straightforward evaluation that equation (\ref{eq:Lagdim4a})
has the Hamiltonian structure (see Appendix \ref{sec:Ham}) with the
following selection for the metric matrix
\begin{equation}
G_{\mathrm{K}}=G_{\mathrm{K}}^{*}=\left[\begin{array}{rr}
2\omega & \mathrm{i}\\
-\mathrm{i} & 0
\end{array}\right],\quad G_{\mathrm{K}}^{-1}=\left[\begin{array}{rr}
0 & \mathrm{i}\\
-\mathrm{i} & -2\omega
\end{array}\right],\quad\det\left[G_{\mathrm{K}}\right]=-1.\label{eq:Lagdim4c}
\end{equation}
The eigenvalues are eigenvectors of metric matrix $G_{\mathrm{K}}$
are as follows:
\begin{equation}
\omega+\sqrt{\omega^{2}+1},\;\left[\begin{array}{r}
\frac{\mathrm{i}}{-\omega+\sqrt{\omega^{2}+1}}\\
1
\end{array}\right];\quad\omega-\sqrt{\omega^{2}+1},\;\left[\begin{array}{r}
-\frac{\mathrm{i}}{\omega+\sqrt{\omega^{2}+1}}\\
1
\end{array}\right].\label{eq:Lagdim4d}
\end{equation}
Using expressions (\ref{eq:Lagdim4b}) and (\ref{eq:Lagdim4c}) for
respectively matrices $A\left(z\right)$ and $G$ one can readily
verify that $A\left(z\right)$ is $G$-skew-Hermitian matrix, that
is
\begin{equation}
G_{\mathrm{K}}A_{\mathrm{K}}\left(z\right)+A_{\mathrm{K}}^{*}\left(z\right)G_{\mathrm{K}}=0,\label{eq:Lagdim4e}
\end{equation}
and that according to Appendix \ref{sec:Ham} implies that the system
(\ref{eq:Lagdim4b}) is Hamiltonian. Consequently, according to Appendix
\ref{sec:Ham} the matrizant $\Phi_{\mathrm{K}}\left(z\right)$ of
the Hamiltonian system (\ref{eq:Lagdim4b}) ) is $G_{\mathrm{K}}$-unitary
and its spectrum $\sigma\left\{ \Phi_{\mathrm{K}}\left(z\right)\right\} $
is symmetric with respect to the unit circle, that is
\begin{equation}
\Phi_{\mathrm{K}}^{*}\left(z\right)G_{\mathrm{K}}\Phi_{\mathrm{K}}\left(z\right)=G,\quad\zeta\in\sigma\left\{ \Phi_{\mathrm{K}}\left(z\right)\right\} \Rightarrow\frac{1}{\bar{\zeta}}\in\sigma\left\{ \Phi_{\mathrm{K}}\left(z\right)\right\} .\label{eq:Lagdim4f}
\end{equation}

\subsection{The monodromy matrix, the dispersion-instability relations and the
gain\label{subsec:mck-mon-gain}}

The MCK monodromy matrix $\mathscr{T}_{\mathrm{K}}$ (see Appendix
\ref{sec:floquet}) is as follows:

\begin{gather}
\mathscr{T}_{\mathrm{K}}=e^{\mathrm{i}\omega}\left[\begin{array}{rr}
\cos\left(f\right)-{\it \mathrm{i}\omega}\mathrm{sinc}\,\left(f\right) & \mathrm{sinc}\,\left(f\right)\\
\mathrm{sinc}\,\left(f\right)\omega^{2}+2{\it \mathrm{i}\omega\left(\cos\left(f\right)+b_{f}\right)}-\frac{2b_{f}\cos\left(f\right)+\cos^{2}\left(f\right)+1}{\mathrm{sinc}\,\left(f\right)} & \mathrm{i}\omega\mathrm{sinc}\,\left(f\right)-\cos\left(f\right)-2b_{f}
\end{array}\right],\label{eq:abcha3dkLs}
\end{gather}
where
\begin{equation}
b_{f}=\frac{B\left(\omega\right)}{2}\mathrm{sinc}\,\left(f\right)-\cos\left(f\right),\quad\mathrm{sinc}\,\left(f\right)=\frac{\sin\left(f\right)}{f},\quad B\left(\omega\right)=\frac{b^{2}\beta_{0}\omega^{2}}{\omega^{2}-\omega_{0}^{2}}.\label{eq:abcha3ekLs}
\end{equation}
We assume that the MCK normalized period $f=\frac{2\pi a}{\lambda_{\mathrm{rp}}}$
, an important parameter that effects the instability, satisfies the
following ineqalities.

\begin{assumption} \label{ass:fpi}(smaller MCK period). The MCK
normalized period $f$ satisfies the following bounds:
\begin{equation}
0<f=\frac{2\pi a}{\lambda_{\mathrm{rp}}}<\pi.\label{eq:bfspm2as}
\end{equation}

\end{assumption}

The MCK gain is defined by the following expression:
\begin{gather}
G=G\left(f,\omega,K_{0}\right)=\left\{ \begin{array}{ccc}
20\left|\log\left(\left|\left|b_{f}\right|+\sqrt{b_{f}^{2}-1}\right|\right)\right| & \text{if} & \left|b_{f}\right|>1\\
0 & \text{if} & \left|b_{f}\right|\leq1
\end{array}\right.,\label{eq:sSbf1cs}\\
b_{f}=b_{f}\left(\omega\right)=K\left(\omega\right)\sin\left(f\right)-\cos\left(f\right),\quad K\left(\omega\right)=K_{0}\frac{\omega^{2}}{\omega^{2}-\omega_{0}^{2}},\quad K_{0}=\frac{b^{2}\beta_{0}}{2f}=\frac{b^{2}g_{\mathrm{B}}}{c_{0}}.\nonumber 
\end{gather}
Note that the following high-frequency decomposition holds for the
instability parameter $b_{f}$:
\begin{equation}
b_{f}\left(\omega\right)=b_{f}^{\infty}+K_{0}\frac{\omega_{0}^{2}}{\omega^{2}-\omega_{0}^{2}},\label{eq:bfKinf0s}
\end{equation}
where

\begin{equation}
b_{f}^{\infty}=\lim_{\omega\rightarrow\infty}b_{f}\left(\omega\right)=K_{0}\sin\left(f\right)-\cos\left(f\right),\quad K_{0}=\frac{b^{2}g_{\mathrm{B}}}{c_{0}},\label{eq:bfKinf1s}
\end{equation}

It turns out that the high-frequency limit $b_{f}^{\infty}$ of instability
parameter $b_{f}$ defined by equation (\ref{eq:bfKinf1s}) plays
significant role in the analysis of the MCK instability and its gain.
In particular, there exists a unique value $f_{\mathrm{cr}}$ on interval
$\left(0,\pi\right)$ of the normalized period $f$ such that

\begin{equation}
b_{f_{\mathrm{cr}}}^{\infty}=1,\quad0<f_{\mathrm{cr}}<\pi,\label{eq:bfomf1cs}
\end{equation}
and we refer to it as the\emph{ critical value} and the following
representation holds
\begin{equation}
f_{\mathrm{cr}}=2\arctan\left(\frac{1}{K_{0}}\right),\quad\text{ where }K_{0}=\frac{b^{2}g_{\mathrm{B}}}{c_{0}},\quad g_{\mathrm{B}}=\frac{\sigma_{\mathrm{B}}}{4\lambda_{\mathrm{rp}}},\quad\left|\arctan\left(*\right)\right|<\frac{\pi}{2}.\label{eq:bfomf1ds}
\end{equation}
The significance of the critical value $f_{\mathrm{cr}}$ is that
for for $f_{\mathrm{cr}}<f<\pi$ any $\omega>\omega_{0}$ is an instability
frequency. One can see that in Figure \ref{fig:mck-gain1s} showing
the frequency dependence of the gain $G$ and its asymptotic behavior
as $\omega\rightarrow+\infty$.

The maximal value
\begin{equation}
G_{\mathrm{max}}=20\left|\log\left(\left|K_{0}+\sqrt{K_{0}^{2}+1}\right|\right)\right|=20\frac{\ln\left(K_{0}+\sqrt{K_{0}^{2}+1}\right)}{\ln\left(10\right)}.\label{eq:sSbf2bs}
\end{equation}
of gain $G$ is attained at $f=f_{\mathrm{max}}$ that satisfies

\begin{equation}
f_{\mathrm{max}}=\pi-\arctan\left(K_{0}\right),\quad f_{\mathrm{cr}}<f_{\mathrm{max}}<\pi.\label{eq:sSbf2as}
\end{equation}
\begin{figure}[h]
\begin{centering}
\includegraphics[scale=0.35]{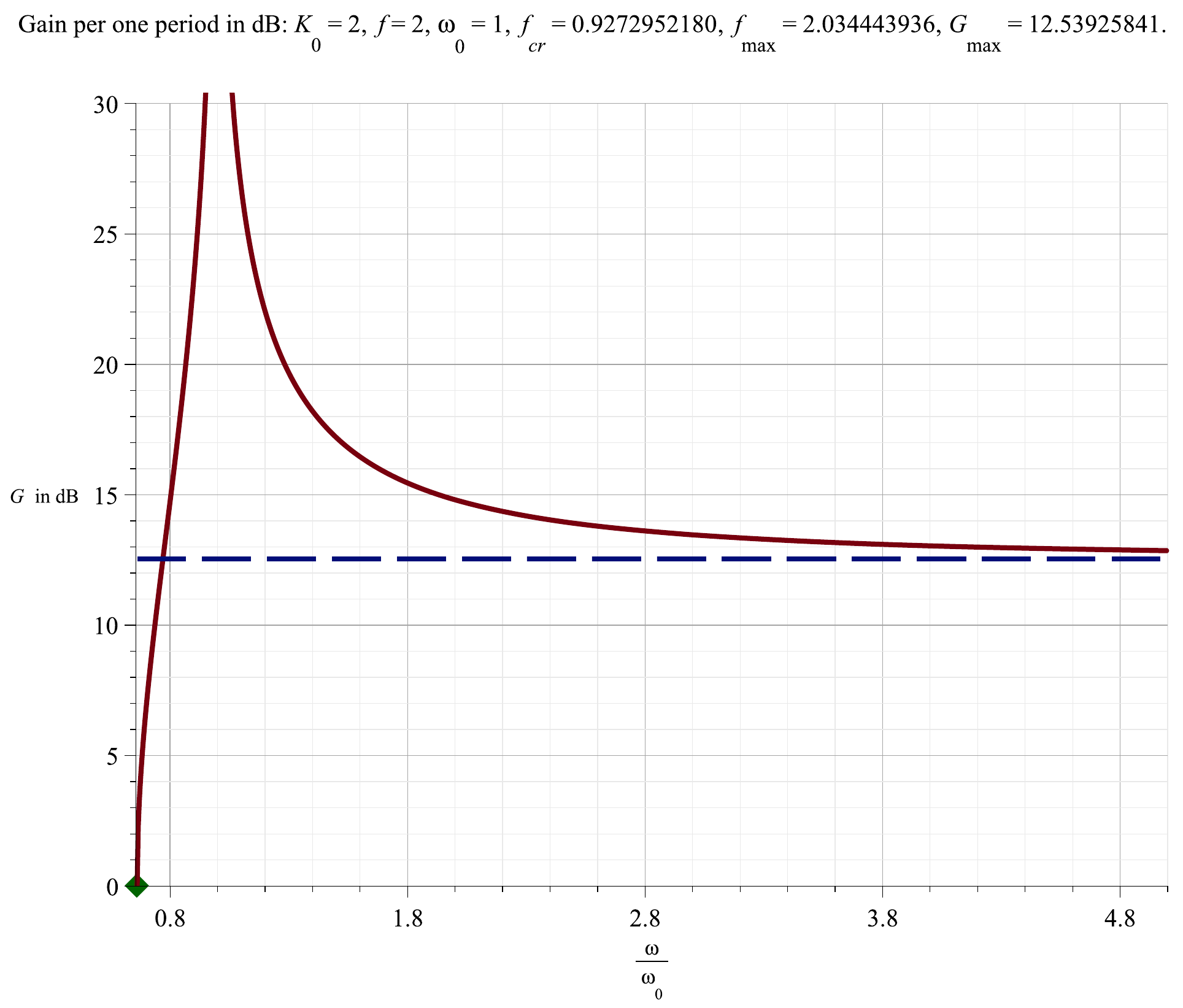}\hspace{1cm}\includegraphics[scale=0.36]{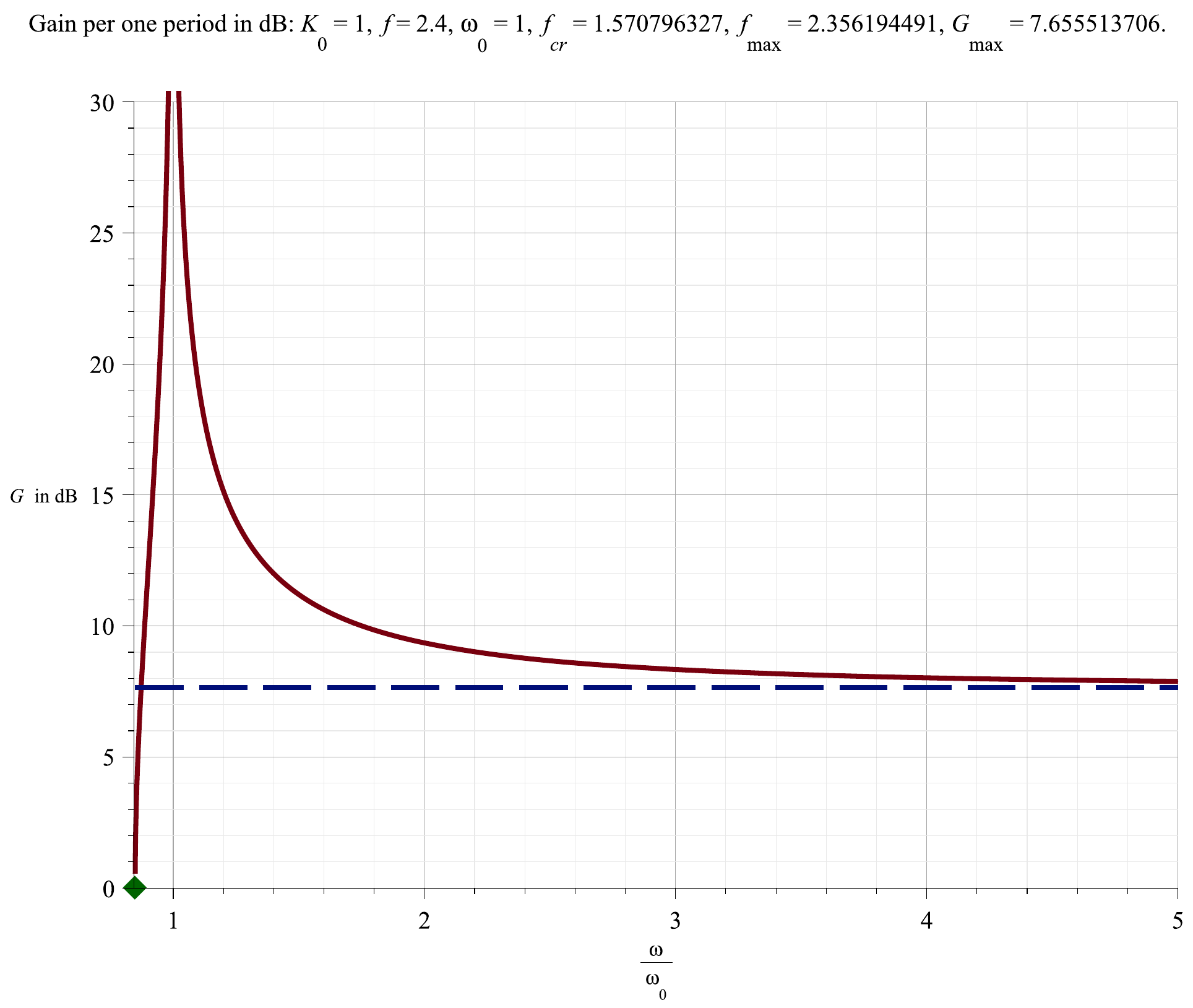}
\par\end{centering}
\centering{}\hspace{1cm}(a)\hspace{7.5cm}(b)\caption{\label{fig:mck-gain1s} Plots of gain $G$ as a function of frequency
$\omega$ defined by equations (\ref{eq:sSbf1cs}) for $\omega_{0}=1$
and (a) $K_{0}=2$, $f=2>f_{\mathrm{cr}}$ with $f_{\mathrm{cr}}\protect\cong0.9272952180$,
$f_{\mathrm{max}}\protect\cong2.034443936$ and $G_{\mathrm{max}}\protect\cong12.53925841$;
(b) $K_{0}=1$, $f=2.4>f_{\mathrm{cr}}$ with $f_{\mathrm{cr}}\protect\cong1.570796327$,
$f_{\mathrm{max}}\protect\cong2.356194491$ and $G_{\mathrm{max}}\protect\cong7.655513706$.
In all plots the horizontal and vertical axes represent respectively
frequency $\omega$ and gain $G$ in $\mathrm{dB}$. The solid (brown)
curves represent gain $G$ as a function of frequency $\omega$, the
dashed (blue) line $G=G_{\mathrm{max}}$ represents the maximal $G_{\mathrm{max}}$value
of $G$ in the high frequency limit. The diamond solid (green) dots
mark the values of $\varOmega_{f}^{-}$ which is the lower frequency
boundary of the instability interval.}
\end{figure}
\begin{figure}[h]
\begin{centering}
\includegraphics[scale=0.6]{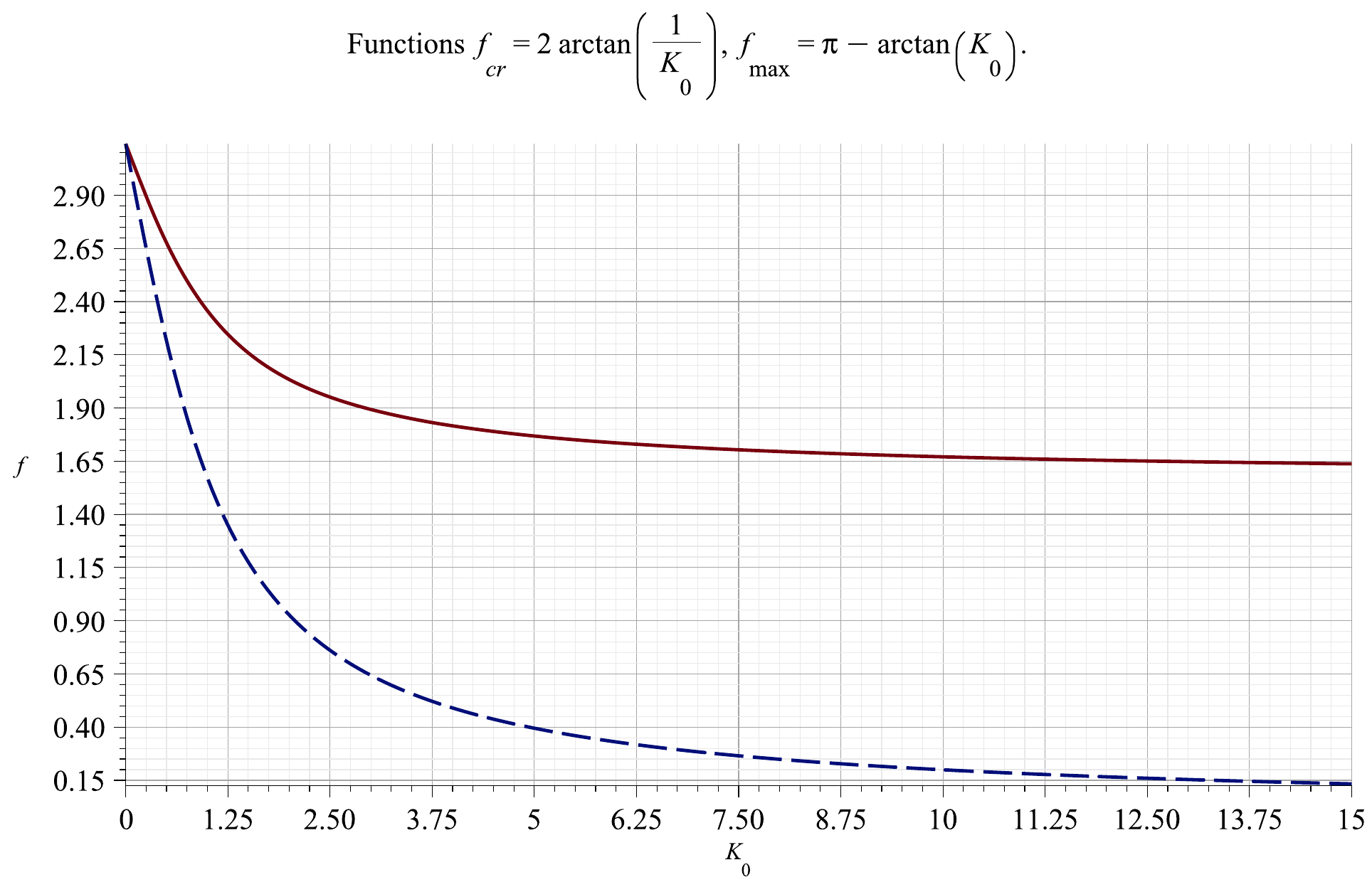}
\par\end{centering}
\centering{}\caption{\label{fig:mck-fcr-fop} Plots of $f_{\mathrm{max}}=\pi-\arctan\left(K_{0}\right)$
as solid (brown) curve and $f_{\mathrm{cr}}=2\arctan\left(\frac{1}{K_{0}}\right)$
as dashed (blue) curve. The horizontal and vertical axes represent
respectively $K_{0}$ and $f$.}
\end{figure}

Let $s=\exp\left\{ \mathrm{i}k\right\} $ where $k$ is the wave number
be the Floquet multiplier of the monodromy matrix $\mathscr{T}_{\mathrm{K}}$
defined by equations (\ref{eq:abcha3dkLs}), (\ref{eq:abcha3ekLs}),
(see Section \ref{sec:floquet} and Remark \ref{rem:disprel}). Then
the two Floquet multipliers $s_{\pm}$ are solutions to the characteristic
equation $\det\left\{ \mathscr{T}_{\mathrm{K}}-s\mathbb{I}\right\} =0$
which is, \cite{FigKly}:
\begin{gather}
s=e^{\mathrm{i}k}=e^{\mathrm{i}\omega}S:\;S^{2}+2b_{f}S+1=0;\quad S=S_{\pm}=-b_{f}\pm\sqrt{b_{f}^{2}-1},\label{eq:bfKpsi1bas}
\end{gather}
readily implying
\begin{gather}
s_{\pm}=e^{\mathrm{i}k_{\pm}}=e^{\mathrm{i}\omega}S_{\pm}=e^{\mathrm{i}\omega}\left(-b_{f}\pm\sqrt{b_{f}^{2}-1}\right),\quad b_{f}=b_{f}\left(\omega\right)=K\left(\omega\right)\sin\left(f\right)-\cos\left(f\right),\label{eq:bfKpsi1bbs}\\
K\left(\omega\right)=K_{0}\frac{\omega^{2}}{\omega^{2}-\omega_{0}^{2}},\quad K_{0}=\frac{b^{2}\beta_{0}}{2f}=\frac{b^{2}g_{\mathrm{B}}}{c_{0}},\nonumber 
\end{gather}
Equations (\ref{eq:bfKpsi1bas}) show that\emph{ parameter $b_{f}$
completely determines the two Floquet multipliers $s_{\pm}$} justifying
its its name the instability parameter.\emph{ Importantly, the characteristic
equation (\ref{eq:bfKpsi1bas}) can be viewed as an expression of
the dispersion relations between the frequency $\omega$ and the wavenumber
$k$} as we discuss in Section \ref{sec:disprel}. Equation (\ref{eq:bfKpsi1bas})
can be readily recast as

\begin{equation}
S+2b_{f}+S^{-1}=0,\quad S=s\exp\left\{ -\mathrm{i}{\it \omega}\right\} =\exp\left\{ \mathrm{i}\left(k-{\it \omega}\right)\right\} ,\quad s=\exp\left\{ \mathrm{i}k\right\} ,\label{eq:dispbf1a}
\end{equation}
or, equivalently, a
\begin{gather}
\cos\left(k-{\it \omega}\right)+b_{f}\left(\omega\right)=0,\quad b_{f}\left(\omega\right)=K\left(\omega\right)\sin\left(f\right)-\cos\left(f\right),\label{eq:dispbf1b}
\end{gather}
Equations (\ref{eq:dispbf1a}) and (\ref{eq:dispbf1b}) can be viewed
as expressions of the dispersion relations between the frequency $\omega$
and the wavenumber $k$ and we will refer to it as the MCK dispersion
relations.\emph{ }Dispersion relation (\ref{eq:dispbf1b}) can be
readily recast as
\begin{gather}
k_{\pm}\left(\omega\right)=\omega\pm\arccos\left(-b_{f}\left(\omega\right)\right),\quad b_{f}\left(\omega\right)=K\left(\omega\right)\sin\left(f\right)-\cos\left(f\right),\label{eq:dispbf1c}\\
K\left(\omega\right)=K_{0}\frac{\omega^{2}}{\omega^{2}-\omega_{0}^{2}},\quad K_{0}=\frac{b^{2}\beta_{0}}{2f}=\frac{b^{2}g_{\mathrm{B}}}{c_{0}}.\nonumber 
\end{gather}

Equation (\ref{eq:dispbf1c}) in turn can recast into even more explicit
form as stated in the following theorem, \cite{FigKly}.
\begin{thm}[MCK dispersion relations]
\label{thm:mckdis} Let $s_{\pm}$ be the MCK Floquet multipliers,
that is solutions to equations (\ref{eq:dispbf1a}), and let $k_{\pm}\left(\omega\right)$
be the corresponding complex-valued wave numbers satisfying
\begin{equation}
s_{\pm}=s_{\pm}\left(\omega\right)=\exp\left\{ \mathrm{i}k_{\pm}\left(\omega\right)\right\} ,\label{eq:spmpo1ds}
\end{equation}
Then the following representation for $k_{\pm}\left(\omega\right)$
holds
\begin{equation}
k_{\pm}\left(\omega\right)=\left\{ \begin{array}{rcr}
-\frac{1+\mathrm{sign}\,\left\{ b_{f}\left(\omega\right)\right\} }{2}\pi+\omega+2\pi m\pm\mathrm{i}\ln\left[\left(\left|b_{f}\left(\omega\right)\right|+\sqrt{b_{f}^{2}\left(\omega\right)-1}\right)\right] & \text{if } & b_{f}^{2}>1\\
-\frac{1+\mathrm{sign}\,\left\{ b_{f}\left(\omega\right)\right\} }{2}\pi+\omega+2\pi m\pm\arccos\left(\left|b_{f}\left(\omega\right)\right|\right) & \text{if } & b_{f}^{2}\leq1
\end{array}\right.,\quad m\in\mathbb{Z},\label{eq:spmpo1es}
\end{equation}
where $0<f<\pi$ and 
\begin{equation}
b_{f}\left(\omega\right)=K_{0}\frac{\omega^{2}}{\omega^{2}-\omega_{0}^{2}}\sin\left(f\right)-\cos\left(f\right),\quad K_{0}=\frac{b^{2}\beta_{0}}{2}=\frac{b^{2}g_{\mathrm{B}}}{c_{0}},\quad g_{\mathrm{B}}=\frac{\sigma_{\mathrm{B}}}{4\lambda_{\mathrm{rp}}}.\label{eq:spmpol1eas}
\end{equation}
Requirement for $\Re\left\{ k_{\pm}\left(\omega\right)\right\} $
to be in the first (main) Brillouin zone $\left(-\pi,\pi\right]$
effectively selects the band number $m$ that depend on $\omega$
as follows. For any given $\omega>0$ and $0<f<\pi$ the band number
$m\in\mathbb{Z}$ is determined by the requirement to satisfy the
following inequalities:
\begin{equation}
\begin{array}{rcr}
-\pi<-\frac{1+\mathrm{sign}\,\left\{ b_{f}\left(\omega\right)\right\} }{2}\pi+\omega+2\pi m\leq\pi, & \text{if } & b_{f}^{2}\left(\omega\right)>1\\
-\pi<-\frac{1+\mathrm{sign}\,\left\{ b_{f}\left(\omega\right)\right\} }{2}\pi\pm\arccos\left(-b_{f}\left(\omega\right)\right)+\omega+2\pi m\leq\pi, & \text{if } & b_{f}^{2}\left(\omega\right)<1
\end{array}.\label{eq:spmpo1hs}
\end{equation}
\emph{The equations (\ref{eq:spmpo1es}) for the complex-valued wave
numbers $k_{\pm}\left(\omega\right)$ represent the dispersion relations
of the MCK.} 
\end{thm}

\begin{rem}
Note that according to expression\emph{ (\ref{eq:spmpo1es}) }in Theorem
\ref{thm:mckdis} we have\emph{
\begin{equation}
\Re\left\{ k_{\pm}\left(\omega\right)\right\} =\pi+\omega+2\pi m,\quad\omega_{0}<\omega<\varOmega_{f}^{+},\label{eq:spmpo1js}
\end{equation}
}where $\varOmega_{f}^{+}$ is the upper boundary of instability frequencies.
Figure \ref{fig:mck-disp3s} illustrates graphically equation\emph{
(\ref{eq:spmpo1js}) }by perfect straight lines parallel to $\Re\left\{ k\right\} =\omega$
in the shadowed area.
\end{rem}

There is yet another form of the dispersion relation (\ref{eq:dispbf1a})
and (\ref{eq:dispbf1b}) which is the high-frequency form:
\begin{gather}
D_{\mathrm{K}}\left(\omega,k\right)=D_{\mathrm{K}}^{\left(0\right)}\left(\omega,k\right)+\frac{K_{0}\omega_{0}^{2}}{\omega^{2}-\omega_{0}^{2}}=0,\label{eq:DKomk1a}\\
D_{\mathrm{K}}^{\left(0\right)}\left(\omega,k\right)=\cos\left({\it \omega}-k\right)+b_{f}^{\infty},\quad b_{f}^{\infty}=K_{0}\sin\left(f\right)-\cos\left(f\right).\label{eq:DKomk1b}
\end{gather}
We refer to function $D_{\mathrm{K}}\left(\omega,k\right)$ as the
\emph{MCK dispersion function}.

This form readily yields the following high-frequency approximation
to the MCK dispersion relations
\begin{equation}
\cos\left({\it \omega}-k\right)+b_{f}^{\infty}=0,\quad b_{f}^{\infty}=K_{0}\sin\left(f\right)-\cos\left(f\right),\quad\left|b_{f}^{\infty}\right|\leq1,\label{eq:DKomk1c}
\end{equation}
or, equivalently
\begin{equation}
\omega=k\pm\arccos\left(-b_{f}^{\infty}\right)+2\pi m,\quad b_{f}^{\infty}=K_{0}\sin\left(f\right)-\cos\left(f\right),\quad\left|b_{f}^{\infty}\right|\leq1,\quad m\in\mathbb{Z},\label{eq:DKom1d}
\end{equation}
where inequality $\left|b_{f}^{\infty}\right|\leq1$ is necessary
and sufficient for the existence of real-valued $\omega$ and $k$
satisfying the dispersion relation.

Theorem \ref{thm:dispfac} shows how the MCK dispersion function $D_{\mathrm{K}}\left(\omega,k\right)$
and its high-frequency approximation $D_{\mathrm{K}}^{\left(0\right)}\left(\omega,k\right)$
are integrated into the relevant dispersion functions associated with
the CCTWT.

Figures \ref{fig:mck-disp2s} and \ref{fig:mck-disp3s} illustrate
graphically the dispersion relations $k_{\pm}\left(\omega\right)$
described by equations (\ref{eq:spmpo1es}). The pairs of nearly straight
lines above the shadowed instability zone depicted in Figure \ref{fig:mck-disp2s}
are consistent with the high-frequency approximation (\ref{eq:DKom1d})
to the MCK dispersion relation.
\begin{figure}[h]
\begin{centering}
\includegraphics[scale=0.2]{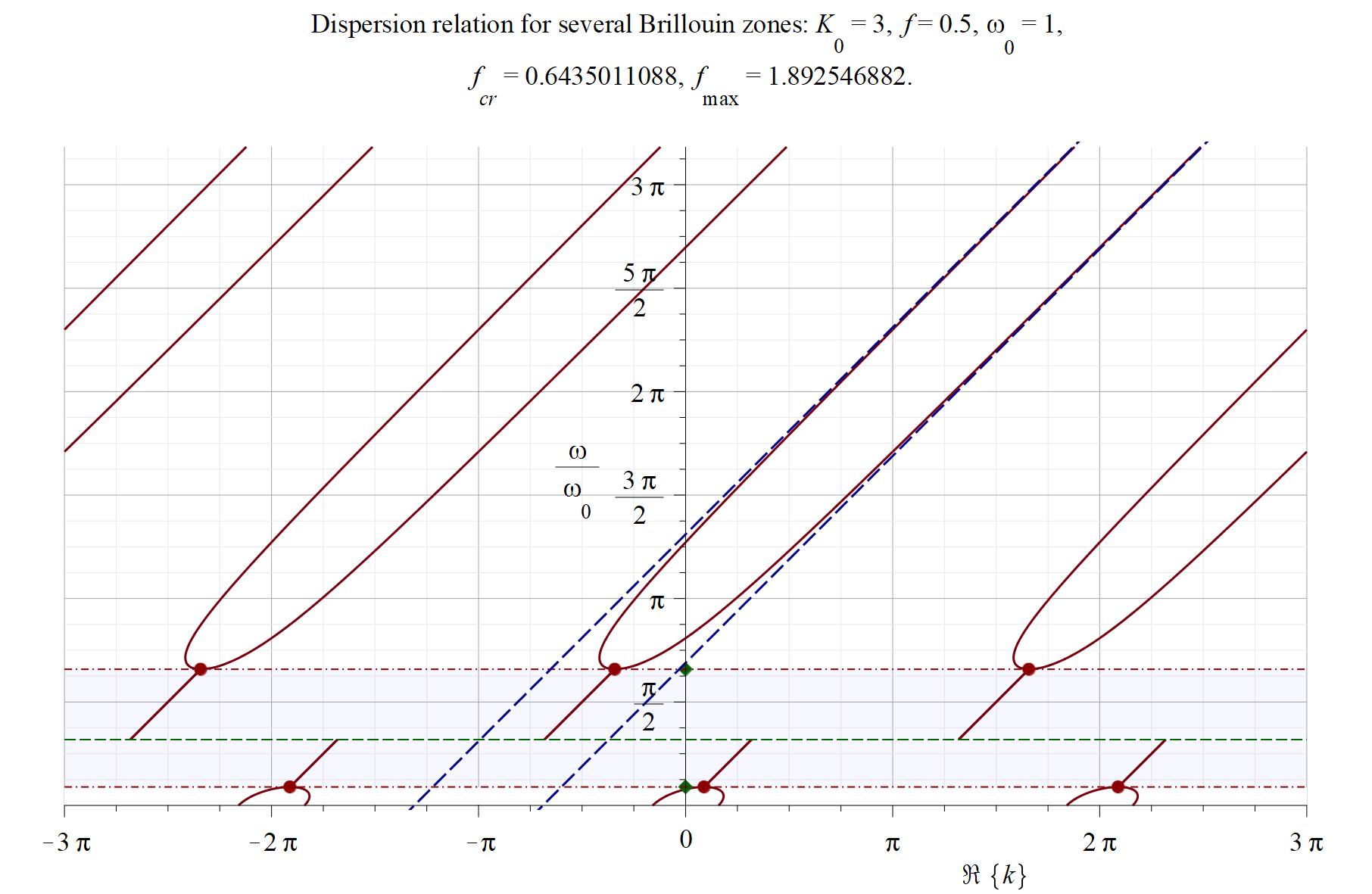}
\par\end{centering}
\centering{}\caption{\label{fig:mck-disp2s} The MCK dispersion-instability plot (solid
brown curves and lines) over 3 Brillouin zones $3\left[-\pi,\pi\right]$
for $K_{0}=3$, $\omega_{0}=1$ for which $f_{\mathrm{cr}}\protect\cong0.6435011088$,
$f_{\mathrm{max}}\protect\cong1.892546882$ and $f=0.5<f_{\mathrm{cr}}\protect\cong0.6435011088$.
The horizontal and vertical axes represent respectively $\Re\left\{ k\right\} $
and $\frac{\omega}{\omega_{0}}$. Two solid (green) diamond dots identify
the values of $\varOmega_{f}^{-}$ and $\varOmega_{f}^{+}$ which
are the frequency boundaries of the instability. Solid (brown) disk
dots identify points of the transition from the instability to the
stability which are also EPD points. Two horizontal (brown) dash-dot
lines $\omega=\varOmega_{f}^{\pm}$ identify the frequency boundaries
of the instability and the shaded (light blue) region between the
lines identify points $\left(\Re\left\{ k\right\} ,\omega\right)$
of instability. Dashed horizontal (green) line $\omega=\omega_{0}$
identifies the resonance frequency $\omega_{0}$. Note the plot has
a jump-discontinuity along the dashed (green) line, namely $\Re\left\{ k_{\pm}\left(\omega\right)\right\} $
jumps by $\pi$ according to equations (\ref{eq:spmpo1es})\emph{
}as the frequency $\omega$ passes through the resonance frequency
$\omega_{0}$ and the sign of $b_{f}\left(\omega\right)$ changes.
The shadowed area marks points $\left(\Re\left\{ k\right\} ,\omega\right)$
associated with the instability. The dashed (blue) straight lines
lines correspond to the high frequency approximation defined by equations
(\ref{eq:DKom1d}).}
\end{figure}
\begin{figure}[h]
\begin{centering}
\includegraphics[scale=0.2]{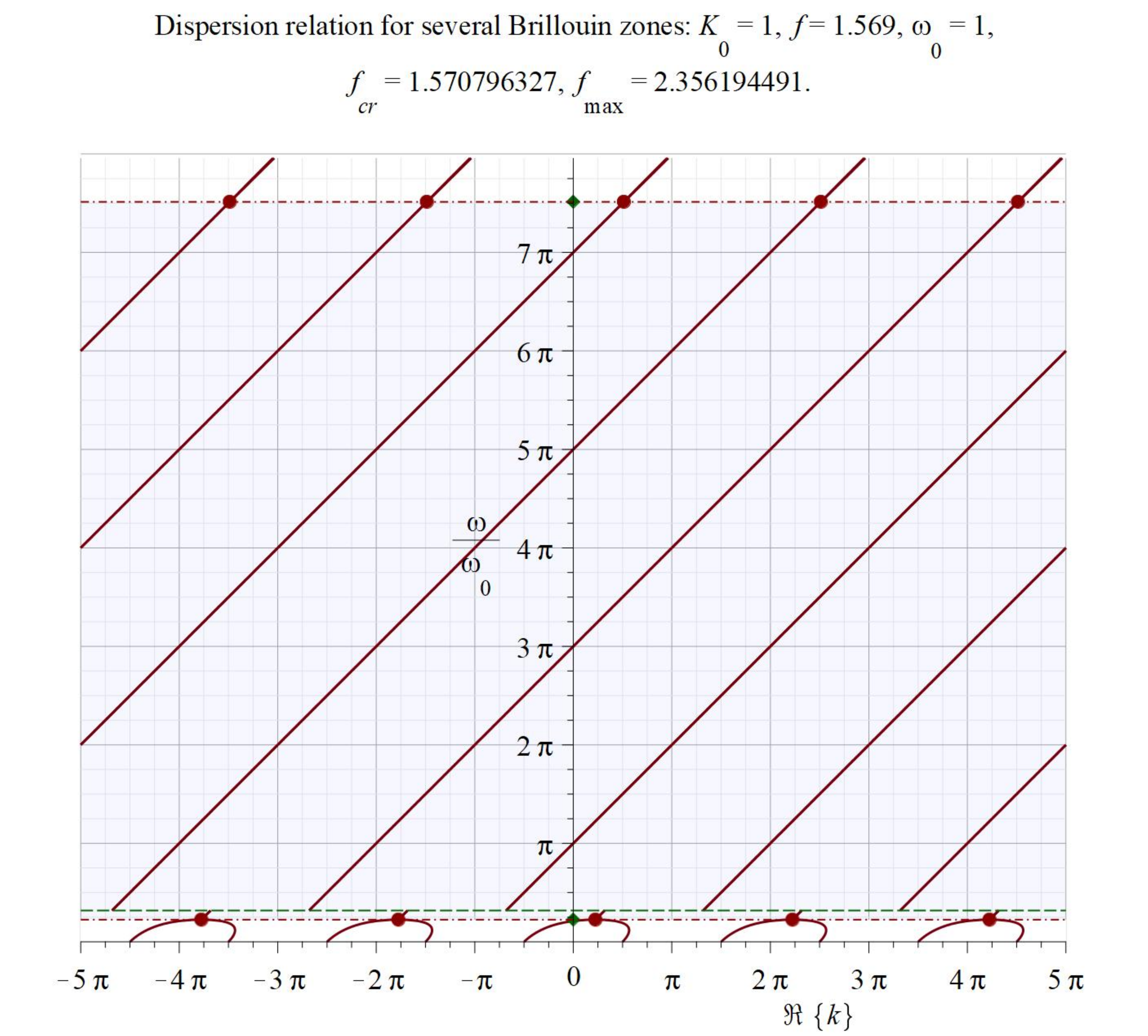}
\par\end{centering}
\centering{}\caption{\label{fig:mck-disp3s} The MCK dispersion-instability plot (solid
brown curves and lines) over 3 Brillouin zones $3\left[-\pi,\pi\right]$
for $K_{0}=1$, $\omega_{0}=1$ for which $f_{\mathrm{cr}}\protect\cong1.570796327$,
$f_{\mathrm{max}}\protect\cong1.892546882$ and $f=1.569\protect\cong f_{\mathrm{cr}}\protect\cong1.570796327$.
The horizontal and vertical axes represent respectively $\Re\left\{ k\right\} $
and $\frac{\omega}{\omega_{0}}$. Two solid (green) diamond dots identify
the values of $\varOmega_{f}^{-}$ and $\varOmega_{f}^{+}$ which
are the frequency boundaries of the instability. Solid (brown) disk
dots identify points of the transition from the instability to the
stability which are also EPD points. Two horizontal (brown) dash-dot
lines $\omega=\varOmega_{f}^{\pm}$ identify the frequency boundaries
of the instability and the shaded (light blue) region between the
lines identify points $\left(\Re\left\{ k\right\} ,\omega\right)$
of instability. Dashed horizontal (green) line $\omega=\omega_{0}$
identifies the resonance frequency $\omega_{0}$. Note the plot has
a jump-discontinuity along the dashed (green) line, namely $\Re\left\{ k_{\pm}\left(\omega\right)\right\} $
jumps by $\pi$ according to equations (\ref{eq:spmpo1es})\emph{
}as the frequency $\omega$ passes through the resonance frequency
$\omega_{0}$ and the sign of $b_{f}\left(\omega\right)$ changes.
The shadowed area marks points $\left(\Re\left\{ k\right\} ,\omega\right)$
associated with the instability.}
\end{figure}

Interestingly, there is an empirical formula due to Tsimring that
shows the dependence of the maximum power gain $G_{\mathrm{T}}\left(N\right)$
on the number $N$ of cavities in the klystron, \cite[7.7.1]{Tsim},
\cite[7.2.6]{Grigo}, \cite[16]{ValMid}:
\begin{equation}
G_{\mathrm{T}}\left(N\right)=15+20\left(N-2\right)\,\mathrm{dB}.\label{eq:GTmaxdB1as}
\end{equation}
Realistically achievable maximum amplification values though are smaller
and are of the order of $50\,\mathrm{dB}$ to $70\,\mathrm{dB}$.
The main limiting factors are noise and self-excitation of the klystron
because of parasitic feedback between cavities.

\section{Coupled cavity structure\label{sec:ccs}}

We introduce and study here basic properties of the \emph{coupled
cavity structure (CCS)}. Since CCS is naturally an integral part of
CCTWT the knowledge of its properties would allow to find out its
contribution to the properties of CCTWT. As to particular designs
of coupled cavities and the way they interact with TWTs see \cite[9.1, 9.3.3]{BenSweScha}.

The Lagrangian $\mathcal{L}_{\mathrm{C}}\left(\left\{ Q\right\} \right)$
of the CCS system can be readily obtained from the Lagrangian $\mathcal{L}$
of CCTWT defined by equations (\ref{eq:aZep1L}), (\ref{eq:aZep1c})
and (\ref{eq:aZep1d}) by assuming $b=0$ and omitting component $\mathcal{L}_{\mathrm{B}}$,
that is
\begin{equation}
\mathcal{L}_{\mathrm{C}}\left(\left\{ Q\right\} \right)=\frac{L}{2}\left(\partial_{t}Q\right)^{2}-\frac{1}{2C}\left(\partial_{z}Q\right)^{2}-\frac{1}{2c_{0}}\sum_{\ell=-\infty}^{\infty}\delta\left(z-a\ell\right)\left[Q\left(a\ell\right)\right]^{2}.\label{eq:ccTL1a}
\end{equation}
Then the corresponding EL equations (\ref{eq:eZep1e}) and (\ref{eq:eZep1fa})
are reduced to
\begin{equation}
L\partial_{t}^{2}Q-C^{-1}\partial_{z}^{2}Q=0,\quad\left[\partial_{z}Q\right]\left(a\ell\right)=C_{0}\left(\frac{\partial_{t}^{2}}{\omega_{0}^{2}}+1\right)Q\left(a\ell\right),\quad C_{0}=\frac{C}{c_{0}},\quad a\ell\in\mathbb{Z},\label{eq:ccTL1b}
\end{equation}
where jumps $\left[Q\right]\left(a\ell\right)$ and $\left[\partial_{z}Q\right]\left(a\ell\right)$
are defined by equation (\ref{eq:limpm1b}), and consequently
\begin{equation}
X\left(a\ell+0\right)=\mathsf{S}_{\mathrm{b}}X\left(a\ell-0\right),\quad\mathsf{S}_{\mathrm{b}}=\left[\begin{array}{rr}
1 & 0\\
C_{0}\left(\frac{\partial_{t}^{2}}{\omega_{0}^{2}}+1\right) & 1
\end{array}\right],\quad X=\left[\begin{array}{r}
Q\\
\partial_{z}Q
\end{array}\right]\label{eq:ccTL1c}
\end{equation}
\emph{Using the same set of dimensionless variables as in Section
\ref{subsec:dim-par} and omitting prime symbol for notation simplicity}
we obtain the following dimensionless form of the EL equations (\ref{eq:ccTL1b})
\begin{equation}
\partial_{t}^{2}Q-\frac{1}{\chi^{2}}\partial_{z}^{2}Q=0,\quad z\neq\ell;\quad\left[\partial_{z}Q\right]\left(\ell\right)=C_{0}\left(\frac{\partial_{t}^{2}}{\omega_{0}^{2}}+1\right)Q\left(\ell\right),\quad\ell\in\mathbb{Z}.\label{eq:ccTL1d}
\end{equation}
The Fourier transform in $t$ (see Appendix \ref{sec:four}) of equations
(\ref{eq:ccTL1d}) yields
\begin{equation}
\partial_{z}^{2}\check{Q}+\frac{\omega^{2}}{\chi^{2}}\check{Q}=0,\quad z\neq\ell;\quad\left[\partial_{z}Q\right]\left(\ell\right)=C_{0}\left(1-\frac{\omega^{2}}{\omega_{0}^{2}}\right)Q\left(\ell\right),\quad\ell\in\mathbb{Z}\label{eq:ccTL1e}
\end{equation}
where $\check{Q}$ and $\check{q}$ are the time Fourier transform
of the corresponding quantities.

An alternative form the system of equations (\ref{eq:ccTL1e}) is
the following second-order vector ODE with the periodic singular potential:
\begin{equation}
\partial_{z}^{2}\check{Q}+\frac{\omega^{2}}{\chi^{2}}\check{Q}+C_{0}\left(1-\frac{\omega^{2}}{\omega_{0}^{2}}\right)\sum_{\ell=-\infty}^{\infty}\delta\left(z-\ell\right)\check{Q}=0.\label{eq:ccTL1f}
\end{equation}

According to Appendix \ref{sec:dif-jord} second-order differential
equation (\ref{eq:ccTL1f}) is equivalent to the first-order differential
equation of the form:
\begin{gather}
\partial_{z}X=A_{\mathrm{C}}\left(z\right)X,\quad A_{\mathrm{C}}\left(z\right)=\left[\begin{array}{rr}
0 & 1\\
-\frac{\omega^{2}}{\chi^{2}}-p\left(z\right) & 0
\end{array}\right],\quad X=\left[\begin{array}{r}
\check{Q}\\
\partial_{z}\check{Q}
\end{array}\right],\label{eq:ccTL2a}\\
p\left(z\right)=C_{0}\left(1-\frac{\omega^{2}}{\omega_{0}^{2}}\right)\sum_{\ell=-\infty}^{\infty}\delta\left(z-\ell\right).\nonumber 
\end{gather}
Using results of Appendix \ref{subsec:speHam} we find that system
(\ref{eq:ccTL2a}) is Hamiltonian for the following choice of nonsingular
Hermitian matrix $G$:
\begin{equation}
G_{\mathrm{C}}=G_{\mathrm{C}}^{*}=\left[\begin{array}{rr}
0 & \mathrm{i}\\
-\mathrm{i} & 0
\end{array}\right],\quad\det\left\{ G_{\mathrm{C}}\right\} =-1\label{eq:ccTL2b}
\end{equation}
In particular, it is an elementary exercise to verify that for each
value of $z$ matrix $A_{\mathrm{C}}\left(z\right)$ is $G_{\mathrm{C}}$-skew-Hermitian,
that is
\begin{equation}
G_{\mathrm{C}}A_{\mathrm{C}}\left(z\right)+A_{\mathrm{C}}^{*}\left(z\right)G_{\mathrm{C}}=0.\label{eq:ccTL2c}
\end{equation}
Then if $\Phi\left(z\right)$ is the matrizant of Hamiltonian equation
(\ref{eq:ccTL2a}) then according to results of Appendix \ref{sec:Ham}
$\Phi_{\mathrm{C}}\left(z\right)$ is $G_{\mathrm{C}}$-unitary matrix
\begin{equation}
\Phi_{\mathrm{C}}^{*}\left(z\right)G_{\mathrm{C}}\Phi_{\mathrm{C}}\left(z\right)=G_{\mathrm{C}},\label{eq:ccTL2d}
\end{equation}
 and consequently its spectrum $\sigma\left\{ \Phi_{\mathrm{C}}\left(z\right)\right\} $
is invariant with respect to the inversion transformation $\zeta\rightarrow\frac{1}{\bar{\zeta}}$,
that is it symmetric with respect to the unit circle: 
\begin{equation}
\zeta\in\sigma\left\{ \Phi_{\mathrm{C}}\left(z\right)\right\} \Rightarrow\frac{1}{\bar{\zeta}}\in\sigma\left\{ \Phi_{\mathrm{C}}\left(z\right)\right\} .\label{eq:ccTL2e}
\end{equation}
To simplify analytic evaluations we assume as before that Assumption
\ref{ass:cavcoup} holds.

\subsection{Monodromy matrix and the dispersion-instability relations\label{subsec:mondis}}

Under simplifying Assumption \ref{ass:cavcoup} ($\chi=\omega_{0}=1$)
the monodromy matrix matrix $\mathscr{T}_{\mathrm{C}}$ defined by
equation (\ref{eq:PTP1b}) takes the form
\begin{equation}
\mathscr{T}_{\mathrm{C}}=\left[\begin{array}{rr}
\cos\left(\omega\right) & \omega^{-1}\sin\left(\omega\right)\\
\left(1-\omega^{2}\right){\it C_{0}}\cos\left(\omega\right)-\omega\,\sin\left(\omega\right) & \cos\left(\omega\right)-C_{0}\left(\omega-\omega^{-1}\right)\sin\left(\omega\right)
\end{array}\right].\label{eq:monTc1a}
\end{equation}
The corresponding characteristic equation (\ref{eq:PTP1d}) turns
into
\begin{gather}
P_{\mathrm{C}}\left(\omega,s\right)=s^{2}+2W_{\mathrm{C}}\left(\omega\right)s+1=0,\quad W_{\mathrm{C}}\left(\omega\right)=\frac{C_{0}}{2}\left(\omega-\frac{1}{\omega}\right)\sin\left(\omega\right)-\cos\left(\omega\right),\label{eq:monTc1b}\\
s=\exp\left\{ \mathrm{i}k\right\} ,\nonumber 
\end{gather}
where quadratic polynomial $P_{\mathrm{C}}$ is referred to as the
\emph{CCS characteristic polynomial} and quantity $W_{\mathrm{C}}\left(\omega\right)$
is the CCS instability parameter which is depicted in Fig. \ref{fig:dis-ccs-wom2}(b).

The two Floquet multipliers $s_{\pm}$ which are the eigenvalues of
the monodromy matrix $\mathscr{T}_{\mathrm{C}}$ defined by equation
(\ref{eq:monTc1a}) and consequently are the solutions to its characteristic
equation (\ref{eq:monTc1b}) can be represented as follows
\begin{equation}
s_{\pm}=e^{\mathrm{i}k_{\pm}}=W_{\mathrm{C}}\pm\sqrt{W_{\mathrm{C}}^{2}-1},\quad W_{\mathrm{C}}=W_{\mathrm{C}}\left(\omega\right)=\frac{C_{0}}{2}\left(\omega-\frac{1}{\omega}\right)\sin\left(\omega\right)-\cos\left(\omega\right).\label{eq:monTc1c}
\end{equation}
Equations (\ref{eq:monTc1c}) imply that\emph{ instability parameter
$W_{\mathrm{C}}\left(\omega\right)$ there completely determines the
two Floquet multipliers $s_{\pm}$} justifying its its name.

\emph{Importantly, the characteristic equation} (\ref{eq:monTc1b})
\emph{can be viewed as an expression of the dispersion relations between
the frequency $\omega$ and the wavenumber $k$}. To obtain an explicit
form of the \emph{dispersion relations for the CCS} under simplifying
Assumption \ref{ass:cavcoup} ($\chi=\omega_{0}=1$) we divide the
characteristic equation (\ref{eq:monTc1b}) by $2s$, substitute $s=\exp\left\{ \mathrm{i}k\right\} $
obtaining the following equations:
\begin{equation}
\cos\left(k\right)+W_{\mathrm{C}}\left(\omega\right)=0,\quad W_{\mathrm{C}}\left(\omega\right)=\frac{C_{0}}{2}\left(\omega-\frac{1}{\omega}\right)\sin\left(\omega\right)-\cos\left(\omega\right),\label{eq:monTc1d}
\end{equation}
where $W_{\mathrm{C}}\left(\omega\right)$ is principle CCS function
(see Fig. \ref{fig:dis-ccs-wom2}(b)). Alternatively, the dispersion-instability
relations (\ref{eq:monTc1c}) can be represented in the form
\begin{gather}
k_{\pm}=k_{\pm}\left(\omega\right)=-\mathrm{i}\ln\left(W_{\mathrm{C}}\pm\sqrt{W_{\mathrm{C}}^{2}-1}\right),\label{eq:monTc1e}\\
W_{\mathrm{C}}=W_{\mathrm{C}}\left(\omega\right)=\frac{C_{0}}{2}\left(\omega-\frac{1}{\omega}\right)\sin\left(\omega\right)-\cos\left(\omega\right).\nonumber 
\end{gather}
Dividing equation (\ref{eq:monTc1d}) by $\omega$ we obtain the following\emph{
high-frequency form of the dispersion relations} for the MCK:
\begin{gather}
D_{\mathrm{C}}\left(\omega,k\right)=D_{\mathrm{C}}^{\left(0\right)}\left(\omega,k\right)+\frac{2\left(\cos\left(k\right)-\cos\left(\omega\right)\right)}{\omega}-\frac{C_{0}\sin\left(\omega\right)}{\omega^{2}}=0.\label{eq:monTc1f}\\
D_{\mathrm{C}}^{\left(0\right)}\left(\omega,k\right)=C_{0}\sin\left(\omega\right).\nonumber 
\end{gather}
We refer to function $D_{\mathrm{C}}\left(\omega,k\right)$ as the
\emph{CCS dispersion function}.

As to the e-beam transforming characteristic equation (\ref{eq:PTP1e})
for the e-beam the same way we obtain the following explicit form
of the \emph{dispersion relations for the e-beam}
\begin{equation}
\cos\left(k-\omega\right)-\cos\left(f\right)=0;\quad\omega=k\pm f.\label{eq:ccWom1a}
\end{equation}

Expression (\ref{eq:monTc1d}) for the CCS dispersion relation readily
implies that its EPD frequencies are solution to the following \emph{CCS
EPD equation}
\begin{equation}
W_{\mathrm{C}}\left(\omega\right)=\frac{C_{0}}{2}\left(\omega-\frac{1}{\omega}\right)\sin\left(\omega\right)-\cos\left(\omega\right)=\pm1.\label{eq:ccWom1b}
\end{equation}
where $W_{\mathrm{C}}\left(\omega\right)$ is principle CCS function
defined by the second equation in (\ref{eq:monTc1b}) and its plot
is depicted in Fig. \ref{fig:dis-ccs-wom2}(b). Straightforward evaluations
show that $W_{\mathrm{C}}\left(\omega\right)$ satisfies the following
equations:
\begin{equation}
W_{\mathrm{C}}\left(\pi n\right)=\left(-1\right)^{n+1},\quad\partial_{\omega}W_{\mathrm{C}}\left(\pi n\right)=\left(-1\right)^{n},\quad n=1,2,\ldots.\label{eq:ccWom1c}
\end{equation}
Consequently $\pi n$ for positive integers $n$ are CCS EPD points.
Remaining set of EPD frequencies $\xi_{m},\:m\geq1$ are found by
solving CCS EPD equation (\ref{eq:ccWom1b}).

Instability (oscillatory spectrum) bands (intervals)
\begin{equation}
\left[0,\xi_{0}\right],\quad\left[\xi_{m},\pi m\right],\quad m\geq1:0<\xi_{0}<\xi_{1}<\pi,\quad\pi\left(m-1\right)<\xi_{m}<\pi m,\quad m\geq2\label{eq:ccWom2a}
\end{equation}
where numbers $\xi_{m}$ satisfy also the following equations:
\begin{equation}
W_{\mathrm{C}}\left(\xi_{0}\right)=-1,\quad W_{\mathrm{C}}\left(\omega\right)\left(\xi_{1}\right)=1,\quad W_{\mathrm{C}}\left(\xi_{m}\right)=\left(-1\right)^{m-1};\;m\geq2.\label{eq:ccWom2b}
\end{equation}
Stability (oscillatory) spectrum bands (intervals)
\begin{equation}
\left[\xi_{0},\xi_{1}\right],\quad\left[\pi m,\xi_{m+1}\right],\quad m\geq1:0<\xi_{0}<\xi_{1}<\pi,\quad\pi\left(m-1\right)<\xi_{m}<\pi m,\quad m\geq2\label{eq:ccWom2c}
\end{equation}
\begin{figure}[h]
\begin{centering}
\includegraphics[scale=0.15]{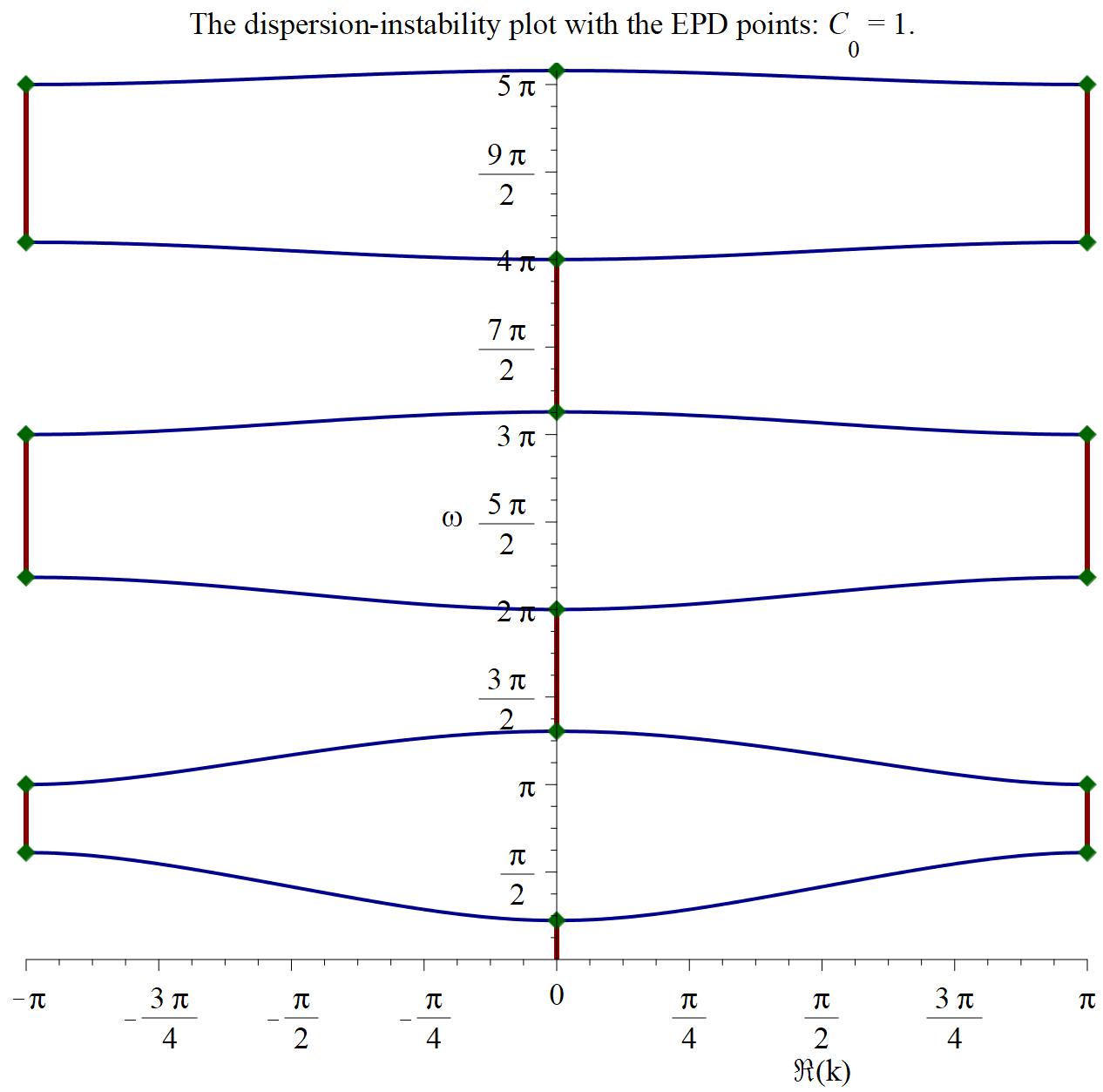}\hspace{0.5cm}\includegraphics[scale=0.15]{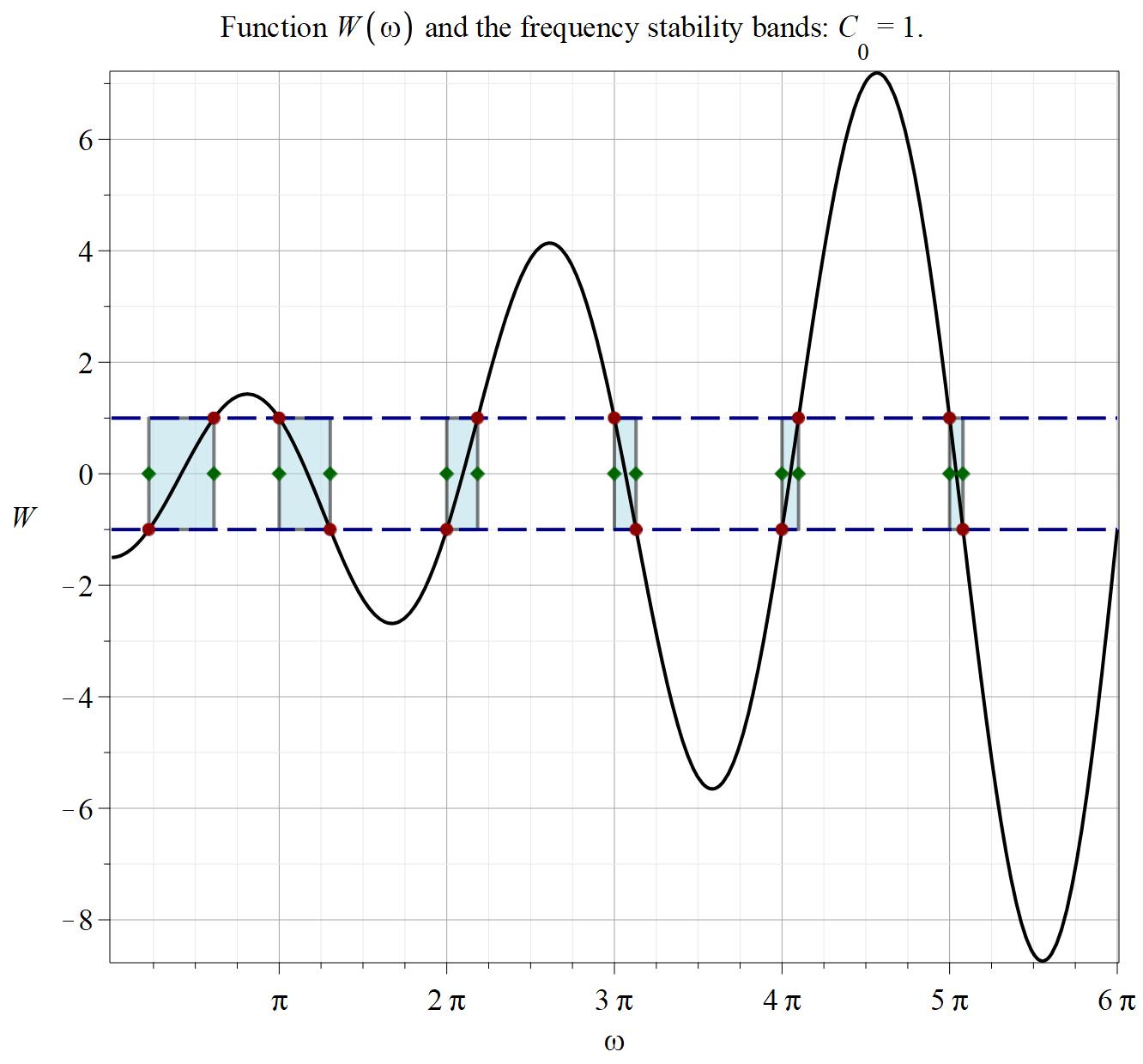}
\par\end{centering}
\centering{}(a)\hspace{7cm}(b)\caption{\label{fig:dis-ccs-wom2} The CCS for $C_{0}=1$: (a) dispersion-instability
graph, horizontal axis is $\Re\left\{ k\right\} $ and vertical axis
is $\omega$; (b) the plot of the instability parameter $W_{\mathrm{C}}\left(\omega\right)$
defined by the second equation in (\ref{eq:monTc1b}). The horizontal
axis is $\Re\left\{ k\right\} $ and the vertical axis is $W$.}
\end{figure}

Based on the prior analysis we introduce the CCS gain $G_{\mathrm{C}}$
in $\mathrm{dB}$ per one period as a the rate of the exponential
growth of the CCS eigenmodes associated with Floquet multipliers $s_{\pm}$
defined by equations (\ref{eq:monTc1e}). More precisely the definition
is as follows.
\begin{defn}[CCS gain per one period]
\label{def:gainccs} Let $s_{\pm}$ be the CCS Floquet multipliers
defined by equations (\ref{eq:monTc1e}). Then the corresponding to
them gain $G_{\mathrm{C}}$ in $\mathrm{dB}$ per one period is defined
by
\begin{gather}
G_{\mathrm{C}}=G_{\mathrm{C}}\left(\omega,C_{0}\right)=\left\{ \begin{array}{ccc}
20\left|\log\left(\left|s_{+}\right|\right)\right|=20\left|\log\left(\left|\left|W_{\mathrm{C}}\right|+\sqrt{W_{\mathrm{C}}^{2}-1}\right|\right)\right| & \text{if} & \left|W_{\mathrm{C}}\right|>1\\
0 & \text{if} & \left|W_{\mathrm{C}}\right|\leq1
\end{array}\right.,\label{eq:ccWom2d}\\
W_{\mathrm{C}}=W_{\mathrm{C}}\left(\omega\right)=\frac{C_{0}}{2}\left(\omega-\frac{1}{\omega}\right)\sin\left(\omega\right)-\cos\left(\omega\right).\nonumber 
\end{gather}
\end{defn}

Fig. \ref{fig:gain-ccs} shows the frequency dependence of the gain
$G_{\mathrm{C}}$ per one period. Growing in magnitude ``bumps''
in Figure \ref{fig:gain-ccs} indicate the presence of gain/amplification
inside of stopbands, known also as spectral gaps in the system (oscillatory)
spectrum, of the CCS, see Remark \ref{rem:ampstop}.
\begin{figure}[h]
\centering{}\includegraphics[scale=0.22]{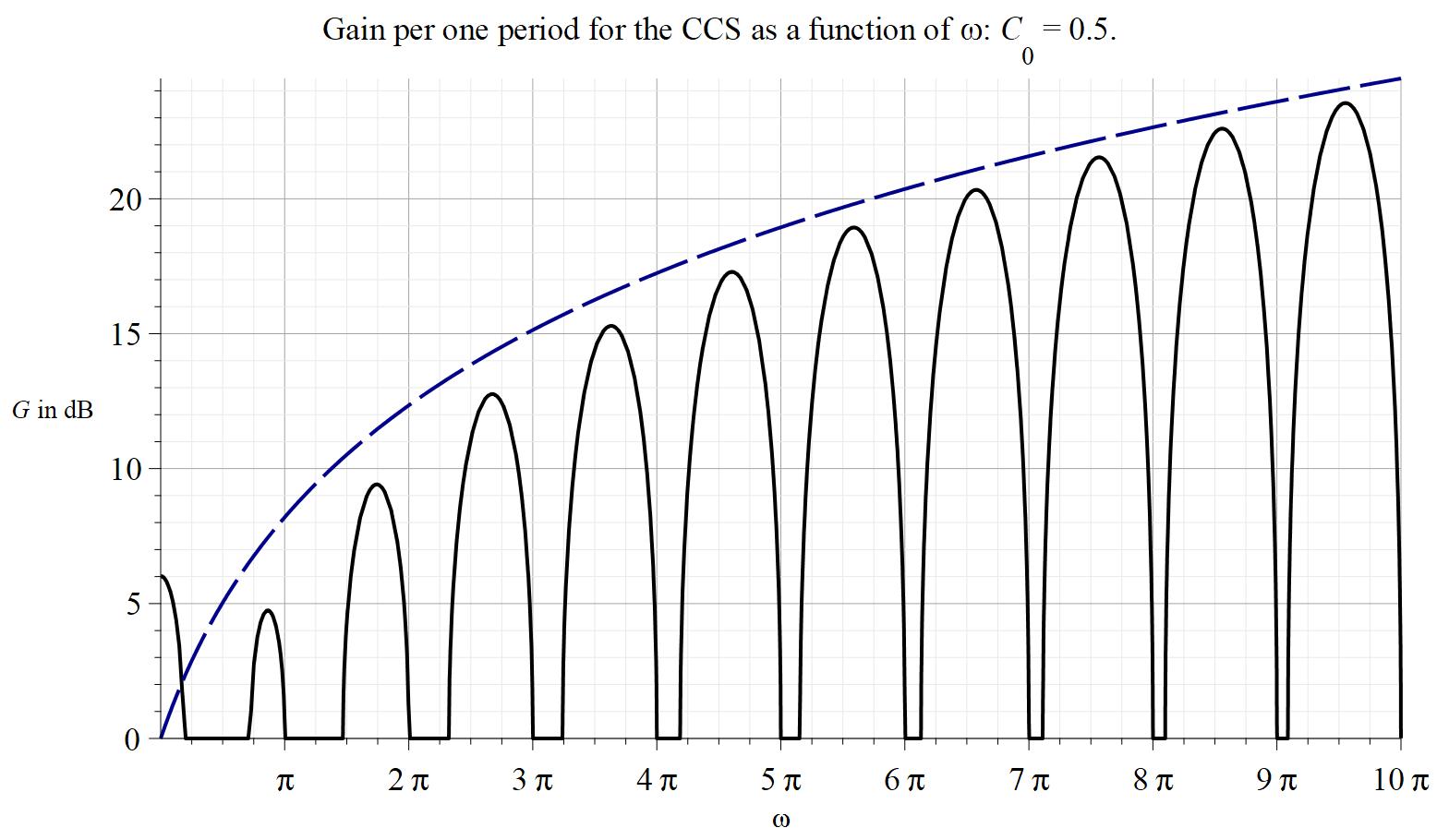}\caption{\label{fig:gain-ccs} The plot of the CCS gain $G_{\mathrm{C}}\left(\omega,C_{0}\right)$
per one period for $C_{0}=0.5.$ The horizontal and vertical axes
represent respectively frequency $\omega$ and gain $G$ in $\mathrm{dB}$.
The instability frequencies $\omega$ are identified by condition
$G_{\mathrm{C}}\left(\omega,C_{0}\right)>0$. The envelope of the
local maxima of the gain $G_{\mathrm{C}}\left(\omega,C_{0}\right)$
behaves asymptotically for large values of frequency $\omega$ as
$20\left|\log\left(C_{0}\omega\right)\right|$ as it follows from
equations (\ref{eq:ccWom2d}). It is shown as dashed (blue) curve.}
\end{figure}

\subsection{Exceptional points of degeneracy}

The monodromy matrix $\mathscr{T}_{\mathrm{C}}$ defined by equation
(\ref{eq:monTc1a}) and its Jordan form at $\omega=\pi n$ are as
follows
\[
\mathscr{T}_{\mathrm{C}}=\left[\begin{array}{rr}
\left(-1\right)^{n} & 0\\
\left(1-\pi^{2}n^{2}\right){\it C_{0}}\left(-1\right)^{n} & \left(-1\right)^{n}
\end{array}\right]=\mathscr{Z}_{\mathrm{C}}\left[\begin{array}{rr}
\left(-1\right)^{n} & 1\\
0 & \left(-1\right)^{n}
\end{array}\right]\mathscr{Z}_{\mathrm{C}}^{-1},\quad\omega=\pi n
\]
where matrix $\mathscr{Z}_{\mathrm{C}}$ is
\[
\mathscr{Z}_{\mathrm{C}}=\left[\begin{array}{rr}
0 & 1\\
\left(1-\pi^{2}n^{2}\right){\it C_{0}}\left(-1\right)^{n} & 0
\end{array}\right],
\]
and columns of the matrix $\mathscr{Z}_{\mathrm{C}}$ is the Jordan
basis of the monodromy matrix $\mathscr{T}_{\mathrm{C}}$

The monodromy matrix expression at EPDs is as follows
\[
\mathscr{T}_{\mathrm{C}}=\left[\begin{array}{rr}
\cos\left(\omega\right) & \omega^{-1}\sin\left(\omega\right)\\
\frac{\omega\sin\left(\omega\right)\left(\cos\left(\omega\right)-1\right)}{\cos\left(\omega\right)+1} & 2-\cos\left(\omega\right)
\end{array}\right]=\mathscr{Z}_{\mathrm{C}}\left[\begin{array}{rr}
1 & 1\\
0 & 1
\end{array}\right]\mathscr{Z}_{\mathrm{C}}^{-1},\quad W_{\mathrm{C}}\left(\omega\right)=-1,\quad\omega\neq\pi n,
\]
where
\[
\mathscr{Z}_{\mathrm{C}}=\left[\begin{array}{rr}
\cos\left(\omega\right)-1 & 1\\
\frac{\omega\sin\left(\omega\right)\left(\cos\left(\omega\right)-1\right)}{\cos\left(\omega\right)+1} & 0
\end{array}\right],
\]
and
\[
\mathscr{T}_{\mathrm{C}}=\left[\begin{array}{rr}
\cos\left(\omega\right) & \omega^{-1}\sin\left(\omega\right)\\
-\frac{\omega\left(\cos\left(\omega\right)+1\right)^{2}}{\sin\left(\omega\right)} & -2-\cos\left(\omega\right)
\end{array}\right]=\mathscr{Z}_{\mathrm{C}}\left[\begin{array}{rr}
-1 & 1\\
0 & -1
\end{array}\right]\mathscr{Z}_{\mathrm{C}}^{-1},\quad W_{\mathrm{C}}\left(\omega\right)=1,\quad\omega\neq\pi n,
\]
where
\[
\mathscr{Z}_{\mathrm{C}}=\left[\begin{array}{rr}
\cos\left(\omega\right)+1 & 1\\
-\frac{\omega\left(\cos\left(\omega\right)+1\right)^{2}}{\sin\left(\omega\right)} & 0
\end{array}\right].
\]

\section{The kinetic and field points of view on the gap interaction\label{sec:fitkit}}

We compare here some of the features of our field theory with the
relevant features of the kinematic/ballistic theory of the CCTWT operation.
Before going into technical details we would like to point out that
from the outset our Lagrangian field theory takes into account the
space-charge forces, that is the electron-to-electron repulsion, whereas
the standard hydrokinetic analysis completely neglects them.

\subsection{Some points from the kinetic theory\label{subsec:kinpoi}}

We briefly review here some points the kinetic/ballistic theory. Kinematic
analysis of the CCTWT operation involves: (i) the electron velocity
modulation in gaps of the klystron cavities; (ii) consequent electron
bunching; (iii) the energy exchange between the e-beam to the EM field;
(iv) the energy transfer from the e-beam to the EM field under proper
conditions and consequent RF signal amplification. The listed subjects
were thoroughly studied by many scholars, see, for instance, \cite{Caryo},
\cite{ChoWes}, \cite[15]{Gilm1}, \cite[7.2]{Grigo}, \cite[6.1-6.3; 7.1-7.7]{Tsim},
\cite[II]{Shev} and references therein. When presenting relevant
to us conclusions of the studies we follow mostly to the \emph{hydrokinetic
(ballistic) approach} that utilizes \emph{the Eulerian (spatial) and
the Lagrangian (material) descriptions (points of view)} as in \cite[7.1-7.7]{Tsim}
and \cite[II]{Shev}. As to general aspects of the hydrokinetic approach
in continua, that includes in particular the Eulerian and the Lagrangian
descriptions, we refer the reader to \cite[I.4-I.8]{Lamb}, \cite[3.1-3.2]{Redd},
\cite[1.7]{Gran}.

Our field theory assumes that the cavity width $l_{\mathrm{g}}$ and
the corresponding transit time $\tau_{g}$ are zeros, see equations
(\ref{eq:transL1d}) and Assumptions \ref{ass:idealmod}). Consequently
the most sophisticated developments of the kinetic theory dealing
with cavity gaps of finite lengths are outside the scope of our studies.
In our simpler case when $l_{\mathrm{g}}=0$ and $\tau_{g}=0$ following
to \cite[II.5]{Shev} we suppose that $\mathring{U}$ is the constant
accelerating voltage so the the stationary dc electron flow velocity
$\mathring{v}=\sqrt{\frac{2e}{m}\mathring{U}}$ where $m$ and $-e$
is respectively the electron mass and its charge. Suppose also that
$U_{1}\sin\left(\omega t\right)$ is the gap voltage. Then based on
the elementary energy conservation law one gets
\begin{equation}
\frac{mv^{2}}{2}=\frac{m\mathring{v}^{2}}{2}+U_{1}\sin\left(\omega t\right)\label{eq:mmvvU1a}
\end{equation}
where $v$ is the modulated velocity. Solving equation (\ref{eq:mmvvU1a})
for $v$ and assuming ``small signal'' approximation we obtain
\begin{equation}
v\cong\mathring{v}\left(1+\frac{\xi}{2}\sin\left(\omega t\right)\right),\quad\xi=\frac{U_{1}}{\mathring{U}}\ll1.\label{eq:mmvvU1b}
\end{equation}
Then following to \cite[II.6, II.7]{Shev} we suppose the velocity-modulated
in the cavity electron beam as described by equation (\ref{eq:mmvvU1b})
enters the field-free\emph{ drift space} beyond the gap. Then, \cite[II.6, II.7]{Shev}:
\begin{quotation}
``Whilst passing through the drift space, some electrons overtake
other, slower, electrons which entered the drift space earlier, and
the initial distribution of charge in the beam is changed. If the
drift space is long enough the initial velocity modulation can lead
to substantial density modulation of the electron beam.''
\end{quotation}
In other words, according the above scenario electron bunching takes
place. More precisely, the velocity-modulated, uniformly-dense electron
beam, becomes a density-modulated beam with nearly constant dc velocity
$\mathring{v}$.

\subsection{Field theory point of view on the kinetic properties of the electron
flow\label{subsec:fiekin}}

According the CCTWT design all the interactions between the electron
flow and the EM field occur in cavity gaps. In what follows to use
notations and results from Section \ref{subsec:cwvel}. Let us consider
first \emph{the action of the cavity ac EM field on the electron flow}.
The cavity ac EM field acts upon the e-beam by accelerating and decelerating
its electrons and effectively modulating their velocities by the relatively
small compare to $\mathring{v}$ electron velocity field $v=v\left(z,t\right)$.
So as to this part of the interaction we may view the electron density
to be essentially constant $\mathring{n}$ whereas its ac velocity
field $v=v\left(z,t\right)$ is modulated by ac EM field. Consider
now \emph{the action of the e-beam on the cavity ac EM field}. The
space charge acts upon the cavity ac EM field essentially quasi electrostatically
through relatively small ac electron number density field $n=n\left(z,t\right)$.
So for this part of the interaction we may view the electron flow
to be of nearly constant velocity $\mathring{v}$ perturbed by relatively
small ac electron number density $n=n\left(z,t\right)$. Following
to the results of Section \ref{sec:cctwtmod} let us take a look at
the variation of the ac electron velocity $v=v\left(z,t\right)$ and
ac electron number density $n=n\left(z,t\right)$ in the vicinity
of centers $a\ell$ of the cavity gaps.

Note first that the action of the ac cavity EM field on the e-beam
is manifested directly through a variation of the electron velocity
$v=v\left(z,t\right)$ in a vicinity of the gap center $a\ell$. The
action of the e-beam on the cavity EM field is produced by the electron
number density $n=n\left(z,t\right)$. As to the quantitative assessment
of the variations note that equations (\ref{eq:chavar1ab}) and (\ref{eq:chavar1ab})
imply that the electron velocity $v=v\left(z,t\right)$ and number
density $n=n\left(z,t\right)$ have the following jumps $\left[n\right]\left(a\ell\right)$
and $\left[v\right]\left(a\ell\right)$ at the interaction points
$a\ell$:
\begin{equation}
\left[n\right]\left(a\ell,t\right)=\frac{\left[\partial_{z}q\right]\left(a\ell,t\right)}{\sigma_{\mathrm{B}}e},\quad\left[v\right]\left(a\ell,t\right)=-\frac{\mathring{v}\left[\partial_{z}q\right]\left(a\ell,t\right)}{e\sigma_{\mathrm{B}}\mathring{n}}=-\frac{v\mathring{n}\left(a\ell,t\right)}{\mathring{n}},\label{eq:naLvaL1a}
\end{equation}
readily implying
\begin{equation}
\frac{\left[v\right]\left(a\ell,t\right)}{\mathring{v}}=-\frac{\left[n\right]\left(a\ell,t\right)}{\mathring{n}}.\label{eq:naLvaL1b}
\end{equation}
Equations (\ref{eq:chavar1a}) and (\ref{eq:naLvaL1b}) in turn yield
\begin{equation}
\left[j\right]\left(a\ell,t\right)=-e\left\{ \mathring{n}\left[v,t\right]+\mathring{v}\left[n\right]\right\} \left(a\ell,t\right)=0\label{eq:naLvaL1c}
\end{equation}
 signifying that the e-beam current density $j$ is continuous in
$z$ at the interaction points $a\ell$. In view of the Poisson equation
(\ref{eq:charvar1c}) and the first equation in (\ref{eq:chavar1a})
the following representation holds for the $\left[\partial_{z}E\right]\left(a\ell\right)$
at the interaction points $a\ell$:
\begin{equation}
\left[\partial_{z}E\right]\left(a\ell,t\right)=-\frac{4\pi}{\sigma_{\mathrm{B}}}\left[\partial_{z}q\right]\left(a\ell,t\right)=-4\pi e\left[n\right]\left(a\ell,t\right).\label{eq:naLvaL1d}
\end{equation}

Note that according to equations (\ref{eq:naLvaL1b}) the jumps in
the velocity the number density are in antiphase.

\subsection{Relation between the kinetic and the field points of view on the
gap interaction\label{subsec:relfiekin}}

An insightful comparative analysis of ``electron-wave theory'' and
the kinetic/ballistic theory of bunching is provided in \cite[II.15]{Shev}:

``A description of the mechanism of phase focusing as a phenomenon
of oscillating space-charge waves, is only a mathematical description
of a process the essence of which is as follows. The initial velocity
modulation gives rise to periodic concentration and dispersion of
electron space charge. The amount of bunching, and the associated
alternating current, increase through the bunching region provided
there are no repulsive space-charge forces affecting this process.
Space-charge forces oppose the initial velocity modulation, and cause
additional retardation and acceleration of the electrons. ... Thus
the law of conservation of energy is obeyed. On the other hand, the
ballistic theory is fundamentally contradictory to this.

In fact, the ballistic theory of bunching assumes that the alternating
velocity acquired by the electrons in the modulator remains constant
along the whole path. However, the potential energy necessarily increases
after electron bunching, and so the total energy of the electron beam
constantly varies, and this conflicts with the law of conservation
of energy. Despite this contradiction, the ballistic theory is a good
enough approximation for many of the cases met with in practice ...In
this case, both ballistic and electron-wave theories lead to identical
results. ``

In agreement with the above quotation our field theory of the space-charge
wave can be viewed as an effective mathematical descriptions of the
underlying physical complexity involving the electron velocity and
the electron number densities.

As to the energy conservation unlike the kinetic theory our Lagrangian
field theory surely provides for that. The field theory under some
conditions agrees at least with some points of the kinetic/ballistic
theory as we discuss below.

Hydrokinetic point of view on our simplifying assumption that the
cavity width $l_{\mathrm{g}}$ and the corresponding transit time
$\tau_{g}$ are zeros, see equations (\ref{eq:transL1d}) and Assumptions
\ref{ass:idealmod}), is as follows, \cite[II.5]{Shev}:
\begin{quotation}
``Let us assume further that the transit time of electrons between
grids 1 and 2 is infinitesimally small, which means a physically small
transit time compared with the period of oscillation of the high-frequency
field. If the transit time is negligible, electrons can be considered
to move through a constant (momentarily) alternating field, i. e.
virtually in a static field. The electrons acquire or lose an amount
of energy equal to the product of the electron charge and the momentary
value of the voltage. Therefore electrons entering the space between
the grids at different moments in time, with equal velocities, pass
out of this space at different velocities which are determined by
the momentary value of the alternating voltage. The electron beam
is thus velocity modulated and has a uniform density of space charge.''
\end{quotation}
The direct link between our field theory and the hydrokinetic theory
is provided by the e-beam Lagrangian $\mathcal{L}_{\mathrm{B}}$ defined
by equations (\ref{eq:T1B1be1b}) and (\ref{eq:aZep1c})
\begin{equation}
\mathcal{L}_{\mathrm{B}}\left(\left\{ q\right\} \right)=\frac{1}{2\beta}\left(D_{t}q\right)^{2}-\frac{2\pi}{\sigma_{\mathrm{B}}}q^{2},\quad D_{t}=\partial_{t}+\mathring{v}\partial_{z}.\label{eq:LaBqD1a}
\end{equation}
Indeed, its first kinetic term $\frac{1}{2\beta}\left(D_{t}q\right)^{2}$
involves the material time derivative which represents an important
concept of ``particle'' in the hydrokinetic theory. The second term
$-\frac{2\pi}{\sigma_{\mathrm{B}}}q^{2}$ in the e-beam Lagrangian
accounts for the electron-to-electron repulsion, a phenomenon neglected
by the standard ballistic analysis of the electron bunching.

Another link between the field and the kinetic theories comes from
our analysis in Section \ref{subsec:fiekin}. In view of equations
(\ref{eq:naLvaL1a}) and (\ref{eq:naLvaL1b}) jumps $\left[\partial_{z}q\right]\left(a\ell,t\right)$
that are explicitly allowed by the field theory represent jumps $\left[n\right]\left(a\ell,t\right)$
and $\left[v\right]\left(a\ell,t\right)$ related the kinetic properties
of the electron flow, see Remark \ref{rem:physjump}. Namely, jump
$\left[n\right]\left(a\ell,t\right)$ manifests the electron bunching,
jump $\left[v\right]\left(a\ell,t\right)$ manifests the ac electron
velocity modulation and equation (\ref{eq:naLvaL1b}) relates the
two of them. 

\section{Lagrangian variational framework\label{sec:lagvar}}

We construct here the Lagrangian variational framework for our model
of CCTWT. According to Assumption \ref{ass:idealmod} the model integrates
into it quantities associated with continuum of real numbers on one
hand and features associated with discrete points on the another hand.
The continuum features are represented by Lagrangian densities $\mathcal{L}_{\mathrm{T}}$
and $\mathcal{L}_{\mathrm{B}}$ in equations (\ref{eq:aZep1c}) whereas
discrete features are represented by Lagrangian $\mathcal{L}_{\mathrm{TB}}$
in equations (\ref{eq:aZep1d}) with energies concentrated in a set
of discrete points $a\mathbb{Z}$. One possibility for constructing
the desired Lagrangian variational framework is to apply the general
approach developed in \cite{FigRey2} when the ``rigidity'' condition
holds. Another possibility is to directly construct the Lagrangian
variational framework using some ideas from \cite{FigRey2} and that
is what we actually pursue here.

Following to the standard procedures of the Least Action principle
\cite[II.3]{ArnMech}, \cite[3]{GantM}, \cite[7]{GelFom}, \cite[8.6]{GoldM}
we start with setting up the action integral $S$ based on the Lagrangian
$\mathcal{L}$ defined by equations (\ref{eq:aZep1L}), (\ref{eq:aZep1c})
and (\ref{eq:aZep1d}). Using notations (\ref{eq:aZep1Q}) and (\ref{eq:eZep1x})
we define the action integral $S$ as follows:
\begin{gather}
S\left(\left\{ x\right\} \right)=\int_{t_{0}}^{t_{1}}\mathrm{d}t\int_{z_{1}}^{z_{2}}\mathcal{L}\left(\left\{ x\right\} \right)\,\mathrm{d}z=S_{\mathrm{T}}\left(\left\{ Q\right\} \right)+S_{\mathrm{B}}\left(\left\{ q\right\} \right)+S_{\mathrm{TB}}\left(x\right),\quad t_{0}<t_{1},\quad z_{0}<z_{1},\label{eq:Sact1a}
\end{gather}
where
\begin{equation}
S_{\mathrm{T}}\left(\left\{ Q\right\} \right)=\int_{t_{0}}^{t_{1}}\mathrm{d}t\int_{z_{1}}^{z_{2}}\mathcal{L}_{\mathrm{T}}\left(\left\{ Q\right\} \right)\,\mathrm{d}z=\int_{t_{0}}^{t_{1}}\mathrm{d}t\int_{z_{1}}^{z_{2}}\left[\frac{L}{2}\left(\partial_{t}Q\right)^{2}-\frac{1}{2C}\left(\partial_{z}Q\right)^{2}\right]\,\mathrm{d}z,\label{eq:Sact1b}
\end{equation}
\begin{equation}
S_{\mathrm{B}}\left(\left\{ q\right\} \right)=\int_{t_{0}}^{t_{1}}\mathrm{d}t\int_{z_{1}}^{z_{2}}\mathcal{L}_{\mathrm{B}}\left(\left\{ q\right\} \right)\,\mathrm{d}z=\int_{t_{0}}^{t_{1}}\mathrm{d}t\int_{z_{1}}^{z_{2}}\left[\frac{1}{2\beta}\left(\partial_{t}q+\mathring{v}\partial_{z}q\right)^{2}-\frac{2\pi}{\sigma_{\mathrm{B}}}q^{2}\right]\,\mathrm{d}z,\label{eq:Sact1c}
\end{equation}
\begin{gather}
S_{\mathrm{TB}}\left(\left\{ Q\right\} \right)=\int_{t_{0}}^{t_{1}}\mathrm{d}t\int_{z_{1}}^{z_{2}}\mathcal{L}_{\mathrm{TB}}\left(Q,q\right)\,\mathrm{d}z=\label{eq:Sact1d}\\
=\sum_{z_{1}<a\ell<z_{2}}\int_{t_{0}}^{t_{1}}\mathrm{d}t\left[\frac{l_{0}}{2}\left(\partial_{t}Q\left(a\ell\right)\right)^{2}-\frac{1}{2c_{0}}\left(Q\left(a\ell\right)+bq\left(a\ell\right)\right)\right]^{2}.\nonumber 
\end{gather}
To make expressions of the action integrals less cluttered we suppress
notationally their dependence on intervals $\left(z_{0},z_{1}\right)$
and $\left(t_{0},t_{1}\right)$ that can be chosen arbitrarily. We
consider then variation $\delta S$ of action $S$ assuming that variations
$\delta Q$ and $\delta q$ of charges $q=q\left(z,t\right)$ and
$Q=Q\left(z,t\right)$ vanish outside intervals $\left(z_{0},z_{1}\right)$
and $\left(t_{0},t_{1}\right)$, that is 
\begin{equation}
\delta Q\left(z,t\right)=\delta q\left(z,t\right)=0,\quad\left(z,t\right)\notin\left(z_{0},z_{1}\right)\times\left(t_{0},t_{1}\right),\label{eq:Sact1e}
\end{equation}
implying, in particular, that $\delta Q$ and $\delta q$ vanish on
the boundary of the rectangle $\left(z_{0},z_{1}\right)\times\left(t_{0},t_{1}\right)$,
that is
\begin{equation}
\delta Q\left(z,t\right)=\delta q\left(z,t\right)=0,\text{if }z=z_{0},\text{\ensuremath{z_{1}} or if }t=t_{0},t_{1}.\label{eq:Sact1f}
\end{equation}
We refer to variations $\delta Q$ and $\delta q$ satisfying equations
(\ref{eq:Sact1e}) and hence (\ref{eq:Sact1e}) for a rectangle $\left(z_{0},z_{1}\right)\times\left(t_{0},t_{1}\right)$
as\emph{ admissible}.

Following to the least action principle we introduce the functional
differential $\delta S$ of the action by the following formula \cite[7(35)]{GelFom}
\begin{equation}
\delta S=\lim_{\varepsilon\rightarrow0}\frac{S\left(\left\{ x+\varepsilon\delta x\right\} \right)-S\left(\left\{ x\right\} \right)}{\varepsilon}.\label{eq:Sact2a}
\end{equation}
Then the system configurations $x=x\left(z,t\right)$ that actually
can occur must satisfy
\begin{equation}
\delta S=\lim_{\varepsilon\rightarrow0}\frac{S\left(\left\{ x+\varepsilon\delta x\right\} \right)-S\left(\left\{ x\right\} \right)}{\varepsilon}=0\text{ for all admissible variations.}\label{eq:Sact2b}
\end{equation}
Let us choose now any $z$ outside lattice $a\mathbb{Z}$. Then there
always exist a sufficiently small $\xi>0$ and an integer $\ell_{0}$
such that
\begin{equation}
a\ell_{0}<z_{0}=z-\xi<z<z_{1}=z+\xi<a\left(\ell_{0}+1\right).\label{eq:Sact2c}
\end{equation}
If we apply now the variational principle (\ref{eq:Sact2c}) for all
admissible variations $\delta Q$ and $\delta q$ such that space
interval $\left(z_{0},z_{1}\right)$ is compliant with inequalities
(\ref{eq:Sact2b}) we readily find that
\begin{equation}
\delta S=\delta S_{\mathrm{T}}+\delta S_{\mathrm{B}}=0,\label{eq:Sact2d}
\end{equation}
where $S_{\mathrm{T}}$ and $S_{\mathrm{B}}$ are defined by expressions
(\ref{eq:Sact1b}) and (\ref{eq:Sact1c}). Using equations (\ref{eq:Sact1f})
and carrying out in the standard way the integration by parts transformations
we arrive at
\begin{gather}
\delta S_{\mathrm{T}}=-\int_{t_{0}}^{t_{1}}\mathrm{d}t\int_{z_{1}}^{z_{2}}\left[L\partial_{t}^{2}Q-C^{-1}\partial_{z}^{2}Q\right]\delta Q\,\mathrm{d}z,\label{eq:Sact2e}\\
\delta S_{\mathrm{B}}=-\int_{t_{0}}^{t_{1}}\mathrm{d}t\int_{z_{1}}^{z_{2}}\left[\frac{1}{\beta}\left(\partial_{t}+\mathring{v}\partial_{z}\right)^{2}q+\frac{4\pi}{\sigma_{\mathrm{B}}}q\right]\delta q\,\mathrm{d}z.\label{eq:Sact2f}
\end{gather}
Combining equations (\ref{eq:Sact2d}), (\ref{eq:Sact2e}) and (\ref{eq:Sact2f})
we arrive at the following EL equations

\begin{gather}
L\partial_{t}^{2}Q-C^{-1}\partial_{z}^{2}Q=0,\quad\frac{1}{\beta}\left(\partial_{t}+\mathring{v}\partial_{z}\right)^{2}q+\frac{4\pi}{\sigma_{\mathrm{B}}}q=0,\quad z\neq a\ell,\quad\ell\in\mathbb{Z}.\label{eq:Sact2g}
\end{gather}

Consider now the case when $z=a\ell_{0}$ for an integer $\ell_{0}$
and select space interval $\left(z_{0},z_{1}\right)$ as follows
\begin{equation}
a\left(\ell_{0}-1\right)<z_{0}=a\left(\ell_{0}-\frac{1}{2}\right)<z=a\ell_{0}<z_{1}=a\left(\ell_{0}+\frac{1}{2}\right)<a\left(\ell_{0}+1\right).\label{eq:Sact2h}
\end{equation}
Notice that in this case all actions $S_{\mathrm{T}}$, $S_{\mathrm{T}}$
and $S_{\mathrm{TB}}$ contribute to the variation $\delta S$. In
particular, as consequence of the presence of delta functions $\delta\left(z-a\ell\right)$
in the expression of the Lagrangian $\mathcal{L}_{\mathrm{TB}}$ defined
by equation (\ref{eq:aZep1d}) the space derivatives $\partial_{z}Q$
and $\partial_{z}q$ can have jumps at $z=a\ell_{0}$ as it was already
acknowledged by Assumption \ref{ass:jumpcon}. Based on this circumstance
we proceed as follows: (i) we split the integral with respect to the
space variable $z$ into two integrals:

\begin{equation}
\int_{z_{0}}^{z_{1}}=\int_{a\left(\ell_{0}-\frac{1}{2}\right)}^{a\ell_{0}}+\int_{a\ell_{0}}^{a\left(\ell_{0}+\frac{1}{2}\right)};\label{eq:Sact3a}
\end{equation}
(ii) we carry out the integration by parts for each of the two integrals
in the right-hand side of equation (\ref{eq:Sact3a}); (iii) we use
already established EL equations (\ref{eq:Sact2g}) to simplify the
integral expressions. When that is all done we arrive at the following:

\begin{equation}
\delta S_{\mathrm{T}}=\int_{t_{0}}^{t_{1}}\frac{1}{C}\left[\partial_{z}Q\right]\left(a\ell_{0},t\right)\delta Q\left(a\ell_{0},t\right)\,\mathrm{d}t,\quad\delta S_{\mathrm{B}}=-\int_{t_{0}}^{t_{1}}\frac{\mathring{v}^{2}}{\beta}\left[\partial_{z}q\right]\left(a\ell_{0},t\right)\delta q\left(a\ell_{0},t\right)\,\mathrm{d}t,\label{eq:Sact3b}
\end{equation}
where jumps $\left[\partial_{z}Q\right]\left(a\ell\right)$ and $\left[\partial_{z}q\right]\left(a\ell\right)$
are defined by equation (\ref{eq:limpm1b}), and
\begin{gather}
\delta S_{\mathrm{TB}}=-\int_{t_{0}}^{t_{1}}\left\{ l_{0}\partial_{t}^{2}Q\left(a\ell\right)+\frac{1}{c_{0}}\left[Q\left(a\ell_{0},t\right)+bq\left(a\ell_{0},t\right)\right]\right\} \delta Q\left(a\ell_{0},t\right)\,\mathrm{d}t-\label{eq:Sact3c}\\
-\int_{t_{0}}^{t_{1}}\left\{ \frac{b}{c_{0}}\left[Q\left(a\ell_{0},t\right)+bq\left(a\ell_{0},t\right)\right]\right\} \delta q\left(a\ell_{0},t\right)\,\mathrm{d}t.\nonumber 
\end{gather}
Using the variational principle (\ref{eq:Sact2b}), that is
\begin{equation}
\delta S_{\mathrm{T}}+\delta S_{\mathrm{B}}+\delta S_{\mathrm{TB}}=0,\label{eq:Sact3d}
\end{equation}
and the fact that variations $\delta Q\left(a\ell_{0},t\right)$ and
$\delta q\left(a\ell_{0},t\right)$ can be chosen arbitrarily we arrive
at the following equations

\begin{gather}
\frac{1}{C}\left[\partial_{z}Q\right]\left(a\ell_{0},t\right)=l_{0}\partial_{t}^{2}Q\left(a\ell\right)+\frac{1}{c_{0}}\left[Q\left(a\ell_{0},t\right)+bq\left(a\ell_{0},t\right)\right],\label{eq:Sact3e}\\
\frac{\mathring{v}^{2}}{\beta}\left[\partial_{z}q\right]\left(a\ell_{0},t\right)=-\frac{b}{c_{0}}\left[Q\left(a\ell_{0}\right)+bq\left(a\ell_{0},t\right)\right],\nonumber 
\end{gather}
where jumps $\left[\partial_{z}Q\right]\left(a\ell\right)$ and $\left[\partial_{z}q\right]\left(a\ell\right)$
are defined by equation (\ref{eq:limpm1b}). Equations (\ref{eq:Sact3e})
can be ready recast into the following boundary conditions
\begin{equation}
\left[\partial_{z}Q\right]\left(a\ell\right)=C_{0}\left[\left(\frac{\partial_{t}^{2}}{\omega_{0}^{2}}+1\right)Q\left(a\ell\right)+bq\left(a\ell\right)\right],\quad\left[\partial_{z}q\right]\left(a\ell\right)=-\frac{b\beta_{0}}{\mathring{v}^{2}}\left[Q\left(a\ell\right)+bq\left(a\ell\right)\right],\label{eq:Sact3ea}
\end{equation}
where
\begin{equation}
C_{0}=\frac{C}{c_{0}},\quad\omega_{0}=\frac{1}{\sqrt{l_{0}c_{0}}},\quad\beta_{0}=\frac{\beta}{c_{0}}.\label{eq:Sact3eb}
\end{equation}
We remind also that as consequence of continuity of $Q$ and $q$
we also have
\begin{equation}
\left[Q\right]\left(a\ell_{0},t\right)=0,\quad\left[q\right]\left(a\ell_{0},t\right)=0.\label{eq:Sact3g}
\end{equation}
Hence equations (\ref{eq:Sact3ea}) and (\ref{eq:Sact3g}) can be
viewed as the EL equations at point $a\ell_{0}$.

Equations (\ref{eq:Sact3ea}) at an interaction point $a\ell_{0}$
are perfectly consistent with boundary conditions (2.12) of the general
treatment in \cite{FigRey2}, which are
\begin{align}
-\frac{\partial L_{\mathrm{D}}}{\partial\partial_{1}\psi_{\mathrm{D}}^{\ell}}(b_{1},t)+\frac{\partial L_{\mathrm{B}}}{\partial\psi_{\mathrm{B}}^{\ell}(b_{1},t)}-\partial_{0}\left(\frac{\partial L_{\mathrm{B}}}{\partial\partial_{0}\psi_{\mathrm{B}}^{\ell}(b_{1},t)}\right) & =0;\label{eq:Sact3f}\\
\frac{\partial L_{\mathrm{D}}}{\partial\partial_{1}\psi_{\mathrm{D}}^{\ell}}(b_{2},t)+\frac{\partial L_{\mathrm{B}}}{\partial\psi_{\mathrm{B}}^{\ell}(b_{2},t)}-\partial_{0}\left(\frac{\partial L_{\mathrm{B}}}{\partial\left(\partial_{0}\psi_{\mathrm{B}}^{\ell}(b_{2},t)\right)}\right) & =0.\nonumber 
\end{align}
where (i) $b_{1}=a\ell_{0}-0$ and $b_{2}=a\ell_{0}+0$; (ii) $L_{\mathrm{D}}$
corresponds to $\mathcal{L}_{\mathrm{T}}+\mathcal{L}_{\mathrm{B}}$;
(iii) $L_{\mathrm{B}}$ corresponds to $\mathcal{L}_{\mathrm{TB}}$;
(iv) fields $\psi_{\mathrm{D}}^{\ell}$ correspond to charges $Q$
and $q$; (v) boundary fields $\psi_{\mathrm{B}}^{\ell}$ correspond
to $Q\left(a\ell_{0},t\right)$ and $q\left(a\ell_{0},t\right)$.
We remind the reader that boundary conditions (2.12) in \cite{FigRey2}
is an implementation of the ``rigidity'' requirement which is appropriate
for Lagrangian $\mathcal{L}_{\mathrm{TB}}$ defined by equation (\ref{eq:aZep1d}).
If fact, the signs of the terms containing $L_{\mathrm{D}}$ in equations
(\ref{eq:Sact3f}) are altered compare to original equations (2.12)
in \cite{FigRey2} to correct an unfortunate typo there.

\emph{Thus equations (\ref{eq:Sact3ea}) and (\ref{eq:Sact3g}) form
a complete set of the EL equations}.

\section{Root degeneracy for a special polynomial of the forth degree\label{sec:degpol4}}

The complex plane transformation $z\rightarrow\frac{1}{\bar{z}}$
is known as the \emph{unit (circle) inversion}, \cite[III.13]{YagCG},
and if a set is invariant under the transformation we refer to it
inversion symmetric set. Let us consider general form of polynomial
equation (\ref{eq:monS2f}) of the order 4
\begin{equation}
S^{4}+aS^{3}+bS^{2}+\bar{a}S+1=0,\quad a\in\mathbb{C},\quad b\in\mathbb{R}.\label{eq:speSab1a}
\end{equation}
If $S$ is a solution to equation (\ref{eq:speSab1a}) which is a
degenerate one then the following equation must hold also
\begin{equation}
\partial_{S}\left(S^{4}+aS^{3}+bS^{2}+\bar{a}S+1\right)=4S^{3}+3aS^{2}+2bS+\bar{a}=0.\label{eq:speSab1b}
\end{equation}
Subtracting from 2 times equation (\ref{eq:speSab1a}) $S$ times
equation (\ref{eq:speSab1b}) and dividing the result by $S^{2}$
we obtain
\begin{equation}
-2S^{2}+\frac{2}{S^{2}}-aS+\frac{\bar{a}}{S^{2}}=0.\label{eq:speSab1c}
\end{equation}
If a solution $S$ to the system of equations (\ref{eq:speSab1a})
and (\ref{eq:speSab1c}) lies on the unit circle, that is $\left|S\right|=1$,
then $S^{-1}=\bar{S}$ and the system is equivalent to the following
system of equations
\begin{equation}
\Re\left\{ S^{2}+aS+\frac{b}{2}\right\} =0,\quad\Im\left\{ 2S^{2}+aS\right\} =0,\quad S=\mathrm{e}^{\mathrm{i}\varphi},\quad\varphi\in\mathbb{R}.\label{eq:speSab1d}
\end{equation}
A trigonometric version of the system equations (\ref{eq:speSab1d})
is
\begin{gather}
\cos\left(2\varphi\right)+\left|a\right|\cos\left(\varphi+\alpha\right)=-\frac{b}{2},\quad2\sin\left(2\varphi\right)+\left|a\right|\sin\left(\varphi+\alpha\right)=0,\label{eq:speSab1e}\\
a=\left|a\right|\exp\left\{ \mathrm{i}\alpha\right\} ,\quad S=\exp\left\{ \mathrm{i}\varphi\right\} ,\quad\varphi,\alpha\in\mathbb{R}.\nonumber 
\end{gather}

\textbf{\vspace{0.2cm}
}

\textbf{ACKNOWLEDGMENT:} This research was supported by AFOSR MURI
under Grant No. FA9550-20-1-0409 administered through the University
of New Mexico. The author is grateful to E. Schamiloglu for sharing
his deep and vast knowledge of high power microwave devices and inspiring
discussions.\textbf{\vspace{0.2cm}
}

\textbf{NOMENCLATURE:} 
\begin{itemize}
\item $\mathbb{C}$ set of complex number.
\item $\mathbb{C}^{n}$ set of $n$ dimensional column vectors with complex
complex-valued entries.
\item $\mathbb{C}^{n\times m}$ set of $n\times m$ matrices with complex-valued
entries.
\item $D\left(\omega,k\right)$ CCTWT dispersion function.
\item $D_{\mathrm{C}}\left(\omega,k\right)$ CCS dispersion function.
\item $D_{\mathrm{K}}\left(\omega,k\right)$ MCK dispersion function.
\item $\det\left\{ A\right\} $ the determinant of matrix $A$.
\item $\mathrm{diag}\,\left(A_{1},A_{2},\ldots,A_{r}\right)$ block diagonal
matrix with indicated blocks.
\item $\mathrm{dim}\,\left(W\right)$ dimension of the vector space $W$.
\item EL the Euler-Lagrange (equations).
\item $\mathbb{I}_{\nu}$ $\nu\times\nu$ identity matrix.
\item $\mathrm{ker}\,\left(A\right)$ kernel of matrix $A$, that is the
vector space of vector $x$ such that $Ax=0$.
\item $M^{\mathrm{T}}$ matrix transposed to matrix $M$.
\item ODE ordinary differential equations.
\item $\bar{s}$ is complex-conjugate to complex number $s$.
\item $\sigma$$\left\{ A\right\} $ spectrum of matrix $A$.
\item $\mathbb{R}^{n\times m}$ set of $n\times m$ matrices with real-valued
entries.
\item $\chi_{A}\left(s\right)=\det\left\{ s\mathbb{I}_{\nu}-A\right\} $
characteristic polynomial of a $\nu\times\nu$ matrix $A$.
\end{itemize}
\renewcommand{\sectionname}{Appendix}
\counterwithin{section}{part}
\renewcommand{\thesection}{\Alph{section}}
\setcounter{section}{0}
\renewcommand{\theequation}{\Alph{section}.\arabic{equation}}

\section{Fourier transform\label{sec:four}}

Our preferred form of the Fourier transforms as in \cite[7.2, 7.5]{Foll},
\cite[20.2]{ArfWeb} is as follows:
\begin{gather}
f\left(t\right)=\int_{-\infty}^{\infty}\hat{f}\left(\omega\right)\mathrm{e}^{-\mathrm{i}\omega t}\,\mathrm{d}\omega,\quad\hat{f}\left(\omega\right)=\frac{1}{2\pi}\int_{-\infty}^{\infty}f\left(t\right)e^{\mathrm{i}\omega t}\,\mathrm{d}t,\label{eq:fourier1a}\\
f\left(z,t\right)=\int_{-\infty}^{\infty}\hat{f}\left(k,\omega\right)\mathrm{e}^{-\mathrm{i}\left(\omega t-kz\right)}\,\mathrm{d}k\mathrm{d}\omega,\label{eq:fourier1b}\\
\hat{f}\left(k,\omega\right)=\frac{1}{\left(2\pi\right)^{2}}\int_{-\infty}^{\infty}f\left(z,t\right)e^{\mathrm{i}\left(\omega t-kz\right)}\,dz\mathrm{d}t.\nonumber 
\end{gather}
This preference was motivated by the fact that the so-defined Fourier
transform of the convolution of two functions has its simplest form.
Namely, the convolution $f\ast g$ of two functions $f$ and $g$
is defined by \cite[7.2, 7.5]{Foll},
\begin{gather}
\left[f\ast g\right]\left(t\right)=\left[g\ast f\right]\left(t\right)=\int_{-\infty}^{\infty}f\left(t-t^{\prime}\right)g\left(t^{\prime}\right)\,\mathrm{d}t^{\prime},\label{eq:fourier2a}\\
\left[f\ast g\right]\left(z,t\right)=\left[g\ast f\right]\left(z,t\right)=\int_{-\infty}^{\infty}f\left(z-z^{\prime},t-t^{\prime}\right)g\left(z^{\prime},t^{\prime}\right)\,\mathrm{d}z^{\prime}\mathrm{d}t^{\prime}.\label{eq:fourier2b}
\end{gather}
 Then its Fourier transform as defined by equations (\ref{eq:fourier1a})
and (\ref{eq:fourier1b}) satisfies the following properties:
\begin{gather}
\widehat{f\ast g}\left(\omega\right)=\hat{f}\left(\omega\right)\hat{g}\left(\omega\right),\label{eq:fourier3a}\\
\widehat{f\ast g}\left(k,\omega\right)=\hat{f}\left(k,\omega\right)\hat{g}\left(k,\omega\right).\label{eq:fourier3b}
\end{gather}

\section{Jordan canonical form\label{sec:jord-form}}

We provide here very concise review of Jordan canonical forms following
mostly to \cite[III.4]{Hale}, \cite[3.1,3.2]{HorJohn}. As to a demonstration
of how Jordan block arises in the case of a single $n$-th order differential
equation we refer to \cite[25.4]{ArnODE}.

Let $A$ be an $n\times n$ matrix and $\lambda$ be its eigenvalue,
and let $r\left(\lambda\right)$ be the least integer $k$ such that
$\mathcal{N}\left[\left(A-\lambda\mathbb{I}\right)^{k}\right]=\mathcal{N}\left[\left(A-\lambda\mathbb{I}\right)^{k+1}\right]$,
where $\mathcal{N}\left[C\right]$ is a null space of a matrix $C$.
Then we refer to $M_{\lambda}=\mathcal{N}\left[\left(A-\lambda\mathbb{I}\right)^{r\left(\lambda\right)}\right]$
is the \emph{generalized eigenspace} of matrix $A$ corresponding
to eigenvalue $\lambda$. Then the following statements hold, \cite[III.4]{Hale}.
\begin{prop}[generalized eigenspaces]
\label{prop:gen-eig} Let $A$ be an $n\times n$ matrix and $\lambda_{1},\ldots,\lambda_{p}$
be its distinct eigenvalues. Then generalized eigenspaces $M_{\lambda_{1}},\ldots,M_{\lambda_{p}}$
are linearly independent, invariant under the matrix $A$ and
\begin{equation}
\mathbb{C}^{n}=M_{\lambda_{1}}\oplus\ldots\oplus M_{\lambda_{p}}.\label{eq:Mgeneig1a}
\end{equation}
Consequently, any vector $x_{0}$ in $\mathbb{C}^{n}$can be represented
uniquely as
\begin{equation}
x_{0}=\sum_{j=1}^{p}x_{0,j},\quad x_{0,j}\in M_{\lambda_{j}},\label{eq:Mgeneig1b}
\end{equation}
and
\begin{equation}
\exp\left\{ At\right\} x_{0}=\sum_{j=1}^{p}e^{\lambda_{j}t}p_{j}\left(t\right),\label{eq:Mgeneig1c}
\end{equation}
where column-vector polynomials $p_{j}\left(t\right)$ satisfy
\begin{gather}
p_{j}\left(t\right)=\sum_{k=0}^{r\left(\lambda_{j}\right)-1}\left(A-\lambda_{j}\mathbb{I}\right)^{k}\frac{t^{k}}{k!}x_{0,j},\quad x_{0,j}\in M_{\lambda_{j}},\quad1\leq j\leq p.\label{eq:Mgeneig1d}
\end{gather}
\end{prop}

For a complex number $\lambda$ a Jordan block $J_{r}\left(\lambda\right)$
of size $r\geq1$ is a $r\times r$ upper triangular matrix of the
form
\begin{gather}
J_{r}\left(\lambda\right)=\lambda\mathbb{I}_{r}+K_{r}=\left[\begin{array}{ccccc}
\lambda & 1 & \cdots & 0 & 0\\
0 & \lambda & 1 & \cdots & 0\\
0 & 0 & \ddots & \cdots & \vdots\\
\vdots & \vdots & \ddots & \lambda & 1\\
0 & 0 & \cdots & 0 & \lambda
\end{array}\right],\quad J_{1}\left(\lambda\right)=\left[\lambda\right],\quad J_{2}\left(\lambda\right)=\left[\begin{array}{cc}
\lambda & 1\\
0 & \lambda
\end{array}\right],\label{eq:Jork1a}
\end{gather}
\begin{equation}
K_{r}=J_{r}\left(0\right)=\left[\begin{array}{ccccc}
0 & 1 & \cdots & 0 & 0\\
0 & 0 & 1 & \cdots & 0\\
0 & 0 & \ddots & \cdots & \vdots\\
\vdots & \vdots & \ddots & 0 & 1\\
0 & 0 & \cdots & 0 & 0
\end{array}\right].\label{eq:Jork1k}
\end{equation}
The special Jordan block $K_{r}=J_{r}\left(0\right)$ defined by equation
(\ref{eq:Jork1k}) is an nilpotent matrix that satisfies the following
identities

\begin{gather}
K_{r}^{2}=\left[\begin{array}{ccccc}
0 & 0 & 1 & \cdots & 0\\
0 & 0 & 0 & \cdots & \vdots\\
0 & 0 & \ddots & \cdots & 1\\
\vdots & \vdots & \ddots & 0 & 0\\
0 & 0 & \cdots & 0 & 0
\end{array}\right],\cdots,\;K_{r}^{r-1}=\left[\begin{array}{ccccc}
0 & 0 & \cdots & 0 & 1\\
0 & 0 & 0 & \cdots & 0\\
0 & 0 & \ddots & \cdots & \vdots\\
\vdots & \vdots & \ddots & 0 & 0\\
0 & 0 & \cdots & 0 & 0
\end{array}\right],\quad K_{r}^{r}=0.\label{eq:Jord1a}
\end{gather}
A general Jordan $n\times n$ matrix $J$ is defined as a direct sum
of Jordan blocks, that is

\begin{equation}
J=\left[\begin{array}{ccccc}
J_{n_{1}}\left(\lambda_{1}\right) & 0 & \cdots & 0 & 0\\
0 & J_{n_{2}}\left(\lambda_{2}\right) & 0 & \cdots & 0\\
0 & 0 & \ddots & \cdots & \vdots\\
\vdots & \vdots & \ddots & J_{n_{q-1}}\left(\lambda_{n_{q}-1}\right) & 0\\
0 & 0 & \cdots & 0 & J_{n_{q}}\left(\lambda_{n_{q}}\right)
\end{array}\right],\quad n_{1}+n_{2}+\cdots n_{q}=n,\label{eq:Jork1b}
\end{equation}
where $\lambda_{j}$ need not be distinct. Any square matrix $A$
is similar to a Jordan matrix as in equation (\ref{eq:Jork1b}) which
is called \emph{Jordan canonical form} of $A$. Namely, the following
statement holds, \cite[3.1]{HorJohn}.
\begin{prop}[Jordan canonical form]
\label{prop:jor-can} Let $A$ be an $n\times n$ matrix. Then there
exists a non-singular $n\times n$ matrix $Q$ such that the following
block-diagonal representation holds
\begin{equation}
Q^{-1}AQ=J\label{eq:QAQC1a}
\end{equation}
where $J$ is the Jordan matrix defined by equation (\ref{eq:Jork1b})
and $\lambda_{j}$, $1\leq j\leq q$ are not necessarily different
eigenvalues of matrix $A$. Representation (\ref{eq:QAQC1a}) is known
as the \emph{Jordan canonical form} of matrix $A$, and matrices $J_{j}$
are called \emph{Jordan blocks}. The columns of the $n\times n$ matrix
$Q$ constitute the \emph{Jordan basis} providing for the Jordan canonical
form (\ref{eq:QAQC1a}) of matrix $A$.
\end{prop}

A function $f\left(J_{r}\left(s\right)\right)$ of a Jordan block
$J_{r}\left(s\right)$ is represented by the following equation \cite[7.9]{MeyCD},
\cite[10.5]{BernM}

\begin{gather}
f\left(J_{r}\left(s\right)\right)=\left[\begin{array}{ccccc}
f\left(s\right) & \partial f\left(s\right) & \frac{\partial^{2}f\left(s\right)}{2} & \cdots & \frac{\partial^{r-1}f\left(s\right)}{\left(r-1\right)!}\\
0 & f\left(s\right) & \partial f\left(s\right) & \cdots & \frac{\partial^{r-2}f\left(s\right)}{\left(r-2\right)!}\\
0 & 0 & \ddots & \cdots & \vdots\\
\vdots & \vdots & \ddots & f\left(s\right) & \partial f\left(s\right)\\
0 & 0 & \cdots & 0 & f\left(s\right)
\end{array}\right].\label{eq:JJordf1a}
\end{gather}
Note that any function $f\left(J_{r}\left(s\right)\right)$ of the
Jordan block $J_{r}\left(s\right)$ is evidently an upper triangular
Toeplitz matrix.

There are two particular cases of formula (\ref{eq:JJordf1a}), which
can also be derived straightforwardly using equations (\ref{eq:Jord1a}),
\begin{gather}
\exp\left\{ K_{r}t\right\} =\sum_{k=0}^{r-1}\frac{t^{k}}{k!}K_{r}^{k}=\left[\begin{array}{ccccc}
1 & t & \frac{t^{2}}{2!} & \cdots & \frac{t^{r-1}}{\left(r-1\right)!}\\
0 & 1 & t & \cdots & \frac{t^{r-2}}{\left(r-2\right)!}\\
0 & 0 & \ddots & \cdots & \vdots\\
\vdots & \vdots & \ddots & 1 & t\\
0 & 0 & \cdots & 0 & 1
\end{array}\right],\label{eq:Jord1c}
\end{gather}
 
\begin{gather}
\left[J_{r}\left(s\right)\right]^{-1}=\sum_{k=0}^{r-1}s^{-k-1}\left(-K_{r}\right)^{k}=\left[\begin{array}{ccccc}
\frac{1}{s} & -\frac{1}{s^{2}} & \frac{1}{s^{3}} & \cdots & \frac{\left(-1\right)^{r-1}}{s^{r}}\\
0 & \frac{1}{s} & -\frac{1}{s^{2}} & \cdots & \frac{\left(-1\right)^{r-2}}{s^{r-1}}\\
0 & 0 & \ddots & \cdots & \vdots\\
\vdots & \vdots & \ddots & \frac{1}{s} & -\frac{1}{s^{2}}\\
0 & 0 & \cdots & 0 & \frac{1}{s}
\end{array}\right].\label{eq:JJordf1b}
\end{gather}

\section{Companion matrix and cyclicity condition\label{sec:co-mat}}

The companion matrix $C\left(a\right)$ for the monic polynomial
\begin{equation}
a\left(s\right)=s^{\nu}+\sum_{1\leq k\leq\nu}a_{\nu-k}s^{\nu-k}\label{eq:compas1a}
\end{equation}
where coefficients $a_{k}$ are complex numbers is defined by \cite[5.2]{BernM}
\begin{equation}
C\left(a\right)=\left[\begin{array}{ccccc}
0 & 1 & \cdots & 0 & 0\\
0 & 0 & 1 & \cdots & 0\\
0 & 0 & 0 & \cdots & \vdots\\
\vdots & \vdots & \ddots & 0 & 1\\
-a_{0} & -a_{1} & \cdots & -a_{\nu-2} & -a_{\nu-1}
\end{array}\right].\label{eq:compas1b}
\end{equation}
Note that
\begin{equation}
\det\left\{ C\left(a\right)\right\} =\left(-1\right)^{\nu}a_{0}.\label{eq:compas1c}
\end{equation}

An eigenvalue is called \emph{cyclic (nonderogatory)} if its geometric
multiplicity is 1. A square matrix is called \emph{cyclic (nonderogatory)}
if all its eigenvalues are cyclic \cite[5.5]{BernM}. The following
statement provides different equivalent descriptions of a cyclic matrix
\cite[5.5]{BernM}.
\begin{prop}[criteria for a matrix to be cyclic]
\label{prop:cyc1} Let $A\in\mathbb{C}^{n\times n}$ be an $n\times n$
matrix with complex-valued entries. Let $\mathrm{spec}\,\left(A\right)=\left\{ \zeta_{1},\zeta_{2},\ldots,\zeta_{r}\right\} $
be the set of all distinct eigenvalues and $k_{j}=\mathrm{ind}{}_{A}\,\left(\zeta_{j}\right)$
is the largest size of Jordan block associated with $\zeta_{j}$.
Then the minimal polynomial $\mu_{A}\left(s\right)$ of the matrix
$A$, that is a monic polynomial of the smallest degree such that
$\mu_{A}\left(A\right)=0$, satisfies
\begin{equation}
\mu_{A}\left(s\right)=\prod_{j=1}^{r}\left(s-\zeta_{j}\right)^{k_{j}}.\label{eq:compas1d}
\end{equation}
Furthermore, the following statements are equivalent:
\end{prop}

\begin{enumerate}
\item $\mu_{A}\left(s\right)=\chi_{A}\left(s\right)=\det\left\{ s\mathbb{I}-A\right\} $.
\item $A$ is cyclic.
\item For every $\zeta_{j}$ the Jordan form of $A$ contains exactly one
block associated with $\zeta_{j}$.
\item $A$ is similar to the companion matrix $C\left(\chi_{A}\right)$.
\end{enumerate}
\begin{prop}[companion matrix factorization]
\label{prop:cyc2} Let $a\left(s\right)$ be a monic polynomial having
degree $\nu$ and $C\left(a\right)$ is its $\nu\times\nu$ companion
matrix. Then, there exist unimodular $\nu\times\nu$ matrices $S_{1}\left(s\right)$
and $S_{2}\left(s\right)$, that is $\det\left\{ S_{m}\right\} =\pm1$,
$m=1,2$, such that
\begin{equation}
s\mathbb{I}_{\nu}-C\left(a\right)=S_{1}\left(s\right)\left[\begin{array}{lr}
\mathbb{I}_{\nu-1} & 0_{\left(\nu-1\right)\times1}\\
0_{1\times\left(\nu-1\right)} & a\left(s\right)
\end{array}\right]S_{2}\left(s\right).\label{eq:compas1e}
\end{equation}
Consequently, $C\left(a\right)$ is cyclic and
\begin{equation}
\chi_{C\left(a\right)}\left(s\right)=\mu_{C\left(a\right)}\left(s\right)=a\left(s\right).\label{eq:compas1f}
\end{equation}
\end{prop}

The following statement summarizes important information on the Jordan
form of the companion matrix and the generalized Vandermonde matrix,
\cite[5.16]{BernM}, \cite[2.11]{LanTsi}, \cite[7.9]{MeyCD}.
\begin{prop}[Jordan form of the companion matrix]
\label{prop:cycJ} Let $C\left(a\right)$ be an $n\times n$ a companion
matrix of the monic polynomial $a\left(s\right)$ defined by equation
(\ref{eq:compas1a}). Suppose that the set of distinct roots of polynomial
$a\left(s\right)$ is $\left\{ \zeta_{1},\zeta_{2},\ldots,\zeta_{r}\right\} $
and $\left\{ n_{1},n_{2},\ldots,n_{r}\right\} $ is the corresponding
set of the root multiplicities such that
\begin{equation}
n_{1}+n_{2}+\cdots+n_{r}=n.\label{eq:compas2a}
\end{equation}
Then 
\begin{equation}
C\left(a\right)=RJR^{-1},\label{eq:compas2b}
\end{equation}
where
\begin{equation}
J=\mathrm{diag}\,\left\{ J_{n_{1}}\left(\zeta_{1}\right),J_{n_{2}}\left(\zeta_{2}\right),\ldots,J_{n_{r}}\left(\zeta_{r}\right)\right\} \label{eq:compas2c}
\end{equation}
is the the Jordan form of companion matrix $C\left(a\right)$ and
$n\times n$ matrix $R$ is the so-called generalized Vandermonde
matrix defined by
\begin{equation}
R=\left[R_{1}|R_{2}|\cdots|R_{r}\right],\label{eq:compas2d}
\end{equation}
 where $R_{j}$ is $n\times n_{j}$ matrix of the form
\begin{equation}
R_{j}=\left[\begin{array}{rrcr}
1 & 0 & \cdots & 0\\
\zeta_{j} & 1 & \cdots & 0\\
\vdots & \vdots & \ddots & \vdots\\
\zeta_{j}^{n-2} & \binom{n-2}{1}\,\zeta_{j}^{n-3} & \cdots & \binom{n-2}{n_{j}-1}\,\zeta_{j}^{n-n_{j}-1}\\
\zeta_{j}^{n-1} & \binom{n-1}{1}\,\zeta_{j}^{n-2} & \cdots & \binom{n-1}{n_{j}-1}\,\zeta_{j}^{n-n_{j}}
\end{array}\right].\label{eq:compas2f}
\end{equation}
As a consequence of representation (\ref{eq:compas2c}) $C\left(a\right)$
is a cyclic matrix.
\end{prop}

As to the structure of matrix $R_{j}$ in equation (\ref{eq:compas2f}),
if we denote by $Y\left(\zeta_{j}\right)$ its first column then it
can be expressed as follows \cite[2.11]{LanTsi}:
\begin{equation}
R_{j}=\left[Y^{\left(0\right)}|Y^{\left(1\right)}|\cdots|Y^{\left(n_{j}-1\right)}\right],\quad Y^{\left(m\right)}=\frac{1}{m!}\partial_{s_{j}}^{m}Y\left(\zeta_{j}\right),\quad0\leq m\leq n_{j}-1.\label{eq:compas3a}
\end{equation}
In the case when all eigenvalues of a cyclic matrix are distinct then
the generalized Vandermonde matrix turns into the standard Vandermonde
matrix
\begin{equation}
V=\left[\begin{array}{rrcr}
1 & 1 & \cdots & 1\\
\zeta_{1} & \zeta_{2} & \cdots & \zeta_{n}\\
\vdots & \vdots & \ddots & \vdots\\
\zeta_{1}^{n-2} & \zeta_{2}^{n-2} & \cdots & \zeta_{n}^{n-2}\\
\zeta_{1}^{n-1} & \zeta_{2}^{n-1} & \cdots & \zeta_{n}^{n-1}
\end{array}\right].\label{eq:compas3c}
\end{equation}

\section{Matrix polynomials\label{sec:mat-poly}}

An important incentive for considering matrix polynomials is that
they are relevant to the spectral theory of the differential equations
of the order higher than 1, particularly the Euler-Lagrange equations
which are the second-order differential equations in time. We provide
here selected elements of the theory of matrix polynomials following
mostly \cite[II.7, II.8]{GoLaRo}, \cite[9]{Baum}. The general matrix
polynomial eigenvalue problem reads
\begin{equation}
A\left(s\right)x=0,\quad A\left(s\right)=\sum_{j=0}^{\nu}A_{j}s^{j},\quad x\neq0,\label{eq:Aux1a}
\end{equation}
where $s$ is a complex number, $A_{k}$ are constant $m\times m$
matrices and $x\in\mathbb{C}^{m}$ is an $m$-dimensional column-vector.
We refer to problem (\ref{eq:Aux1a}) of funding complex-valued $s$
and non-zero vector $x\in\mathbb{C}^{m}$ as the \emph{polynomial
eigenvalue problem}. 

If a pair of a complex $s$ and non-zero vector $x$ solves problem
(\ref{eq:Aux1a}) we refer to $s$ as an \emph{eigenvalue} or as a\emph{
characteristic value} and to $x$ as the corresponding value to the
$s$ \emph{eigenvector}. Evidently the characteristic values of problem
(\ref{eq:Aux1a}) can be found from polynomial \emph{characteristic
equation} as follows:
\begin{equation}
\det\left\{ A\left(s\right)\right\} =0.\label{eq:Aux1b}
\end{equation}
We refer to matrix polynomial $A\left(s\right)$ as \emph{regular}
if $\det\left\{ A\left(s\right)\right\} $ is not identically zero.
We denote by $m\left(s_{0}\right)$ the \emph{multiplicity} (called
also \emph{algebraic multiplicity}) of eigenvalue $s_{0}$ as a root
of polynomial $\det\left\{ A\left(s\right)\right\} $. In contrast,
the \emph{geometric multiplicity} of eigenvalue $s_{0}$ is defined
as $\dim\left\{ \ker\left\{ A\left(s_{0}\right)\right\} \right\} $,
where $\ker\left\{ A\right\} $ defined for any square matrix $A$
stands for the subspace of solutions $x$ to equation $Ax=0$. Evidently,
the geometric multiplicity of eigenvalue does not exceed its algebraic
one, see Corollary \ref{cor:dim-ker}. 

It turns out that the matrix polynomial eigenvalue problem (\ref{eq:Aux1a})
can be always recast as the standard ``linear'' eigenvalue problem,
namely
\begin{equation}
\left(s\mathsf{B}-\mathsf{A}\right)\mathsf{x}=0,\label{eq:Aux1c}
\end{equation}
where $m\nu\times m\nu$ matrices $\mathsf{A}$ and $\mathsf{B}$
are defined by
\begin{gather}
\mathsf{B}=\left[\begin{array}{ccccc}
\mathbb{I} & 0 & \cdots & 0 & 0\\
0 & \mathbb{I} & 0 & \cdots & 0\\
0 & 0 & \ddots & \cdots & \vdots\\
\vdots & \vdots & \ddots & \mathbb{I} & 0\\
0 & 0 & \cdots & 0 & A_{\nu}
\end{array}\right],\quad\mathsf{A}=\left[\begin{array}{ccccc}
0 & \mathbb{I} & \cdots & 0 & 0\\
0 & 0 & \mathbb{I} & \cdots & 0\\
0 & 0 & 0 & \cdots & \vdots\\
\vdots & \vdots & \ddots & 0 & \mathbb{I}\\
-A_{0} & -A_{1} & \cdots & -A_{\nu-2} & -A_{\nu-1}
\end{array}\right],\label{eq:CBA1b}
\end{gather}
with $\mathbb{I}$ being $m\times m$ the identity matrix. Matrix
$\mathsf{A}$, particularly in the monic case, is often referred to
as \emph{companion matrix}. In the case of \emph{monic polynomial}
$A\left(\lambda\right)$, when $A_{\nu}=\mathbb{I}$ is the $m\times m$
identity matrix, matrix $\mathsf{B}=\mathsf{I}$ is the $m\nu\times m\nu$
identity matrix. The reduction of original polynomial problem (\ref{eq:Aux1a})
to an equivalent linear problem (\ref{eq:Aux1c}) is called \emph{linearization}.

The linearization is not unique, and one way to accomplish is by introducing
the so-called known ``companion polynomial'', which is the $m\nu\times m\nu$
matrix
\begin{gather}
\mathsf{C}_{A}\left(s\right)=s\mathsf{B}-\mathsf{A}=\left[\begin{array}{ccccc}
s\mathbb{I} & -\mathbb{I} & \cdots & 0 & 0\\
0 & s\mathbb{I} & -\mathbb{I} & \cdots & 0\\
0 & 0 & \ddots & \cdots & \vdots\\
\vdots & \vdots & \vdots & s\mathbb{I} & -\mathbb{I}\\
A_{0} & A_{1} & \cdots & A_{\nu-2} & sA_{\nu}+A_{\nu-1}
\end{array}\right].\label{eq:CBA1a}
\end{gather}
Notice that in the case of the EL equations the linearization can
be accomplished by the relevant Hamilton equations.

To demonstrate the equivalency between the eigenvalue problems for
the $m\nu\times m\nu$ companion polynomial $\mathsf{C}_{A}\left(s\right)$
and the original $m\times m$ matrix polynomial $A\left(s\right)$
we introduce two $m\nu\times m\nu$ matrix polynomials $\mathsf{E}\left(s\right)$
and $\mathsf{F}\left(s\right)$. Namely,
\begin{gather}
\mathsf{E}\left(s\right)=\left[\begin{array}{ccccc}
E_{1}\left(s\right) & E_{2}\left(s\right) & \cdots & E_{\nu-1}\left(s\right) & \mathbb{I}\\
-\mathbb{I} & 0 & 0 & \cdots & 0\\
0 & -\mathbb{I} & \ddots & \cdots & \vdots\\
\vdots & \vdots & \ddots & 0 & 0\\
0 & 0 & \cdots & -\mathbb{I} & 0
\end{array}\right],\label{eq:CBA1c}\\
\det\left\{ \mathsf{E}\left(s\right)\right\} =1,\nonumber 
\end{gather}
where $m\times m$ matrix polynomials $E_{j}\left(s\right)$ are defined
by the following recursive formulas
\begin{gather}
E_{\nu}\left(s\right)=A_{\nu},\quad E_{j-1}\left(s\right)=A_{j-1}+sE_{j}\left(s\right),\quad j=\nu,\ldots,2.\label{eq:CBA1d}
\end{gather}
Matrix polynomial $\mathsf{F}\left(s\right)$ is defined by
\begin{gather}
\mathsf{F}\left(s\right)=\left[\begin{array}{ccccc}
\mathbb{I} & 0 & \cdots & 0 & 0\\
-s\mathbb{I} & \mathbb{I} & 0 & \cdots & 0\\
0 & -s\mathbb{I} & \ddots & \cdots & \vdots\\
\vdots & \vdots & \ddots & \mathbb{I} & 0\\
0 & 0 & \cdots & -s\mathbb{I} & \mathbb{I}
\end{array}\right],\quad\det\left\{ \mathsf{F}\left(s\right)\right\} =1.\label{eq:CBA1e}
\end{gather}
Notice, that both matrix polynomials $\mathsf{E}\left(s\right)$ and
$\mathsf{F}\left(s\right)$ have constant determinants readily implying
that their inverses $\mathsf{E}^{-1}\left(s\right)$ and $\mathsf{F}^{-1}\left(s\right)$
are also matrix polynomials. Then, it is straightforward to verify
that
\begin{gather}
\mathsf{E}\left(s\right)\mathsf{C}_{A}\left(s\right)\mathsf{F}^{-1}\left(s\right)=\mathsf{E}\left(s\right)\left(s\mathsf{B}-\mathsf{A}\right)\mathsf{F}^{-1}\left(s\right)=\left[\begin{array}{ccccc}
A\left(s\right) & 0 & \cdots & 0 & 0\\
0 & \mathbb{I} & 0 & \cdots & 0\\
0 & 0 & \ddots & \cdots & \vdots\\
\vdots & \vdots & \ddots & \mathbb{I} & 0\\
0 & 0 & \cdots & 0 & \mathbb{I}
\end{array}\right].\label{eq:CBA1f}
\end{gather}
The identity (\ref{eq:CBA1f}) where matrix polynomials $\mathsf{E}\left(s\right)$
and $\mathsf{F}\left(s\right)$ have constant determinants can be
viewed as the definition of equivalency between matrix polynomial
$A\left(s\right)$ and its companion polynomial $\mathsf{C}_{A}\left(s\right)$. 

Let us take a look at the eigenvalue problem for eigenvalue $s$ and
eigenvector $\mathsf{x}\in\mathbb{C}^{m\nu}$ associated with companion
polynomial $\mathsf{C}_{A}\left(s\right)$, that is
\begin{gather}
\left(s\mathsf{B}-\mathsf{A}\right)\mathsf{x}=0,\quad\mathsf{x}=\left[\begin{array}{c}
x_{0}\\
x_{1}\\
x_{2}\\
\vdots\\
x_{\nu-1}
\end{array}\right]\in\mathbb{C}^{m\nu},\quad x_{j}\in\mathbb{C}^{m},\quad0\leq j\leq\nu-1,\label{eq:CBAx1a}
\end{gather}
where
\begin{equation}
\left(s\mathsf{B}-\mathsf{A}\right)\mathsf{x}=\left[\begin{array}{c}
sx_{0}-x_{1}\\
sx_{1}-x_{2}\\
\vdots\\
sx_{\nu-2}-x_{\nu-1}\\
\sum_{j=0}^{\nu-2}A_{j}x_{j}+\left(sA_{\nu}+A_{\nu-1}\right)x_{\nu-1}
\end{array}\right].\label{eq:CBAx1b}
\end{equation}
With equations (\ref{eq:CBAx1a}) and (\ref{eq:CBAx1b}) in mind we
introduce the following vector polynomial
\begin{equation}
\mathsf{x}_{s}=\left[\begin{array}{c}
x_{0}\\
sx_{0}\\
\vdots\\
s^{\nu-2}x_{0}\\
s^{\nu-1}x_{0}
\end{array}\right],\quad x_{0}\in\mathbb{C}^{m}.\label{eq:CBAx1c}
\end{equation}
Not accidentally, the components of the vector $\mathsf{x}_{s}$ in
its representation (\ref{eq:CBAx1c}) are in evident relation with
the derivatives $\partial_{t}^{j}\left(x_{0}\mathrm{e}^{st}\right)=s^{j}x_{0}\mathrm{e}^{st}$.
That is just another sign of the intimate relations between the matrix
polynomial theory and the theory of systems of ordinary differential
equations, see Appendix \ref{sec:dif-jord}. 
\begin{thm}[eigenvectors]
\label{thm:matpol-eigvec} Let $A\left(s\right)$ as in equations
(\ref{eq:Aux1a}) be regular, that $\det\left\{ A\left(s\right)\right\} $
is not identically zero, and let $m\nu\times m\nu$ matrices $\mathsf{A}$
and $\mathsf{B}$ be defined by equations (\ref{eq:Aux1b}). Then,
the following identities hold
\begin{equation}
\left(s\mathsf{B}-\mathsf{A}\right)\mathsf{x}_{s}=\left[\begin{array}{c}
0\\
0\\
\vdots\\
0\\
A\left(s\right)x_{0}
\end{array}\right],\;\mathsf{x}_{s}=\left[\begin{array}{c}
x_{0}\\
sx_{0}\\
\vdots\\
s^{\nu-2}x_{0}\\
s^{\nu-1}x_{0}
\end{array}\right],\label{eq:CBAx1d}
\end{equation}
\begin{gather}
\det\left\{ A\left(s\right)\right\} =\det\left\{ s\mathsf{B}-\mathsf{A}\right\} ,\quad\det\left\{ \mathsf{B}\right\} =\det\left\{ A_{\nu}\right\} ,\label{eq:CBAx1g}
\end{gather}
where $\det\left\{ A\left(s\right)\right\} =\det\left\{ s\mathsf{B}-\mathsf{A}\right\} $
is a polynomial of the degree $m\nu$ if $\det\left\{ \mathsf{B}\right\} =\det\left\{ A_{\nu}\right\} \neq0$.
There is one-to-one correspondence between solutions of equations
$A\left(s\right)x=0$ and $\left(s\mathsf{B}-\mathsf{A}\right)\mathsf{x}=0$.
Namely, a pair $s,\:\mathsf{x}$ solves eigenvalue problem $\left(s\mathsf{B}-\mathsf{A}\right)\mathsf{x}=0$
if and only if the following equalities hold
\begin{gather}
\mathsf{x}=\mathsf{x}_{s}=\left[\begin{array}{c}
x_{0}\\
sx_{0}\\
\vdots\\
s^{\nu-2}x_{0}\\
s^{\nu-1}x_{0}
\end{array}\right],\quad A\left(s\right)x_{0}=0,\quad x_{0}\neq0;\quad\det\left\{ A\left(s\right)\right\} =0.\label{eq:CBAx1e}
\end{gather}
\end{thm}

\begin{proof}
Polynomial vector identity (\ref{eq:CBAx1d}) readily follows from
equations (\ref{eq:CBAx1b}) and (\ref{eq:CBAx1c}). Identities (\ref{eq:CBAx1g})
for the determinants follow straightforwardly from equations (\ref{eq:CBAx1c}),
(\ref{eq:CBAx1e}) and (\ref{eq:CBA1f}). If $\det\left\{ \mathsf{B}\right\} =\det\left\{ A_{\nu}\right\} \neq0$
then the degree of the polynomial $\det\left\{ s\mathsf{B}-\mathsf{A}\right\} $
has to be $m\nu$ since $\mathsf{A}$ and $\mathsf{B}$ are $m\nu\times m\nu$
matrices.

Suppose that equations (\ref{eq:CBAx1e}) hold. Then combining them
with proven identity (\ref{eq:CBAx1d}) we get $\left(s\mathsf{B}-\mathsf{A}\right)\mathsf{x}_{s}=0$
proving that expressions (\ref{eq:CBAx1e}) define an eigenvalue $s$
and an eigenvector $\mathsf{x}=\mathsf{x}_{s}$.

Suppose now that $\left(s\mathsf{B}-\mathsf{A}\right)\mathsf{x}=0$
where $\mathsf{x}\neq0$. Combing that with equations (\ref{eq:CBAx1b})
we obtain
\begin{gather}
x_{1}=sx_{0},\quad x_{2}=sx_{1}=s^{2}x_{0},\cdots,\quad x_{\nu-1}=s^{\nu-1}x_{0},\label{eq:CBAx2a}
\end{gather}
implying that
\begin{equation}
\mathsf{x}=\mathsf{x}_{s}=\left[\begin{array}{c}
x_{0}\\
sx_{0}\\
\vdots\\
s^{\nu-2}x_{0}\\
s^{\nu-1}x_{0}
\end{array}\right],\quad x_{0}\neq0,\label{eq:CBAx2b}
\end{equation}
 and 
\begin{equation}
\sum_{j=0}^{\nu-2}A_{j}x_{j}+\left(sA_{\nu}+A_{\nu-1}\right)x_{\nu-1}=A\left(s\right)x_{0}.\label{eq:CBAx2c}
\end{equation}
Using equations (\ref{eq:CBAx2b}) and identity (\ref{eq:CBAx1d})
we obtain
\begin{equation}
0=\left(s\mathsf{B}-\mathsf{A}\right)\mathsf{x}=\left(s\mathsf{B}-\mathsf{A}\right)\mathsf{x}_{s}=\left[\begin{array}{c}
0\\
0\\
\vdots\\
0\\
A\left(s\right)x_{0}
\end{array}\right].\label{eq:CBAx2d}
\end{equation}
 Equations (\ref{eq:CBAx2d}) readily imply $A\left(s\right)x_{0}=0$
and $\det\left\{ A\left(s\right)\right\} =0$ since $x_{0}\neq0$.
That completes the proof.
\end{proof}
\begin{rem}[characteristic polynomial degree]
\label{rem:char-pol-deg} Note that according to Theorem \ref{thm:matpol-eigvec}
the characteristic polynomial $\det\left\{ A\left(s\right)\right\} $
for $m\times m$ matrix polynomial $A\left(s\right)$ has the degree
$m\nu$, whereas in linear case $s\mathbb{I}-A_{0}$ for $m\times m$
identity matrix $\mathbb{I}$ and $m\times m$ matrix $A_{0}$ the
characteristic polynomial $\det\left\{ s\mathbb{I}-A_{0}\right\} $
is of the degree $m$. This can be explained by observing that in
the non-linear case of $m\times m$ matrix polynomial $A\left(s\right)$
we are dealing effectively with many more $m\times m$ matrices $A$
than just a single matrix $A_{0}$.
\end{rem}

Another problem of our particular interest related to the theory of
matrix polynomials is eigenvalues and eigenvectors degeneracy and
consequently the existence of non-trivial Jordan blocks, that is Jordan
blocks of dimensions higher or equal to 2. The general theory addresses
this problem by introducing so-called ``Jordan chains'' which are
intimately related to the theory of system of differential equations
expressed as $A\left(\partial_{t}\right)x\left(t\right)=0$ and their
solutions of the form $x\left(t\right)=p\left(t\right)e^{st}$ where
$p\left(t\right)$ is a vector polynomial, see Appendix \ref{sec:dif-jord}
and \cite[I, II]{GoLaRo}, \cite[9]{Baum}. Avoiding the details of
Jordan chains developments we simply notice that an important to us
point of Theorem \ref{thm:matpol-eigvec} is that there is one-to-one
correspondence between solutions of equations $A\left(s\right)x=0$
and $\left(s\mathsf{B}-\mathsf{A}\right)\mathsf{x}=0$, and it has
the following immediate implication.
\begin{cor}[equality of the dimensions of eigenspaces]
\label{cor:dim-ker} Under the conditions of Theorem \ref{thm:matpol-eigvec}
for any eigenvalue $s_{0}$, that is, $\det\left\{ A\left(s_{0}\right)\right\} =0$,
we have
\begin{equation}
\dim\left\{ \ker\left\{ s_{0}\mathsf{B}-\mathsf{A}\right\} \right\} =\dim\left\{ \ker\left\{ A\left(s_{0}\right)\right\} \right\} .\label{eq:CBAx2e}
\end{equation}
In other words, the geometric multiplicities of the eigenvalue $s_{0}$
associated with matrices $A\left(s_{0}\right)$ and $s_{0}\mathsf{B}-\mathsf{A}$
are equal. In view of identity (\ref{eq:CBAx2e}) the following inequality
holds for the (algebraic) multiplicity $m\left(s_{0}\right)$
\begin{equation}
m\left(s_{0}\right)\geq\dim\left\{ \ker\left\{ A\left(s_{0}\right)\right\} \right\} .\label{eq:CBAx2f}
\end{equation}
\end{cor}

The next statement shows that if the geometric multiplicity of an
eigenvalue is strictly less than its algebraic one than there exist
non-trivial Jordan blocks, that is Jordan blocks of dimensions higher
or equal to 2.
\begin{thm}[non-trivial Jordan block]
\label{thm:Jord-block} Assuming notations introduced in Theorem
\ref{thm:matpol-eigvec} let us suppose that the multiplicity $m\left(s_{0}\right)$
of eigenvalue $s_{0}$ satisfies
\begin{equation}
m\left(s_{0}\right)>\dim\left\{ \ker\left\{ A\left(s_{0}\right)\right\} \right\} .\label{eq:CBAAx3a}
\end{equation}
Then the Jordan canonical form of companion polynomial $\mathsf{C}_{A}\left(s\right)=s\mathsf{B}-\mathsf{A}$
has a least one nontrivial Jordan block of the dimension exceeding
2.

In particular, if 
\begin{equation}
\dim\left\{ \ker\left\{ s_{0}\mathsf{B}-\mathsf{A}\right\} \right\} =\dim\left\{ \ker\left\{ A\left(s_{0}\right)\right\} \right\} =1,\label{eq:CBAAx3b}
\end{equation}
 and $m\left(s_{0}\right)\geq2$ then the Jordan canonical form of
companion polynomial $\mathsf{C}_{A}\left(s\right)=s\mathsf{B}-\mathsf{A}$
has exactly one Jordan block associated with eigenvalue $s_{0}$ and
its dimension is $m\left(s_{0}\right)$.
\end{thm}

The proof of Theorem \ref{thm:Jord-block} follows straightforwardly
from the definition of the Jordan canonical form and its basic properties.
Note that if equations (\ref{eq:CBAAx3b}) hold, it implies that the
eigenvalue $0$ is cyclic (nonderogatory) for matrix $A\left(s_{0}\right)$
and eigenvalue $s_{0}$ is cyclic (nonderogatory) for matrix $\mathsf{B}^{-1}\mathsf{A}$
provided $\mathsf{B}^{-1}$ exists, see Appendix  \ref{sec:co-mat}.

\section{Vector differential equations and the Jordan canonical form\label{sec:dif-jord}}

In this section we relate the vector ordinary differential equations
to the matrix polynomials reviewed in Appendix \ref{sec:mat-poly}
following \cite[5.1, 5.7]{GoLaRo2}, \cite[II.8.3]{GoLaRo}, \cite[III.4]{Hale},
\cite[7.9]{MeyCD}.

Equation $A\left(s\right)x=0$ with polynomial matrix $A\left(s\right)$
defined by equations (\ref{eq:Aux1a}) corresponds to the following
$m$-vector $\nu$-th order ordinary differential
\begin{equation}
A\left(\partial_{t}\right)x\left(t\right)=0,\text{ where }A\left(\partial_{t}\right)=\sum_{j=0}^{\nu}A_{j}\partial_{t}^{j},\label{eq:Adtx1}
\end{equation}
where $A_{j}=A_{j}\left(t\right)$ are $m\times m$ matrices. Introducing
$m\nu$-column-vector function
\begin{equation}
Y\left(t\right)=\left[\begin{array}{c}
x\left(t\right)\\
\partial_{t}x\left(t\right)\\
\vdots\\
\partial_{t}^{\nu-2}x\left(t\right)\\
\partial_{t}^{\nu-1}x\left(t\right)
\end{array}\right]\label{eq:yxt1b}
\end{equation}
and under the assumption that matrix $A_{\nu}\left(t\right)$ is the
identity matrix the differential equation (\ref{eq:Adtx1}) can be
recast and the first order differential equation

\begin{equation}
\partial_{t}Y\left(t\right)=\mathsf{A}Y\left(t\right),\label{eq:yxt1ba}
\end{equation}
where $\mathsf{A}$ is $m\nu\times m\nu$ matrix defined by
\begin{gather}
\mathsf{A}=\mathsf{A}\left(t\right)=\left[\begin{array}{ccccc}
0 & \mathbb{I} & \cdots & 0 & 0\\
0 & 0 & \mathbb{I} & \cdots & 0\\
0 & 0 & 0 & \cdots & \vdots\\
\vdots & \vdots & \ddots & 0 & \mathbb{I}\\
-A_{0}\left(t\right) & -A_{1}\left(t\right) & \cdots & -A_{\nu-2}\left(t\right) & -A_{\nu-1}\left(t\right)
\end{array}\right],\quad A_{\nu}\left(t\right)=\mathbb{I}.\label{eq:yxt1bb}
\end{gather}

\subsection{Constant coefficients case}

Let us consider an important special case of equation (\ref{eq:Adtx1})
when matrices $A_{j}$ are $m\times m$ that do not depended on $t$.
Then equation (\ref{eq:Adtx1}) can be recast as
\begin{equation}
\mathsf{B}\partial_{t}Y\left(t\right)=\mathsf{A}Y\left(t\right),\label{eq:yxt1a}
\end{equation}
where $\mathsf{A}$ and $\mathsf{B}$ are $m\nu\times m\nu$ companion
matrices defined by equations (\ref{eq:CBA1b}) and

In the case when $A_{\nu}$ is an invertible $m\times m$ matrix equation
(\ref{eq:yxt1a}) can be recast further as
\begin{equation}
\partial_{t}Y\left(t\right)=\dot{\mathsf{A}}Y\left(t\right),\label{eq:yxt1c}
\end{equation}
where
\begin{gather}
\dot{\mathsf{A}}=\left[\begin{array}{ccccc}
0 & \mathbb{I} & \cdots & 0 & 0\\
0 & 0 & \mathbb{I} & \cdots & 0\\
0 & 0 & 0 & \cdots & \vdots\\
\vdots & \vdots & \ddots & 0 & \mathbb{I}\\
-\dot{A}_{0} & -\dot{A}_{1} & \cdots & -\dot{A}_{\nu-2} & -\dot{A}_{\nu-1}
\end{array}\right],\quad\dot{A}_{j}=A_{\nu}^{-1}A_{j},\quad0\leq\nu-1.\label{eq:yxt1d}
\end{gather}
Note that one can interpret equation (\ref{eq:yxt1c}) as a particular
case of equation (\ref{eq:yxt1a}) where matrices $A_{\nu}$ and $\mathsf{B}$
are identity matrices of the respective dimensions $m\times m$ and
$m\nu\times m\nu$, and that polynomial matrix $A\left(s\right)$
defined by equations (\ref{eq:Aux1a}) becomes monic matrix polynomial
$\dot{A}\left(s\right)$, that is
\begin{gather}
\dot{A}\left(s\right)=\mathbb{I}s^{\nu}+\sum_{j=0}^{\nu-1}\dot{A}_{j}s^{j},\quad\dot{A}_{j}=A_{\nu}^{-1}A_{j},\quad0\leq\nu-1.\label{eq:yxt1e}
\end{gather}
Note that in view of equation (\ref{eq:yxt1b}), one recovers $x\left(t\right)$
from $Y\left(t\right)$ by using the following formula:
\begin{equation}
x\left(t\right)=P_{1}Y\left(t\right),\quad P_{1}=\left[\begin{array}{ccccc}
\mathbb{I} & 0 & \cdots & 0 & 0\end{array}\right],\label{eq:yxt2b}
\end{equation}
where $P_{1}$ evidently is $m\times m\nu$ matrix.

Observe also that, \cite[Prop. 5.1.2]{GoLaRo2}, \cite[14]{LanTsi}
\begin{gather}
\left[\dot{A}\left(s\right)\right]^{-1}=P_{1}\left[\mathbb{I}s-\dot{\mathsf{A}}\right]^{-1}R_{1},\quad P_{1}=\left[\begin{array}{ccccc}
\mathbb{I} & 0 & \cdots & 0 & 0\end{array}\right],\quad R_{1}=\left[\begin{array}{c}
0\\
0\\
\vdots\\
0\\
\mathbb{I}
\end{array}\right],\label{eq:yxt2a}
\end{gather}
where $P_{1}$ and $R_{1}$ evidently respectively $m\times m\nu$
and $m\nu\times m$ matrices.

The general form for the solution to vector differential equation
(\ref{eq:yxt1c}) is
\begin{equation}
Y\left(t\right)=\exp\left\{ \dot{\mathsf{A}}t\right\} Y_{0},\quad Y_{0}\in\mathbb{C}^{m\nu}.\label{eq:yxt2c}
\end{equation}
Then, using formulas (\ref{eq:yxt2b}) and (\ref{eq:yxt2c}) and Proposition
\ref{prop:gen-eig}, we arrive the following statement.
\begin{prop}[solution to the vector differential equation]
\label{prop:dif-sol-g} Let $\dot{\mathsf{A}}$ be $m\nu\times m\nu$
companion matrix defined by equations (\ref{eq:yxt1d}), $\zeta_{1},\ldots,\zeta_{p}$
be its distinct eigenvalues, and $M_{\zeta_{1}},\ldots,M_{\zeta_{p}}$
be the corresponding generalized eigenspaces of the corresponding
dimensions $r\left(\zeta_{j}\right)$, $1\leq j\leq p$. Then the
$m\nu$ column-vector solution $Y\left(t\right)$ to differential
equation (\ref{eq:yxt1c}) is of the form
\begin{gather}
Y\left(t\right)=\exp\left\{ \dot{\mathsf{A}}t\right\} Y_{0}=\sum_{j=1}^{p}e^{\zeta_{j}t}p_{j}\left(t\right),\quad Y_{0}=\sum_{j=1}^{p}Y_{0,j},\quad Y_{0,j}\in M_{\zeta_{j}},\label{eq:yxt2d}
\end{gather}
where $m\nu$-column-vector polynomials $p_{j}\left(t\right)$ satisfy
\begin{gather}
p_{j}\left(t\right)=\sum_{k=0}^{r\left(\zeta_{j}\right)-1}\frac{t^{k}}{k!}\left(\dot{\mathsf{A}}-\zeta_{j}\mathbb{I}\right)^{k}Y_{0,j},\quad1\leq j\leq p.\label{eq:yxt2e}
\end{gather}
Consequently, the general $m$-column-vector solution $x\left(t\right)$
to differential equation (\ref{eq:Adtx1}) is of the form
\begin{gather}
x\left(t\right)=\sum_{j=1}^{p}e^{\zeta_{j}t}P_{1}p_{j}\left(t\right),\quad P_{1}=\left[\begin{array}{ccccc}
\mathbb{I} & 0 & \cdots & 0 & 0\end{array}\right].\label{eq:yxt2f}
\end{gather}
\end{prop}

Note that $\chi_{\dot{\mathsf{A}}}\left(s\right)=\det\left\{ s\mathbb{I}-\dot{\mathsf{A}}\right\} $
is the characteristic function of the matrix $\dot{\mathsf{A}}$.
Then, using notations of Proposition \ref{prop:dif-sol-g}, we obtain
\begin{equation}
\chi_{\dot{\mathsf{A}}}\left(s\right)=\prod_{j=1}^{p}\left(s-\zeta_{j}\right)^{r\left(\zeta_{j}\right)}.\label{eq:yxt3a}
\end{equation}
Note also that for any values of complex-valued coefficients $b_{k}$
we have
\begin{gather}
\left(\partial_{t}-\zeta_{j}\right)^{r\left(\zeta_{j}\right)}\left[e^{\zeta_{j}t}p_{j}\left(t\right)\right]=0,\quad p_{j}\left(t\right)=\sum_{k=0}^{r\left(\zeta_{j}\right)-1}b_{k}t^{k},\label{eq:yxt3b}
\end{gather}
implying together with representation (\ref{eq:yxt3a}) that
\begin{gather}
\chi_{\dot{\mathsf{A}}}\left(\partial_{t}\right)\left[e^{\zeta_{j}t}p_{j}\left(t\right)\right]=0,\quad p_{j}\left(t\right)=\sum_{k=0}^{r\left(\zeta_{j}\right)-1}b_{k}t^{k}.\label{eq:yxt3c}
\end{gather}
Combining now Proposition \ref{prop:dif-sol-g} with equation (\ref{eq:yxt3c}),
we obtain the following statement.
\begin{cor}[property of a solution to the vector differential equation]
\label{cor:dif-sol-g} Let $x\left(t\right)$ be the general $m$-column-vector
solution $x\left(t\right)$ to differential equation (\ref{eq:Adtx1}).
Then $x\left(t\right)$ satisfies
\begin{equation}
\chi_{\dot{\mathsf{A}}}\left(\partial_{t}\right)x\left(t\right)=0.\label{eq:yxt3d}
\end{equation}
\end{cor}

\section{Floquet theory\label{sec:floquet}}

We provide here a concise review of the Floquet theory following \cite[III]{DalKre},
\cite[III.7]{Hale} and \cite[II.2]{YakSta}. The primary subject
of the Floquet theory is the general form of solutions to the ordinary
differential equations with periodic coefficients. With that in mind
suppose that: (i) $z$ is real valued variable, (ii) $x\left(z\right)$
is an $n$-vector valued function of $z$, (iii) $A\left(z\right)$
is an $n\times n$ matrix-valued $\varsigma$-periodic function of
$z$, and consider the following homogeneous linear periodic system:
\begin{equation}
\partial_{z}x\left(z\right)=A\left(z\right)x\left(z\right),\quad A\left(z+\varsigma\right)=A\left(z\right),\quad\varsigma>0.\label{eq:floq1a}
\end{equation}
We would like to give a complete characterization of the general structure
of the solutions to equation (\ref{eq:floq1a}). We start with the
following statement showing how to define the logarithm $B$ of a
matrix $C$ so that $C=\mathrm{e}^{B}$.
\begin{lem}[logarithm of a matrix]
\label{lem:log-mat} Let $C$ be an $n\times n$ matrix with $\det\left\{ C\right\} \neq0$.
Suppose that $C=Z^{-1}JZ$ where $J$ is the Jordan canonical form
of $C$ as described in Proposition \ref{prop:jor-can}. Then using
the block representation (\ref{eq:Jork1b}) for $J$, that is
\begin{equation}
J=J=\mathrm{diag}\,\left\{ J_{n_{1}}\left(\zeta_{1}\right),J_{n_{2}}\left(\zeta_{2}\right),\ldots,J_{n_{r}}\left(\zeta_{r}\right)\right\} ,\quad n_{1}+n_{2}+\cdots n_{q}=n,\label{eq:floq1b}
\end{equation}
we decompose $J$ into its diagonal and nilpotent components:
\begin{equation}
J=\mathrm{diag}\,\left\{ \lambda_{1}\mathbb{I}_{n_{1}},\lambda_{2}\mathbb{I}_{n_{2}},\ldots,\lambda_{q}\mathbb{I}_{n_{q}}\right\} +K\label{eq:floq1c}
\end{equation}
where
\begin{gather}
D=\mathrm{diag}\,\left\{ \lambda_{1}\mathbb{I}_{n_{1}},\lambda_{2}\mathbb{I}_{n_{2}},\ldots,\lambda_{q}\mathbb{I}_{n_{q}}\right\} ,\quad K=\mathrm{diag}\,\left\{ K_{n_{1}},K_{n_{2}},\ldots,K_{n_{q}}\right\} ,\label{eq:floq1d}\\
K_{n_{j}}=J_{n_{j}}\left(\lambda_{j}\right)-\lambda_{j}\mathbb{I}_{n_{j}},\quad1\leq j\leq q.\nonumber 
\end{gather}
Then let $\ln\left(\ast\right)$ be a branch of the logarithm and
let
\begin{equation}
H=\ln J=\mathrm{diag}\,\left\{ \ln\left(\lambda_{1}\right)\mathbb{I}_{n_{1}},\ln\left(\lambda_{2}\right)\mathbb{I}_{n_{2}},\ldots,\ln\left(\lambda_{q}\right)\mathbb{I}_{n_{q}}\right\} +S\label{eq:floq1e}
\end{equation}
where $\mathbb{I}_{n_{j}}$ are identity matrices of identified dimensions
and
\begin{equation}
S=\mathrm{diag}\,\left\{ S_{n_{1}},S_{n_{2}},\ldots,S_{n_{q}}\right\} ,\quad S_{n_{j}}=\sum_{m=1}^{n_{j}-1}\left(-1\right)^{m-1}\frac{1}{m\lambda_{j}^{m}}K_{n_{j}}^{m},\quad1\leq j\leq q.\label{eq:floq1g}
\end{equation}
Then
\begin{equation}
C=\mathrm{e}^{B},\quad B=\ln C=Z^{-1}HZ,\label{eq:floq1h}
\end{equation}
where matrix $H$ is defined by equation (\ref{eq:floq1e}).
\end{lem}

Note that matrix $S$ in equations (\ref{eq:floq1e}) and (\ref{eq:floq1g})
is associated with the nilpotent part of Jordan canonical form $J$.
The expression for $S_{n_{j}}$ originates in the series
\begin{equation}
\ln\left(1+s\right)=\sum_{m=1}^{\infty}\left(-1\right)^{m-1}\frac{1}{m}s^{m}=s-\frac{s^{2}}{2}+\frac{s^{3}}{3}+\cdots,\label{eq:floq2a}
\end{equation}
and it is a finite sum since $K_{n_{j}}$ is a nilpotent matrix such
that 
\begin{equation}
K_{n_{j}}^{m}=0,\quad m\geq n_{j},\quad1\leq j\leq q.\label{eq:floq2b}
\end{equation}

An $n\times n$ matrix $\Phi\left(z\right)$ is called \emph{matrizant
(matriciant)} of equation (\ref{eq:floq1a}) if it satisfies the following
equation:
\begin{equation}
\partial_{z}\Phi\left(z\right)=A\left(z\right)\Phi\left(z\right),\quad\Phi\left(0\right)=\mathbb{I},\quad A\left(z+\varsigma\right)=A\left(z\right),\quad\varsigma>0,\label{eq:floq2c}
\end{equation}
where $\mathbb{I}$ is the $n\times n$ identity matrix. Matrix $\Phi\left(z\right)$
is also called \emph{principal fundamental matrix} solution to equation
(\ref{eq:floq1a}). Evidently $x\left(z\right)=\Phi\left(z\right)x_{0}$
is the a solution to equation (\ref{eq:floq1a}) with the initial
condition $x\left(0\right)=x_{0}$. Using the fundamental solution
$\Phi\left(z\right)$ we can represent any matrix solution $\Psi\left(z\right)$
to equation (\ref{eq:floq1a}) based on its initial values as follows
\begin{equation}
\partial_{z}\Psi\left(z\right)=A\left(z\right)\Psi\left(z\right),\quad\Psi\left(z\right)=\Phi\left(z\right)\Psi\left(0\right).\label{eq:floq2d}
\end{equation}
In the case of $\varsigma$-periodic matrix function $A\left(z\right)$
the matrix function $\Psi\left(z\right)=\Phi\left(z+\varsigma\right)$
is evidently a solution to equation (\ref{eq:floq2d}) and consequently
\begin{equation}
\Phi\left(z+\varsigma\right)=\Phi\left(z\right)\Phi\left(\varsigma\right).\label{eq:floq2ca}
\end{equation}
It turns out that matrix $M_{\varsigma}=\Phi\left(\varsigma\right)$
called the \emph{monodromy matrix} is of particular importance for
the analysis of solutions to equation (\ref{eq:floq2c}) with $\varsigma$-periodic
matrix function $A\left(z\right)$. 

The monodromy matrix is integrated into the formulation of the main
statement of the Floquet theory describing the structure of solutions
to equation (\ref{eq:floq2d}) for $\varsigma$-periodic matrix function
$A\left(z\right)$.
\begin{thm}[Floquet]
\label{thm:floquet} Suppose that $A\left(z\right)$ is a $\varsigma$-periodic
continuous function of $z$. Let $\Phi\left(z\right)$ be the matrizant
of equation (\ref{eq:floq2c}) and let $M_{\varsigma}=\Phi\left(\varsigma\right)$
be the corresponding monodromy matrix. Using the statement of Lemma
\ref{lem:log-mat} we introduce matrix $\Gamma$ defined by
\begin{equation}
\Gamma=\frac{1}{\varsigma}\ln M_{\varsigma}=\frac{1}{\varsigma}\ln\Phi\left(\varsigma\right),\text{implying }M_{\varsigma}=\Phi\left(\varsigma\right)=\mathrm{e}^{\Gamma\varsigma}.\label{eq:floq2f}
\end{equation}
Then matrizant $\Phi\left(z\right)$ satisfies the following equation
called Floquet representation
\begin{equation}
\Phi\left(z\right)=P\left(z\right)\mathrm{e}^{\Gamma z},\quad P\left(z+\varsigma\right)=P\left(z\right),\quad P\left(0\right)=\mathbb{I},\label{eq:floq2g}
\end{equation}
where $P\left(z\right)$ is a differentiable $\varsigma$-periodic
matrix function of $z$. 
\end{thm}

\begin{proof}
Let us define matrix $P\left(z\right)$ by the following equation
\begin{equation}
P\left(z\right)=\Phi\left(z\right)\mathrm{e}^{-\Gamma z}.\label{eq:floq3a}
\end{equation}
Then combining representation (\ref{eq:floq3a}) for $P\left(z\right)$
with equations (\ref{eq:floq2ca}) and (\ref{eq:floq2f}) we obtain
\begin{equation}
P\left(z+\varsigma\right)=\Phi\left(z+\varsigma\right)\mathrm{e}^{-\Gamma\left(z+\varsigma\right)}=\Phi\left(z\right)\Phi\left(\varsigma\right)\mathrm{e}^{-\Gamma\varsigma}\mathrm{e}^{-\Gamma z}=\Phi\left(z\right)\mathrm{e}^{-\Gamma z}=P\left(z\right),\label{eq:floq3b}
\end{equation}
that is $P\left(z\right)$ is a differentiable $\varsigma$-periodic
matrix function of $z$. Equality $P\left(0\right)=\mathbb{I}$ readily
follows from equation (\ref{eq:floq3a}) and equality $\varPhi\left(0\right)=\mathbb{I}$.
\end{proof}
The eigenvalues of the monodromy matrix $\Phi\left(\varsigma\right)=\mathrm{e}^{\Gamma\varsigma}$
are called\emph{ Floquet (characteristic) multipliers} and their logarithms
(not uniquely defined) are called\emph{ characteristic exponents.}
\begin{defn}[Floquet multipliers, characteristic exponents and eigenmodes]
\label{def:floqmul} Using notation of Theorem \ref{thm:floquet}
let us consider complex numbers $\kappa$, $s_{\kappa}$ and vector
$y_{\kappa}$ satisfying the following equations
\begin{equation}
\Gamma y_{\kappa}=\kappa y_{\kappa},\quad\Phi\left(\varsigma\right)y_{\kappa}=\mathrm{e}^{-\Gamma\varsigma}y_{\kappa}=s_{\kappa}y_{\kappa},\quad s_{\kappa}=\mathrm{e}^{\kappa\varsigma},\label{eq:Gamyk1a}
\end{equation}
where evidently $\kappa$ and $y_{\kappa}$ are respectively an eigenvalue
and the corresponding eigenvector of matrix $\Gamma$. We refer to
$\kappa$ and $s_{\kappa}$ respectively as the \emph{Floquet characteristic
exponent} and the \emph{Floquet (characteristic) multiplier}.

Using $\kappa$ and $y_{\kappa}$ defined above we introduce the following
special solution to the original differential equation (\ref{eq:floq1a}):
\begin{equation}
\psi_{\kappa}\left(z\right)=p_{\kappa}\left(z\right)\mathrm{e}^{\kappa z}=\Phi\left(z\right)y_{\kappa}=P\left(z\right)\mathrm{e}^{\Gamma z}y_{\kappa},\quad p_{\kappa}\left(z\right)=P\left(z\right)y_{\kappa},\label{eq:Gamyk1b}
\end{equation}
and we refer to it as the \emph{Floquet eigenmode}. Note that $p_{\kappa}\left(z\right)$
in equations (\ref{eq:Gamyk1b}) is a $\varsigma$-periodic vector-function
of $z$.
\end{defn}

\begin{rem}[Floquet eigenmodes]
 If $\psi_{\kappa}\left(z\right)$ is the Floquet eigenmode defined
by equations (\ref{eq:Gamyk1b}) and $\Re\left\{ \kappa\right\} >0$
or, equivalently, $\left|s_{\kappa}\right|>1$ then $\psi_{\kappa}\left(z\right)$
grows exponentially as $z\rightarrow+\infty$, and we refer to such
$\psi_{\kappa}\left(z\right)$ as \emph{exponentially growing Floquet
eigenmode}. In the case when $\Re\left\{ \kappa\right\} =0$ or equivalently
$\left|s_{\kappa}\right|=1$ function $\psi_{\kappa}\left(z\right)$
is bounded and we refer to such $\psi_{\kappa}\left(z\right)$ as
an \emph{oscillatory Floquet eigenmode}. 
\end{rem}

\begin{rem}[dispersion relations]
 \label{rem:disprel} In physical applications of the Floquet theory
$\varsigma$-periodic matrix valued function $A\left(z\right)$ in
differential equation (\ref{eq:floq1a}) depends on the frequency
$\omega$, that is $A\left(z\right)=A\left(z,\omega\right)$. In this
case we also have $\kappa=\kappa\left(\omega\right)$. If we naturally
introduce the wave number $k$ by
\begin{equation}
k=k\left(\omega\right)=-\mathrm{i}\kappa\left(\omega\right),\label{eq:Gamyk1c}
\end{equation}
then the relation between $\omega$ and $k$ provided by equation
(\ref{eq:Gamyk1c}) is called the \emph{dispersion relation}.
\end{rem}

\section{Hamiltonian system of linear differential equations\label{sec:Ham}}

We follow here to \cite[I.8, V.1]{DalKre} and \cite[III]{YakSta}.
We introduce first\emph{ indefinite scalar product} $\left\langle x,y\right\rangle $
on the vector space $\mathbb{C}^{n}$ associated with a nonsingular
Hermitian $n\times n$ matrix $G$, namely
\begin{equation}
\left\langle x,y\right\rangle =\overline{\left\langle y,x\right\rangle }=x^{*}Gy,\quad G^{*}=G,\quad\det\left\{ G\right\} \neq0,\quad x,y\in\mathbb{C}^{n}.\label{eq:Gindef1a}
\end{equation}
We refer to matrix the $G$ \emph{metric matrix}. We define, then,
for any $n\times n$ matrix $A$ another matrix $A^{\dagger}$ called
adjoint by using the following relations:
\begin{equation}
\left\langle Ax,y\right\rangle =\left\langle x,A^{\dagger}y\right\rangle \text{ or equivalently }A^{\dagger}=G^{-1}A^{*}G.\label{eq:Gindef1b}
\end{equation}
Note that relations (\ref{eq:Gindef1b}) readily imply that
\begin{equation}
\left(AB\right)^{\dagger}=B^{\dagger}A^{\dagger}.\label{eq:Gindef1c}
\end{equation}
\begin{table}
\centering{}%
\begin{tabular}{|c|c|c|}
\hline 
$G$-unitary & $G$-skew-Hermitian & $G$-Hermitian\tabularnewline
\hline 
\hline 
$\left\langle Ax,Ay\right\rangle =\left\langle x,y\right\rangle $ & $\left\langle Ax,y\right\rangle =-\left\langle x,Ay\right\rangle $ & $\left\langle Ax,y\right\rangle =\left\langle x,Ay\right\rangle $\tabularnewline
\hline 
$A^{\dagger}A=G^{-1}A^{*}GA=\mathbb{I},$ & $A^{\dagger}=G^{-1}A^{*}G=-A$, & $A^{\dagger}=G^{-1}A^{*}G=A,$\tabularnewline
\hline 
$A^{*}=GA^{-1}G^{-1}$ & $GA+A^{*}G=0$ & $GA-A^{*}G=0$\tabularnewline
\hline 
$A^{*}GA=G$ & $A=\mathrm{i}G^{-1}H,\;H=H^{*}$ & $A=G^{-1}H,\;H=H^{*}$\tabularnewline
\hline 
\end{tabular}\vspace{0.3cm}
\caption{\label{tab:Gunit}$G$-unitary, $G$-skew-Hermitian and $G$-Hermitian
matrices. }
\end{table}

Let $G$ and $H\left(t\right)$ be Hermitian $n\times n$ matrices
and suppose that matrix $G$ is nonsingular. We define the \emph{Hamiltonian
system of equations} to be a system of the form.
\begin{equation}
-\mathrm{i}G\partial_{t}x\left(t\right)=H\left(t\right)x\left(t\right),\quad H^{*}\left(t\right)=H\left(t\right).\label{eq:Gindefe1d}
\end{equation}
If based on matrices $G$ and $H\left(t\right)$ we introduce $G$-skew-Hermitian
matrix,
\begin{equation}
A\left(t\right)=\mathrm{i}G^{-1}H\left(t\right),\label{eq:Gindefe1da}
\end{equation}
we can recast the Hamiltonian system (\ref{eq:Gindefe1d}) in the
following equivalent form:
\begin{equation}
\partial_{t}x\left(t\right)=A\left(t\right)x\left(t\right),\quad A^{\dagger}\left(t\right)=-A\left(t\right).\label{eq:Gindefe1e}
\end{equation}
It turns out that the matrizant $\Phi\left(t\right)$ of equation
(\ref{eq:Gindefe1e}) with $G$-skew-Hermitian matrix $A\left(t\right)$
is a $G$-unitary matrix for each value of $t$. Indeed, using equation
(\ref{eq:Gindefe1e}) together with equations (\ref{eq:Gindef1b})
and (\ref{eq:Gindef1c}), we obtain
\begin{gather}
\partial_{t}\left[\Phi^{\dagger}\left(t\right)\Phi\left(t\right)\right]=\left\{ \partial_{t}\left[\Phi\left(t\right)\right]\right\} ^{\dagger}\Phi\left(t\right)+\Phi^{\dagger}\left(t\right)\partial_{t}\left[\Phi\left(t\right)\right]=\label{eq:Gindefe1f}\\
=-\Phi^{\dagger}\left(t\right)A\left(t\right)\Phi\left(t\right)+\Phi^{\dagger}\left(t\right)A\left(t\right)\Phi\left(t\right)=0,\nonumber 
\end{gather}
implying that matrizant $\Phi\left(t\right)$ satisfies
\begin{equation}
\Phi^{\dagger}\left(t\right)\Phi\left(t\right)=\mathbb{I},\text{or equivalently }\Phi^{*}\left(t\right)G\Phi\left(t\right)=G,\label{eq:Gindefe2a}
\end{equation}
implying that $\Phi\left(t\right)$ is a $G$-unitary matrix for each
value of $t$. Identity (\ref{eq:Gindefe2a}) implies in turn that
for any two solutions $x\left(t\right)$ and $y\left(t\right)$ to
the Hamiltonian system (\ref{eq:Gindefe1d}) we always have
\begin{equation}
\left\langle x\left(t\right),y\left(t\right)\right\rangle =x^{*}\left(t\right)Gy\left(t\right)=x^{*}\left(0\right)\Phi^{*}\left(t\right)G\Phi\left(t\right)y\left(0\right)=\left\langle x\left(0\right),y\left(0\right)\right\rangle ,\label{eq:Gindefe2aa}
\end{equation}
that is $\left\langle x\left(t\right),y\left(t\right)\right\rangle $
does not depend on $t$.

\subsection{Symmetry of the spectra}

$G$-unitary, $G$-skew-Hermitian and $G$-Hermitian matrices have
special properties described in Table \ref{tab:Gunit}. These properties
can viewed as symmetries, and not surprisingly, they imply consequent
symmetries of the spectra of the matrices. Let $\sigma$$\left\{ A\right\} $
denote the spectrum of matrix $A$. It is a straightforward exercise
to verify based on matrix properties described in Table \ref{tab:Gunit}
that the following statements hold.
\begin{thm}[spectral symmetries]
\label{thm:sp-sym} Suppose that matrix $A$ is either $G$-unitary
or $G$-skew-Hermitian or $G$-Hermitian. Then the following statements
hold:
\begin{enumerate}
\item If $A$ is $G$-unitary then $\sigma$$\left\{ A\right\} $ is symmetric
with respect to the unit circle, that is
\begin{equation}
\zeta\in\sigma\left\{ \Phi\right\} \Rightarrow\frac{1}{\bar{\zeta}}\in\sigma\left\{ \Phi\right\} .\label{eq:Gindef2b}
\end{equation}
\item If $A$ is $G$-skew-Hermitian then $\sigma$$\left\{ A\right\} $
is symmetric with respect the imaginary axis, that is
\begin{equation}
\zeta\in\sigma\left\{ \Phi\right\} \Rightarrow-\bar{\zeta}\in\sigma\left\{ \Phi\right\} .\label{eq:Gindef2ba}
\end{equation}
\item If $A$ is $G$-Hermitian then $\sigma$$\left\{ A\right\} $ is symmetric
with respect to real axis, that is
\begin{equation}
\zeta\in\sigma\left\{ \Phi\right\} \Rightarrow\bar{\zeta}\in\sigma\left\{ \Phi\right\} .\label{eq:Gindef2bb}
\end{equation}
\end{enumerate}
\end{thm}

The following statement describes the $G$-orthogonality of invariant
subspaces of $G$-unitary, $G$-skew-Hermitian and $G$-Hermitian
matrices, \cite[1.8]{DalKre}.
\begin{thm}[eigenspaces]
\label{thm:G-eig} Suppose that matrix $A$ is either $G$-unitary
or $G$-skew-Hermitian or $G$-Hermitian. Then the following statements
hold. Let $\Lambda\subset\sigma\left\{ A\right\} $ be a subset of
the spectrum $\sigma\left\{ A\right\} $ of the matrix $A$, and let
$\tilde{\Lambda}$ be the relevant symmetric image of $\Lambda$ defined
by
\[
\tilde{\Lambda}=\left\{ \begin{array}{rrr}
\left\{ \frac{1}{\bar{\zeta}}:\zeta\in\Lambda\right\}  & \text{if} & A\text{ is \ensuremath{G}-unitary }\\
\left\{ -\bar{\zeta}:\zeta\in\Lambda\right\}  & \text{if} & A\text{ is \ensuremath{G}-skew-Hermitian }\\
\left\{ \bar{\zeta}:\zeta\in\Lambda\right\}  & \text{if} & \text{ is \ensuremath{G}-Hermitian }
\end{array}\right..
\]
Let $\Lambda_{1},\Lambda_{2}\subset\sigma\left\{ A\right\} $ be two
subsets of the spectrum $\sigma\left\{ A\right\} $ so that $\tilde{\Lambda}_{1}$
and $\Lambda_{2}$ are separated from each other by non-intersecting
contours $\tilde{\Gamma}_{1}$ and $\varGamma_{2}$. Then the invariant
subspaces $E_{1}$ and $E_{1}$ of the matrix $A$ corresponding to
$\Lambda_{1}$ and $\Lambda_{2}$ are $G$-orthogonal.
\end{thm}

The statement below describes a special property of eigenvectors of
a $G$-unitary matrix.
\begin{lem}[isotropic eigenvector]
\label{lem:isoeig} Let $A$ be a $G$-unitary matrix and $\zeta$
be its eigenvalue that does not lie on the unit circuit, that is,
$\left|\zeta\right|\neq1$. Then if $x$ is the eigenvector corresponding
to $\zeta$ it is isotropic, that is
\begin{equation}
\left\langle x,x\right\rangle =x^{*}Gx=0,\quad Ax=\zeta x,\quad\left|\zeta\right|\neq1.\label{eq:Gindef2d}
\end{equation}
\end{lem}

\begin{proof}
Since $Ax=\zeta x$ and $A$ is a $G$-unitary, we have

\[
\left\langle Ax,Ax\right\rangle =\left\langle \zeta x,\zeta x\right\rangle =\left|\zeta\right|^{2}\left\langle x,x\right\rangle ,\quad\left\langle Ax,Ax\right\rangle =\left\langle x,x\right\rangle .
\]
Combining the two equation above with $\left|\zeta\right|\neq1$ we
conclude that $\left\langle x,x\right\rangle =0$ which is the desired
equation (\ref{eq:Gindef2b}).
\end{proof}

\subsection{Special Hamiltonian systems\label{subsec:speHam}}

With equation (\ref{eq:Lagdim2e}) in mind we introduce the following
system
\begin{equation}
\partial_{t}x\left(t\right)=A\left(t\right)x\left(t\right),\label{eq:spHer1a}
\end{equation}
where 4$\times4$ matrix function $A\left(z\right)$ is of the following
form special form
\begin{equation}
A\left(z\right)=\left[\begin{array}{rrrr}
0 & 0 & 1 & 0\\
0 & 0 & 0 & 1\\
a_{1}\left(z\right) & a\left(z\right)a_{2} & \mathrm{i}c_{1} & 0\\
a\left(z\right)a_{3} & a_{4}\left(z\right) & 0 & \mathrm{i}c_{2}
\end{array}\right],\quad a_{1}\left(z\right),a_{2},a_{3},a_{4}\left(z\right),c_{1},c_{2},a\left(z\right)\in\mathbb{R}.\label{eq:spHer1b}
\end{equation}
The system (\ref{eq:spHer1a}), (\ref{eq:spHer1a}) is Hamiltonian
if we select Hermitian matrix $G$ to be
\begin{equation}
G=\left[\begin{array}{rrrr}
c_{1} & 0 & \mathrm{i} & 0\\
0 & \frac{c_{2}a_{2}}{a_{3}}0 & 0 & \mathrm{i}\frac{a_{2}}{a_{3}}\\
-\mathrm{i} & 0 & 0 & 0\\
0 & -\mathrm{i}\frac{a_{2}}{a_{3}} & 0 & 0
\end{array}\right],\quad G^{-1}=\left[\begin{array}{rrrr}
0 & 0 & \mathrm{i} & 0\\
0 & 0 & 0 & \mathrm{i}\frac{a_{3}}{a_{2}}\\
-\mathrm{i} & 0 & -c_{1} & 0\\
0 & -\mathrm{i}\frac{a_{3}}{a_{2}} & 0 & -\frac{c_{2}a_{3}}{a_{2}}
\end{array}\right],\quad\det\left\{ G\right\} =\frac{a_{2}^{2}}{a_{3}^{2}}.\label{eq:spHer1c}
\end{equation}
Indeed, it is an elementary exercise to verify that for each value
of $z$ matrix $A\left(z\right)$ is $G$-skew-Hermitian, that is
\begin{equation}
GA\left(z\right)+A^{*}\left(z\right)G=0.\label{eq:spHer1d}
\end{equation}

\section{Folded waveguide TWT dispersion relations}

Using a number of approximations the authors of \cite{GanArm} arrive
at the following expression of the dispersion relation similar to
that of the Pierce theory
\begin{equation}
\delta k^{\prime}\left(\Delta-\delta k^{\prime}\right)^{2}\left(1+\frac{\delta k^{\prime}}{2}\right)-C^{3}=0,\quad k^{\prime}=\frac{k}{k_{\mathrm{c}}},\label{eq:fowgd1a}
\end{equation}
where (i) $\omega_{\mathrm{c}}$ is the cut-off frequency of $\mathrm{TE}_{10}$-mode;
(ii) $k_{\mathrm{c}}=\frac{\omega_{\mathrm{c}}}{c}$, $c$ is the
velocity of light; (iii) $C$ are respectively $\Delta$ are the coupling
(Pierce) parameter and the detuning parameter represented by explicit
formulas involving folded waveguide TWT parameters and frequency $\omega$;
(iv) $\delta k$ is defined by
\begin{equation}
k_{m}=k_{m}^{0}+\delta k,\quad k_{m}^{0}=\frac{2\pi m}{a}+\frac{a+h}{a}\sqrt{\frac{\omega^{2}-\omega_{\mathrm{c}}^{2}}{c^{2}}},\label{eq:fowgd1b}
\end{equation}
where $k_{m}^{0}$ is unperturbed propagation constant. In the case
$C^{3}\rightarrow0$ there four solutions to equation (...): $\delta k=0$
(unperturbed forward propagating wave), $2k_{0}^{0}$ (unperturbed
contra propagating wave), and $k=k_{\mathrm{c}}\Delta$ (degenerate
e-beam mode). If the interaction with the contra propagating wave
is neglected, $\left|\delta k\right|\ll2\left|k_{0}^{0}\right|$,
then we obtain from (\ref{eq:fowgd1a}) the following third-order
dispersion equation
\begin{equation}
\delta k^{\prime}{}^{3}-2\Delta\delta k^{\prime}{}^{2}+\Delta^{2}\delta k^{\prime}-C^{3}=0.\label{eq:fowgd1c}
\end{equation}

\section{Capacitance}

According to \cite[I.1]{GreEM}, \cite[3.5.2]{Zahn} the following
formulas hold for capacitance of capacitors of different geometries
in Gaussian system of units.

Capacitance for the parallel-plate capacitor consisting of two parallel
plates of area $A$ that are separated by distance $d$ is
\begin{equation}
C=\frac{A}{4\pi d}.\label{eq:capaci1a}
\end{equation}

Capacitance for the spherical capacitor consisting of two concentric
spherical shells of radii $r_{1}<r_{2}$ is
\begin{equation}
C=\frac{r_{1}r_{2}}{r_{2}-r_{1}};\quad C\cong\frac{r_{1}^{2}}{2\left(r_{2}-r_{1}\right)},\text{ if }0<\frac{r_{2}-r_{1}}{r_{1}}\ll1.\label{eq:capaci1b}
\end{equation}

Capacitance of the cylindrical capacitor consisting of two coaxial
cylinders of radii $r_{1}<r_{2}$ and height $h$ is
\begin{equation}
C=\frac{h}{2\ln\left\{ \frac{r_{2}}{r_{1}}\right\} };\quad C\cong\frac{hr_{1}}{2\left(r_{2}-r_{1}\right)},\text{ if }0<\frac{r_{2}-r_{1}}{r_{1}}\ll1.\label{eq:capaci1c}
\end{equation}

Capacitances for a number of different geometric shapes are available
in \cite[II.3]{Landk} including capacitance of the disk of radius
$r$:
\begin{equation}
C=\frac{2r}{\pi^{2}}.\label{eq:capaci1d}
\end{equation}

Since often the data is available in $\mathrm{SI}$ system of units
rather than in Gaussian it is useful to know that the capacitances
in these two systems are related as follows, \cite[App. on units, 4]{Jack}
\begin{equation}
C_{\mathrm{SI}}=4\pi\varepsilon_{0}C_{\mathrm{Gaussian}},\quad\varepsilon_{0}=8.854187813\cdot10^{-12}\,\mathrm{F/m}.\label{eq:capaci1e}
\end{equation}

\textbf{\vspace{0.1cm}
}

\textbf{DATA AVAILABILITY:} The data that supports the findings of
this study are available within the article.

\end{document}